\documentclass[authoryear]{imsart}

\usepackage{amsthm,amsmath,amsfonts,amssymb}
\usepackage[authoryear]{natbib}
\usepackage[colorlinks,citecolor=blue,urlcolor=blue]{hyperref}
\usepackage{graphicx}
\usepackage{mathtools}

\startlocaldefs
\theoremstyle{plain}

\newtheorem{Lemma}{Lemma}

\newtheorem{Theorem}{Theorem}
\theoremstyle{remark}
\newtheorem{Assumption}{Assumption}

\newtheorem{Example}{Example}

\usepackage{newtxmath} %for indicator function
\usepackage{enumitem}
\usepackage{chapterbib}

\newcommand{\R}{\mathbb{R}}
\newcommand{\N}{\mathbb{N}}

\newcommand{\E}{\mathbb{E}}
\newcommand{\Var}{\operatorname{Var}}   
\newcommand{\Cov}{\operatorname{Cov}}

\DeclarePairedDelimiter{\floor}{\lfloor}{\rfloor}

\DeclarePairedDelimiter{\Bfloor}{\Big\lfloor}{\Big\rfloor}

\DeclarePairedDelimiter{\bfloor}{\big\lfloor}{\big\rfloor}
            
\newcommand{\EINS}{\vmathbb{1}}           

\newcommand{\Pj}{\mathbb{P}}

\newcommand\argmin{\operatornamewithlimits{argmin}}
\newcommand\argmax{\operatornamewithlimits{argmax}}

\newcommand{\inprob}{\stackrel{p}{\to}}
\endlocaldefs

\begin{document}
\begin{frontmatter}

\title{Cross-validation for change-point regression:\\pitfalls and solutions}
%\title{A sample article title with some additional note\thanksref{T1}}
\runtitle{Cross-validation for change-point regression}
\runauthor{F. Pein and R. D. Shah}
%\thankstext{T1}{A sample of additional note to the title.}

\begin{aug}
%%%%%%%%%%%%%%%%%%%%%%%%%%%%%%%%%%%%%%%%%%%%%%%
%% Only one address is permitted per author. %%
%% Only division, organization and e-mail is %%
%% included in the address.                  %%
%% Additional information can be included in %%
\author[A]{\fnms{Florian} \snm{Pein}\ead[label=e1]{f.pein@lancaster.ac.uk}\orcid{0000-0002-7994-9264}} \author[B]{\fnms{Rajen} \snm{D.\ Shah}\ead[label=e2]{r.shah@statslab.cam.ac.uk}\orcid{0000-0001-9073-3782}}
%%%%%%%%%%%%%%%%%%%%%%%%%%%%%%%%%%%%%%%%%%%%%%
%% Addresses                                %%
%%%%%%%%%%%%%%%%%%%%%%%%%%%%%%%%%%%%%%%%%%%%%%
\address[A]{Lancaster University, UK\printead[presep={,\ }]{e1}}

\address[B]{University of Cambridge, UK\printead[presep={,\ }]{e2}}
\end{aug}

\begin{abstract}
Cross-validation is the standard approach for tuning parameter selection in many non-parametric regression problems. However its use is less common in change-point regression, perhaps as its prediction error-based criterion may appear to permit small spurious changes and hence be less well-suited to estimation of the number and location of change-points. We show that in fact the problems of cross-validation with squared error loss are more severe and can lead to systematic under- or over-estimation of the number of change-points, and highly suboptimal estimation of the mean function in simple settings where changes are easily detectable. We propose two simple approaches to remedy these issues, the first involving the use of absolute error rather than squared error loss, and the second involving modifying the holdout sets used. For the latter, we provide conditions that permit consistent estimation of the number of change-points for a general change-point estimation procedure. We show these conditions are satisfied for least squares estimation using new results on its performance when supplied with the incorrect number of change-points. Numerical experiments show that our new approaches are competitive with common change-point methods using classical tuning parameter choices when error distributions are well-specified, but can substantially outperform these in misspecified models. An implementation of our methodology is available in the \texttt{R} package \texttt{crossvalidationCP} on CRAN.
\end{abstract}



\end{frontmatter}
%%%%%%%%%%%%%%%%%%%%%%%%%%%%%%%%%%%%%%%%%%%%%%
%% Please use \tableofcontents for articles %%
%% with 50 pages and more                   %%
%%%%%%%%%%%%%%%%%%%%%%%%%%%%%%%%%%%%%%%%%%%%%%
%\tableofcontents

%%%%%%%%%%%%%%%%%%%%%%%%%%%%%%%%%%%%%%%%%%%%%%
%%%% Main text entry area:
\begin{cbunit}
\section{Introduction}\label{sec:introduction}
Driven by a need to study datasets that exhibit abrupt changes in distribution, often across time, the field of change-point analysis has received a great deal of attention in recent years. Application areas where such data are common include biochemistry \citep{pein2021analysis}, finance \citep{bai2003computation, kim2005structural}, genomics \citep{olshen2004circular}, quality  monitoring \citep{d2011incipient} and speech processing \citep{harchaoui2009regularized}, to name a few. Perhaps the simplest model studied involves data $Y=(Y_1,\ldots,Y_n) \in \R^n$ satisfying
\begin{equation} \label{eq:uni_mod}
	Y_i = \mu_i + \varepsilon_i, \qquad i=1,\ldots,n,
\end{equation}
where the $\varepsilon_i$ are independent mean-zero errors, and $\mu := (\mu_1,\ldots,\mu_n)$ is piecewise constant with change-points $\tau_1 < \cdots <\tau_K$; that is $\mu_i \neq \mu_{i+1}$ if and only if $i = \tau_k$ for some $k$, for $i=1,\ldots,n-1$.

A variety of methods exists for estimating $\mu$ and the unknown change-points, for instance binary segmentation \citep{vostrikova1981detecting} and its variants \citep{olshen2004circular, fryzlewicz2014wild, fryzlewicz2020detecting, kovacs2020seeded}, (penalized) cost methods such as Segment Neighbourhood \citep{auger1989algorithms, jackson2005algorithm, zhang2007modified, killick2012optimal, maidstone2017optimal, verzelen2020optimal}, multiscale methods \citep{frick2014multiscale, li2016fdr} and Bayesian approaches \citep{fearnhead2006exact, du2016stepwise}, among many others; for further detail see \citet{niu2016multiple, truong2020selective, fearnhead2020relating} and references therein. The empirical and theoretical properties of these methods typically rely on selecting appropriate choices for tuning parameters. For instance, Segment Neighbourhood \citep{auger1989algorithms} requires selection of the number of change-points to estimate, and then determines the location of these to minimise the residual sum of squares. Given the correct number of change-points, the least squares estimate of the change-point locations is minimax rate-optimal and its $L_2$-error rate is optimal up to log-factors (see Theorem~\ref{theorem:detectionPrecision} in Section~\ref{sec:optimalPartitioning} and the following discussion as well as \eqref{eq:L2riskModifiedCriterion}). Other approaches require different tuning parameters, and for some of these methods, theoretically motivated choices of those tuning parameters can be very successful in idealised settings where the joint distribution of the $\varepsilon_i$ is known. However, their performance can deteriorate when the error distribution is misspecified, as is likely to be the case in practice, particularly when the errors have heavy tails or are heteroscedastic.

The problem of selecting regression procedures indexed by tuning parameters is of course encountered in more general regression settings, and here cross-validation is typically the method of choice. One of the appeals of cross-validation is its inherently model-free nature which confers a certain robustness. It has been shown to be very successful empirically in non-parametric and high-dimensional regression settings and theoretical guarantees have provided additional support for its usage \citep{wong1983consistency, yang2007consistency, arlot2010survey, yu2014modified, chetverikov2021cross}.

The use of cross-validation is however less common in change-point regression, and the only theoretical contributions we are aware of are \citet{arlot2011segmentation} and \citet{zou2020consistent}. The former provides results on the quality of a cross-validation estimation of the prediction error for a least squares estimate based on a given fixed set of putative change-points. These results however do not directly tackle the problem of whether cross-validation can provide a consistent estimate of the true number of changes; this latter problem is studied in \citet{zou2020consistent} which we discuss further in the following.

A misgiving one may have about cross-validation is that it typically targets procedures with good prediction properties, and a method that introduces many spurious changes with small jump sizes may not suffer too much from this perspective. This may be concerning given that the goal in change-point regression is often accurate estimation of the number and locations of the change-points.

In this work however, we show that the shortcomings of standard cross-validation with squared error loss are more serious, and can lead to both under- and over-selection of the number of change-points, and perhaps surprisingly, poor performance in terms of mean squared error. The central issue is that if one of the holdout sets includes a point immediately following or preceding a large change-point, the squared error incurred when predicting at that data point can dominate the cross-validation error criterion. We detail this in Section~\ref{sec:L2loss} where we provide some formal negative results for the use of cross-validation with least squares estimation. In Sections~\ref{sec:LargeExample},~\ref{sec:simulationUnderestimation}~and~\ref{sec:simulationOverestimation} we demonstrate empirically that this can substantially affect the performance of cross-validation in practice, and moving a single change-point from an even to an odd location for example, can lead to a drastic deterioration in performance. This issue extends also to the cross-validation procedures of \citet{arlot2011segmentation}, which are also based on squared error loss.

One reason this problem has not (to our knowledge) been highlighted in prior literature, is that existing theoretical results on cross-validation consider asymptotic regimes which may either implicitly or explicitly assume that a change-point procedure trained on a subset of the data will make bounded expected squared errors on the remaining data as is the case in \citet[Prop.~1]{arlot2011segmentation}. However by bypassing the issue of poor predictive performance described above, the insights of such asymptotic regimes for finite samples are thus perhaps somewhat limited.

Our second contribution is to further advance the point made in \citet{zou2020consistent} that the basic intuition that cross-validation encourages too many small spurious changes, is not necessarily well-founded. In Sections~\ref{sec:cvL1}~and~\ref{sec:cvL2} we propose two simple approaches to avoid the problems associated with large changes. The first involves using absolute error rather than squared error in the cross-validation criterion. The second involves modifying the cross-validation score for squared error loss to avoid the problematic points. For the latter, we provide relatively mild conditions on the underlying change-point regression procedure under which the cross-validation approach is consistent for selecting the number of change-points; no additional penalisation (or need for choosing appropriate tuning parameters that would come with it) is required to achieve this consistency.

Our theory builds on the work of \citet{zou2020consistent}, who show consistency of a cross-validation scheme in their Theorems~1~and~2. However inspection of their proofs shows that their conclusions as stated may need some caveats. Firstly, though not explicitly stated, the proofs require all change-points to occur at even locations. This simplification of the model may be justified when the noise and signal strength are bounded away from $0$ and $\infty$ respectively such that the expected errors of the fitted regression function are bounded; however, as explained above, this asymptotic regime may then have less relevance to practice. Secondly, it is unclear to us how the arguments in their proofs may be extended to allow for the maximum number of change-points considered and the true number of change-points to diverge; see the discussion before and after \eqref{eq:whatZouetalHaveShown} in Section~\ref{sec:proof:RescaledCV} in the supplementary material of our paper. This seems particularly relevant given that the number of change-points is the object of inference. In contrast, our result (Theorem~\ref{theorem:positiveResultRescaledCV}) allows for the number of change-points to tend to infinity. We verify the conditions of our general result for the case of least squares estimation by employing a new result (Theorem~\ref{theorem:detectionPrecision} in Section~\ref{sec:optimalPartitioning}) on the existence of estimated change-points in the neighbourhood of true changes even when the number of changes has been incorrectly specified, which may be of independent interest.

In Section~\ref{sec:simulations}, we present numerical experiments that illustrate the performance of our new cross-validation schemes in comparison with commonly used change-point procedures using classical tuning parameter choices. We see that cross-validation with absolute error loss is competitive when the error distribution has been well-specified, but substantially outperforms classical methods in settings with heteroscedastic or heavy-tailed errors, or when outliers are present. We conclude with a discussion in Section~\ref{sec:discussion}. The Appendix contains descriptions of generalisations of our methodology; additional numerical experiments as well as all proofs are contained in the supplementary material.

\section{Pitfalls of using cross-validation with squared error loss}\label{sec:L2loss}
In this section, we give examples of simple settings where changes are easily detectable but where using cross-validation with squared error loss can lead to both systematic under- and over-estimation of the number of change-points. For these negative results, we focus on the univariate mean change-point regression problem \eqref{eq:uni_mod}, where additionally $\varepsilon_i \sim \mathcal{N}(0, \sigma^2)$. Furthermore, we consider a version of two-fold cross-validation, termed the $\operatorname{COPSS}$ procedure in \citet{zou2020consistent}, with least squares estimation for estimating the change-point locations\footnote{Throughout the paper we will use the terminology least squares estimation. \citet{zou2020consistent} called it optimal partitioning, highlighting the fact that least squares estimation is performed for various putative numbers of change-points. Elsewhere in the literature however, this is known as Segment Neighbourhood \citep{auger1989algorithms}, with optimal partitioning referring to a penalised version. We finally remark that least squares estimation coincides with maximum likelihood estimation if the noise is Gaussian.}.
However it will be clear that our constructions highlighting the undesirable properties of cross-validation, and our conclusions, can be generalised to other settings and other forms of cross-validation employing squared error loss.

\subsection{Setting}\label{sec:gaussianSetting}
In order to describe and study the $\operatorname{COPSS}$ procedure, we introduce some notation. Let $\tau_0=0$ and $\tau_{K+1}=n$. Let $\beta_k$ for $k=1,\ldots,K$ be the mean of the signal in the $k$th constant segment, so
\begin{equation}\label{eq:signalGaussian}
	\mu_i = \beta_k,\ \text{ if and only if } \tau_k < i \leq \tau_{k+1},\ k = 0,\ldots,K;\ i = 1,\ldots,n.
\end{equation}
For an arbitrary vector of observations $Z:=(Z_1,\ldots,Z_m)$ and a putative number of change-points $L$, least squares estimation obtains the change-points as 
\begin{equation}\label{eq:optimalPartitioning}
	\left(\hat{\tau}^Z_{L,0},\ldots, \hat{\tau}^Z_{L,L + 1}\right) := \argmin_{0 =: t_0 < t_1 < \cdots < t_L < t_{L+1} := n}{\sum_{l = 0}^{L}{\sum_{i = t_l + 1}^{t_{l+1}}{\left(Z_i - \overline{Z}_{t_l :t_{l + 1}}\right)^2}}},
\end{equation}
where $\overline{Z}_{a:b} := (b - a)^{-1}\sum_{i = a+1}^{b}{Z_i}$. In Section~\ref{sec:optimalPartitioning} we will give theoretical guarantees for the estimates $\hat{\tau}^Z_{L,l}$. We will assume for simplicity here and throughout that $n$ is even; if not, the final observation may be dropped.

In order to perform $\operatorname{COPSS}$ with least squares estimation, we will apply the above to the odd and even indexed observations separately. (We will assume for simplicity here and throughout that $n$ is even; if not, the final observation may be dropped.) To study this, we introduce
\begin{equation}\label{eq:2foldsSplitting}
	\begin{aligned}
		Y^O_i & := Y_{2 i - 1}, & \mu^O_i & := \mu_{2 i - 1}, & \varepsilon^O_i & := \varepsilon_{2 i - 1}, & i = &1,\ldots, n / 2 ,\\
		Y^E_i & := Y_{2 i},     & \mu^E_i & := \mu_{2 i },    & \varepsilon^E_i & := \varepsilon_{2 i },    & i = & 1,\ldots,n / 2 .
	\end{aligned}
\end{equation}
We write $\mathcal{T}^O := \big\{\tau^O_0, \ldots,\tau^O_{K+1}\big\}$ for the set of true change-points among the odd observations $Y^O :=\big(Y^O_1,\ldots,Y^O_{ n/2 }\big)$, so
\begin{equation*}
	\mu^O_i = \beta_k,\ \text{ if and only if } \tau^O_k < i \leq \tau^O_{k+1},\ k = 0,\ldots,K;\ i = 1,\ldots, n/2 .
\end{equation*}
We denote the estimated change-point set obtained by applying least squares estimation \eqref{eq:optimalPartitioning} to $Y^O$ by
\begin{equation}\label{eq:estimatedCpO}
	\hat{\mathcal{T}}^O_L := \big\{0 = \hat{\tau}^O_{L,0} < \hat{\tau}^O_{L,1} < \cdots < \hat{\tau}^O_{L,L} < \hat{\tau}^O_{L, L + 1} =  n / 2 \big\},
\end{equation}
and define $\mathcal{T}^E$ and $\hat{\mathcal{T}}^E_L$ analogously for the even observations. (Here and throughout, we assume that the original change-points $\tau_0, \tau_1, \ldots, \tau_{K+1}$ are at least two apart, so $\min_{k=0,\ldots,K} \tau_{k+1} - \tau_k \geq 2$; this ensures that each of $\mathcal{T}^E$ and $\mathcal{T}^0$ have $K+2$ elements.)

With these, we may now define the (two-fold) cross-validation criterion for $L$ change-points using squared error loss as
\begin{equation}\label{eq:cvL2}
		\operatorname{CV}_{(2)}(L) := 
		\sum_{l = 0}^{L}\sum_{i = \hat{\tau}^O_{L,l} + 1}^{\hat{\tau}^O_{L,l + 1}}{ \left(Y^E_i - \overline{Y}^O_{\hat{\tau}^O_{L,l} :\hat{\tau}^O_{L,l + 1}}\right)^2 }\\
		+ \sum_{l = 0}^{L}\sum_{i = \hat{\tau}^E_{L,l} + 1}^{\hat{\tau}^E_{L,l + 1}}{ \left(Y^O_i - \overline{Y}^E_{\hat{\tau}^E_{L,l} :\hat{\tau}^E_{L,l + 1}}\right)^2 },
\end{equation}
where the subscript $(2)$ in $\operatorname{CV}_{(2)}(L)$ emphasises the use of squared error loss. Lastly, the estimated number of change-points $K$ is given by
\begin{equation}\label{eq:estimatedKgaussian}
	\hat{K} := \argmin_{L = 0,\ldots,K_{\max}}{\operatorname{CV}_{(2)}(L)},
\end{equation}
where $K_{\max} \geq 1$ is a pre-specified upper bound for the number of change-points. Given $\hat{K}$, we may perform least squares estimation with $L=\hat{K}$ on the full vector of observations $Y$ to produce a final set of estimated change-points $0 =: \hat{\tau}_{0} < \hat{\tau}_{1} < \cdots < \hat{\tau}_{\hat{K}} < \hat{\tau}_{\hat{K} + 1} := n$ and an estimate $\hat{f}:\ [0,1] \to \R,\ t \mapsto  \sum_{k = 0}^{\hat{K}}{\overline{Y}_{\hat{\tau}_{k} :\hat{\tau}_{k + 1}} \EINS_{(\hat{\tau}_{k} / n, \hat{\tau}_{k + 1} / n]}(t)}$ of the true mean function
$f:\ [0,1] \to \R,\ t \mapsto  \sum_{k = 0}^{K}{\beta_k \EINS_{(\tau_{k} / n, \tau_{k + 1} / n]}(t)}$.

The motivation for this approach is that in \eqref{eq:cvL2}, $Y^E_i$ is (almost) an unbiased proxy for $\mu^O_i$, and similarly for the odd and even designators interchanged. The ``almost'' qualification is due to the fact that the unbiasedness may fail to hold immediately after a change, so for example $\mu^O_i$ and $\E Y^E_i = \mu^E_i$ can be very different when $i=\tau^O_k$. Whilst this will only occur at isolated points, the fact that the errors are squared in \eqref{eq:cvL2} can lead these discrepancies to dominate the cross-validation criterion when changes are large, and this has severe consequences for the quality of the estimate $\hat{K}$ as we show formally in the following.

\subsection{Underestimation}\label{sec:underestimation}

In this section we present a scenario in which cross-validation will under-estimate the number of change-points with high probability.

\begin{Example}[Underestimation]\label{example:underestmation}
Let $Y \in \R^n$ be of the form in \eqref{eq:uni_mod} with mean vector $\mu$ as in \eqref{eq:signalGaussian}. Suppose $n$ is even, but $n/2$ odd. Let $K=2$ and for odd $\underline{\lambda} < n / 4$, let $\tau_1 = n / 2 - \underline{\lambda},\, \tau_2 = n / 2$. We set $\Delta_1 = \beta_1$, $\Delta_2 = \beta_3$ and $\beta_2 = 0$ and suppose $\Delta_1 < \Delta_2$, so $\Delta_2$ and $\Delta_1$ are the sizes of the largest and smallest jumps respectively. An illustration of this construction is given in Figure \ref{fig:example}. 
\end{Example}

\begin{figure}[!htb]
\includegraphics[width = 0.99\textwidth]{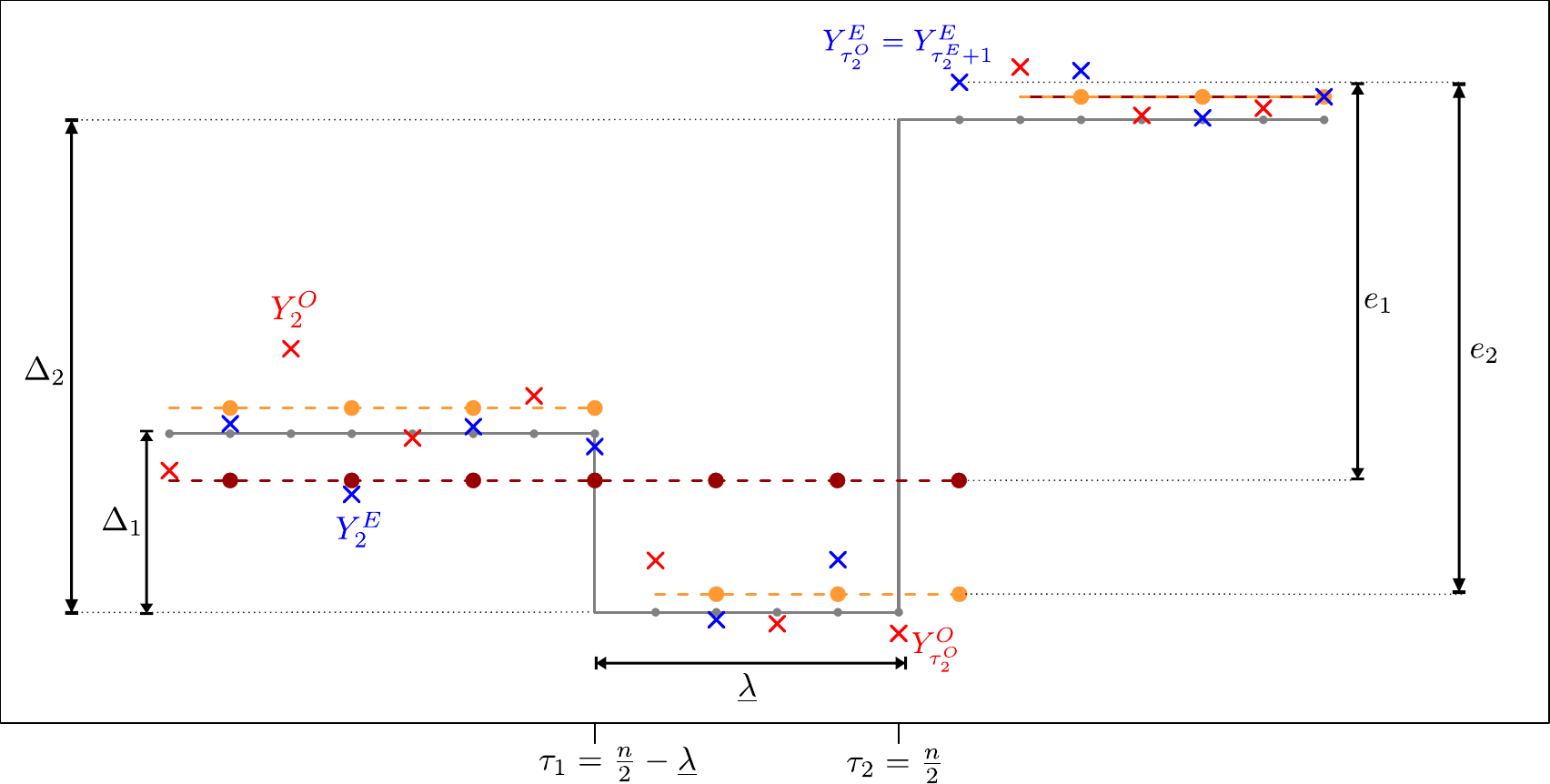} 
\caption{Schematic of Example~\ref{example:underestmation}. Expectations are visualized by grey dots and observations by coloured crosses, split into the two folds given by $Y^O$ (red) and $Y^E$ (blue). The predictions $\big(\overline{Y}^O_{\hat{\tau}^O_{L,l} :\hat{\tau}^O_{L,l + 1}}\big)_{l=0}^L$ from the odd observations, are shown for $L = 2$ (orange dots) and $L = 1$ (brown dots). Distances $e_2$ and $e_1$ between $Y^E_{\tau^O_2}$ and the corresponding predictions are much larger than those corresponding to other observations, and can dominate the cross-validation criterion. The larger size of $e_2$ results in $\operatorname{CV}_{(2)}$ being minimised at $L=1$.}
\label{fig:example}
\end{figure} 
 
In Example~\ref{example:underestmation} above, we have that as $n/2$ is odd, $\tau^O_2 = (n/2 + 1)/2$, but $\tau^E_2 = (n/2 - 1)/2$, and so $\mu^E_{\tau^O_2} = \Delta_2 \neq \mu^O_{\tau^O_2} = 0$. Thus, if $\Delta_2$ is large, the point following the large second change, $Y^E_{\tau^E_2 + 1} = Y^E_{\tau^O_2}$, contributes heavily to 
the cross-validation criterion. To see why this is problematic, it is instructive to first consider a noiseless setting where $\sigma=0$. Then with a correctly specified $L=K=2$, we obtain
\begin{equation*}
	\begin{split}
		\operatorname{CV}_{(2)}(2) =& \left(\mu^O_{\tau^E_2 + 1} - \beta_2\right)^2 + \left(\mu^E_{\tau^O_2} - \beta_1\right)^2
		= \left(\mu^O_{\tau^O_2} - \beta_2\right)^2 + \left(\mu^E_{\tau^E_2 + 1} - \beta_1\right)^2 \\
		= &\left(\beta_1 - \beta_2\right)^2 + \left(\beta_2 - \beta_1\right)^2  
		= \Delta_2^2 + \Delta_2^2.
	\end{split}
\end{equation*}
On the other hand, when $L = 1$, least squares estimation recovers only the large second change-point. Hence, the first segment consists of $n / 2 - \underline{\lambda}$ observations with value $\beta_0 = \Delta_1$ and $\underline{\lambda}$ observations with value $\beta_1 = 0$. Hence, $\overline{\mu}^O_{0:\hat{\tau}^O_{1,1}} = \overline{\mu}^E_{0:\hat{\tau}^E_{1,1}} \approx \frac{n - 2\underline{\lambda}}{n} \Delta_1$. Thus,
\begin{align*}
	\operatorname{CV}_{(2)}(1)	&\approx \left(\frac{n}{2} - \underline{\lambda}\right) \left(\frac{2\underline{\lambda}}{n}\right)^2\Delta_1^2 + \underline{\lambda} \left(\frac{n - 2\underline{\lambda}}{n}\right)^2\Delta_1^2
	+ \left(\mu^O_{\tau^E_2 + 1} - \beta_2\right)^2 + \left(\mu^E_{\tau^O_2} - \frac{n - 2\underline{\lambda}}{n} \Delta_1 \right)^2\\
	&=  \left(1-\frac{2\underline{\lambda}}{n}\right) \underline{\lambda} \Delta_1^2
	+ \Delta_2^2 + \left(\Delta_2 - \frac{n - 2\underline{\lambda}}{n} \Delta_1 \right)^2.
\end{align*}
We then obtain
\begin{equation} \label{eq:noiseless}
	\operatorname{CV}_{(2)}(2) - \operatorname{CV}_{(2)}(1) \approx \Delta_1^2 \left( 2\frac{\Delta_2}{\Delta_1} - \underline{\lambda} \right)\left(1-\frac{2\underline{\lambda}}{n}\right) > \frac{1}{2}\Delta_1^2 \left( 2\frac{\Delta_2}{\Delta_1} - \underline{\lambda} \right).
\end{equation}
Thus when $2\Delta_2 / ( \underline{\lambda}\Delta_1) > 1$ we can expect cross-validation to favour $L=1$ change-point.

More precisely, and taking account of the presence of noise, we have the following asymptotic result. Note that here and in the sequel, all parameters (in the current case $\Delta_1$, $\Delta_2$, $\underline{\lambda}$ and $\sigma$) are permitted to change with $n$, though we suppress this in the notation. Also, although our result is stated for simplicity for the case of $\operatorname{COPSS}$, a similar conclusion would hold for other squared error cross-validation schemes, possibly with the large change at a different location depending on the details of the method employed such as the number of folds used.

One example is the $\operatorname{LooVF}$ procedure from \citet{arlot2011segmentation}, where the issues  described above occur when one considers Example~\ref{example:underestmation} in reverse order. This is because they extrapolate to the left when predicting at unseen time points. In general, a choice has to be made to extrapolate in either direction, and depending on the underlying signal, each strategy can result in poor performance.
\begin{Theorem}
	\label{theorem:underestimation}
	Let $Y \in \R^n$ be as in Example~\ref{example:underestmation}. Suppose that the following hold:
	\begin{align}
		\frac{\underline{\lambda}\Delta_1^2}{\sigma^2 \log(n)} &\to \infty, \label{eq:minmax} \\
		\liminf_{n \to \infty}{\frac{2}{\underline{\lambda}}\frac{\Delta_2}{\Delta_1}} &> 1. \label{eq:conditionUnderestimation}
	\end{align}
Then $\hat{K}$ \eqref{eq:estimatedKgaussian} satisfies  $\Pj(\hat{K} = 2) \to 0$.
Moreover, if additionally $K_{\max} = K = 2$,
\begin{equation}\label{eq:L2riskSquaredErrorLoss}
\left[\int_0^1{ \left(\hat{f}(t) - f(t)\right)^2 dt}\right]^{-1} = \mathcal{O}_\Pj\left( \frac{n}{\underline{\lambda} \Delta_1^2} \right).
\end{equation}
\end{Theorem}
Condition \eqref{eq:minmax} is the minimum requirement that ensures that both change-points are detectable; see the discussion after Theorem~\ref{thm:consistency}. Cross-validation however is unlikely to select the correct number of change-points provided the size $\Delta_2$ of the largest jump is large compared to the product of the minimum gap $\underline{\lambda}$ between the change-points and the smallest jump size $\Delta_1$ \eqref{eq:conditionUnderestimation}.

If overestimation is not permitted, i.e. $K_{\max} = K$, \eqref{eq:L2riskSquaredErrorLoss} shows that the $L_2$-loss of the estimated function is at least of order $\underline{\lambda} \Delta_1^2 / n$. We may contrast this with the corresponding $L_2$-loss realised by our proposed criterion, which for general signals is $\mathcal{O}_{\Pj}\big(n^{-1} \sigma^2 \log\log \overline{\lambda}\big)$, see \eqref{eq:L2riskModifiedCriterion}. It follows from \eqref{eq:minmax} that in the setting of Theorem~\ref{theorem:underestimation}, the quotient of these losses converges to zero, underlining the suboptimality of cross-validation based on squared error loss even with respect to $L_2$-loss.

\subsection{Overestimation}\label{sec:overestimation}
We now introduce an example where cross-validation with least squares loss has a tendency to overestimate the number of change-points.

\begin{Example}[Overestimation]\label{example:overestmation}
As in Example~\ref{example:underestmation} let $Y \in \R^n$ be of the form in \eqref{eq:uni_mod} with $n$ even but $n/2$ odd. Let $K=1$, $\tau_1 = n / 2$, $\beta_0 = 0$ and $\beta_1 = \Delta_1$.
\end{Example}

In this example, $\tau^O_1 = (n/2 + 1)/2$ and $\mu^E_{\tau^O_1} - \mu^O_{\tau^O_1} =  \Delta_1$. Thus, if $\Delta_1$ is large, the cross-validation criterion will be heavily influenced by $Y^E_{\tau^O_1}$. More precisely, $\tau^O_1$ is estimated exactly with high probability if $\Delta_1$ is large. Thus on this event, the criterion contains the term
\begin{equation} \label{eq:over_est_term}
\left( Y^E_{\tau^O_1} - \overline{Y}^O_{\hat{\tau} :\tau^O_1}\right)^2 = \left( \varepsilon^E_{\tau^O_1} + \Delta_1 - \overline{\varepsilon}^O_{\hat{\tau} :\tau^O_1}\right)^2,
\end{equation}
where $\hat{\tau}$ denotes the last estimated change-point before $\tau^O_1$. If $L = K = 1$, $\hat{\tau} = 0$. However, if $L > K$, $\Pj\big(\hat{\tau} > 0\big)\geq \frac{1}{2}$ and $\Pj\big(\overline{\varepsilon}^O_{\hat{\tau} :\tau^O_1} > \overline{\varepsilon}^O_{0:\tau^O_1}\big) \geq \frac{1}{2}$. Thus, if $\Delta_1$ is large, the cross-term $\Delta_1\overline{\varepsilon}^O_{\hat{\tau} :\tau^O_1}$ in \eqref{eq:over_est_term} can outweigh the costs incurred by the additional change-points and cause overestimation. This is formalised in Theorem~\ref{theorem:overestimation}. Note that for this effect it is only important that $\tau_1$ is at an odd location; other choices of the setup are just made to simplify the analysis.

\begin{Theorem}
	\label{theorem:overestimation}
Let $Y$ be as in Example \ref{example:overestmation} and $\hat{K}$ be as in \eqref{eq:estimatedKgaussian} with $K_{\max} > 1$. Then, if
\begin{equation}\label{eq:conditionOverestimation}
	\frac{\Delta_1}{\sigma \sqrt{n\log\log{n}}} \to \infty,\quad \text{as } n\to \infty,
\end{equation}
we have
\begin{equation}\label{eq:statementOverestimation}
\liminf_{n \to \infty} \Pj\left( \hat{K} > 1 \right) > 0.
\end{equation}
\end{Theorem}
The simulations in Section~\ref{sec:simulationOverestimation} in the supplementary material suggest that the probability in \eqref{eq:statementOverestimation} can be roughly $2/3$ when the signal-to-noise ratio is large. In these settings we also see that the probability drops to less than $9\%$ when the change-point is moved by a single design point to an even location, illustrating that the single point is solely responsible for the unfavourable behaviour. Similarly to Example~\ref{example:underestmation}, for this setup there exist procedures that correctly estimate the number of change-points to be $1$ with probability converging to $1$. In particular, our modified squared error loss criterion, which we describe in the next section, would achieve this; see Theorem~\ref{thm:consistency}.

\section{New cross-validation criteria}\label{sec:methodology}
We now introduce two new cross-validation criteria which circumvent the issues of squared error loss presented in the previous section. In Section~\ref{sec:cvL1} we  motivate the use of absolute errors and in Section~\ref{sec:cvL2} we propose a  modified criterion with squared errors. For the latter we prove theoretical guarantees on consistent estimation of the number of change-points in a more general multivariate model introduced in Section~\ref{sec:model}.

\subsection{Cross-validation with absolute error loss}\label{sec:cvL1}
We have seen how, when standard squared error loss is used, a single point following a large change-point can dominate the cross-validation criterion. 
For instance in Example~\ref{example:underestmation}, $Y^E_{\tau^E_2 + 1} = Y^E_{\tau^O_2}$ contributes $e_L^2$ to $\operatorname{CV}_{(2)}(L)$ for $L=1,2$, where
\[
e_1^2 \approx \bigg(\varepsilon^E_{\tau^O_2} + \Delta_2 - \frac{n - 2\underline{\lambda}}{n}\Delta_1\bigg)^2 \qquad \text{and} \qquad e_2^2 \approx \Big(\varepsilon^E_{\tau^O_2} + \Delta_2\Big)^2;
\]
see Figure~\ref{fig:example}. If $\Delta_2$ is large, the difference $e_2^2-e_1^2$ will be large, resulting in cross-validation erroneously favouring $\hat{K}=1$ change-points.

Consider now cross-validation with absolute error loss:
\begin{equation}\label{eq:cvL1}
		\operatorname{CV}_{(1)}(L) := 
		\sum_{l = 0}^{L}\sum_{i = \hat{\tau}^O_{L,l} + 1}^{\hat{\tau}^O_{L,l + 1}}{ \left\vert Y^E_i - \overline{Y}^O_{\hat{\tau}^O_{L,l} :\hat{\tau}^O_{L,l + 1}}\right\vert }
		+ \sum_{l = 0}^{L}\sum_{i = \hat{\tau}^E_{L,l} + 1}^{\hat{\tau}^E_{L,l + 1}}{ \left\vert Y^O_i - \overline{Y}^E_{\hat{\tau}^E_{L,l} :\hat{\tau}^E_{L,l + 1}}\right\vert }.
\end{equation}
Then the difference of the contribution of $Y^E_{\tau^O_2}$ to the criterion will be $|e_2| - |e_1|$, which when $\Delta_2$ is large will, with high probability, simply be
\[
e_2 - e_1 \approx \frac{n - 2\underline{\lambda}}{n}\Delta_1.
\]
The quantity on the right-hand side importantly does not feature $\Delta_2$, and is of the same order as the additional loss incurred at other points. Specifically, considering a noiseless version of Example~\ref{example:underestmation} as in \eqref{eq:noiseless}, we have
\begin{align*}
	\operatorname{CV}_{(1)}(1) - \operatorname{CV}_{(1)}(2) &\approx \frac{2\underline{\lambda}\Delta_1}{n}\left(\frac{n}{2} - \underline{\lambda}\right) +\underline{\lambda}\Delta_1 \frac{n-2\underline{\lambda}}{n} +\Delta_2 + \left(\Delta_2 - \Delta_1\frac{n-2\underline{\lambda}}{n}\right) - 2\Delta_2 \\
	&= \Delta_1\left((\underline{\lambda}-1)\frac{n-2\underline{\lambda}}{n} + \left(\frac{n}{2} - \underline{\lambda}\right)\frac{2\underline{\lambda}}{n}\right) > 0.
\end{align*}
We show in Section~\ref{sec:simulations} through numerical experiments that a 5-fold version of $\operatorname{CV}_{(1)}$ \eqref{eq:cvL1} works well in general, and can substantially outperform other methods in the realistic scenario where models have been misspecified. We therefore recommend this as a reasonable default option for cross-validation in change-point regression. In Appendix~\ref{sec:generalisedProcedure} we extend \eqref{eq:cvL1} to multivariate settings and general $V$-fold cross-validation.

\subsection{Modified cross-validation with squared error loss}\label{sec:cvL2}
A second approach to addressing the issues of standard cross-validation with squared error loss is to remove those observations that have the potential to dominate the criterion as follows:
\begin{equation}\label{eq:cvModifiedL2}
\begin{split}
\operatorname{CV}_{\mathrm{mod}}(L) := 
&\sum_{l = 0}^{L}\frac{\hat{\tau}^O_{L,l + 1} - \hat{\tau}^O_{L,l}}{\hat{\tau}^O_{L,l + 1} - \hat{\tau}^O_{L,l} - 1}\sum_{i = \hat{\tau}^O_{L,l} + 1}^{\hat{\tau}^O_{L,l + 1} - 1}{ \left(Y^E_i - \overline{Y}^O_{\hat{\tau}^O_{L,l} :\hat{\tau}^O_{L,l + 1}}\right)^2}\\
 +&\sum_{l = 0}^{L}\frac{\hat{\tau}^E_{L,l + 1} - \hat{\tau}^E_{L,l}}{\hat{\tau}^E_{L,l + 1} - \hat{\tau}^E_{L,l} - 1} \sum_{i = \hat{\tau}^E_{L,l} + 2}^{\hat{\tau}^E_{L,l + 1}}{ \left(Y^O_i - \overline{Y}^E_{\hat{\tau}^E_{L,l} :\hat{\tau}^E_{L,l + 1}}\right)^2}.
\end{split}
\end{equation}
Observe that compared to \eqref{eq:cvL2}, the squared error terms involving $Y^E_{\hat{\tau}^O_{L,l + 1}}$ and $Y^O_{\hat{\tau}^E_{L,l} + 1}$, $l = 1,\ldots,L$ have been removed. Out-of-sample prediction at these points necessarily involves extrapolation and thus large change-points immediately following or preceding these can be problematic. To compensate for the removal of these observations, we rescale the sums of squared error terms so that each estimated constant segment gives a contribution proportional to its length in the case where all change-points have been estimated correctly. We thus require that $\hat{\tau}^O_{L,l + 1} - \hat{\tau}^O_{L,l} \geq 2$ and  $\hat{\tau}^E_{L,l + 1} - \hat{\tau}^E_{L,l} \geq 2\ \forall\ l = 0,\ldots,L$.

In Theorem~\ref{thm:consistency} below, we show that this criterion does indeed deliver consistent estimation of the number of change-points and our numerical results in Section~\ref{sec:simulations} confirm its good performance in practice. We consider a generalisation of the setting of Section~\ref{sec:gaussianSetting} where we relax the Gaussian assumption on the errors and instead require the $\varepsilon_i$ to be independent (possibly with different distributions) and sub-Gaussian with variance proxy $\sigma^2$, so $\max_i \E[\exp(s\varepsilon_i)]\leq \exp(s^2\sigma^2/2)$ for all $s \in \R$. 
We assume additionally that all observations on the same segment have the same variance, i.e.\ $\Var[\varepsilon_i] = \Var[\varepsilon_j]$ if there exists a $k$ such that $\tau_k < i < j \leq \tau_{k+1}$, and writing $\underline{\sigma}^2 := \min_{i=1,\ldots,n}\Var[\varepsilon_i]$ we have $\limsup_{n \to \infty}\sigma / \underline{\sigma} < \infty$.

In order to present our result, we introduce some further notation. We denote by $\Delta_k := \vert \beta_k - \beta_{k - 1}\vert$ the size of the $k$th change-point and by $\Delta_{(k)}$ the $k$th smallest order statistic of $\Delta_1,\ldots,\Delta_K$, i.e.\ $\{\Delta_1, \ldots, \Delta_K\} = \{\Delta_{(1)}, \ldots, \Delta_{(K)}\}$ and $\Delta_{(1)} \leq \cdots \leq \Delta_{(K)}$. Further, we denote by $\underline{\lambda}$ and $\overline{\lambda}$ the minimal and maximal distances respectively between two change-points, i.e.\
\begin{equation*}
	\underline{\lambda} := \min_{k = 0,\ldots,K}{ \tau_{k + 1} - \tau_{k}} \text{ and } \overline{\lambda} := \max_{k = 0,\ldots,K}{ \tau_{k + 1} - \tau_{k}}.
\end{equation*}

We require the following bounds on the speed at which the number of change-points $K$ and its upper bound $K_{\max}$ are permitted to increase.
\begin{Assumption}[Number of Change-points]\label{assumption:cpNumberMultivariate}
\begin{enumerate}[label=(\roman*)]
\item \label{assumption:cpNumberMultivariateKmax}
$K_{\max} \geq K$ for all $n$ sufficiently large, and $K$ is non-decreasing.
\item \label{assumption:cpNumberBound}
$K = o(\underline{\lambda})$ and $K (\log(K\vee 1))^2 = o(\log\log\overline{\lambda})$.
\item \label{assumption:cpNumberKmaxBound} $(K_{\max} \log K_{\max})^{1/2} = o(\log\log\overline{\lambda})$.
\end{enumerate}
\end{Assumption}
Condition \ref{assumption:cpNumberMultivariateKmax} ensures that the true number of change-points is considered by the method. This is not too difficult to satisfy as comparing \ref{assumption:cpNumberBound} and \ref{assumption:cpNumberKmaxBound} shows that $K_{\max}$ can increase more than quadratically in $K$. Condition \ref{assumption:cpNumberBound} however restricts the growth of $K$ quite substantially: we have $\overline{\lambda} \geq n / K$  and hence $K (\log(K\vee 1))^2 = o(\log\log n)$. This condition is sufficient to ensure that the cost of adding a false positive, which is of order $\sigma^2 \log\log\overline{\lambda}$, dominates the cost of misestimating change-point locations when $L=K$. On the other hand, \ref{assumption:cpNumberKmaxBound} ensures that the former also dominates the variance of the cross-validation error criterion for all $K \leq L \leq K_{\max}$. Relative to other results on (non-cross-validation based) change-point approaches, the condition on the growth of $K$ is quite restrictive \citep{frick2014multiscale,Garreau2018Consistent,verzelen2020optimal}.
	We remark however that \ref{assumption:cpNumberBound} and \ref{assumption:cpNumberKmaxBound} are likely artefacts of our proof strategy, which considers worst case scenarios for each estimated change-point location, rather than a fundamental limitation of cross-validation. In practice, $K_{\max}$ does not have to be fixed in advance. Calculating $\operatorname{CV}_{(1)}(L)$ for increasing $L$ and stopping once a clear local minimum is found is possible and is the approach we use in our numerical results; further details are given in Section~\ref{sec:adaptiveKmax}.

\begin{Theorem}\label{thm:consistency} 
Suppose Assumption~\ref{assumption:cpNumberMultivariate} holds, and in the case where $K > 0$ eventually,
\begin{equation}\label{eq:conditionConsistency}
\liminf_{n\to \infty}{\frac{\underline{\lambda} \Delta_{(1)}^2}{K\sigma^2 (\log \overline{\lambda})^2}} = \infty \;\;\text{ and }\;\; K\log\log\left((\log(K) \vee 1) \frac{\sigma^2}{\Delta_{(1)}^2} \vee e \right) = o(\log\log \overline{\lambda}).
\end{equation}
Then we have
\begin{align}
\Pj\left( \hat{K} = K \right) &\to 1, \text{ as } n\to \infty, \notag \\
\int_0^1{ \left(\hat{f}(t) - f(t)\right)^2 dt}
& = \mathcal{O}_\Pj\left( n^{-1}\sigma^2 \left(\log(K) \vee 1\right) \log\log\overline{\lambda}\right). \label{eq:L2riskModifiedCriterion}
\end{align}
\end{Theorem}

Theorem~\ref{thm:consistency} shows that $\operatorname{CV}_{\mathrm{mod}}$ \eqref{eq:cvRescaled}, used with least squares estimation, estimates the number of change-points consistently, under mild conditions.
%It is a consequence of our general result on the consistency of $\operatorname{CV}_{\mathrm{mod}}$ (Theorem~\ref{theorem:positiveResultRescaledCV}) which requires certain estimation error properties of the underlying change-point procedure, and Theorem~\ref{theorem:detectionPrecision} which bounds the estimation error of optimal partitioning.
The first condition in \eqref{eq:conditionConsistency} is slightly stronger than the minimax rate for change-points to be detectable, which is $\liminf_{n\to \infty} \underline{\lambda} \Delta_{(1)}^2 /(\sigma^2 \log \overline{\lambda}) = \infty$; see \citet[Equation~(2.2)]{chan2013detection} and \citet{dumbgen2001multiscale, frick2014multiscale, fryzlewicz2014wild} for related discussions. This first condition may be an artefact of our proof strategy and could perhaps be relaxed. The second condition in \eqref{eq:conditionConsistency} ensures that change-point locations are estimated accurately enough such that the costs caused by misestimating change-point locations are smaller than the costs of adding a false positive.

To the best of our knowledge, existing $L_2$-error bounds for other estimators are $\mathcal{O}_{\Pj}\big(n^{-1}\sigma^2\log(n)\big)$ or worse \citep[Remarks~4~and~10]{lin2016approximate}. In comparison, our bound is only a factor $K(\log(K) \vee 1) \log\log\overline{\lambda} \leq  K(\log(K)\vee 1) \log\log n$ larger than the lower bound in \citet[Theorem 1(ii)]{li2019multiscale}. As a side effect we also see that the same bound holds for least squares estimation given the correct number of change-points.

We note that if the second condition in \eqref{eq:conditionConsistency} does not hold, but if instead the left hand side is $o(\log(\overline{\lambda}))$, then we may conclude $\Pj\left( \hat{K} = K \text{ or } \hat{K} = K + 1 \right) \to 1, \text{ as } n\to \infty$, i.e.\ the number change-points will be over-estimated by at most one change-point.

\subsection{General multivariate model}\label{sec:model}
Our result on the consistency of $\operatorname{CV}_{\mathrm{mod}}$ used with least squares estimation is based on a general result placing conditions on an arbitrary estimation procedure that yield consistency, which we now present. We consider a multivariate parametric change-point regression model with potentially non sub-Gaussian errors. We build on earlier notation but with estimated change-point sets $\hat{\mathcal{T}}^O_L$ and $\hat{\mathcal{T}}^E_L$ \eqref{eq:estimatedCpO} now not necessarily estimated by least squares estimation. 

We consider a multivariate version of $\operatorname{CV}_{\mathrm{mod}}$:
\begin{equation}\label{eq:cvRescaled}
\begin{split}
\operatorname{CV}_{\mathrm{mod}}(L) := 
&\sum_{l = 0}^{L}\sum_{i = \hat{\tau}^O_{L,l} + 1}^{\hat{\tau}^O_{L,l + 1} - 1}{ \frac{\hat{\tau}^O_{L,l + 1} - \hat{\tau}^O_{L,l}}{\hat{\tau}^O_{L,l + 1} - \hat{\tau}^O_{L,l} - 1} \left\|Y^E_i - \overline{Y}^O_{\hat{\tau}^O_{L,l} :\hat{\tau}^O_{L,l + 1}}\right\|_2^2}\\
 + &\sum_{l = 0}^{L}\sum_{i = \hat{\tau}^E_{L,l} + 2}^{\hat{\tau}^E_{L,l + 1}}{ \frac{\hat{\tau}^E_{L,l + 1} - \hat{\tau}^E_{L,l}}{\hat{\tau}^E_{L,l + 1} - \hat{\tau}^E_{L,l} - 1} \left\|Y^O_i - \overline{Y}^E_{\hat{\tau}^E_{L,l} :\hat{\tau}^E_{L,l + 1}}\right\|_2^2},
\end{split}
\end{equation}
where here $Y_i^E, Y_i^0 \in \R^d$ for all $i$. We first establish some additional notation in order to state our result. Let $\Sigma_k := \Cov\left[ Y_i \right]$, for $i = \tau_k + 1,\ldots,\tau_{k+1},\ k = 0,\ldots,K$ and let $\overline{\sigma}(\Sigma_k)^2$ be the maximum eigenvalue of $\Sigma_k$. Further let $\overline{\sigma}^2 := \max_{k=0, \ldots,K} \overline{\sigma}(\Sigma_k)^2$. Finally, for any set of candidate change-points $\mathcal{U}=\{t_0 < t_1 < \cdots <t_M < t_{M + 1}\}$  and any collection of vectors $X = (X_1,\ldots,X_{t_{M + 1}}) \in \R^{d \times t_{M+1}}$ we use the notation
\begin{equation} \label{eq:S_notation}
S_X(\mathcal{U}) := \sum_{k = 0}^{M}\sum_{i = t_k + 1}^{t_{k+1}}{\left\| X_i - \overline{X}_{t_k :t_{k+1}}\right\|_2^2}.
\end{equation}

We now introduce assumptions under which consistent estimation of $K$ holds. In addition to Assumption~\ref{assumption:cpNumberMultivariate}, we require the following assumptions for consistent estimation of $K$. These parallel assumptions required for \citet[Theorems~1~and~2]{zou2020consistent}, but are in places stronger as they guarantee consistency even when $K$ is allowed to increase; see the discussion before and after \eqref{eq:whatZouetalHaveShown} in Section~\ref{sec:proof:RescaledCV} in the supplementary material.

We require the following uniform Bernstein condition on the errors.
\begin{Assumption}[Noise]\label{assumption:NoiseBernstein}
The covariance matrices $\Sigma_k$ are positive-definite and there exists a constant $c>0$ such that
\begin{equation*}
\limsup_{n \to \infty}\max_{k = 0,\ldots,K} \max_{\tau_k < i \leq \tau_{k+1}} \E\left[ \left\| \Sigma_k^{-1/2} \varepsilon_{i} \right\|_2^{q} \right] \leq  \frac{q!}{2} c^{q - 2}\quad \forall\ q \geq 3.
\end{equation*} 
\end{Assumption}

The next assumption ensures precise estimation of the change-point locations. In the special case where $K = 0$ for all $n$, this and Assumption~\ref{assumption:minimumSignalMultivariate} are not required.
\begin{Assumption}[Estimation precision]\label{assumption:detectionPrecision} We state the below in terms of the odd sequence, but require analogous properties with all instances of $O$ replaced by $E$.
%Let $Q := K_{\max} - K$.
There exist triangular arrays $(\delta_{q,k})_{k=1}^K$ for $q=0,1$ and 
%a sequence of matrices $(\delta_{q,k}) \in \mathbb{Z}_{\geq 0}^{(Q+1) \times K}$ where $q=0,\ldots,Q$ and $k=1,\ldots,K$,
a positive sequence $(C_n)$ with $C_n \to \infty$ such that the following are satisfied:
\begin{enumerate}[label=(\roman*)]
\item \label{assumption:detectionPrecision:L>=K} $\Pj\Big(|\hat{\tau}^O_{K,k} - \tau^O_k| \leq \delta_{0,k} \wedge \Big( \frac{\underline{\lambda}}{2}-1\Big), \; k=1,\ldots,K\Big) \to 1 $ and
\begin{equation*}
\begin{split}
	& \Pj\Big(\forall\ K < L \leq K_{\max},\,\, \exists\, \hat{\tau}^O_{L,i_1},\ldots, \hat{\tau}^O_{L,i_{K}} \in \hat{\mathcal{T}}^O_L\, :\, \left\vert \hat{\tau}^O_{L,i_k} - \tau^O_{k} \right\vert \leq \delta_{1,k} \wedge \Big(\frac{\underline{\lambda}}{2}-1\Big),\ k = 1,\ldots,K \Big)\\
	& \to 1,
	\end{split}
\end{equation*}
%\begin{equation*}
%\Pj\Big(\forall\ K \leq L \leq K_{\max},\,\, \exists\, \hat{\tau}^O_{L,i_1},\ldots, \hat{\tau}^O_{L,i_{K}} \in \hat{\mathcal{T}}^O_L\, :\, \left\vert \hat{\tau}^O_{L,i_k} - \tau^O_{k} \right\vert \leq \delta_{L-K,k} \wedge \Big(\frac{\underline{\lambda}}{2}-1\Big),\ k = 1,\ldots,K \Big) \to 1,
%\end{equation*}

\item \label{assumption:detectionPrecision:Rate}
$\max_{q=0,1;\ 1\leq k \leq K }  K \log\log{ (\delta_{q,k} \vee e)} = o(\log\log\overline{\lambda})$,
\item \label{assumption:detectionPrecision:Infq0} $\sum_{k = 1}^{K} \delta_{0, k} \Delta_k^2 = o( \overline{\sigma}^2\log\log \overline{\lambda})$,
%for each $k=1,\ldots,K$, if $\log(eK) \overline{\sigma}^2 / \Delta_k^2 \leq C_n$, then $\delta_{q,k} = 0$ for all $q =0,\ldots,Q$,
\item \label{assumption:detectionPrecision:Infq>0}
writing $\mathcal{K}_n := \{k : \overline{\sigma}^2 \log\log \overline{\lambda} / (K \Delta_k^2) \leq C_n\}$, we have that for all $k \in \mathcal{K}_n$,
%for each $k=1,\ldots,K$, if $\overline{\sigma}^2 \log\log \overline{\lambda} / (K \Delta_k^2) \leq C_n$, then
 $\delta_{0,k} =\delta_{1,k}= 0$,
\item \label{assumption:detectionPrecision:L<K} %writing $\mathcal{K}_n := \{k : \log(eK) \overline{\sigma}^2 / \Delta_k^2 > C_n\}$, we have that
%writing $\mathcal{K}_n := \{k : \overline{\sigma}^2 \log\log \overline{\lambda} / (K \Delta_k^2) \leq C_n\}$, we have that
for a constant $C > 0$ a constant and writing $\overline{\mu}^O_{L,i} := \sum_{l = 0}^{L}{\EINS_{\{\hat{\tau}^O_{L,l} + 1 \leq i \leq \hat{\tau}^O_{L,l + 1}\}} \overline{\mu}_{\hat{\tau}^O_{L,l} :\hat{\tau}^O_{L,l + 1}}}$, we have
\begin{equation*}
	\Pj\Bigg(\forall\, L < K,\, \forall\, k \in \mathcal{K}_n, \sum_{i = \tau^O_{k} - \floor{\underline{\lambda}/4} + 1}^{\tau^O_{k} + \floor{\underline{\lambda}/4}}{\big\|\mu^O_i - \overline{\mu}^O_{L,i}\big\|_2^2} \geq  C\underline{\lambda}\Delta_k^2 \text{ or } \exists\, \hat{\tau} \in \hat{\mathcal{T}}^O_L\, :\, \hat{\tau} = \tau_{k} \Bigg) \to 1.
\end{equation*}
%as $n\to\infty$, with $C > 0$ a constant and $\overline{\mu}^O_{L,i} := \sum_{l = 0}^{L}{\EINS_{\{\hat{\tau}^O_{L,l} + 1 \leq i \leq \hat{\tau}^O_{L,l + 1}\}} \overline{\mu}_{\hat{\tau}^O_{L,l} :\hat{\tau}^O_{L,l + 1}}}$.
\end{enumerate}
\end{Assumption}
Assumption \ref{assumption:detectionPrecision:L>=K} in particular states that the misestimation in the cases where $L$ is the true number of change-points, and when $L$ is too large, is uniformly bounded by triangular arrays $(\delta_{0,k})$ and $(\delta_{1,k})$ respectively. The size of those values are limited by \ref{assumption:detectionPrecision:Rate}--\ref{assumption:detectionPrecision:Infq>0} to ensure that the influence of misestimation on the cross-validation criterion is under control. While \ref{assumption:detectionPrecision:Rate} ensures that the errors are generally small, \ref{assumption:detectionPrecision:Infq0}~and~\ref{assumption:detectionPrecision:Infq>0} bound the errors of large changes when $q = 0$ and $q \geq 0$, respectively. Part \ref{assumption:detectionPrecision:L<K} deals with the case where $L < K$ and states that for large changes either the change-point is estimated precisely or the mean estimate is poor. Theorem~\ref{theorem:detectionPrecision} shows that all of these conditions are satisfied in the case of least squares estimation as long as changes are large enough; see Section~\ref{sec:optimalPartitioning} and the proof of Theorem~\ref{thm:consistency}.

Next, we have to assume that the costs of over-fitting are at least of order $\overline{\sigma}^2 \log\log\overline{\lambda}$; see the discussion after Assumption~\ref{assumption:cpNumberMultivariate}. One can show that this is satisfied by least squares estimation; see for instance Lemma~\ref{lemma:overestimation} in the supplementary material which relies on \citet[Theorem~2]{zou2020consistent}.
\begin{Assumption}[Over-fitting]\label{assumption:overestimation}
For all $M>0$, 
\begin{equation*}
	\Pj \left(\frac{\min_{L = K + 1,\ldots, K_{\max}}\left\{S_{\varepsilon^O}\left(\mathcal{T}^{O}_{K}\right) - S_{\varepsilon^O}\left(\hat{\mathcal{T}}^O_L\cup \mathcal{T}^{O}_{K}\right)\right\}}{\overline{\sigma}^2 \log\log\overline{\lambda}} < M \right)  \to 0,
\end{equation*}
and as above but with all instances of $O$ replaced by $E$.
\end{Assumption}

Finally, we assume that all change-points are sufficiently large; see also the discussion following Theorem~\ref{thm:consistency}.
\begin{Assumption}[Minimum Signal]\label{assumption:minimumSignalMultivariate}
The minimum jump size $\Delta_{(1)}$ satisfies
\begin{equation*}
\frac{\underline{\lambda} \Delta_{(1)}^2}{K \overline{\sigma}^2 (\log \overline{\lambda})^2} \to \infty, \quad \text{as } n\to\infty.
\end{equation*}
\end{Assumption}

We may now state our general result on consistency of our modified cross-validation $\operatorname{CV}_{\mathrm{mod}}$.
\begin{Theorem}[Consistency]\label{theorem:positiveResultRescaledCV}
Suppose that Assumptions \ref{assumption:cpNumberMultivariate}--\ref{assumption:minimumSignalMultivariate} hold. Then,
\begin{equation*}
\Pj\left( \hat{K} = K \right) \to 1, \text{ as } n\to \infty.
\end{equation*}
\end{Theorem}

\subsection{Least Squares Estimation}\label{sec:optimalPartitioning}
In this section we give theoretical guarantees for the change-points estimated by least squares estimation (i.e.\ the Segment Neighbourhood algorithm) \citep{auger1989algorithms}. This allows us to verify Assumption~\ref{assumption:detectionPrecision} for least squares estimation in the case where we have sub-Gaussian noise. More precisely, we assume the setting of Section~\ref{sec:cvL2}, with the exception that the condition
$\limsup_{n \to \infty}\sigma / \underline{\sigma} < \infty$ is not required. Similarly to Section~\ref{sec:gaussianSetting}, we denote the unknown set of true change-points by
\begin{equation}\label{eq:trueCP}
\mathcal{T} := \{0 = \tau_0 < \tau_1 < \cdots < \tau_{K} < \tau_{K + 1} = n\}.
\end{equation}
Furthermore, for $L \in \N$, we write
\begin{equation}\label{eq:estimatedCP}
\hat{\mathcal{T}}_L := \big\{0 = \hat{\tau}_{L,0} < \hat{\tau}_{L,1} < \cdots < \hat{\tau}_{L,L} < \hat{\tau}_{L, L + 1} = n\big\}.
\end{equation}
for the set of estimated change-points using least squares estimation, i.e.\ the output of \eqref{eq:optimalPartitioning} with $Z$ being $Y$. Then, we have the following guarantees for the estimated change-point locations.

\begin{Theorem}\label{theorem:detectionPrecision}
Let $1 \leq K \leq \overline{\lambda}$ eventually and assume that
\begin{equation}\label{eq:detectableCPs}
\frac{\underline{\lambda} \Delta_{(1)}^2}{K \sigma^2 \log(\overline{\lambda})} \to \infty, \text{ as } n \to \infty.
\end{equation}
Then, for any triangular arrays $(\gamma_{q,k}),\ q = 0,1,\ k = 1,\ldots,K$, such that 
\begin{equation}\label{eq:precisionRate}
\begin{split}
\max_{k = 1, \ldots, K}\big(\gamma_{0,k}\big)^{-1}(\log(K) \vee 1) \frac{\sigma^2}{\Delta_k^2} =& o\left(1\right),\\
\max_{k = 1, \ldots, K}\big(\gamma_{1,k}\big)^{-1}\left(\log\left(K\frac{\sigma^2}{\Delta_k^2} \right) \vee 1 \right) \frac{\sigma^2}{\Delta_k^2} =& o\left(1\right),
\end{split}
\end{equation}
we have
\begin{equation}\label{detectionPrecision:L>=K}
\begin{split}
& \Pj\Big(|\hat{\tau}_{K,j} - \tau_j| \leq \gamma_{0,k}, \; k=1,\ldots,K\Big) \to 1,\\
&\Pj\Bigg(\forall\, L > K\, \exists\, \hat{\tau}_{L,i_1},\ldots, \hat{\tau}_{L,i_{K}} \in \hat{\mathcal{T}}_L\, :\, \left\vert \hat{\tau}_{L,i_k} - \tau_{k} \right\vert \leq  \gamma_{1,k},\ k = 1,\ldots,K \Bigg) \to 1,\\
\end{split}
\end{equation}
and
\begin{equation}\label{detectionPrecision:L<K}
\begin{split}
\Pj\Bigg(& \forall\, L < K,\, \forall\, k = 1,\ldots,K, \, \sum_{i = \tau_{k} - \floor{\underline{\lambda}/2} + 1}^{\tau_{k} + \floor{\underline{\lambda}/2}}{\big\|\mu_i - \overline{\mu}_{L, i}\big\|^2} \geq  \frac{\underline{\lambda}\Delta_k^2}{200} \text{ or }\\
& \qquad \exists \, \hat{\tau} \in \hat{\mathcal{T}}_L : |\hat{\tau} - \tau_k| \leq \gamma_{1,k}
%\exists\, \hat{\tau}_{L,i_k} \in \hat{\mathcal{T}}_L,\ i_k = i_l \Leftrightarrow k = l\, :\, \left\vert \hat{\tau}_{L,i_k} - \tau_{k} \right\vert \leq  \gamma_{L,k}
\Bigg) \to 1,
\end{split}
\end{equation}
where $\overline{\mu}_{L, i} := \sum_{l = 0}^{L}{\EINS_{\{\hat{\tau}_{L,l} + 1 \leq i \leq \hat{\tau}_{L,l + 1}\}} \overline{\mu}_{\hat{\tau}_{L,l} :\hat{\tau}_{L,l + 1}}}$.
\end{Theorem}

Assumption~\eqref{eq:detectableCPs} ensures that each change-point is detectable; see the discussion after Theorem~\ref{thm:consistency}. The assumption $K \leq \overline{\lambda}$ is rather weak. Moreover, if it is not satisfied, we can  obtain a similar statement by replacing $\log \overline{\lambda}$ by $\log K$ in \eqref{eq:detectableCPs}. The guarantees obtained  for the estimated change-point locations are essential for the proof of Theorem~\ref{thm:consistency}. However, while \eqref{detectionPrecision:L<K} is rather technical, \eqref{detectionPrecision:L>=K} may be of independent interest. Related results about least squares estimation, which is equivalent to maximum likelihood estimation in the Gaussian case, exist in the literature \citep{yao1988estimating, yao1989least}. However, to the best of our knowledge, no existing result covers the case of $L \neq K$.
Moreover, Theorem~\ref{theorem:detectionPrecision} permits all parameters to vary with $n$. For finite sample results where $L=K$ (beyond estimating changes in mean in piecewise constant signals with sub-Gaussian noise) see \citep[Theorem 3.1]{Garreau2018Consistent} and following discussion. If we only want versions of \eqref{detectionPrecision:L>=K} and \eqref{detectionPrecision:L<K} to hold for a specific change-point (rather than for all simultaneously), the $\log K$ factors may be dropped in \eqref{eq:precisionRate}. 

The rate obtained in \eqref{eq:precisionRate} is optimal when $L = K$  \citep[Proposition~6]{verzelen2020optimal}.
Interestingly, for $L > K$ an additional $\log(\sigma^2 / \Delta_k^2)$ factor is required, which we believe to be necessary. The factor may be seen as a consequence of the fact that adding one false positive in a setting when $K=0$ (i.e.\ considering $L=1$) only decreases the least square loss \eqref{eq:optimalPartitioning} by $\mathcal{O}_{\Pj}(\sigma^2 \log\log n)$, while adding two false positives results in a decrease of $\mathcal{O}_{\Pj}(\sigma^2 \log n)$, see Lemmas~\ref{lemma:maxOneSided}~and~\ref{lemma:maxYao}, respectively. Hence, placing one incorrect change-point due to noise has a larger contribution when there is at least one false positive, i.e.\ $L > K$.

\subsection{Data-driven choice of $K_{\max}$}\label{sec:adaptiveKmax}
For our theoretical results, we have considered a deterministic choice for the maximum number of potential change-points $K_{\max}$ (though this is permitted to increase with $n$). In practice however, it is helpful to be able to set this in a data-driven way. To do this, we can seek a local minimum of the cross-validation criterion by evaluating $\operatorname{CV}(L)$ for each $L=0, 1, \ldots$ and stopping when $\operatorname{CV}(L)>\operatorname{CV}(L-1)$, i.e.\ taking $K_{\max}$ to be the first $L$ where this occurs, and $\hat{K} = K_{\max} - 1$. In order to protect against the fact that such a local minimum may occur largely as a result of the noise, we can instead insist that $\hat{K}$ is further away from the boundary of values at which we have evaluated the criterion. We take such an approach in our numerical experiments, with precise implementation details given in Appendix~\ref{sec:adaptiveKmaxDetails}.

\section{Numerical experiments}\label{sec:simulations}
In this section we study the performance of different methods and tuning parameter selection procedures empirically\footnote{Code for the simulations is available at \url{https://github.com/FlorianPein/SimulationsCrossvalidationCP}}. We consider a variety of univariate change in mean settings and report for each setting, the proportion of times the number of change-points is underestimated ($\hat{K} < K$), correctly estimated ($\hat{K} = K$) and overestimated ($\hat{K} > K$) over $M=10\,000$ simulation runs. We also report the mean integrated squared error $\operatorname{MISE}$.

In Section~\ref{sec:LargeExample}, we consider a more complicated and perhaps more realistic signal than those described in Examples~\ref{example:underestmation}~and~\ref{example:overestmation} in Section~\ref{sec:L2loss}, but which still demonstrates the sensitivity of vanilla cross-validation to the locations of large change-points. In Section~\ref{sec:arlot}, we apply our procedures to a simulation setting from \citet{arlot2011segmentation}. In Section~\ref{sec:blocksSignal}, we investigate the performance of our new criteria in settings with the famous `block' \citep{donoho1994ideal} and 'stairs' \citep{fryzlewicz2014wild} signals, and in Section~\ref{sec:robustness} we consider settings with non-Gaussian and non-i.i.d.\ error distributions. Section~\ref{sec:furtherSimulations} in the supplementary material contains the results of further simulations relating to Examples~\ref{example:underestmation}~and~\ref{example:overestmation}, and a systematic study of the detection power.

We compare our cross-validation approaches $\operatorname{CV}_{(1)}$ \eqref{eq:cvL1}, $\operatorname{CV}_{\operatorname{mod}}$ \eqref{eq:cvModifiedL2}, and a $V$-fold version of $\operatorname{CV}_{(1)}$ (see Appendix~\ref{sec:generalisedProcedure} and \eqref{eq:cvL1generlised} in particular), all used with least squares estimation, to a number of competitors. Note that for $V$-fold cross-validation with $\operatorname{CV}_{(1)}$ we use our adaptive choice for $K_{\max}$, see Section~\ref{sec:adaptiveKmax} and Appendix~\ref{sec:adaptiveKmaxDetails} for details. For all other procedures and all competitors that require such a choice we set $K_{\max} = 30$, in order to allow for a more direct comparison with the original \texttt{COPSS} procedure. Note that the choice of $K_{\max}$ and whether it is adaptively chosen or not has no significant influence on the results; see also Appendix~\ref{sec:simulationKmax} for a brief simulation study. For the simulations, we use the functions \texttt{CV1}, \texttt{CVmod}, and \texttt{VfoldCV} from our \texttt{R} package \texttt{crossvalidationCP}, available on CRAN, which is coded entirely in \texttt{R} for maximum flexibility; we note that if a fast implementation for a specialised setting is required, compiled code will be beneficial. We do however use the function \texttt{Fpsn} from the package \texttt{fpopw}, which is a fast implementation of the Segment Neighbourhood algorithm with functional pruning (see \citet{rigaill2015pruned, maidstone2017optimal}), for the runtime intensive calculation of least squares estimation for $L = 0,\ldots,K_{\max}$ change-points.

The list of competitors includes the $\operatorname{COPSS}$ procedure from \citet{zou2020consistent} described in Section~\ref{sec:gaussianSetting} (we use our own implementation given by the function \texttt{COPSS} in the \texttt{R} package \texttt{crossvalidationCP}) and the two stage $V$-fold cross-validation procedure $\operatorname{LooVF}$ from \citet{arlot2011segmentation}\footnote{Matlab implementation available at https://www.imo.universite-paris-saclay.fr/~sylvain.arlot/code/CHPTCV.htm};\ note that due to the slower run-time of the latter, we only include it in our smaller-scale simulations. We also compare to classical change-point procedures such as $\operatorname{PELT}$ \citep{killick2012optimal}, with the $\operatorname{SIC}$-penalty implemented in the R-package \texttt{changepoint}; wild binary segmentation ($\operatorname{WBS}$) \citep{fryzlewicz2014wild}, with penalty $1.3$ times the $\operatorname{SIC}$-penalty implemented in the \texttt{R} package \texttt{wbs}; $\operatorname{FDRSeg}$ \citep{li2016fdr} with $\alpha = 0.9$ as implemented by the \texttt{R} package \texttt{FDRSeg}; and $\operatorname{Ms.FPOP}$ \citep{liehrmann2023ms, verzelen2020optimal}, with recommend parameters $\alpha = 9 + 2.25 \log(n)$ and $\beta = 2.25$. We also included the robust methods from \citet{fearnhead2019changepoint}, implemented in the R-package \texttt{robseg}. We use Biweight loss (which outperformed the other available option of Huber loss in all simulations) with default parameters $\lambda = 2\log(n)$ and threshold $3$. These provide genuinely different estimation methods to least squares, and so a comparison with the cross-validation approaches using the latter loss function should be interpreted with care. We also experimented with $\operatorname{HSMUCE}$ \citep{pein2017heterogeneous}; however the strong error control it provides meant that it lacked sufficient power to detect change-points so we do not present these results.

\subsection{Sensitivity to locations of large changes}\label{sec:LargeExample}
Here we consider a setting with $n = 2048$ observations and $K = 11$ change-points. The signal, shown in Figure~\ref{fig:largerExample}, has change-points at $204, 470, 778, 878, 883, 894, 984, 1414, 1638, 1680, 1740$ and function values $-2.32$, $15.98$, $5$, $20$, $0$, $70$, $0$, $-15$, $-7.32$, $8.42$, $-2.93$, $4.76$. To form the observations, we add to the signal $\mathcal{N}(0, \sigma^2)$ errors with $\sigma = 7$. We also consider a modified signal where the change-point at $883$ is shifted to $884$, an even location.

\begin{figure}[!htb]
	\centering
\includegraphics[width = 0.9\textwidth]{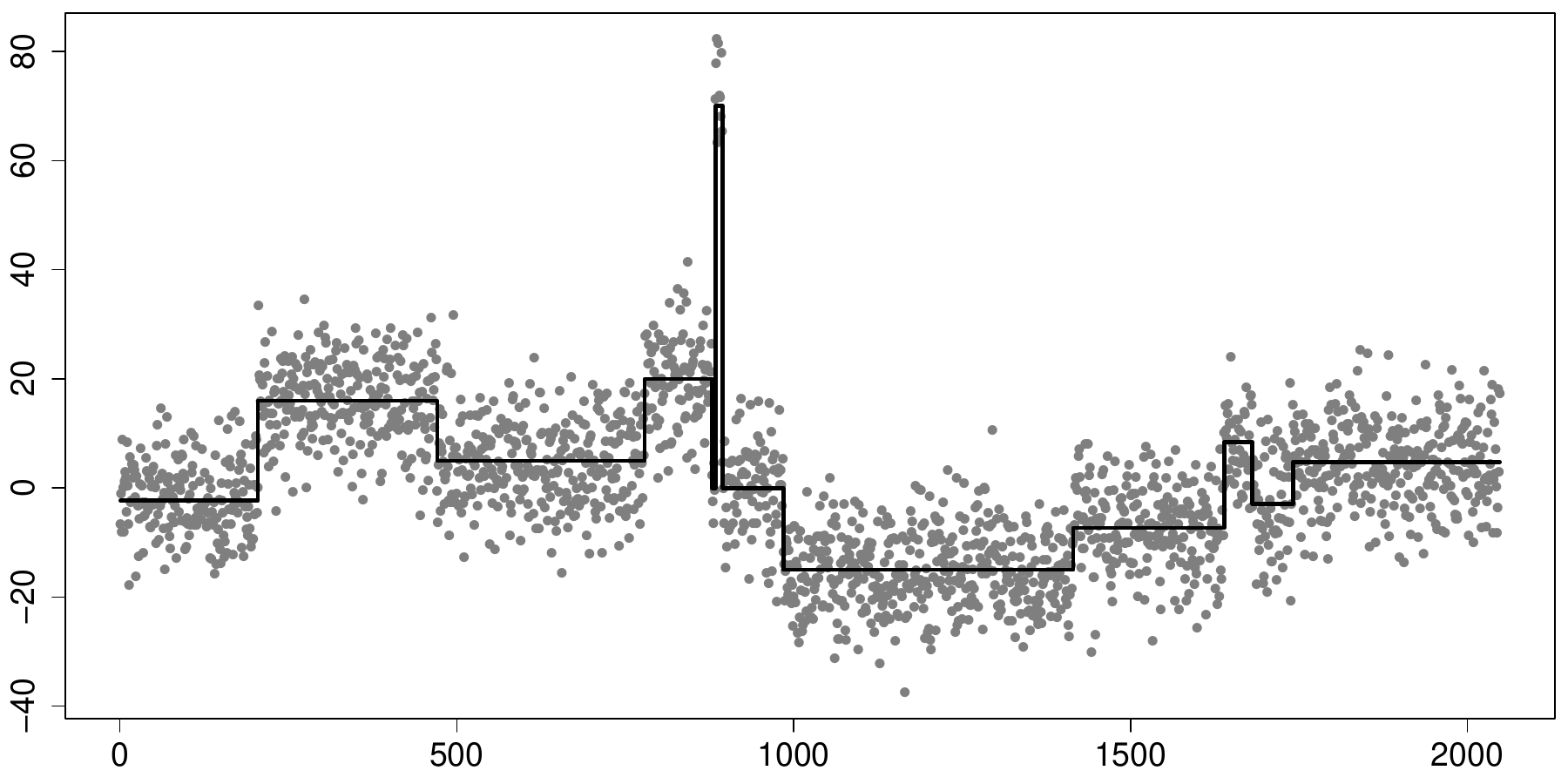} 
\caption{Signal with $n = 2048$ observations and $K = 11$ change-points.}
\label{fig:largerExample}
\end{figure}

\begin{table}[ht]
\centering
\begin{tabular}{|l|cccc|cccc|}
\hline
  &  \multicolumn{4}{c|}{Original signal} & \multicolumn{4}{c|}{Modified signal}\\
  \hline
Method & $\hat{K} < K$ & $\hat{K} = K$ & $\hat{K} > K$ & $\operatorname{MISE}$ & $\hat{K} < K$ & $\hat{K} = K$ & $\hat{K} > K$ & $\operatorname{MISE}$ \\ 
  \hline
$5$-fold $\operatorname{CV}_{(1)}$ &  3.56 & 81.15 & 15.29 & 0.9061 &  3.46 & 81.14 &  15.4 & 0.9011 \\ 
  $\operatorname{CV}_{(1)}$ & 16.27 & 73.49 & 10.24 & 1.004 & 13.28 & 75.54 & 11.18 & 0.9811 \\ 
  $\operatorname{CV}_{\operatorname{mod}}$ & 15.61 & 74.94 &  9.45 & 1.006 & 15.55 & 75.07 &  9.38 & 1.003 \\ 
    $\operatorname{COPSS}$ &  60.7 &  33.5 &   5.8 & 1.493 & 12.84 & 74.78 & 12.38 & 0.9794 \\ 
  $\operatorname{PELT}$ &  0.91 &  95.9 &  3.19 & 0.8369 &  0.78 & 95.83 &  3.39 & 0.831 \\ 
  $\operatorname{WBS}$ &   4.7 & 40.65 & 54.65 &  1.14 &  4.73 & 39.95 & 55.32 & 1.135 \\ 
  $\operatorname{FDRSeg}$ &  1.05 & 74.22 & 24.73 & 0.9133 &  0.78 & 75.17 & 24.05 & 0.9086 \\ 
  $\operatorname{Ms.FPOP}$ &  6.62 & 92.83 &  0.55 & 0.8754 &  6.04 & 93.47 &  0.49 & 0.865 \\ 
  $\operatorname{Biweight}$ &  2.69 & 94.65 &  2.66 & 0.8935 &  2.37 &  94.9 &  2.73 & 0.8888 \\ 
   \hline
\end{tabular}
\caption{Results of simulations with the signal as in Figure~\ref{fig:largerExample}. $\operatorname{COPSS}$ systematically underestimates the number of change-points in the setting with the original signal, while the suggested alternatives perform well.} 
\label{tab:largerExample}
\end{table}

We see that COPSS performs poorly compared to PELT and other change-point procedures when the large change is at an odd location, both in terms of estimation of $K$ and MISE; in contrast, the new cross-validation approaches perform well in both settings, with $5$-fold $\operatorname{CV}_{(1)}$ the best among these. This is mirrored in the additional results presented in Sections~\ref{sec:simulationUnderestimation}~and~\ref{sec:simulationOverestimation} in the supplementary material. 

\subsection{Small sample size}\label{sec:arlot}
In this section we consider a simulation setting of \citet{arlot2011segmentation} involving their $s_1$ signal (see Fig.~2 in their paper). This has four change-points with each constant segment containing 20 time points and jumps of size one in alternating directions. We consider the standard deviation functions $\sigma_c = 0.25$ (constant), $\sigma_{pc,3} = 0.6 \EINS_{[0,1/3)}(t) + 0.15 \EINS_{[1/3,1]}(t)$ (piecewise constant), and $\sigma_s = 0.5 \sin(t \pi / 4)$.

\begin{table}[ht]
\centering
\begin{tabular}{|l|cccc|cccc|}
\hline
  &  \multicolumn{4}{c|}{$\sigma_c$} & \multicolumn{4}{c|}{$\sigma_{pc,3}$}\\
  \hline
Method & $\hat{K} < K$ & $\hat{K} = K$ & $\hat{K} > K$ & $\operatorname{MISE}$ & $\hat{K} < K$ & $\hat{K} = K$ & $\hat{K} > K$ & $\operatorname{MISE}$ \\ 
  \hline
$2$-fold $\operatorname{CV}_{(1)}$ &     0 &  83.5 &  16.5 & 0.00613 &  0.34 & 67.74 & 31.92 & 0.02539 \\ 
  $5$-fold $\operatorname{CV}_{(1)}$ &     0 & 77.74 & 22.26 & 0.006659 &   0.2 & 70.66 & 29.14 & 0.02445 \\ 
  $10$-fold $\operatorname{CV}_{(1)}$ &     0 & 73.41 & 26.59 & 0.007089 &  0.13 & 68.61 & 31.26 & 0.02479 \\ 
  $20$-fold $\operatorname{CV}_{(1)}$ &     0 & 69.42 & 30.58 & 0.007407 &  0.16 & 66.76 & 33.08 & 0.02529 \\ 
  $\operatorname{CV}_{(1)}$ &     0 & 84.26 & 15.74 & 0.006199 &  0.21 & 66.97 & 32.82 & 0.02725 \\ 
  $\operatorname{CV}_{\operatorname{mod}}$ &     0 & 90.41 &  9.59 & 0.005648 &  2.87 & 81.36 & 15.77 & 0.02342 \\ 
   $\operatorname{COPSS}$ &     0 &  82.3 &  17.7 & 0.006193 &  0.76 &  71.6 & 27.64 & 0.02531 \\ 
  $\operatorname{LooVF}_2$ &     0 & 74.03 & 25.97 & 0.006866 &  1.26 & 65.33 & 33.41 & 0.02598 \\ 
  $\operatorname{LooVF}_5$ &     0 & 69.47 & 30.53 & 0.007342 &   0.6 & 68.66 & 30.74 & 0.0243 \\ 
  $\operatorname{PELT}$ &     0 & 87.11 & 12.89 & 0.006028 &     0 &  0.36 & 99.64 & 0.08918 \\ 
  $\operatorname{WBS}$ &     0 & 90.95 &  9.05 & 0.005569 &     0 &  0.13 & 99.87 & 0.08213 \\ 
  $\operatorname{FDRSeg}$ &     0 & 80.51 & 19.49 & 0.006702 &     0 &  0.02 & 99.98 & 0.1032 \\ 
  $\operatorname{Ms.FPOP}$ &     0 & 99.25 &  0.75 & 0.005122 &     0 & 12.32 & 87.68 & 0.05875 \\ 
  $\operatorname{Biweight}$ &     0 &  89.3 &  10.7 & 0.006057 &     0 & 12.35 & 87.65 & 0.0559 \\ 
   \hline
\end{tabular}
\caption{Results for the $s_1$ signal of \citet{arlot2011segmentation} with various standard deviation functions.}
\label{tab:arlotI}
\end{table}

\begin{table}[ht]
\centering
\begin{tabular}{|l|cccc|}
\hline
  &  \multicolumn{4}{c|}{$\sigma_s$}\\
  \hline
Method & $\hat{K} < K$ & $\hat{K} = K$ & $\hat{K} > K$ & $\operatorname{MISE}$ \\ 
  \hline
$2$-fold $\operatorname{CV}_{(1)}$ &     0 & 80.41 & 19.59 & 0.005233 \\ 
  $5$-fold $\operatorname{CV}_{(1)}$ &     0 & 76.49 & 23.51 & 0.005743 \\ 
  $10$-fold $\operatorname{CV}_{(1)}$ &     0 &  72.9 &  27.1 & 0.006056 \\ 
  $20$-fold $\operatorname{CV}_{(1)}$ &     0 & 70.73 & 29.27 & 0.00623 \\ 
  $\operatorname{CV}_{(1)}$ &     0 & 77.57 & 22.43 & 0.005724 \\ 
  $\operatorname{CV}_{\operatorname{mod}}$ &     0 & 87.77 & 12.23 & 0.004541 \\ 
    $\operatorname{COPSS}$ &     0 & 78.35 & 21.65 & 0.00549 \\ 
  $\operatorname{LooVF}_2$ &     0 & 72.75 & 27.25 & 0.005826 \\ 
  $\operatorname{LooVF}_5$ &     0 &  70.4 &  29.6 & 0.006228 \\ 
  $\operatorname{PELT}$ &     0 &  8.68 & 91.32 & 0.01827 \\ 
  $\operatorname{WBS}$ &     0 &  5.24 & 94.76 & 0.01726 \\ 
  $\operatorname{FDRSeg}$ &     0 &  2.76 & 97.24 & 0.02421 \\ 
  $\operatorname{Ms.FPOP}$ &     0 &  49.7 &  50.3 & 0.009256 \\ 
  $\operatorname{Biweight}$ &     0 & 24.41 & 75.59 & 0.0137 \\ 
   \hline
\end{tabular}
\caption{Continuation of Table~\ref{tab:arlotI}.}
\label{tab:arlotII}
\end{table}

We see that the proposed cross-validation criteria perform well, and in particular dominate $\operatorname{LooVF}$ here. $\operatorname{PELT}$, $\operatorname{WBS}$, $\operatorname{Ms.FPOP}$ and $\operatorname{Biweight}$ all perform well in the constant variance setting of $\sigma_c$, but their performances greatly deteriorate in the heteroscedastic settings. In contrast, our proposed  $\operatorname{CV}_{\operatorname{mod}}$ is competitive in the idealised homoscedastic setting, and is the best performer in the heteroscedastic settings.

\subsection{Blocks and stairs signal}\label{sec:blocksSignal}
We consider the famous block signal with $K = 11$ change-points. As in \citet{zou2020consistent} we choose $2048$ observations, $\mathcal{N}(0, 7^2)$ errors as before, and set change-points at $205$, $267$, $308$, $472$, $512$, $820$, $902$, $1332$, $1557$, $1598$, $1659$ with corresponding function values $0$, $14.64$, $-3.66$, $7.32$, $-7.32$, $10.98$, $-4.39$, $3.29$, $19.03$, $7.68$, $15.37$, $0$; see Figure~\ref{fig:signalBlocks}. Secondly, we consider the stairs example from \citet{fryzlewicz2014wild} with $n = 150$, errors distributed as $\mathcal{N}(0, 0.3^2)$, and a change-point with jump size $1$ every ten observations, so $K = 14$.

\begin{figure}[!htb]
	\centering
\includegraphics[width = 0.9\textwidth]{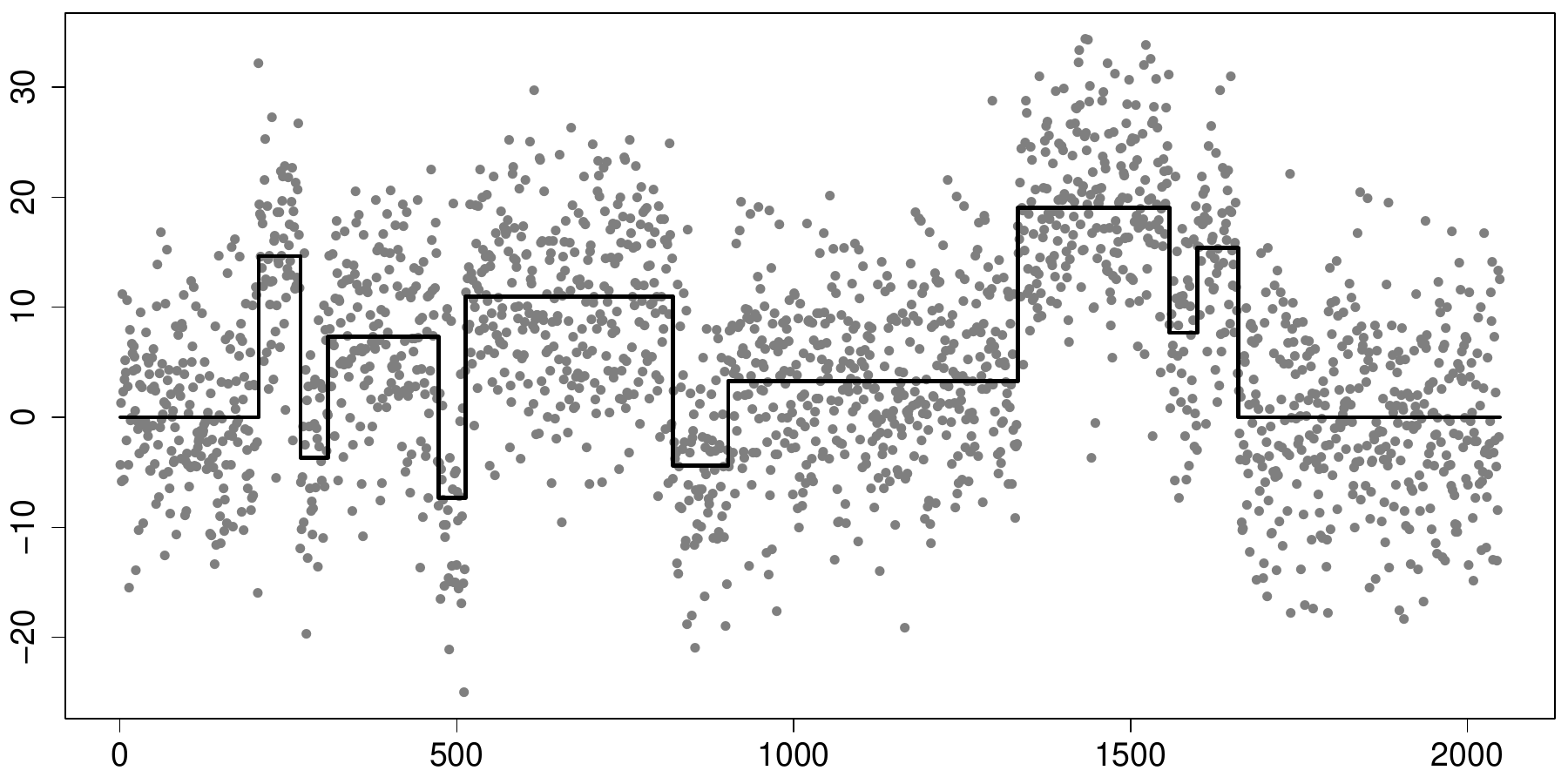} 
\caption{Blocks signal with $K = 11$ change-points.}
\label{fig:signalBlocks}
\end{figure}  

\begin{table}[ht]
\centering
\begin{tabular}{|l|cccc|cccc|}
\hline
  &  \multicolumn{4}{c|}{Blocks signal} & \multicolumn{4}{c|}{Stairs signal}\\
  \hline
Method & $\hat{K} < K$ & $\hat{K} = K$ & $\hat{K} > K$ & $\operatorname{MISE}$ & $\hat{K} < K$ & $\hat{K} = K$ & $\hat{K} > K$ & $\operatorname{MISE}$ \\ 
  \hline
$5$-fold $\operatorname{CV}_{(1)}$ &  7.78 & 76.46 & 15.76 & 1.047 &  0.27 & 75.57 & 24.16 & 0.02192 \\ 
  $\operatorname{CV}_{(1)}$ & 23.13 & 66.59 & 10.28 & 1.109 &  0.53 & 67.64 & 31.83 & 0.02275 \\ 
  $\operatorname{CV}_{\operatorname{mod}}$ & 21.76 & 69.15 &  9.09 & 1.099 &  8.63 & 82.21 &  9.16 & 0.02377 \\ 
    $\operatorname{COPSS}$ &  24.5 & 66.27 &  9.23 & 1.114 &  0.54 & 64.93 & 34.53 & 0.02296 \\ 
  $\operatorname{PELT}$ &  3.76 &  92.9 &  3.34 & 0.9741 &  0.42 &    94 &  5.58 & 0.02066 \\ 
  $\operatorname{WBS}$ & 25.66 & 72.44 &   1.9 & 1.205 &  6.14 & 86.05 &  7.81 & 0.02778 \\ 
  $\operatorname{FDRSeg}$ &  2.51 & 75.16 & 22.33 & 1.046 &  1.43 & 82.28 & 16.29 & 0.02231 \\ 
  $\operatorname{Ms.FPOP}$ & 10.08 & 89.18 &  0.74 & 0.994 &  9.74 &  90.2 &  0.06 & 0.02326 \\ 
  $\operatorname{Biweight}$ &  4.39 &  92.7 &  2.91 & 1.006 &  0.47 &  95.5 &  4.03 & 0.021 \\ 
   \hline
\end{tabular}
\caption{Results for the blocks, see Figure~\ref{fig:signalBlocks}, and stairs signal.}
\label{tab:blocks}
\end{table}

In these settings with well-specified errors, we might expect classical change-point procedures to have a noticeable advantage. However, from Table~\ref{tab:blocks} we see that cross-validation remains quite competitive. 

\subsection{Robustness}\label{sec:robustness}
In this section, we explore  the robustness of cross-validation to a misspecified models. We continue to use the blocks signal (Figure~\ref{fig:signalBlocks}) but with various violations of the change-point model in Section~\ref{sec:L2loss} as detailed below.

\subsubsection*{Misspecified error distribution} (Table~\ref{tab:wrongErrorDistribution}) We consider $t$-distributed errors with $5$ degrees of freedom and exponentially distributed errors with mean $1$ which we then mean-centre. All errors are standardised such that the standard deviation is $\sigma = 7$.

\subsubsection*{Heteroscedastic errors} (Table~\ref{tab:hetero311}) 
We consider two settings with heteroscedastic Gaussian errors: the first where the errors have different standard deviations on each piecewise constant segment, and the second where the standard deviation instead changes after each block of $32$ observations. The standard deviations are drawn independently from $U[0, 8]$.

\subsubsection*{Outliers} (Table~\ref{tab:outlier})
We use the same setting as that of Section~\ref{sec:blocksSignal} but randomly sample ten observations and add a Poisson distributed random variable with intensity $\lambda \in \{20, 30\}$.

\begin{table}[ht]
\centering
\begin{tabular}{|l|cccc|cccc|}
  \hline
Method & $\hat{K} < K$ & $\hat{K} = K$ & $\hat{K} > K$ & $\operatorname{MISE}$ & $\hat{K} < K$ & $\hat{K} = K$ & $\hat{K} > K$ & $\operatorname{MISE}$ \\ 
\hline
  &  \multicolumn{4}{c|}{$t_5$ error distribution} & \multicolumn{4}{c|}{exponential error distribution}\\
  \hline
$5$-fold $\operatorname{CV}_{(1)}$ &  17.4 & 58.08 & 24.52 & 1.722 & 27.51 & 50.73 & 21.76 & 1.488 \\ 
  $\operatorname{CV}_{(1)}$ & 28.87 & 35.62 & 35.51 & 2.153 & 40.87 & 27.13 &    32 &  1.99 \\ 
  $\operatorname{CV}_{\operatorname{mod}}$ &  72.8 & 26.71 &  0.49 & 3.565 & 83.77 & 15.88 &  0.35 & 2.893 \\ 
    $\operatorname{COPSS}$ & 48.56 & 35.54 &  15.9 & 2.152 & 65.09 &  25.3 &  9.61 & 1.945 \\ 
  $\operatorname{PELT}$ &  0.02 &  3.24 & 96.74 & 4.261 &     0 &     0 &   100 & 11.61 \\ 
  $\operatorname{WBS}$ &  0.75 &  6.65 &  92.6 & 3.324 &     0 &     0 &   100 & 7.304 \\ 
  $\operatorname{FDRSeg}$ &     0 &  0.12 & 99.88 & 6.265 &     0 &     0 &   100 & 18.81 \\ 
  $\operatorname{Ms.FPOP}$ &  0.84 & 22.98 & 76.18 & 2.873 &  0.01 &  0.03 & 99.96 & 6.978 \\ 
  $\operatorname{Biweight}$ &   1.6 & 93.11 &  5.29 & 0.8514 &  0.13 &  75.1 & 24.77 & 1.898 \\ 
   \hline
\end{tabular}
\caption{Results with misspecified error distributions.} 
\label{tab:wrongErrorDistribution} 
\end{table}

\begin{table}[ht]
\centering
\begin{tabular}{|l|cccc|cccc|}
  \hline
Method & $\hat{K} < K$ & $\hat{K} = K$ & $\hat{K} > K$ & $\operatorname{MISE}$ & $\hat{K} < K$ & $\hat{K} = K$ & $\hat{K} > K$ & $\operatorname{MISE}$ \\ 
  \hline
$5$-fold $\operatorname{CV}_{(1)}$ &  0.88 & 80.11 & 19.01 & 0.4409 &   0.2 & 81.66 & 18.14 & 0.4148 \\ 
  $\operatorname{CV}_{(1)}$ &  3.72 & 80.98 &  15.3 & 0.4408 &  1.95 & 84.44 & 13.61 & 0.4063 \\ 
  $\operatorname{CV}_{\operatorname{mod}}$ &  6.51 & 82.22 & 11.27 & 0.4454 &  4.17 & 85.69 & 10.14 & 0.4074 \\ 
  $\operatorname{COPSS}$ &  7.56 & 80.79 & 11.65 & 0.4531 &   4.9 &  84.3 &  10.8 & 0.4128 \\ 
  $\operatorname{PELT}$ &  0.01 &  3.75 & 96.24 &   4.8 &     0 &  0.01 & 99.99 & 5.034 \\ 
  $\operatorname{WBS}$ &     0 &  5.01 & 94.99 & 2.628 &     0 &  0.01 & 99.99 & 2.764 \\ 
  $\operatorname{FDRSeg}$ &     0 &  0.24 & 99.76 & 8.905 &     0 &     0 &   100 & 9.567 \\ 
  $\operatorname{Ms.FPOP}$ &  0.01 & 12.82 & 87.17 & 2.881 &     0 &  1.24 & 98.76 & 2.515 \\ 
  $\operatorname{Biweight}$ &  0.02 &  9.62 & 90.36 &  2.41 &     0 &  1.98 & 98.02 &  2.04 \\ 
   \hline
\end{tabular}
\caption{Results with heteroscedastic errors.}
\label{tab:hetero311} 
\end{table}

\begin{table}[ht]
\centering
\begin{tabular}{|l|cccc|cccc|}
  \hline
Method & $\hat{K} < K$ & $\hat{K} = K$ & $\hat{K} > K$ & $\operatorname{MISE}$ & $\hat{K} < K$ & $\hat{K} = K$ & $\hat{K} > K$ & $\operatorname{MISE}$ \\ 
\hline
  &  \multicolumn{4}{c|}{$\lambda = 20$} & \multicolumn{4}{c|}{$\lambda = 30$}\\
  \hline
$5$-fold $\operatorname{CV}_{(1)}$ & 10.21 & 77.51 & 12.28 & 1.101 & 18.86 &    71 & 10.14 & 1.264 \\ 
  $\operatorname{CV}_{(1)}$ & 27.59 & 62.76 &  9.65 & 1.204 & 40.52 & 41.85 & 17.63 & 1.573 \\ 
  $\operatorname{CV}_{\operatorname{mod}}$ & 30.88 & 64.32 &   4.8 & 1.206 & 64.02 & 35.29 &  0.69 & 1.816 \\ 
    $\operatorname{COPSS}$ & 31.42 & 62.07 &  6.51 & 1.208 & 54.54 & 39.62 &  5.84 & 1.572 \\ 
  $\operatorname{PELT}$ &  3.41 & 76.94 & 19.65 &  1.17 &  0.74 & 14.33 & 84.93 &  2.73 \\ 
  $\operatorname{WBS}$ & 22.85 & 61.79 & 15.36 & 1.364 &  7.19 & 18.13 & 74.68 & 2.372 \\ 
  $\operatorname{FDRSeg}$ &  1.54 & 41.13 & 57.33 & 1.505 &  0.08 &  1.62 &  98.3 & 3.937 \\ 
  $\operatorname{Ms.FPOP}$ & 10.81 &  87.2 &  1.99 &  1.06 &  7.96 & 59.39 & 32.65 & 1.575 \\ 
  $\operatorname{Biweight}$ &     5 & 92.13 &  2.87 & 1.034 &  5.33 & 91.83 &  2.84 & 1.031 \\ 
   \hline
\end{tabular}
\caption{Results with $10$ Poisson-distributed outliers.} 
\label{tab:outlier}  
\end{table}

The results above indicate that cross-validation and change-point regression with Biweight loss are more robust to model misspecifications than classical approaches designed for homoscedastic Gaussian errors.
Comparing these former two approaches, Biweight loss gives an impressive $75-93 \%$ probability of recovering the true number of changes in cases with $t$-distributed and exponentially distributed errors, and cases with outliers (Tables~\ref{tab:wrongErrorDistribution} and~\ref{tab:outlier}); cross-validation and in particular 5-fold $\mathrm{CV}_{(1)}$ is also fairly robust to these departures from the standard setting, delivering $50-77\%$ recovery probabilities here. However, in the heteroscedastic settings studied (Table~\ref{tab:hetero311}), the Biweight loss performs poorly, with recovery probabilities of $2-10\%$, compared to $80-82\%$ for 5-fold $\mathrm{CV}_{(1)}$. Overall, we find that cross-validation performs reasonably well without any knowledge of the type of violation to be expected, and without sacrificing much performance in idealised settings.

\section{Discussion}\label{sec:discussion}
In sharp contrast to its ubiquity in high-dimensional and non-parametric regression, cross-validation has received little attention and use in change-point problems. There is good reason for this: as we show in this work, standard cross-validation with squared error loss may not correctly estimate the number of change-points in settings where all changes are easily detectable, and can yield an estimated regression function that has an integrated squared error that is orders of magnitude larger than achievable by other methods. On the other hand, there may be much to be gained from  deeper investigation of the use of cross-validation in change-point problems. We propose two simple approaches to remedy these deficiencies and show empirically that they perform well in settings with Gaussian errors, and are relatively robust to heavy-tailed, non-symmetric and heteroscedastic errors.

We expect that cross-validation-type approaches for tuning parameter selection may be even more successful in more complex change-point settings, for example those involving piecewise smooth mean functions, in part due to the fact that no prior estimate of the noise variance is required. As well as exploring such settings, it would also be of interest to develop theory for the absolute error approach similar to that which we present for our modified squared error criterion. It may also be fruitful to develop alternatives to cross-validation that are also model agnostic, for example based on the bootstrap \citep{antoch1995change, huvskova2008bootstrapping, sharipov2016sequential}.

\begin{appendix}

\section{Generalised procedure}\label{sec:generalisedProcedure}
In this section we consider the setting of Section~\ref{sec:model} and describe a more general $V$-fold cross-validation procedure that selects a tuning parameter $\psi$ of an arbitrary change-point estimation procedure $\mathcal{A}(\psi, Y_1,\ldots,Y_n)$. The estimation procedure $\mathcal{A}$ requires a tuning parameter $\psi$ and a vector of observations $(Y_1, \ldots, Y_n)$ and returns the estimated change-point locations $0 = \hat{\tau}_{\psi, 0} < \hat{\tau}_{\psi, 1} < \cdots < \hat{\tau}_{\psi,\hat{K}_{\psi}} < \hat{\tau}_{\psi, \hat{K}_{\psi} + 1} = n$ and parameter estimates $\hat{\beta}_{\psi, 0} \neq \hat{\beta}_{\psi, 1} \neq \cdots \neq \hat{\beta}_{\psi, \hat{K}_{\psi}}$. The tuning parameter may for instance be the number of change-points as in the case of least squares estimation / Segment Neighbourhood
\citep{auger1989algorithms}.

Let $\Psi$ be a set of potential tuning parameters and define fold $F_v := \{v + i \cdot V,\ i\in \N\, : \, v + i \cdot V \leq n\}$ for $v=1,\ldots,V$. We then select $\psi \in \Psi$ as follows.

For each fold $F_v$ and every tuning parameter $\psi \in \Psi$, the procedure $\mathcal{A}$ is applied using all observations with indices not in $F_v$. We denote the estimated change-points by $0 =: \hat{\tau}^{-F_v}_{\psi,0} < \hat{\tau}^{-F_v}_{\psi,1} < \cdots < \hat{\tau}^{-F_v}_{\psi,\hat{K}^{-F_v}_{\psi}} < \hat{\tau}^{-F_v}_{\psi, \hat{K}^{-F_v}_{\psi} + 1} := n$ and the parameters by $\hat{\beta}^{-F_v}_{\psi,0},\ldots,\hat{\beta}^{-F_v}_{\psi,\hat{K}^{-F_v}_{\psi}}$. Note that we ask the estimated change-points to be a subset of $\{1,\ldots,n\}$, i.e.\ if $F_v^c = \{i_1,\ldots,i_m\}$ and $\mathcal{A}$ returns change-points $0 = \tilde{\tau}_{\psi, 0} < \tilde{\tau}_{\psi, 1} < \cdots < \tilde{\tau}_{\psi,\hat{K}^{-F_v}_{\psi}} < \tilde{\tau}_{\psi, \hat{K}^{-F_v}_{\psi} + 1} = m$, we set $\hat{\tau}^{-F_v}_{\psi, l} := i_{\tilde{\tau}_{\psi, l}},\ l = 1,\ldots,\hat{K}^{-F_v}_{\psi}$.

To evaluate the quality of the estimates, we use extended versions of the criteria we have proposed in Section~\ref{sec:methodology}. We extend cross-validation with absolute error loss \eqref{eq:cvL1} to
\begin{equation}\label{eq:cvL1generlised}
\begin{split}
& \operatorname{CV}^V_{(1)}(\psi) := \sum_{v = 1}^{V} \sum_{k = 0}^{\hat{K}^{-F_v}_{\psi}} \sum_{\substack{i \in F_v,\\\hat{\tau}^{-F_v}_{\psi, k} < i \leq \hat{\tau}^{-F_v}_{\psi, k + 1}}}{ \left\| Y_i - \hat{\beta}^{-F_v}_{\psi,k} \right\|_2}.
\end{split}
\end{equation}
Finally, the tuning parameter $\psi \in \Psi$ that minimizes cross-validation criterion is selected, i.e.\
\begin{equation}
\hat{\psi} := \argmin_{\psi \in \Psi}{\operatorname{CV}_{(1)}^V(\psi)}.
\end{equation}
%Once again, the tuning parameter $\psi$ can be the number of change-points $K$ as in the case of  optimal partitioning with $W$-weighted squared loss.
In the case where $\mathcal{A}$ is least squares estimation, but also for many other estimators, a natural choice for $\hat{\beta}^{-F_v}_{\psi,k}$ is 
\begin{equation}
\overline{Y}^{-F_v}_{\psi, k}
:= \big\vert \big\{ i \notin F_v;\  \hat{\tau}^{-F_v}_{\psi, k} < i \leq \hat{\tau}^{-F_v}_{\psi, k + 1}\big\}\big\vert^{-1} \sum_{\substack{i \notin F_v,\\ \hat{\tau}^{-F_v}_{\psi, k} < i \leq \hat{\tau}^{-F_v}_{\psi, k + 1}}}{  Y_i }.
\end{equation}

Alternatively, a slight generalisation of  \eqref{eq:cvModifiedL2} and \eqref{eq:cvRescaled} gives %$\operatorname{CV}_{(1)}^V(\psi)$ can be replaced by
\begin{equation}\label{eq:criterionVfold}
\begin{split}
&\operatorname{CV}_{\mathrm{mod}}^V(\psi) := \sum_{v = 1}^{V} \sum_{k = 0}^{\hat{K}^{-F_v}_{\psi}} \sum_{\substack{i \in F_v,\\ \hat{\tau}^{-F_v}_{\psi, k} + 1 < i \leq \hat{\tau}^{-F_v}_{\psi, k + 1}}}{ \frac{\hat{\tau}^{-F_v}_{\psi, k + 1} - \hat{\tau}^{-F_v}_{\psi, k}}{\hat{\tau}^{-F_v}_{\psi, k + 1} - \hat{\tau}^{-F_v}_{\psi, k} - 1}\left\| Y_i - \hat{\beta}^{-F_v}_{\psi,k} \right\|_2^2 }.
\end{split}
\end{equation}
%which is an extension of the cross-validation criteria \eqref{eq:cvModifiedL2} and \eqref{eq:cvRescaled}.
Note that for this criterion we require that each segment is at least of length $2(V - 1)$, i.e.\ $\hat{\tau}^{-F_v}_{\psi, k + 1}  - \hat{\tau}^{-F_v}_{\psi, k} \geq 2(V - 1)$.

%Finally, we note that we do not have to use the observations $Y_i$ directly. Following \citet{zou2020consistent} we can also consider the following transformation. Let $l(\beta, Y_i)$ be a loss function for observation $Y_i$ and parameter $\beta$ and let $s(\cdot, \cdot)$ be the partial derivative of $l(\beta, Y_i)$ with respect to $\beta$. We then define the transformation $Z_i := s(\gamma, Y_i)$, where $\gamma$ is a tuning parameter. As discussed in \citet{zou2020consistent}, the choice of the parameter $\gamma$ is often of low importance. However its choice can for instance be used to ensure that $Z_1,\ldots,Z_n$ satisfy Assumptions~\ref{assumption:cpNumberMultivariate}--\ref{assumption:minimumSignalMultivariate}. Similarly to before, we define $\mu_i := \E[Z_i]$ and $\varepsilon_i = Z_i - \mu_i$. We then replace $Y_i$ by $Z_i$, $i = 1,\ldots,n$, in all of the calculations. For example when $Y_i$ is univariate and has a Gaussian distribution and $l$ is chosen to be the log-likelihood function, then $Z_i = s(\gamma, Y_i) = - (Y_i - \gamma)$ and $\gamma$  cancels in \eqref{eq:cvL1generlised} and \eqref{eq:criterionVfold}.

\subsection{Adaptive choice of $K_{\max}$}\label{sec:adaptiveKmaxDetails}
As discussed in Section~\ref{sec:adaptiveKmax}, in practice, rather than fixing $K_{\max}$ in advance, we can choose $K_{\max}$ in a data-driven way so as to ensure that $\hat{K}$ is not too close to $K_{\max}$. To implement this, since the function \texttt{Fpsn} from the package \texttt{fpopw}, which we use to calculate the least squares estimator, only allows one to perform calculations for $L=1, 2, \ldots, K_{\max}$ for any given $K_{\max}$, we proceed as follows.  We start with $K_{\max} = 8$. If $\hat{K} < K_{\max} - 3$, we stop and return $\hat{K}$. Otherwise, we double $K_{\max}$ and rerun the procedure until either our stopping criterium is satisfied or in extreme cases $K_{\max} \geq n / 2$. The choices to start at $K_{\max} = 8$ and to subtract $3$ do not follow any specific considerations and results does not depend much on it as we see in the following section.

\subsection{Influence of $K_{\max}$ in simulations }\label{sec:simulationKmax}

In this simulation study we compare different starting values for $K_{\max}$ for the adaptive procedure and also use the non-adaptive procedure where we fixed $K_{\max} = 30$. We use the blocks setting from Section~\ref{sec:blocksSignal} with 5-fold $\operatorname{CV}_{(1)}$.

\begin{table}[ht]
\centering
\begin{tabular}{l|cccc}
  \hline
Method & $\hat{K} < K$ & $\hat{K} = K$ & $\hat{K} > K$ & $\operatorname{MISE}$ \\ 
  \hline
$K_{\max} = 30$, fixed &  7.47 & 77.26 & 15.27 & 1.041 \\ 
$K_{\max} = 5$, adaptive &  7.47 & 77.26 & 15.27 & 1.041 \\ 
$K_{\max} = 6$, adaptive &  7.47 & 77.26 & 15.27 & 1.041 \\ 
$K_{\max} = 7$, adaptive &  7.49 & 77.26 & 15.25 & 1.041 \\ 
$K_{\max} = 8$, adaptive &  7.47 & 77.31 & 15.22 &  1.04 \\ 
$K_{\max} = 9$, adaptive &  7.47 & 77.26 & 15.27 & 1.041 \\ 
$K_{\max} = 10$, adaptive &  7.47 & 77.26 & 15.27 & 1.041 \\ 
$K_{\max} = 11$, adaptive &  7.47 & 77.26 & 15.27 & 1.041 \\ 
$K_{\max} = 12$, adaptive &  7.47 & 77.26 & 15.27 & 1.041 \\ 
$K_{\max} = 13$, adaptive &  7.47 & 77.26 & 15.27 & 1.041 \\ 
$K_{\max} = 14$, adaptive &  7.49 & 77.26 & 15.25 & 1.041 \\ 
   \hline
\end{tabular}
\caption{Different choices for $K_{\max}$.}
\label{tab:adaptiveKmax} 
\end{table}

We see from Table~\ref{tab:adaptiveKmax} that results are nearly the same for all choices. We obtained similar results when we tried different choices for a fixed $K_{\max}$ (not displayed) as long as the number was a bit larger than the true number of change-points. This confirms our reasoning for the design of the adaptive procedure.

\end{appendix}

%%%%%%%%%%%%%%%%%%%%%%%%%%%%%%%%%%%%%%%%%%%%%%
%% Support information, if any,             %%
%% should be provided in the                %%
%% Acknowledgements section.                %%
%%%%%%%%%%%%%%%%%%%%%%%%%%%%%%%%%%%%%%%%%%%%%%
%\begin{acks}[Acknowledgments]
% The authors would like to thank ...
%\end{acks}
\begin{acks}[Acknowledgments]
The authors would like to thank the associate editor and two anonymous referees for their valuable feedback that helped us to improve the manuscript.
\end{acks}
%%%%%%%%%%%%%%%%%%%%%%%%%%%%%%%%%%%%%%%%%%%%%%
%% Funding information, if any,             %%
%% should be provided in the                %%
%% funding section.                         %%
%%%%%%%%%%%%%%%%%%%%%%%%%%%%%%%%%%%%%%%%%%%%%%
\begin{funding}
The authors were supported by EPSRC grant EP/N031938/1.
\end{funding}

%%%%%%%%%%%%%%%%%%%%%%%%%%%%%%%%%%%%%%%%%%%%%%
%% Supplementary Material, including data   %%
%% sets and code, should be provided in     %%
%% {supplement} environment with title      %%
%% and short description. It cannot be      %%
%% available exclusively as external link.  %%
%% All Supplementary Material must be       %%
%% available to the reader on Project       %%
%% Euclid with the published article.       %%
%%%%%%%%%%%%%%%%%%%%%%%%%%%%%%%%%%%%%%%%%%%%%%
\begin{supplement}
The supplementary material contains additional simulations and all proofs.
%\stitle{Supplementary material to Cross-validation for change-point regression: pitfalls and solutions}
%\sdescription{Supplementary material contains all proofs.}
\end{supplement}

%%  MR numbers will be added by VTeX.                      %%
%%                                                         %%
%%  Use \cite{...} to cite references in text.             %%
%%                                                         %%
%%%%%%%%%%%%%%%%%%%%%%%%%%%%%%%%%%%%%%%%%%%%%%%%%%%%%%%%%%%%%

%\bibliographystyle{imsart-nameyear}
%\bibliography{Literature}

%
%
%
\end{cbunit}
%\end{document}
\clearpage
\setcounter{page}{1}
\noindent \begin{center}
{\large \bf Supplementary material to\\
Cross-validation for change-point regression:\\
pitfalls and solutions}
\end{center} 

\quad

%\title{Supplementary material to\\
%Cross-validation for change-point regression:\\pitfalls and solutions}
%\author{Florian Pein
%	\and
%	Rajen D.\ Shah}
%%\emptythanks
%\maketitle

\begin{cbunit}

\setcounter{section}{0}
\renewcommand{\thesection}{S\arabic{section}}
\renewcommand{\thesubsection}{S1.\arabic{subsection}}

In the following we collect additional simulations and all of our proofs. First of all, in Section~\ref{sec:furtherSimulations} we present further simulation settings. Section~\ref{sec:proof:detectionPrecision} begins  with the proof of Theorem~\ref{theorem:detectionPrecision} as the strategy used to split the least squares objective will be helpful in other proofs as well.  Theorems~\ref{theorem:underestimation},~\ref{theorem:overestimation}~and~\ref{theorem:positiveResultRescaledCV} are proved in Sections~\ref{sec:proof:underestimation}, \ref{sec:proof:overestimation} and \ref{sec:proof:RescaledCV} respectively. Finally, Section~\ref{sec:proof:consistency} gives a proof of Theorem~\ref{thm:consistency}, which largely follows from Theorems~\ref{theorem:positiveResultRescaledCV}~and~\ref{theorem:detectionPrecision} and is hence shown last.

\section{Additional simulations}\label{sec:furtherSimulations}
In this section we include further simulation results. The methods under consideration include $V$-fold $\operatorname{CV}_{(1)}$ procedures as well as those mentioned in Section~\ref{sec:simulations}.%, with leave-one-out cross-validation denoted by $\operatorname{LOOCV} \operatorname{CV}_{(1)}$.

\subsection{Underestimation example}\label{sec:simulationUnderestimation}
Observations are as in Example~\ref{example:underestmation}. We take $n = 202$, $\Delta_1 = 10$, $\sigma = 1$, and $\lambda = 5$. We vary $\Delta_2 = D \Delta_1$ as factor of $\Delta_1$, with $D \in \{2, 3, 5\}$. For all $D$ the same set of seeds was used.

\begin{table}[ht]
\centering
\begin{tabular}{|l|cccc|cccc|}
\hline
  &  \multicolumn{4}{c|}{$D = 5$} & \multicolumn{4}{c|}{$D = 3$}\\
  \hline
Method & $\hat{K} < K$ & $\hat{K} = K$ & $\hat{K} > K$ & $\operatorname{MISE}$ & $\hat{K} < K$ & $\hat{K} = K$ & $\hat{K} > K$ & $\operatorname{MISE}$ \\ 
  \hline
$2$-fold $\operatorname{CV}_{(1)}$ &     0 & 86.98 & 13.02 & 0.02159 &     0 & 86.98 & 13.02 & 0.02159 \\ 
  $5$-fold $\operatorname{CV}_{(1)}$ &     0 & 79.39 & 20.61 & 0.02747 &     0 & 79.39 & 20.61 & 0.02747 \\ 
  $10$-fold $\operatorname{CV}_{(1)}$ &     0 & 73.25 & 26.75 & 0.03249 &     0 & 73.25 & 26.75 & 0.03249 \\ 
  $20$-fold $\operatorname{CV}_{(1)}$ &     0 & 67.72 & 32.28 & 0.0374 &     0 & 67.72 & 32.28 & 0.0374 \\  
  $\operatorname{CV}_{(1)}$ &     0 & 87.18 & 12.82 & 0.02183 &     0 & 87.18 & 12.82 & 0.02183 \\ 
  $\operatorname{CV}_{\operatorname{mod}}$ &     0 & 92.03 &  7.97 & 0.0191 &     0 & 92.03 &  7.97 & 0.0191 \\ 
    $\operatorname{COPSS}$ & 99.95 &     0 &  0.05 & 2.361 & 91.28 &  0.06 &  8.66 & 2.166 \\
  $\operatorname{LooVF}_2$ &   100 &     0 &     0 & 2.362 & 99.87 &  0.08 &  0.05 & 2.359 \\ 
  $\operatorname{LooVF}_5$ &   100 &     0 &     0 & 2.362 & 99.81 &  0.09 &   0.1 & 2.358 \\ 
  $\operatorname{PELT}$ &     0 & 91.34 &  8.66 & 0.02118 &     0 & 91.34 &  8.66 & 0.02118 \\ 
  $\operatorname{WBS}$ &     0 & 87.38 & 12.62 & 0.0204 &     0 & 87.38 & 12.62 & 0.0204 \\ 
  $\operatorname{FDRSeg}$ &     0 & 82.19 & 17.81 & 0.02409 &     0 & 82.19 & 17.81 & 0.02409 \\ 
  $\operatorname{Ms.FPOP}$ &     0 & 98.89 &  1.11 & 0.01549 &     0 & 98.89 &  1.11 & 0.01549 \\ 
  $\operatorname{Biweight}$ &     0 & 93.21 &  6.79 & 0.05512 &     0 & 93.21 &  6.79 & 0.03418 \\  
   \hline
\end{tabular}
\caption{Simulation results relating to Example~\ref{example:underestmation}. $\operatorname{LooVF}_2$ and $\operatorname{LooVF}_5$ were applied to the observations in reverse order.}
\label{tab:underestimationExampleI}
\end{table}

\begin{table}[ht]
\centering
\begin{tabular}{|l|cccc|}
\hline
 & \multicolumn{4}{c|}{$D = 2$}\\
  \hline
Method & $\hat{K} < K$ & $\hat{K} = K$ & $\hat{K} > K$ & $\operatorname{MISE}$ \\ 
  \hline
$2$-fold $\operatorname{CV}_{(1)}$ &     0 & 86.98 & 13.02 & 0.02159 \\ 
  $5$-fold $\operatorname{CV}_{(1)}$ &     0 & 79.39 & 20.61 & 0.02747 \\ 
  $10$-fold $\operatorname{CV}_{(1)}$ &     0 & 73.25 & 26.75 & 0.03249 \\ 
  $20$-fold $\operatorname{CV}_{(1)}$ &     0 & 67.72 & 32.28 & 0.0374 \\ 
  $\operatorname{CV}_{(1)}$ &     0 & 87.18 & 12.82 & 0.02183 \\ 
  $\operatorname{CV}_{\operatorname{mod}}$ &     0 & 92.03 &  7.97 & 0.0191 \\ 
    $\operatorname{COPSS}$ &  0.47 & 69.69 & 29.84 & 0.04654 \\ 
  $\operatorname{LooVF}_2$ & 50.14 & 40.25 &  9.61 & 1.196 \\ 
  $\operatorname{LooVF}_5$ &  9.43 & 63.72 & 26.85 & 0.252 \\ 
  $\operatorname{PELT}$ &     0 & 91.34 &  8.66 & 0.02118 \\ 
  $\operatorname{WBS}$ &     0 & 87.38 & 12.62 & 0.0204 \\ 
  $\operatorname{FDRSeg}$ &     0 & 82.19 & 17.81 & 0.02409 \\ 
  $\operatorname{Ms.FPOP}$ &     0 & 98.89 &  1.11 & 0.01549 \\ 
  $\operatorname{Biweight}$ &     0 & 93.21 &  6.79 & 0.02723 \\  
   \hline
\end{tabular}
\caption{Simulation results relating to Example~\ref{example:underestmation}. $\operatorname{LooVF}_2$ and $\operatorname{LooVF}_5$ were applied to the observations in reverse order.}
\label{tab:underestimationExampleII}
\end{table}

The results in Tables~\ref{tab:underestimationExampleI}~and~\ref{tab:underestimationExampleII} support our theoretical findings from Section~\ref{sec:underestimation} that cross-validation with least squares loss ($\operatorname{COPSS}$ and $\operatorname{LooVF}_v$) underestimates the number of change-points when $2\Delta_2 > \underline{\lambda}\Delta_1$ in Example~\ref{example:underestmation} (see Theorem~\ref{theorem:underestimation}). Moreover these methods also have very large $\operatorname{MISE}$. Note that the results above for $\operatorname{LooVF}$ were obtained by applying the method to the observations in reverse order, which as explained in Example~\ref{example:underestmation}, is where the issue of $\operatorname{LooVF}$ occurs here.
%We think that this modification is fair, since Example~\ref{example:underestmation} and these simulations were specifically designed to highlight the issues of cross-validation with squared error loss, but this requires small variations of the setting for different methods. Applying $\operatorname{LooVF}$ to the observations in normal order gives results that are slightly worse than the ones for our proposed methods due to a higher likelihood for overestimation as observed in other simulations as well.
Other approaches do not underestimate the number of change-points and our new cross-validation approaches (particularly $\operatorname{CV}_{\operatorname{mod}}$ here) are competitive with these. With a moderate number of folds, our new cross-validation approaches are competitive with classical change-point approaches.

\subsection{Overestimation example}\label{sec:simulationOverestimation}
Consider observations as in Example \ref{example:overestmation}. We choose $n = 202$, $\Delta_1 = 1$ and vary $\sigma$. Moreover, in one simulation the change-point will be at $\tau_1 = n / 2 + 1$, an even location.

\begin{table}[ht]
\centering
\begin{tabular}{|l|cccc|cccc|}
\hline
  &  \multicolumn{4}{c|}{$\sigma = 1$} & \multicolumn{4}{c|}{$\sigma = 0.1$}\\
  \hline
Method & $\hat{K} < K$ & $\hat{K} = K$ & $\hat{K} > K$ & $\operatorname{MISE}$ & $\hat{K} < K$ & $\hat{K} = K$ & $\hat{K} > K$ & $\operatorname{MISE}$ \\ 
  \hline
$2$-fold $\operatorname{CV}_{(1)}$ & 38.54 & 53.03 &  8.43 & 0.194 &     0 & 88.13 & 11.87 & 0.0001636 \\ 
  $5$-fold $\operatorname{CV}_{(1)}$ & 25.21 & 58.53 & 16.26 & 0.1931 &     0 &  79.6 &  20.4 & 0.0002293 \\ 
  $10$-fold $\operatorname{CV}_{(1)}$ & 21.37 & 56.51 & 22.12 & 0.2055 &     0 & 73.69 & 26.31 & 0.0002821 \\ 
  $20$-fold $\operatorname{CV}_{(1)}$ & 17.45 & 54.51 & 28.04 & 0.2182 &     0 & 67.94 & 32.06 & 0.0003258 \\  
  $\operatorname{CV}_{(1)}$ & 38.62 & 52.82 &  8.56 & 0.1949 &     0 & 86.57 & 13.43 & 0.0001711 \\ 
  $\operatorname{CV}_{\operatorname{mod}}$ & 37.53 & 56.56 &  5.91 & 0.1842 &     0 & 91.33 &  8.67 & 0.0001449 \\ 
    $\operatorname{COPSS}$ & 37.18 & 56.08 &  6.74 & 0.1865 &     0 & 74.72 & 25.28 & 0.0002274 \\
  $\operatorname{LooVF}_2$ & 40.96 & 51.61 &  7.43 & 0.1938 &     0 & 76.99 & 23.01 & 0.0002136 \\ 
  $\operatorname{LooVF}_5$ & 26.01 & 59.55 & 14.44 & 0.1869 &     0 & 73.04 & 26.96 & 0.0002593 \\ 
  $\operatorname{PELT}$ & 23.32 & 66.57 & 10.11 & 0.1765 &     0 &  91.5 &   8.5 & 0.0001619 \\ 
  $\operatorname{WBS}$ & 39.42 & 50.16 & 10.42 & 0.1988 &     0 & 86.32 & 13.68 & 0.0001598 \\ 
  $\operatorname{FDRSeg}$ & 30.42 & 64.59 &  4.99 & 0.1722 &     0 & 83.93 & 16.07 & 0.0001845 \\ 
  $\operatorname{Ms.FPOP}$ &  33.4 & 65.35 &  1.25 & 0.1625 &     0 &  98.9 &   1.1 & 0.0001059 \\ 
  $\operatorname{Biweight}$ & 25.14 & 67.37 &  7.49 & 0.1736 &     0 &    93 &     7 & 0.0001644 \\ 
   \hline
\end{tabular}
\caption{Simulation results relating to Example~\ref{example:overestmation}.}
\label{tab:overestimationExampleI} 
\end{table}

\begin{table}[ht]
\centering
\begin{tabular}{|l|cccc|cccc|}
\hline
 & \multicolumn{4}{c|}{$\sigma = 0.01$} &   \multicolumn{4}{c|}{$\sigma = 0.001$}\\
  \hline
Method & $\hat{K} < K$ & $\hat{K} = K$ & $\hat{K} > K$ & $\operatorname{MISE}$ & $\hat{K} < K$ & $\hat{K} = K$ & $\hat{K} > K$ & $\operatorname{MISE}$ \\ 
  \hline
$2$-fold $\operatorname{CV}_{(1)}$ &     0 & 87.54 & 12.46 & 1.649e-06 &     0 & 87.66 & 12.34 & 1.648e-08 \\ 
  $5$-fold $\operatorname{CV}_{(1)}$ &     0 & 80.75 & 19.25 & 2.227e-06 &     0 & 79.81 & 20.19 & 2.264e-08 \\ 
  $10$-fold $\operatorname{CV}_{(1)}$ &     0 & 74.41 & 25.59 & 2.753e-06 &     0 & 73.77 & 26.23 & 2.74e-08 \\ 
  $20$-fold $\operatorname{CV}_{(1)}$ &     0 & 68.21 & 31.79 & 3.225e-06 &     0 & 68.07 & 31.93 & 3.214e-08 \\ 
  $\operatorname{CV}_{(1)}$ &     0 & 86.53 & 13.47 & 1.72e-06 &     0 & 86.51 & 13.49 & 1.712e-08 \\ 
  $\operatorname{CV}_{\operatorname{mod}}$ &     0 & 91.33 &  8.67 & 1.453e-06 &     0 & 91.31 &  8.69 & 1.463e-08 \\ 
    $\operatorname{COPSS}$ &     0 & 32.34 & 67.66 & 1.227e-05 &     0 &  18.6 &  81.4 & 1.982e-07 \\ 
  $\operatorname{LooVF}_2$ &     0 & 39.78 & 60.22 & 6.111e-05 &     0 & 23.11 & 76.89 & 5.024e-05 \\ 
  $\operatorname{LooVF}_5$ &     0 & 40.16 & 59.84 & 0.009707 &     0 & 20.25 & 79.75 & 0.0451 \\ 
  $\operatorname{PELT}$ &     0 & 91.71 &  8.29 & 1.603e-06 &     0 & 91.12 &  8.88 & 1.669e-08 \\ 
  $\operatorname{WBS}$ &     0 & 86.39 & 13.61 & 1.561e-06 &     0 & 86.61 & 13.39 & 1.595e-08 \\ 
  $\operatorname{FDRSeg}$ &     0 & 87.88 & 12.12 & 1.659e-06 &     0 & 86.97 & 13.03 & 1.739e-08 \\ 
  $\operatorname{Ms.FPOP}$ &     0 & 98.85 &  1.15 & 1.074e-06 &     0 & 98.75 &  1.25 & 1.091e-08 \\ 
  $\operatorname{Biweight}$ &     0 & 93.21 &  6.79 & 1.735e-05 &     0 & 93.04 &  6.96 & 2.081e-05 \\ 
   \hline
\end{tabular}
\caption{Simulation results relating to Example~\ref{example:overestmation}.}
\label{tab:overestimationExampleII}
\end{table}

\begin{table}[ht]
\centering
\begin{tabular}{|l|cccc|cccc|}
\hline
  & \multicolumn{4}{c|}{$\sigma = 0.0001$} & \multicolumn{4}{c|}{$\sigma = 0.0001$, $\tau_1 = n / 2 + 1$}\\
  \hline
Method & $\hat{K} < K$ & $\hat{K} = K$ & $\hat{K} > K$ & $\operatorname{MISE}$ & $\hat{K} < K$ & $\hat{K} = K$ & $\hat{K} > K$ & $\operatorname{MISE}$ \\ 
  \hline
$2$-fold $\operatorname{CV}_{(1)}$ &     0 & 87.15 & 12.85 & 1.691e-10 &     0 & 86.74 & 13.26 & 1.726e-10 \\ 
  $5$-fold $\operatorname{CV}_{(1)}$ &     0 & 79.72 & 20.28 & 2.289e-10 &     0 &  79.9 &  20.1 & 2.294e-10 \\ 
  $10$-fold $\operatorname{CV}_{(1)}$ &     0 & 72.89 & 27.11 & 2.838e-10 &     0 & 73.92 & 26.08 & 2.81e-10 \\ 
  $20$-fold $\operatorname{CV}_{(1)}$ &     0 & 67.37 & 32.63 & 3.288e-10 &     0 & 68.21 & 31.79 & 3.237e-10 \\ 
  $\operatorname{CV}_{(1)}$ &     0 & 86.13 & 13.87 & 1.752e-10 &     0 & 88.22 & 11.78 & 1.674e-10 \\ 
  $\operatorname{CV}_{\operatorname{mod}}$ &     0 & 90.94 &  9.06 & 1.49e-10 &     0 & 91.17 &  8.83 & 1.48e-10 \\ 
    $\operatorname{COPSS}$ &     0 & 18.47 & 81.53 & 2.034e-09 &     0 &  91.5 &   8.5 & 1.522e-10 \\ 
  $\operatorname{LooVF}_2$ &     0 & 22.01 & 77.99 & 5e-05 &     0 & 22.64 & 77.36 & 4.951e-05 \\ 
  $\operatorname{LooVF}_5$ &     0 &  18.9 &  81.1 & 0.05105 &     0 & 19.95 & 80.05 & 0.05203 \\ 
  $\operatorname{PELT}$ &     0 & 91.29 &  8.71 & 1.659e-10 &     0 & 90.66 &  9.34 & 1.707e-10 \\ 
  $\operatorname{WBS}$ &     0 & 86.84 & 13.16 & 1.624e-10 &     0 & 86.18 & 13.82 & 1.637e-10 \\ 
  $\operatorname{FDRSeg}$ &     0 &  88.8 &  11.2 & 1.64e-10 &     0 & 88.63 & 11.37 & 1.651e-10 \\ 
  $\operatorname{Ms.FPOP}$ &     0 & 98.82 &  1.18 & 1.097e-10 &     0 & 98.92 &  1.08 & 1.092e-10 \\ 
  $\operatorname{Biweight}$ &     0 &    93 &     7 & 1.535e-05 &     0 & 92.59 &  7.41 & 8.911e-06 \\ 
   \hline
\end{tabular}
\caption{Simulation results relating to Example~\ref{example:overestmation}.}
\label{tab:overestimationExampleIII}
\end{table}

Tables~\ref{tab:overestimationExampleI}--\ref{tab:overestimationExampleIII} show that cross-validation with least squares loss is likely to overestimate the number of change-points in the setting of Example~\ref{example:overestmation} when the signal-to-noise ratio increases. In accordance with the explanations in Section~\ref{sec:L2loss}, this phenomenon does not occur when the change-point is at an even location. Cross-validation with our proposed criteria does not suffer from this deficiency and is insensitive to whether the change-point is at an odd or even location. Similarly to the previous set of results, the likelihood of overestimation increases with the number of folds, and classical change-point approaches have a slightly smaller tendency to overestimate than cross-validation.

\subsection{Detection power}\label{sec:detectionPower}
In this section we provide a more systematic study of detection power. To this end, we consider a signal with a single bump in the middle. We chose $n = 200$, $K = 2$, $\tau_1 = (n - \lambda) / 2$, $\tau_2 = (n + \lambda) / 2$, $\beta_0 = \beta_2 = 0$ and $\beta_1 = \Delta_1 = -\Delta_2 = \delta$, and $\sigma = 1$. We vary $\lambda$ and $\delta$ to change the length and size of the bump.

\begin{table}[ht]
\centering
\begin{tabular}{|l|cccc|cccc|}
\hline
  &  \multicolumn{4}{c|}{$\lambda = 6$, $\delta = 2$} & \multicolumn{4}{c|}{$\lambda = 6$, $\delta = 3$}\\
  \hline
Method & $\hat{K} < K$ & $\hat{K} = K$ & $\hat{K} > K$ & $\operatorname{MISE}$ & $\hat{K} < K$ & $\hat{K} = K$ & $\hat{K} > K$ & $\operatorname{MISE}$ \\ 
  \hline
$2$-fold $\operatorname{CV}_{(1)}$ & 57.91 & 32.98 &  9.11 & 0.09507 & 11.29 & 75.34 & 13.37 & 0.06415 \\ 
  $5$-fold $\operatorname{CV}_{(1)}$ & 31.82 & 50.91 & 17.27 & 0.07822 &  2.09 &  76.7 & 21.21 & 0.04884 \\ 
  $10$-fold $\operatorname{CV}_{(1)}$ & 26.91 & 50.92 & 22.17 & 0.07852 &  1.62 &  71.4 & 26.98 & 0.05183 \\ 
  $20$-fold $\operatorname{CV}_{(1)}$ & 24.35 & 48.35 &  27.3 & 0.08044 &  1.49 & 66.44 & 32.07 & 0.05528 \\  
  $\operatorname{CV}_{(1)}$ & 67.33 &  25.7 &  6.97 & 0.1033 & 27.37 & 62.29 & 10.34 & 0.103 \\ 
  $\operatorname{CV}_{\operatorname{mod}}$ & 53.08 & 37.46 &  9.46 & 0.08971 &  6.86 & 80.83 & 12.31 & 0.05266 \\ 
    $\operatorname{COPSS}$ & 67.59 & 26.38 &  6.03 & 0.1029 & 26.68 & 62.75 & 10.57 & 0.1017 \\
  $\operatorname{LooVF}_2$ & 61.92 &    29 &  9.08 & 0.09794 & 18.61 & 66.11 & 15.28 & 0.08379 \\ 
  $\operatorname{LooVF}_5$ & 33.03 & 48.53 & 18.44 & 0.07943 &  2.26 & 73.44 &  24.3 & 0.05007 \\ 
  $\operatorname{PELT}$ & 30.06 & 62.79 &  7.15 & 0.07085 &  0.56 &  91.1 &  8.34 & 0.03865 \\ 
  $\operatorname{WBS}$ &  22.9 & 65.67 & 11.43 & 0.06609 &  0.28 & 87.44 & 12.28 & 0.03777 \\ 
  $\operatorname{FDRSeg}$ & 56.55 & 35.87 &  7.58 & 0.09902 &  4.87 & 84.75 & 10.38 & 0.05035 \\ 
  $\operatorname{Ms.FPOP}$ &  60.7 & 38.56 &  0.74 & 0.09193 &  3.97 & 94.97 &  1.06 & 0.04143 \\ 
  $\operatorname{Biweight}$ & 38.47 & 56.35 &  5.18 & 0.07644 &  1.78 & 91.56 &  6.66 & 0.04118 \\ 
   \hline
\end{tabular}
\caption{Detection of a single bump of length $\lambda$ and size $\delta$.}
\label{tab:detectionPowerI} 
\end{table}

\begin{table}[ht]
\centering
\begin{tabular}{|l|cccc|cccc|}
\hline
  &  \multicolumn{4}{c|}{$\lambda = 6$, $\delta = 4$} & \multicolumn{4}{c|}{$\lambda = 8$, $\delta = 2$}\\
  \hline
Method & $\hat{K} < K$ & $\hat{K} = K$ & $\hat{K} > K$ & $\operatorname{MISE}$ & $\hat{K} < K$ & $\hat{K} = K$ & $\hat{K} > K$ & $\operatorname{MISE}$ \\ 
  \hline
$2$-fold $\operatorname{CV}_{(1)}$ &  1.11 & 85.19 &  13.7 & 0.03563 & 36.39 & 51.07 & 12.54 & 0.09098 \\ 
  $5$-fold $\operatorname{CV}_{(1)}$ &  0.05 & 78.73 & 21.22 & 0.0366 & 14.44 & 65.34 & 20.22 & 0.07011 \\ 
  $10$-fold $\operatorname{CV}_{(1)}$ &  0.05 & 73.53 & 26.42 & 0.04118 & 11.18 & 63.67 & 25.15 & 0.07032 \\ 
  $20$-fold $\operatorname{CV}_{(1)}$ &  0.02 & 67.78 &  32.2 & 0.0454 &    10 & 59.21 & 30.79 & 0.07334 \\ 
  $\operatorname{CV}_{(1)}$ &  9.21 & 77.75 & 13.04 & 0.07238 & 30.77 & 53.74 & 15.49 & 0.08562 \\ 
  $\operatorname{CV}_{\operatorname{mod}}$ &  0.33 & 89.91 &  9.76 & 0.03021 & 31.48 &  56.9 & 11.62 & 0.08361 \\ 
   $\operatorname{COPSS}$ &  9.56 & 75.71 & 14.73 & 0.07406 &  28.7 & 57.08 & 14.22 & 0.08195 \\ 
  $\operatorname{LooVF}_2$ &  3.96 & 78.78 & 17.26 & 0.05039 & 33.15 & 53.26 & 13.59 & 0.08703 \\ 
  $\operatorname{LooVF}_5$ &  0.06 & 74.79 & 25.15 & 0.03826 & 15.48 & 64.13 & 20.39 & 0.07086 \\ 
  $\operatorname{PELT}$ &     0 & 91.43 &  8.57 & 0.02975 & 12.12 & 79.95 &  7.93 & 0.06073 \\ 
  $\operatorname{WBS}$ &     0 & 88.42 & 11.58 & 0.02889 &  8.78 & 79.34 & 11.88 & 0.05714 \\ 
  $\operatorname{FDRSeg}$ &  0.01 & 87.06 & 12.93 & 0.03281 & 35.91 & 56.87 &  7.22 & 0.09241 \\ 
  $\operatorname{Ms.FPOP}$ &     0 & 99.01 &  0.99 & 0.0244 & 34.08 & 65.01 &  0.91 & 0.08182 \\ 
  $\operatorname{Biweight}$ &  0.02 &  93.1 &  6.88 & 0.02935 & 16.75 &  77.1 &  6.15 & 0.06554 \\ 
   \hline
\end{tabular}
\caption{Detection of a single bump of length $\lambda$ and size $\delta$.}
\label{tab:detectionPowerII} 
\end{table}

\begin{table}[ht]
\centering
\begin{tabular}{|l|cccc|cccc|}
\hline
  &  \multicolumn{4}{c|}{$\lambda = 12$, $\delta = 2$} & \multicolumn{4}{c|}{$\lambda = 20$, $\delta = 2$}\\
  \hline
Method & $\hat{K} < K$ & $\hat{K} = K$ & $\hat{K} > K$ & $\operatorname{MISE}$ & $\hat{K} < K$ & $\hat{K} = K$ & $\hat{K} > K$ & $\operatorname{MISE}$ \\ 
  \hline
$2$-fold $\operatorname{CV}_{(1)}$ & 13.29 &  71.2 & 15.51 & 0.07345 &  1.23 & 83.32 & 15.45 & 0.05142 \\ 
  $5$-fold $\operatorname{CV}_{(1)}$ &  2.68 & 76.03 & 21.29 & 0.05779 &  0.06 & 78.33 & 21.61 & 0.05305 \\ 
  $10$-fold $\operatorname{CV}_{(1)}$ &  1.89 & 71.16 & 26.95 & 0.06026 &  0.05 & 73.14 & 26.81 & 0.05731 \\ 
  $20$-fold $\operatorname{CV}_{(1)}$ &  1.73 & 66.36 & 31.91 & 0.06397 &  0.02 & 67.75 & 32.23 & 0.06125 \\  
  $\operatorname{CV}_{(1)}$ & 10.61 & 72.83 & 16.56 & 0.0693 &  0.97 & 83.14 & 15.89 & 0.05101 \\ 
  $\operatorname{CV}_{\operatorname{mod}}$ &  9.75 & 77.59 & 12.66 & 0.06499 &  0.75 & 86.53 & 12.72 & 0.0479 \\ 
    $\operatorname{COPSS}$ &  8.92 & 75.85 & 15.23 & 0.06492 &  0.66 & 84.38 & 14.96 & 0.04919 \\
  $\operatorname{LooVF}_2$ & 10.92 & 71.37 & 17.71 & 0.06956 &  1.01 & 80.68 & 18.31 & 0.05152 \\ 
  $\operatorname{LooVF}_5$ &  2.61 & 74.58 & 22.81 & 0.05815 &  0.04 & 77.58 & 22.38 & 0.0525 \\ 
  $\operatorname{PELT}$ &  1.38 &  90.1 &  8.52 & 0.04915 &  0.01 & 90.75 &  9.24 & 0.04632 \\ 
  $\operatorname{WBS}$ &  1.01 & 86.82 & 12.17 & 0.04803 &  0.01 & 87.84 & 12.15 & 0.0455 \\ 
  $\operatorname{FDRSeg}$ &  6.03 & 81.59 & 12.38 & 0.06085 &  0.07 & 86.51 & 13.42 & 0.04913 \\ 
  $\operatorname{Ms.FPOP}$ &  6.59 & 92.32 &  1.09 & 0.05332 &  0.05 & 98.94 &  1.01 & 0.04055 \\ 
  $\operatorname{Biweight}$ &  2.01 &  91.1 &  6.89 & 0.05034 &  0.01 & 92.64 &  7.35 & 0.04613 \\ 
   \hline
\end{tabular}
\caption{Detection of a single bump of length $\lambda$ and size $\delta$.}
\label{tab:detectionPowerIII} 
\end{table}

We see from Tables~\ref{tab:detectionPowerI}--\ref{tab:detectionPowerIII} that the newly proposed criteria have only a slightly smaller detection power than the one based on least square loss. Note that the issues of least square loss do not have a significance influence in this setting. Detection power increases with the number of folds, but this also increases the likelihood of false positives. We find that overall, $5$-fold cross-validation with absolute error offers a good balance between the risks of under- or overestimate, and also performs well with respect to $\operatorname{MISE}$. Its performance is also quite competitive with classical change-point regression approaches.

\section{Proof of Theorem~\ref{theorem:detectionPrecision}}\label{sec:proof:detectionPrecision}
%\Rajen{TODO:
%\begin{itemize}
%	\item Ceiling / floor issues throughout
%	\item $C_n$ issue in Thm 3
%	\item Moving to sequence of half $\pm 1$ length in Thm 3. 
%	\item `Without loss of generality' in Thm 5
%	\item Last part of Thm 5 unclear.
%	\end{itemize}} all done
Our proof strategy is as follows.
We first show that with high probability, for each true change-point, there is an estimated change-point that is $\underline{\lambda} / 4$ close; see \eqref{eq:CPdetection}. Then working on this event, we show that any change-point set that violates the condition in Theorem~\ref{theorem:detectionPrecision} has larger costs than the true change-point set. To this end, we use Lemma~\ref{lemma:differenceCosts} to split the difference of these costs into several terms, which we use the following large deviations bounds to control. Throughout the proofs, we write $k\in [a, b]$ instead of $k \in \{\lceil a \rceil,\ldots,\lfloor b \rfloor\}$ when it is clear from the context that $k$ is an integer.

\begin{Lemma}[Lemma~4 on p. 33 in \citet{verzelen2020optimal}]\label{lemma:largeDeviationInterval}
	Let $\varepsilon_1,\varepsilon_2,\ldots$ be independent centred sub-Gaussian random variables with variance proxy $1$. Then, for any integer $d \geq 1$, any $\alpha > 0$ and any $x > 0$,
	\begin{equation*}
		\Pj\left( \max_{k\in [d,(1+\alpha)d]} \frac{\sum_{i=1}^{k}{\varepsilon_i}}{\sqrt{k}} \geq x \right)\leq \exp\left(-\frac{x^2}{2(1+\alpha)}\right).
	\end{equation*}
\end{Lemma}
\begin{Lemma}\label{lemma:largeDeviationLargeLeftIndex}
	Let $\varepsilon_1,\varepsilon_2,\ldots$ be independent centred sub-Gaussian random variables with variance proxy $1$. Then, for any integer $c \geq 1$ and any $x > 0$ such that $cx^2 > 4$,
	\begin{equation*}
		\Pj\left( \sup_{k \geq c} \frac{\sum_{i=1}^{k}{\varepsilon_i}}{k} \geq x \right)\leq 2\exp\left(-\frac{1}{4}cx^2\right).
	\end{equation*}
\end{Lemma}
\begin{proof}
	It follows from a union bound and Lemma \ref{lemma:largeDeviationInterval} that
	\begin{equation*}
		\begin{split}
			&\Pj\left( \sup_{k \geq c} \frac{\sum_{i=1}^{k}{\varepsilon_i}}{k} \geq x \right)
			=  \Pj\left( \sup_{s \in \mathbb{Z} : s\geq 0}\max_{k \in [c 2^s, c 2^{s+1}]} \frac{\sum_{i=1}^{k}{\varepsilon_i}}{k} \geq x \right)\\
			\leq & \sum_{s = 0}^{\infty}{\Pj\left(\max_{k \in [c 2^s, c 2^{s+1}]} \frac{\sum_{i=1}^{k}{\varepsilon_i}}{\sqrt{k}} \geq x \sqrt{c 2^s}\right)}\\
			\leq & \sum_{s = 0}^{\infty}{\exp\left(- \frac{1}{4} c x^2 2^{s}\right)}
			\leq  \sum_{s = 1}^{\infty}{\exp\left(- \frac{1}{4} c x^2 s\right)}
			=  \frac{\exp\left(- \frac{1}{4} c x^2\right)}{1 - \exp\left(-\frac{1}{4} c x^2\right)}
			\leq 2\exp\left(- \frac{1}{4} c x^2\right).
		\end{split}
	\end{equation*}
\end{proof}

\begin{Lemma}\label{lemma:largeDeviationShortInterval}
	Let $\varepsilon_1,\varepsilon_2,\ldots$ be independent centred sub-Gaussian random variables with variance proxy $1$. Then, for any integer $c \geq 1$ and any $x > 0$ such that $cx^2 \geq 4$,
	\begin{equation*}
		\Pj\left( \sup_{k \geq c,\, l \geq 1} \frac{\sum_{i = k + 1}^{k + l}{\varepsilon_i}}{\sqrt{l}\sqrt{k + l}} \geq x \right)\leq 2c^2\exp\left(-\frac{1}{2}cx^2\right).
	\end{equation*}
\end{Lemma}
\begin{proof}
We have that
\begin{equation*}
\Pj\left( \sup_{k \geq c,\, l \geq 1} \frac{\sum_{i = k + 1}^{k + l}{\varepsilon_i}}{\sqrt{l}\sqrt{k + l}} \geq x \right)
\leq \Pj\left( \sup_{k \geq c,\, l \in [1,c-1]} \frac{\sum_{i = k + 1}^{k + l}{\varepsilon_i}}{\sqrt{l}\sqrt{k + l}} \geq x \right)
+ \Pj\left( \sup_{k \geq c,\, l \geq c} \frac{\sum_{i = k + 1}^{k + l}{\varepsilon_i}}{\sqrt{l}\sqrt{k + l}} \geq x \right).
\end{equation*}

From a union bound, sub-Gaussianity and $cx^2 \geq 4$, we have that
\begin{equation*}
	\begin{split}
		\Pj\left( \sup_{k \geq c,\, l \in [1,c-1]} \frac{\sum_{i = k + 1}^{k + l}{\varepsilon_i}}{\sqrt{l}\sqrt{k + l}} \geq x \right)
		\leq & \sum_{l = 1}^{c-1} \sum_{s = 1}^{\infty}\sum_{k\in[cs, c(s+1))} \Pj\left( \frac{\sum_{i = k + 1}^{k + l}{\varepsilon_i}}{\sqrt{l}} \geq x \sqrt{cs+l} \right) \\
		\leq &c \sum_{s = 1}^{\infty} \sum_{l = 1}^{c-1} \exp\left( -\frac{1}{2}x^2(cs+l)\right) \\
		\leq &c^2 \frac{\exp\left(-\frac{1}{2}cx^2 \right)}{1 - \exp\left(-\frac{1}{2}cx^2 \right)} \\
		= & c^2 \frac{\exp\left(-\frac{1}{4}cx^2 \right)}{1 - \exp\left(-\frac{1}{2}cx^2 \right)}  \exp\left(-\frac{1}{4}cx^2 \right) \leq c^2 \exp\left(-\frac{1}{4}cx^2\right).
	\end{split}
\end{equation*}
It follows from a union bound, Lemma \ref{lemma:largeDeviationInterval} and the fact that $cx^2 \geq 4$, that
\begin{equation*}
\begin{split}
 \Pj\left( \sup_{k \geq c,\, l \geq c} \frac{\sum_{i = k + 1}^{k + l}{\varepsilon_i}}{\sqrt{l}\sqrt{k + l}} \geq x \right)
\leq & \sum_{s = 1}^{\infty}\sum_{t = 0}^{\infty}\sum_{k \in [s c, (s + 1) c)}{\Pj\left( \max_{l \in [c 2^t, c 2^{t + 1}]} \frac{\sum_{i = k + 1}^{k + l}{\varepsilon_i}}{\sqrt{l}} \geq x \sqrt{s c + 2^t c} \right)}\\
\leq & \sum_{s = 1}^{\infty}\sum_{t = 0}^{\infty}{c \exp\left(-\frac{1}{4}c x^2 (s + 2^t)\right)}\\
\leq & c\left(\sum_{s = 1}^{\infty}{ \exp\left(-\frac{1}{4}c x^2 s\right)}\right) \left( \sum_{t = 1}^{\infty}{ \exp\left(-\frac{1}{4}c x^2 t\right)} \right)\\
\leq & c \left(\frac{\exp\left(-\frac{1}{4}c x^2\right)}{1 - \exp\left(-\frac{1}{4}c x^2\right)}\right)^2
\leq  c \exp\left(-\frac{1}{4}c x^2 \right),
\end{split}
\end{equation*}
in the final line using the fact that $\{u/(1-u)\}^2 \leq u$ for $0\leq u \leq (3-\sqrt(5))/2$ and that $e^{-1} \leq (3-\sqrt(5))/2$.
Combing all inequalities shows the statement.
\end{proof}
For a set of putative change-points $\mathcal{U} = \{t_0 < t_1 < \cdots < t_L < t_{L + 1}\}$
%and any vector $Z = (Z_1,\ldots,Z_n) \in \R^n$, we use the notation
%\begin{equation*}
%S_Z(\mathcal{U}) := \sum_{l = 0}^{L}\sum_{i = t_l + 1}^{t_{l+1}}{\left(Z_i - \overline{Z}_{t_l :t_{l+1}}\right)^2}.
%\end{equation*}
%For ease of presentation,
we will write $\mathcal{U} \setminus t_1$ and $\mathcal{U} \cup t_2$ instead of $\mathcal{U} \setminus \{t_1\}$ and $\mathcal{U} \cup \{t_2\}$ throughout the paper, for ease of presentation.
\begin{Lemma}\label{lemma:differenceCosts}
Let $Y_i = \mu_i + \varepsilon_i$ and let $\mathcal{U} = \{0 = t_0 < t_1 < \cdots < t_{L} < t_{L + 1} = n\}$ be an arbitrary set of candidate change-points. Suppose integer $t$ is such that $t_{l - 1} < t < t_l$ and $\overline{\mu}_{t : t_l} = 0$. Then,
\begin{equation*}
\begin{split}
& S_Y\left( \mathcal{U} \right)  - S_Y\left( (\mathcal{U} \setminus t_l) \cup t \right)\\
= & \left(t_l - t\right)\left[
 \frac{t - t_{l - 1}}{t_l - t_{l - 1}} \overline{\mu}_{t_{l - 1} : t}^2 - 
 \frac{t_{l + 1} - t_l}{t_{l + 1} - t} \overline{\mu}_{t_l : t_{l + 1}}^2\right]\\
 & + 2\left(t_l - t\right)\left[
 \frac{t - t_{l - 1}}{t_l - t_{l - 1}} \overline{\mu}_{t_{l - 1} : t}\left(\overline{\varepsilon}_{t_{l - 1} : t} - \overline{\varepsilon}_{t : t_l}\right) + \frac{t_{l + 1} - t_l}{t_{l + 1} - t} \overline{\mu}_{t_l : t_{l + 1}}\left(\overline{\varepsilon}_{t : t_l} - \overline{\varepsilon}_{t_l : t_{l + 1}} \right)\right]\\
& + \left(t_l - t\right)\left[ \frac{\left(t_l - t\right) \overline{\varepsilon}_{t :t_l}^2 + 2 \left(t_{l + 1} - t_l\right) \overline{\varepsilon}_{t_l :t_{l + 1}} \overline{\varepsilon}_{t :t_l} - \left(t_{l + 1} - t_l\right)\overline{\varepsilon}_{t_l :t_{l + 1}}^2}{t_{l + 1} - t}\right.\\
& \hspace*{52pt} + \left. \frac{\left(t - t_{l - 1}\right) \overline{\varepsilon}_{t_{l - 1} :t}^2 - 2\left(t - t_{l - 1}\right) \overline{\varepsilon}_{t_{l - 1} :t}\overline{\varepsilon}_{t :t_l} - \left(t_l - t\right) \overline{\varepsilon}_{t :t_l}^2}{t_l - t_{l - 1}} 
 \right].
\end{split}
\end{equation*}
\end{Lemma}
The assumption $\overline{\mu}_{t : t_l} = 0$ may appear restrictive at first glance. But note that we can replace $Y_i$ with $Y_i - \overline{\mu}_{t : t_l}$, since it does not change the costs. Alternatively, one can replace every $\mu_i$ with $\mu_i - \overline{\mu}_{t : t_l}$ in the right hand side to obtain a more general lemma that does not require this assumption.
\begin{proof}[Proof of Lemma~\ref{lemma:differenceCosts}]
We have that
\begin{equation}\label{eq:costSplit}
\begin{split}
&\sum_{i = t_l + 1}^{t_{l+1}}{\left(Y_i - \overline{Y}_{t_l :t_{l + 1}}\right)^2}\\
=& \sum_{i = t_l + 1}^{t_{l+1}}{\left(\mu_i + \varepsilon_i - \overline{\mu}_{t_l :t_{l + 1}} - \overline{\varepsilon}_{t_l :t_{l + 1}}\right)^2} \\
=& \sum_{i = t_l + 1}^{t_{l+1}}{\left(\mu_i - \overline{\mu}_{t_l :t_{l + 1}}\right)^2}
+ 2\sum_{i = t_l + 1}^{t_{l+1}}{\left(\mu_i - \overline{\mu}_{t_l :t_{l + 1}} \right) \left(\varepsilon_i - \overline{\varepsilon}_{t_l :t_{l + 1}}\right)} 
+ \sum_{i = t_l + 1}^{t_{l+1}}{\left(\varepsilon_i - \overline{\varepsilon}_{t_l :t_{l + 1}}\right)^2}\\
=& \sum_{i = t_l + 1}^{t_{l+1}}{\mu_i\left(\mu_i - \overline{\mu}_{t_l :t_{l + 1}}\right)}
+ 2\sum_{i = t_l + 1}^{t_{l+1}}{\mu_i\left(\varepsilon_i - \overline{\varepsilon}_{t_l :t_{l + 1}} \right)} 
- \left(t_{l+1} - t_l\right)\overline{\varepsilon}_{t_l :t_{l + 1}}^2
+ \sum_{i = t_l + 1}^{t_{l+1}}{\varepsilon_i^2},
\end{split}
\end{equation}
since
\begin{equation*}
\sum_{i = t_l + 1}^{t_{l+1}}{\overline{\mu}_{t_l :t_{l + 1}}}\left(\mu_i - \overline{\mu}_{t_l :t_{l + 1}} \right)  = 0 = \sum_{i = t_l + 1}^{t_{l+1}}{\overline{\mu}_{t_l :t_{l + 1}} \left(\varepsilon_i - \overline{\varepsilon}_{t_l :t_{l + 1}} \right) }.
\end{equation*}

Thus,
\begin{equation*}
\begin{split}
& S_Y\left( \mathcal{U} \right)  - S_Y\left( (\mathcal{U} \setminus t_l) \cup t \right)\\
= & -\sum_{i =  t_{l - 1} + 1}^{t_l}{\mu_i \overline{\mu}_{t_{l - 1} :t_l}}
-\sum_{i =  t_l + 1}^{t_{l + 1}}{\mu_i \overline{\mu}_{t_l :t_{l + 1}}}
 + \sum_{i = t_{l - 1} + 1}^{t}{\mu_i \overline{\mu}_{t_{l - 1} :t}}
+ \sum_{i =  t + 1}^{t_{l + 1}}{\mu_i \overline{\mu}_{t :t_{l + 1}}}\\
& - 2\sum_{i =  t_{l - 1} + 1}^{t_l}{\mu_i \overline{\varepsilon}_{t_{l - 1} :t_l}}
- 2\sum_{i =  t_l + 1}^{t_{l + 1}}{\mu_i \overline{\varepsilon}_{t_l :t_{l + 1}}}
 + 2\sum_{i = t_{l - 1} + 1}^{t}{\mu_i \overline{\varepsilon}_{t_{l - 1} :t}}
+ 2\sum_{i =  t + 1}^{t_{l + 1}}{\mu_i \overline{\varepsilon}_{t :t_{l + 1}}}\\
& + \left(t - t_{l - 1}\right) \overline{\varepsilon}_{t_{l - 1} :t}^2
+ \left(t_{l + 1} - t\right) \overline{\varepsilon}_{t :t_{l + 1}}^2
- \left(t_l - t_{l - 1}\right) \overline{\varepsilon}_{t_{l - 1} :t_l}^2
- \left(t_{l + 1} - t_l\right) \overline{\varepsilon}_{t_l :t_{l + 1}}^2.
\end{split}
\end{equation*}

Since $\overline{\mu}_{t : t_l} = 0$, we have that $\left(t_l - t_{l - 1}\right) \overline{\mu}_{t_{l - 1} :t_l} = \left(t - t_{l - 1}\right) \overline{\mu}_{t_{l - 1} :t}$ and\\ $\left(t_{l + 1} - t\right) \overline{\mu}_{t :t_{l + 1}} = \left(t_{l + 1} - t_l\right) \overline{\mu}_{t_l :t_{l + 1}}$. Thus,
\begin{equation*}
\begin{split}
& -\sum_{i =  t_{l - 1} + 1}^{t_l}{\mu_i \overline{\mu}_{t_{l - 1} :t_l}}
-\sum_{i =  t_l + 1}^{t_{l + 1}}{\mu_i \overline{\mu}_{t_l :t_{l + 1}}}
 + \sum_{i = t_{l - 1} + 1}^{t}{\mu_i \overline{\mu}_{t_{l - 1} :t}}
+ \sum_{i =  t + 1}^{t_{l + 1}}{\mu_i \overline{\mu}_{t :t_{l + 1}}}\\
= & -\sum_{i =  t_{l - 1} + 1}^{t_l}{\mu_i \frac{\left(t - t_{l - 1}\right) \overline{\mu}_{t_{l - 1} :t}}{t_l - t_{l - 1}}}
 + \sum_{i = t_{l - 1} + 1}^{t}{\mu_i \overline{\mu}_{t_{l - 1} :t}}\\
& + \sum_{i =  t + 1}^{t_{l + 1}}{\mu_i \frac{\left(t_{l + 1} - t_l\right) \overline{\mu}_{t_l :t_{l + 1}}}{t_{l + 1} - t}}
- \sum_{i =  t_l + 1}^{t_{l + 1}}{\mu_i \overline{\mu}_{t_l :t_{l + 1}}}\\
= & \left(t_l - t\right)\left[
 \frac{t - t_{l - 1}}{t_l - t_{l - 1}} \overline{\mu}_{t_{l - 1} : t}^2 - 
 \frac{t_{l + 1} - t_l}{t_{l + 1} - t} \overline{\mu}_{t_l : t_{l + 1}}^2\right].
\end{split}
\end{equation*}

Since $\left(t_l - t_{l - 1}\right) \overline{\varepsilon}_{t_{l - 1} :t_l} = \left(t - t_{l - 1}\right) \overline{\varepsilon}_{t_{l - 1} :t} + \left(t_l - t\right) \overline{\varepsilon}_{t :t_l}$ and\\ $\left(t_{l + 1} - t\right) \overline{\varepsilon}_{t :t_{l + 1}} = \left(t_{l + 1} - t_l\right) \overline{\varepsilon}_{t_l :t_{l + 1}} + \left(t_l - t\right) \overline{\varepsilon}_{t :t_l}$, we have that
\begin{equation*}
\begin{split}
& - \sum_{i =  t_{l - 1} + 1}^{t_l}{\mu_i \overline{\varepsilon}_{t_{l - 1} :t_l}}
- \sum_{i =  t_l + 1}^{t_{l + 1}}{\mu_i \overline{\varepsilon}_{t_l :t_{l + 1}}}
 + \sum_{i = t_{l - 1} + 1}^{t}{\mu_i \overline{\varepsilon}_{t_{l - 1} :t}}
+ \sum_{i =  t + 1}^{t_{l + 1}}{\mu_i \overline{\varepsilon}_{t :t_{l + 1}}}\\
= & -\sum_{i =  t_{l - 1} + 1}^{t_l}{\mu_i \frac{\left(t - t_{l - 1}\right) \overline{\varepsilon}_{t_{l - 1} :t} + \left(t_l - t\right) \overline{\varepsilon}_{t :t_l}}{t_l - t_{l - 1}}}
 + \sum_{i = t_{l - 1} + 1}^{t}{\mu_i \overline{\varepsilon}_{t_{l - 1} :t}}\\
& + \sum_{i =  t + 1}^{t_{l + 1}}{\mu_i \frac{\left(t_{l + 1} - t_l\right) \overline{\varepsilon}_{t_l :t_{l + 1}} + \left(t_l - t\right) \overline{\varepsilon}_{t :t_l}}{t_{l + 1} - t}}
- \sum_{i =  t_l + 1}^{t_{l + 1}}{\mu_i \overline{\varepsilon}_{t_l :t_{l + 1}}}\\
= & \left(t_l - t\right)
 \left[ \frac{t - t_{l - 1}}{t_l - t_{l - 1}} \overline{\mu}_{t_{l - 1} :t}  \left(\overline{\varepsilon}_{t_{l - 1} :t} - \overline{\varepsilon}_{t :t_l}\right)  + \frac{t_{l + 1} - t_l}{t_{l + 1} - t} \overline{\mu}_{t_l :t_{l + 1}} \left(\overline{\varepsilon}_{t :t_l} - \overline{\varepsilon}_{t_l :t_{l + 1}}\right) \right],
\end{split}
\end{equation*}
where we have used the assumption $\overline{\mu}_{t : t_l} = 0$. Moreover, it follows from the same splitting of the errors that
\begin{equation*}
\begin{split}
& \left(t - t_{l - 1}\right) \overline{\varepsilon}_{t_{l - 1} :t}^2
+ \left(t_{l + 1} - t\right) \overline{\varepsilon}_{t :t_{l + 1}}^2
- \left(t_l - t_{l - 1}\right) \overline{\varepsilon}_{t_{l - 1} :t_l}^2
- \left(t_{l + 1} - t_l\right) \overline{\varepsilon}_{t_l :t_{l + 1}}^2\\
= &\frac{\left[ \left(t_{l + 1} - t_l\right) \overline{\varepsilon}_{t_l :t_{l + 1}} + \left(t_l - t\right) \overline{\varepsilon}_{t :t_l} \right]^2}{t_{l + 1} - t} - \left(t_{l + 1} - t_l\right) \overline{\varepsilon}_{t_l :t_{l + 1}}^2\\
&- \frac{\left[ \left(t - t_{l - 1}\right) \overline{\varepsilon}_{t_{l - 1} :t} + \left(t_l - t\right) \overline{\varepsilon}_{t :t_l} \right]^2}{t_l - t_{l - 1}} 
+ \left(t - t_{l - 1}\right) \overline{\varepsilon}_{t_{l - 1} :t}^2\\
= &\left(t_l - t\right)
\left[ \frac{\left(t_l - t\right) \overline{\varepsilon}_{t :t_l}^2 + 2 \left(t_{l + 1} - t_l\right) \overline{\varepsilon}_{t_l :t_{l + 1}} \overline{\varepsilon}_{t :t_l} - \left(t_{l + 1} - t_l\right)\overline{\varepsilon}_{t_l :t_{l + 1}}^2}{t_{l + 1} - t}\right.\\
& \hspace*{42pt} + \left. \frac{\left(t - t_{l - 1}\right) \overline{\varepsilon}_{t_{l - 1} :t}^2 - 2\left(t - t_{l - 1}\right) \overline{\varepsilon}_{t_{l - 1} :t}\overline{\varepsilon}_{t :t_l} - \left(t_l - t\right) \overline{\varepsilon}_{t :t_l}^2}{t_l - t_{l - 1}} 
 \right].
\end{split}
\end{equation*}
Combining all equalities completes the proof.
\end{proof}

\begin{proof}[Proof of Theorem~\ref{theorem:detectionPrecision}]
	\hfill
\subsubsection*{Proof of \eqref{detectionPrecision:L>=K}} We begin by showing that with high probability, for $L \geq K$ there is always an estimated change-point in each $\underline{\lambda} / 4$-neighbourhood of a true change, i.e.\ that the sequence of events
\begin{equation}\label{eq:CPdetection}
	\Omega_{1n} := \left\{\forall\, L \geq K\, \exists\, \hat{\tau}_{L,i_1},\ldots, \hat{\tau}_{L,i_{K}} \in \hat{\mathcal{T}}_L\, :\, \max_{1\leq k \leq K}{\left\vert \hat{\tau}_{L,i_k} - \tau_{k} \right\vert} \leq \frac{\underline{\lambda}}{4} \right\}
\end{equation}
has $\Pj(\Omega_{1n}) \to 1$.

\subsubsection*{Step 1: $\Pj(\Omega_{1n}) \to 1$.} We denote by $\hat{\mathcal{T}}_L\left( [a, b] \right)$ the set of $L\in \N$ change-points estimated by least squares estimation, with the restriction that no change-point is in $[a, b]$, so
\begin{equation*}
\hat{\mathcal{T}}_L\left( [a, b] \right) :=
\argmin_{\substack{0 =: t_0 < t_1 < \cdots < t_L < t_{L+1} := n,\\ t_1,\ldots,t_l \not\in [a, b]}}{\sum_{l = 0}^{L}{\sum_{i = t_l + 1}^{t_{l+1}}{\left(Y_i - \overline{Y}_{t_l :t_{l + 1}}\right)^2}}}.
\end{equation*}
We also use the notation $\hat{\mathcal{T}}_L^{-k} := \hat{\mathcal{T}}_L\left( \left[ \tau_k - \frac{\underline{\lambda}}{4}, \tau_k + \frac{\underline{\lambda}}{4}\right]\right)$. Further define $\hat{\mathcal{T}}^{\varepsilon}_{L}$ to be the set of $L$ change-points that in conjunction with $\mathcal{T}$ minimises $S_\varepsilon$, i.e.\ $\hat{\mathcal{T}}^{\varepsilon}_{L} := \argmin_{\mathcal{U}:|\mathcal{U}|=L} S_\varepsilon\left(\mathcal{U} \cup \mathcal{T}\right)$.

First note that the event
\begin{equation}
 \Big\{ S_Y\left( \hat{\mathcal{T}}_L^{-k} \right) > S_Y\left( \hat{\mathcal{T}}_L\right)\ \forall\ L \geq K,\ \forall\ k = 1,\ldots,K \Big\},
\end{equation}
is contained in $\Omega_{1n}$ \eqref{eq:CPdetection}; thus is suffices to show the above as probability converging to one. To this end,  observe that
\begin{equation}
	S_Y\left( \hat{\mathcal{T}}_L\right) \leq S_Y\left( \mathcal{T} \cup \hat{\mathcal{T}}^{\varepsilon}_{L - K} \right) = S_\varepsilon\left( \mathcal{T} \cup \hat{\mathcal{T}}^{\varepsilon}_{L - K} \right). \label{eq:S_Y_bd}
\end{equation}
Now from \eqref{eq:costSplit}, for any $k=0,\ldots,K$, and for any $a,b\in\N$ such that $\tau_k \leq a < b\leq \tau_{k+1}$, we have that
\begin{equation*}
\begin{split}
&\sum_{i = a + 1}^{b}{\left(Y_i - \overline{Y}_{a :\tau_{l + 1}}\right)^2}\\
=& \sum_{i = a + 1}^{b}{\left(\mu_i - \overline{\mu}_{a :\tau_{l + 1}}\right)^2}
+ 2\sum_{i = a + 1}^{b}{\left(\mu_i - \overline{\mu}_{a :\tau_{l + 1}} \right) \left(\varepsilon_i - \overline{\varepsilon}_{a :\tau_{l + 1}}\right)} 
+ \sum_{i = a + 1}^{b}{\left(\varepsilon_i - \overline{\varepsilon}_{a :\tau_{l + 1}}\right)^2}\\
=& \sum_{i = a + 1}^{b}{\left(\mu_i - \overline{\mu}_{a :\tau_{l + 1}}\right)^2}
+ 2\sum_{i = a + 1}^{b}{\left(\mu_i - \overline{\mu}_{a :\tau_{l + 1}} \right) \varepsilon_i} 
+ \sum_{i = a + 1}^{b}{\left(\varepsilon_i - \overline{\varepsilon}_{a :\tau_{l + 1}}\right)^2},
\end{split}
\end{equation*}
since $\sum_{i = a + 1}^{b}\left(\mu_i - \overline{\mu}_{a :\tau_{l + 1}} \right) = 0$. Thus, and since adding change-points can never increase the cost,
\begin{align}
 S_Y\bigg( \hat{\mathcal{T}}_L^{-k} \bigg)
 \geq & S_Y\bigg( \hat{\mathcal{T}}_L^{-k} \cup \bigcup_{\substack{j = 1,\ldots,K,\\ j \neq k}} \tau_j \cup \left\{ \tau_k - \Bfloor{\frac{\underline{\lambda}}{4}}, \tau_k + \Bfloor{\frac{\underline{\lambda}}{4}} \right\} \bigg) \notag\\
= & \sum_{i =  \tau_k - \floor{\frac{\underline{\lambda}}{4}} + 1}^{\tau_k + \floor{\frac{\underline{\lambda}}{4}}}{\left[\left(\mu_i - \overline{\mu}_{\tau_k - \floor{\frac{\underline{\lambda}}{4}} :\tau_k + \floor{\frac{\underline{\lambda}}{4}}}\right)^2 + 2\left(\mu_i - \overline{\mu}_{\tau_k - \floor{\frac{\underline{\lambda}}{4}} :\tau_k + \floor{\frac{\underline{\lambda}}{4}}}\right)\varepsilon_i\right]} \notag\\
& + S_\varepsilon\bigg( \hat{\mathcal{T}}_L^{-k} \cup \bigcup_{\substack{j = 1,\ldots,K,\\ j \neq k}} \tau_j \cup \left\{ \tau_k - \Bfloor{\frac{\underline{\lambda}}{4}}, \tau_k + \Bfloor{\frac{\underline{\lambda}}{4}} \right\} \bigg) \notag\\
\geq & 2\Bfloor{\frac{\underline{\lambda}}{4}}\left(\frac{\Delta_k}{2}\right)^2
- 2\operatorname{sgn}(\beta_k - \beta_{k-1}) \Bfloor{\frac{\underline{\lambda}}{4}}\frac{\Delta_k}{2}\left(\overline{\varepsilon}_{\tau_k - \floor{\frac{\underline{\lambda}}{4}}:\tau_k} - \overline{\varepsilon}_{\tau_k :\tau_k + \floor{\frac{\underline{\lambda}}{4}}} \right) \notag\\
& + S_\varepsilon\bigg( \hat{\mathcal{T}}_L^{-k} \cup \bigcup_{\substack{j = 1,\ldots,K,\\ j \neq k}} \tau_j \cup \left\{ \tau_k - \Bfloor{\frac{\underline{\lambda}}{4}}, \tau_k + \Bfloor{\frac{\underline{\lambda}}{4}} \right\} \bigg). \label{eq:T_K_bd1}
\end{align}
We now derive inequalities to bound the terms in \eqref{eq:T_K_bd1}. We will combine all of them in \eqref{eq:summaryboundcpsclose}, which will complete the proof of $\Pj(\Omega_{1n}) \to 1$. As a consequence of the sub-Gaussianity of $\varepsilon$,
\begin{equation}
\max_{k = 1,\ldots,K}\left\vert\overline{\varepsilon}_{\tau_k - \floor{\frac{\underline{\lambda}}{4}}:\tau_k} - \overline{\varepsilon}_{\tau_k :\tau_k + \floor{\frac{\underline{\lambda}}{4}}}\right\vert = \mathcal{O}_\Pj\left( \frac{\sigma \sqrt{\log(e K)}}{\sqrt{\underline{\lambda}}} \right). \label{eq:eps_bar_bd}
\end{equation}
Since removing restrictions and adding change-points can never increase the cost,
\begin{equation}
S_\varepsilon\bigg( \hat{\mathcal{T}}_L^{-k} \cup \bigcup_{\substack{j = 1,\ldots,K,\\ j \neq k}} \tau_j \cup \left\{ \tau_k - \Bfloor{\frac{\underline{\lambda}}{4}}, \tau_k + \Bfloor{\frac{\underline{\lambda}}{4}} \right\} \bigg)
\geq  S_\varepsilon\left( \mathcal{T} \cup \hat{\mathcal{T}}^{\varepsilon}_{L + 2} \right). \label{eq:add_change}
\end{equation}

We next aim to bound
\[\max_{L \geq K}\left\{S_\varepsilon\left( \mathcal{T} \cup \hat{\mathcal{T}}^{\varepsilon}_{L - K} \right) - S_\varepsilon\left( \mathcal{T} \cup \hat{\mathcal{T}}^{\varepsilon}_{L + 2} \right)\right\}. \]
To this end, we first argue that for any fixed $M \in \N$, we have
\begin{equation*}
S_\varepsilon\left( \mathcal{T} \cup \hat{\mathcal{T}}^{\varepsilon}_{M} \right) - S_\varepsilon\left( \mathcal{T} \cup \hat{\mathcal{T}}^{\varepsilon}_{M + 1} \right) 
	\leq 2\max_{1\leq t_1 < t_2 \leq n}{\left(t_2 - t_1\right)\overline{\varepsilon}_{t_1 :t_2}^2}.
\end{equation*}
Let $\tilde{\mathcal{T}}^{\varepsilon}_{M}$ be a set of $M$ change-points, where all $M$ change-points coincide with change-points in $\hat{\mathcal{T}}^{\varepsilon}_{M + 1}$, i.e. we have removed an arbitrary one. Let $0 = \hat{\tau}^M_0 < \hat{\tau}^M_1 < \cdots < \hat{\tau}^M_{K_M} < \hat{\tau}^M_{K_M + 1} = n$ be the change-points in $\mathcal{T} \cup \hat{\mathcal{T}}^{\varepsilon}_{M + 1}$. Then, there exists a $j$ such that $\hat{\tau}^M_0 < \cdots < \hat{\tau}^M_{j - 1} < \hat{\tau}^M_{j + 1} < \cdots < \hat{\tau}^M_{K_M + 1}$ are the change-points in $\mathcal{T} \cup \tilde{\mathcal{T}}^{\varepsilon}_{M}$. Finally, it follows from the fact that for any putative change-points $t_1, t_2$,
$\sum_{i = t_1 + 1}^{t_2}{\left(\varepsilon_i - \overline{\varepsilon}_{(t_1+1):t_{2}}\right)^2}
= \sum_{i = t_1 + 1}^{t_{2}}{\varepsilon_i^2} - \left(t_{2} - t_1\right)\overline{\varepsilon}_{t_1 :t_{2}}^2$ that
\begin{align*}
& S_\varepsilon\left( \mathcal{T} \cup \hat{\mathcal{T}}^{\varepsilon}_{M} \right) - S_\varepsilon\left( \mathcal{T} \cup \hat{\mathcal{T}}^{\varepsilon}_{M + 1} \right)\\
\leq & S_\varepsilon\left( \mathcal{T} \cup \tilde{\mathcal{T}}^{\varepsilon}_{M} \right) - S_\varepsilon\left( \mathcal{T} \cup \hat{\mathcal{T}}^{\varepsilon}_{M + 1} \right)\\
= & (\hat{\tau}^M_{j} - \hat{\tau}^M_{j - 1}) \overline{\varepsilon}_{(\hat{\tau}^M_{j - 1} + 1) : \hat{\tau}^M_{j}}^2 + (\hat{\tau}^M_{j + 1} - \hat{\tau}^M_{j}) \overline{\varepsilon}_{(\hat{\tau}^M_{j} + 1) : \hat{\tau}^M_{j + 1}}^2 - (\hat{\tau}^M_{j + 1} - \hat{\tau}^M_{j - 1}) \overline{\varepsilon}_{(\hat{\tau}^M_{j - 1} + 1) : \hat{\tau}^M_{j + 1}}^2\\
\leq & (\hat{\tau}^M_{j} - \hat{\tau}^M_{j - 1}) \overline{\varepsilon}_{(\hat{\tau}^M_{j - 1} + 1) : \hat{\tau}^M_{j}}^2 + (\hat{\tau}^M_{j + 1} - \hat{\tau}^M_{j}) \overline{\varepsilon}_{(\hat{\tau}^M_{j} + 1) : \hat{\tau}^M_{j + 1}}^2\\
\leq & 2\max_{1\leq t_1 < t_2 \leq n}{\left(t_2 - t_1\right)\overline{\varepsilon}_{t_1 :t_2}^2}.
\end{align*}
Thus
\begin{align}
	&\max_{L \geq K}\left\{S_\varepsilon\left( \mathcal{T} \cup \hat{\mathcal{T}}^{\varepsilon}_{L - K} \right) - S_\varepsilon\left( \mathcal{T} \cup \hat{\mathcal{T}}^{\varepsilon}_{L + 2} \right)\right\} \notag\\
	\leq & 2(K + 2) \max_{1\leq t_1 < t_2 \leq n}{\left(t_2 - t_1\right)\overline{\varepsilon}_{t_1 :t_2}^2} \notag\\
	= & \mathcal{O}_\Pj\left(K \sigma^2 \log(n) \right) = \mathcal{O}_\Pj\left(K \sigma^2 \log(K \overline{\lambda}) \right) = \mathcal{O}_\Pj\left(K \sigma^2 \log(\overline{\lambda}) \right), \label{eq:T_K_bd2}
\end{align}
appealing to sub-Gaussianity of the errors in the final line.

%Fixing $L \geq K$, we have that $S_\varepsilon\left( \mathcal{T} \cup \hat{\mathcal{T}}^{\varepsilon}_{L - K} \right) \leq S_\varepsilon\left( \hat{\mathcal{T}}^{\varepsilon}_{L - K} \right)$, since removing change-points only increases the costs. We also have that $S_\varepsilon\left( \mathcal{T} \cup \hat{\mathcal{T}}^{\varepsilon}_{L + 2} \right) \geq S_\varepsilon\left( \hat{\mathcal{T}}^{\varepsilon}_{L + K + 2} \right)$, since removing restrictions only decreases the costs. Hence, it follows from the arguments above for the first inequality, applying \eqref{eq:costsdifferenceonechangepoint} iteratively for the second inequality, sub-Gaussianity for the first equality and $K < \overline{\lambda}$ for the last inequality that
%\begin{align}
%&\max_{L \geq K}\left\{S_\varepsilon\left( \mathcal{T} \cup \hat{\mathcal{T}}^{\varepsilon}_{L - K} \right) - S_\varepsilon\left( \mathcal{T} \cup \hat{\mathcal{T}}^{\varepsilon}_{L + 2} \right)\right\} \notag\\
%\leq &\max_{L \geq K}\left\{S_\varepsilon\left( \hat{\mathcal{T}}^{\varepsilon}_{L - K} \right) - S_\varepsilon\left( \hat{\mathcal{T}}^{\varepsilon}_{L + K + 2} \right)\right\} \notag\\
%\leq & 2(2K + 2) \max_{1\leq t_1 < t_2 \leq n}{\left(t_2 - t_1\right)\overline{\varepsilon}_{t_1 :t_2}^2} \notag\\
%= & \mathcal{O}_\Pj\left(K \sigma^2 \log(n) \right) = \mathcal{O}_\Pj\left(K \sigma^2 \log(K \overline{\lambda}) \right) = \mathcal{O}_\Pj\left(K \sigma^2 \log(\overline{\lambda}) \right). \label{eq:T_K_bd2}
%\end{align}

Hence combining \eqref{eq:S_Y_bd}-\eqref{eq:T_K_bd2} yields
\begin{equation}\label{eq:summaryboundcpsclose}
\begin{split}
& \min_{L > K}\min_{k = 1,\ldots,K} \left\{S_Y\left( \hat{\mathcal{T}}_L^{-k} \right) - S_Y\left( \hat{\mathcal{T}}_L\right)\right\}\\
\geq & 2\Bfloor{\frac{\underline{\lambda}}{4}}\left(\frac{\Delta_{(1)}}{2}\right)^2
+ \mathcal{O}_\Pj\left(\sqrt{\underline{\lambda} \Delta_{(1)}^2 \log(eK)} \sigma  \right) + \mathcal{O}_\Pj\left(K \sigma^2 \log(\overline{\lambda}) \right),
\end{split}
\end{equation}
where the second inequality holds for large $n$ because of $\eqref{eq:detectableCPs}$. Moreover, the right-hand side is positive with probability converging to 1 due to \eqref{eq:detectableCPs}. Thus $\Pj(\Omega_{1n}) \to 1$ as claimed.
%
%\begin{equation*}
%\begin{split}
%&\Pj\left(\forall\, L \geq K\, \exists\, \hat{\tau}^O_{L,i_1},\ldots, \hat{\tau}^O_{L,i_{K}} \in \hat{\mathcal{T}}_L\, :\, \max_{1\leq k \leq K}{\left\vert \hat{\tau}^O_{L,i_k} - \tau_{k} \right\vert} \leq \frac{\underline{\lambda}}{4} \right)\\
%\geq & \Pj\left(S_Y\left( \hat{\mathcal{T}}_L^{-k} \right) > S_Y\left( \hat{\mathcal{T}}_L\right)\ \forall\ L \geq K,\ \forall\ k = 1,\ldots,K \right) \to 1,
%\end{split}
%\end{equation*}
%as $n \to \infty$, which shows $\Pj(\Omega_{1n}) \to 1$.

\subsubsection*{Step 2: Bounding cost differences.}
We will now introduce classes of change-point sets $\mathfrak{T}_L(\cdot)$ that violate the statement for $\hat{\mathcal{T}}_L$ in Theorem~\ref{theorem:detectionPrecision}. We then compute the difference of the costs of such change-point sets and of $\hat{\mathcal{T}}_L$. To complete the proof, it will remain to show that with probability converging to $1$, this difference is positive simultaneously for all $L$ and all change-point sets.

Let
\begin{equation*}
c_{L, k} := 
\begin{cases}
\log(eK)  \frac{\sigma^2}{\Delta_k^2}, &\text{ if }  L = K,\\
\log\left(eK\frac{\sigma^2}{\Delta_k^2} \right) \frac{\sigma^2}{\Delta_k^2}, &\text{ if }  L \neq K.
\end{cases}
\end{equation*}
Then, by definition $\max_{L,k}c_{L, k}\gamma_{L,k}^{-1} = o\left(1\right)$. Now let $C_n$ be a sequence with $C_n \to \infty$, but 
\begin{equation}\label{eq:conditionOnC}
C_n\frac{\sigma^2 \log(\overline{\lambda})}{\underline{\lambda}\Delta_{(1)}^2} \to 0,\quad \text{as } n \to \infty,
\end{equation}
which is possible because of \eqref{eq:detectableCPs}. Then, for fixed $n$ and $L \geq K$, for any $\mathcal{I} \subseteq \{1,\ldots,K\}$, let $\mathfrak{T}_L(\mathcal{I})$ be the class of change-point sets $\mathcal{U}$ of size $L$ such that the following two properties hold:
\begin{gather}
	\text{for each } k=1,\ldots,K, \,\, \exists \ \tau \in \mathcal{U} : |\tau - \tau_k| \leq \underline{\lambda}/4, \notag\\
\{k : |\tau_k - \tau| > C_nc_{L,k} \,\, \forall \tau \in \mathcal{U}\} = \mathcal{I}. \label{eq:definitionCounterExampleCpSet}
\end{gather}
We note that $\mathcal{I}$ denotes the indices of those change-points which are not well estimated. Conversely, $\{1,\ldots,K\} \setminus \mathcal{I}$ are the indices that satisfy the result in Theorem~\ref{theorem:detectionPrecision}. Hence defining
\begin{equation}\label{eq:eventToShow}
\Omega_{2n} := \big\{ S_Y\left( \mathcal{U} \right)  > S_Y\left( \hat{\mathcal{T}}_L\right)\ \forall\ \mathcal{U} \in \mathfrak{T}_L(\mathcal{I}),\ \forall\ \mathcal{I} \subseteq \{1,\ldots,K\},\ \mathcal{I} \neq \emptyset, \text{ and } \forall\ L \geq K \big\},
\end{equation}
we have that $\Omega_{1n}\cap \Omega_{2n}$ 
implies the event in Theorem~\ref{theorem:detectionPrecision}. Indeed, suppose that $\hat{\mathcal{I}}_L := \{k : |\tau_k - \hat{\tau}| > C_nc_{L,k} \,\, \forall \hat{\tau} \in \hat{\mathcal{T}}_L\}$ is non-empty for some $L \geq K$. Then on $\Omega_{1n}$ we have $\hat{\mathcal{T}}_L \in \mathfrak{T}_L(\hat{\mathcal{I}}_L)$, implying that $\Omega_{2n}$ has not occurred. Thus it suffices to show $\Pj(\Omega_{2n})\to 1$, as $n\to\infty$. To this end, we will bound $S_Y\left( \mathcal{U} \right) - S_Y\big( \hat{\mathcal{T}}_L\big)$ next. More precisely, we will bound the cost difference iteratively, starting with the largest change-point.
%Additionally, note that $\Pj(\Omega_{2n}) \geq \Pj(\Omega_{2n} \cup \Omega_{1n}) - \Pj(\Omega_{1n}^C)$, but $\Pj(\Omega_{1n}^C)\to 0$. Hence, we can restrict us further to those $\mathcal{U}$ that satisfy the constraint in $\Omega_{1n}$.

Let $n$ and $L \geq K$ be fixed. Take a non-empty $\mathcal{I} = \{i_1,\ldots,i_{\tilde{K}}\} \subseteq \{1,\ldots,K\}$, and $i_1,\ldots,i_{\tilde{K}}$ ordered such that $\Delta_{i_k}\geq \Delta_{i_l} \Leftrightarrow i_k < i_l$. Further let $\mathcal{U} = \{0 = \tilde{\tau}_0 < \tilde{\tau}_1 < \cdots \tilde{\tau}_L < \tilde{\tau}_{L + 1} = n\} \in \mathfrak{T}_L(\mathcal{I})$ be fixed.
%and such that constraint in $\Omega_{1n}$ \eqref{eq:CPdetection} is satisfied, i.e. $\exists\, \{j_1,\ldots,j_K\}\subset \{1,\ldots,L\}\, :\, \max_{1\leq k \leq K}{\left\vert \tilde{\tau}_{j_k} - \tau_{k} \right\vert} \leq \underline{\lambda}/4$ and $j_k = j_l \Rightarrow k = l$ (no change-point is matched twice).
 For each $k \in \mathcal{I}$, let $j_k$ be such that $\tilde{\tau}_{j_k} \in \mathcal{U}$ is the closest change-point in $\mathcal{U}$ to $\tau_{i_k}$ (in the case of a tie we pick the change-point on the right-hand side). Note that the $j_k$ for $k \in \mathcal{I}$ are all unique elements of $\{1,\ldots,L\}$ due to the definition of $\mathcal{U}$.
 
 By the definition of $\hat{\mathcal{T}}_L$, we have
\begin{equation*}
\begin{split}
&S_Y\left( \mathcal{U} \right) - S_Y\left( \hat{\mathcal{T}}_L\right)\\
\geq & S_Y\left( \mathcal{U} \right) - 
S_Y\left( \left[ \mathcal{U} \setminus  \bigcup_{l = 1}^{\tilde{K}}{\tilde{\tau}_{j_l}} \right] \cup  \bigcup_{l = 1}^{\tilde{K}}{\tau_{i_l}} \right)\\
= & \sum_{k = 1}^{\tilde{K}}{\left[ S_Y\left( \left[ \mathcal{U} \setminus  \bigcup_{l = 1}^{k - 1}{\tilde{\tau}_{j_l}} \right] \cup  \bigcup_{l = 1}^{k - 1}{\tau_{i_l}} \right) - S_Y\left( \left[ \mathcal{U} \setminus  \bigcup_{l = 1}^{k}{\tilde{\tau}_{j_l}} \right] \cup  \bigcup_{l = 1}^{k}{\tau_{i_l}} \right)   \right]}.
\end{split}
\end{equation*}
The above telescoping sum replaces poorly estimated changepoints with their true versions one at a time, in the order of their jump sizes.
We now show that with probability converging to $1$, each summand, and hence the sum, is positive, simultaneously for all $L$, $\mathcal{I}$ and $\mathcal{U} \in \mathfrak{T}_L(\mathcal{I})$.

For any fixed $\tilde{k} \in \{1,\ldots,\tilde{K}\}$, write
\begin{equation}\label{eq:definitionTildeU}
\tilde{\mathcal{U}} := \left( \mathcal{U} \setminus  \bigcup_{m = 1}^{\tilde{k} - 1}{\tilde{\tau}_{j_m}} \right) \cup  \bigcup_{m = 1}^{\tilde{k} - 1}{\tau_{i_m}} .
\end{equation}
Further write $l:=j_{\tilde{k}}$ and $k:=i_{\tilde{k}}$. We distinguish the cases $\tau_{k} < \tilde{\tau}_{l}$ and $\tau_k > \tilde{\tau}_{l}$. We will detail the calculation for $\tau_k < \tilde{\tau}_{l}$. The case $\tau_k > \tilde{\tau}_{l}$ follows similarly. Also, since replacing every $\mu_i$ by $\mu_i - c$, with $c$ a global constant leaves the costs unchanged, we can assume $\overline{\mu}_{\tau_k : \tilde{\tau}_{l}} = 0$ for the calculations below without loss of generality. Indeed, each subsequent instance of $\mu_i$ can be interpreted as the original mean with $\overline{\mu}_{\tau_k : \tilde{\tau}_{l}}$ subtracted. Hence, it follows from Lemma~\ref{lemma:differenceCosts} that
\begin{equation}\label{eq:differenceCosts}
\begin{split}
& S_Y\left( \tilde{\mathcal{U}} \right)  - S_Y\left( (\tilde{\mathcal{U}} \setminus \tilde{\tau}_{l}) \cup \tau_k \right)\\
= & \left(\tilde{\tau}_{l} - \tau_k\right)\left[
 \frac{\tau_k - \tilde{\tau}_{l - 1}}{\tilde{\tau}_{l} - \tilde{\tau}_{l - 1}} \overline{\mu}_{\tilde{\tau}_{l - 1} : \tau_k}^2 - 
 \frac{\tilde{\tau}_{l + 1} - \tilde{\tau}_{l}}{\tilde{\tau}_{l + 1} - \tau_k} \overline{\mu}_{\tilde{\tau}_{l} : \tilde{\tau}_{l + 1}}^2\right]\\
 & + 2\left(\tilde{\tau}_{l} - \tau_k\right)\left[
 \frac{\tau_k - \tilde{\tau}_{l - 1}}{\tilde{\tau}_{l} - \tilde{\tau}_{l - 1}} \overline{\mu}_{\tilde{\tau}_{l - 1} : \tau_k}\left(\overline{\varepsilon}_{\tilde{\tau}_{l - 1} : \tau_k} - \overline{\varepsilon}_{\tau_k : \tilde{\tau}_{l}}\right)\right.\\
 & \hspace*{64pt} + \left. \frac{\tilde{\tau}_{l + 1} - \tilde{\tau}_{l}}{\tilde{\tau}_{l + 1} - \tau_k} \overline{\mu}_{\tilde{\tau}_{l} : \tilde{\tau}_{l + 1}}\left(\overline{\varepsilon}_{\tau_k : \tilde{\tau}_{l}} - \overline{\varepsilon}_{\tilde{\tau}_{l} : \tilde{\tau}_{l + 1}} \right)\right]\\
& + \left(\tilde{\tau}_{l} - \tau_k\right)
\left[ \frac{\left(\tilde{\tau}_{l} - \tau_k\right) \overline{\varepsilon}_{\tau_k :\tilde{\tau}_{l}}^2 + 2 \left(\tilde{\tau}_{l + 1} - \tilde{\tau}_{l}\right) \overline{\varepsilon}_{\tilde{\tau}_{l} :\tilde{\tau}_{l + 1}} \overline{\varepsilon}_{\tau_k :\tilde{\tau}_{l}} - \left(\tilde{\tau}_{l + 1} - \tilde{\tau}_{l}\right)\overline{\varepsilon}_{\tilde{\tau}_{l} :\tilde{\tau}_{l + 1}}^2}{\tilde{\tau}_{l + 1} - \tau_k}\right.\\
& \hspace*{58pt} + \left. \frac{\left(\tau_k - \tilde{\tau}_{l - 1}\right) \overline{\varepsilon}_{\tilde{\tau}_{l - 1} :\tau_k}^2 - 2\left(\tau_k - \tilde{\tau}_{l - 1}\right) \overline{\varepsilon}_{\tilde{\tau}_{l - 1} :\tau_k}\overline{\varepsilon}_{\tau_k :\tilde{\tau}_{l}} - \left(\tilde{\tau}_{l} - \tau_k\right) \overline{\varepsilon}_{\tau_k :\tilde{\tau}_{l}}^2}{\tilde{\tau}_{l} - \tilde{\tau}_{l - 1}} 
 \right].
\end{split}
\end{equation}
We now bound this cost difference. We first bound the mean parts and then the noise parts before we combine all results in \eqref{eq:positiveDifference}. For notational convenience we assume from now on that the $k$th jump at $\tau_k$ is downwards. Note that this does not change the final bound, since we bound absolute values. Thus, since $\overline{\mu}_{\tau_k : \tilde{\tau}_{l}} = 0 = \beta_k$ (as $\tilde{\tau}_l < \tau_{k+1}$), we have that $\beta_{k-1} = \Delta_k$, $\beta_{k-2} = \Delta_k \pm \Delta_{k - 1}$ and $\beta_{k + 1} = \pm \Delta_{k + 1}$. Recall that we consider $\tau_k < \tilde{\tau}_{l}$. 
Thus, if $\tilde{\tau}_{l + 1} \leq \tau_{k + 1}$, then $\overline{\mu}_{\tilde{\tau}_{l} : \tilde{\tau}_{l + 1}} = 0$. If $\tilde{\tau}_{l + 1} > \tau_{k + 1}$, then
\begin{equation*}
\overline{\mu}_{\tilde{\tau}_{l} : \tilde{\tau}_{l + 1}}\leq \frac{\tilde{\tau}_{l + 1} - \tau_{k + 1}}{\tilde{\tau}_{l + 1} - \tilde{\tau}_{l}} \Delta_{k + 1} \leq \frac{1}{4} \Delta_k,
\end{equation*}
since $\tilde{\tau}_{l + 1} - \tilde{\tau}_{l} \geq \tilde{\tau}_{l + 1} - \tau_{k + 1} + \frac{3}{4} \underline{\lambda}$ and $\tilde{\tau}_{l + 1} - \tau_{k + 1} \leq \frac{1}{4} \underline{\lambda}$. If $\Delta_{k + 1} > \Delta_k$, we have also that $\tilde{\tau}_{l + 1} - \tau_{k + 1} \leq C_n c_{L,k + 1}$, since we have ordered the changes in $\mathcal{I}$ according to their size and hence a larger misestimated change-point would have been replaced earlier in the telescoping sum in \eqref{eq:definitionTildeU}. Hence, it follows from \eqref{eq:detectableCPs} and \eqref{eq:conditionOnC} that
\begin{equation*}
\begin{split}
\frac{\overline{\mu}_{\tilde{\tau}_{l} : \tilde{\tau}_{l + 1}}}{\Delta_k}
\leq & \frac{\tilde{\tau}_{l + 1} - \tau_{k + 1}}{\tilde{\tau}_{l + 1} - \tilde{\tau}_{l}} \frac{\Delta_{k + 1}}{\Delta_k}
\leq \frac{C_n c_{L,k+1} \Delta_{k + 1}}{\frac{3}{4} \underline{\lambda}\Delta_k}\\
\leq & C_n\frac{(\log(K) \vee 1) \sigma^2}{\frac{3}{4} \underline{\lambda}\Delta_{(1)}^2}
+ C_n\frac{\sigma^2 \log(\overline{\lambda})}{\frac{3}{4} \underline{\lambda}\Delta_{(1)}^2} \frac{\log\left(\frac{\sigma^2}{\Delta_{(1)}^2} \vee 1\right)}{\log(\overline{\lambda})} \to 0,
\end{split}
\end{equation*}
as $n\to\infty$. Thus for all $n$ sufficiently large, we have $\overline{\mu}_{\tilde{\tau}_{l} : \tilde{\tau}_{l + 1}} \leq \frac{1}{4} \Delta_k$ uniformly in $k$ and $\mathcal{U}$.

Using similar arguments, we also have $\overline{\mu}_{\tilde{\tau}_{l - 1} : \tau_k} \geq \frac{4}{5} \Delta_k$ for $n$ sufficiently large, uniformly in $k$ and $\mathcal{U}$. Hence for such $n$,
\begin{equation}\label{eq:boundMu}
\frac{\tau_k - \tilde{\tau}_{l - 1}}{\tilde{\tau}_{l} - \tilde{\tau}_{l - 1}} \overline{\mu}_{\tilde{\tau}_{l - 1} : \tau_k}^2 - 
 \frac{\tilde{\tau}_{l + 1} - \tilde{\tau}_{l}}{\tilde{\tau}_{l + 1} - \tau_k} \overline{\mu}_{\tilde{\tau}_{l} : \tilde{\tau}_{l + 1}}^2
 \geq \left(\frac{1}{2}\frac{16}{25} - 1 \frac{1}{16} \right) \Delta_k^2 > \frac{1}{4} \Delta_k^2,
\end{equation}
since $\tau_k - \tilde{\tau}_{l - 1} \geq \tilde{\tau}_{l} - \tau_k$ as  we have chosen $l$ such that $\tilde{\tau}_{l}\in \mathcal{U}$ is the closest change-point in $\mathcal{U}$ to $\tau_k$.

We now argue that the averages of the noise terms in \eqref{eq:differenceCosts} are small with high probability. To this end, we introduce, for a fixed $\tilde{C} > 0$ to be chosen, the collection of events
\begin{equation}\label{eq:conditionsOnMeanEpsilon}
	\Omega_{<}(\mathcal{U}, k) := \left\{
	\vert\overline{\varepsilon}_{\tau_k : \tilde{\tau}_{l}}\vert \leq \tilde{C} \Delta_k,\
	\vert\overline{\varepsilon}_{\tilde{\tau}_{l - 1} : \tau_k}\vert \leq \tilde{C} \Delta_k \text{ and }
	\left\vert  \frac{\tilde{\tau}_{l + 1} - \tilde{\tau}_{l}}{\tilde{\tau}_{l + 1} - \tau_k} \overline{\varepsilon}_{\tilde{\tau}_{l} : \tilde{\tau}_{l + 1}}^2\right\vert \leq \tilde{C}^2 \Delta_k^2
	\right\}.
\end{equation}
On this event, it follows from similar arguments to that used in \eqref{eq:boundMu} and above, that
\begin{equation}\label{eq:boundMixedLeft}
\begin{split}
\left \vert \frac{\tau_k - \tilde{\tau}_{l - 1}}{\tilde{\tau}_{l} - \tilde{\tau}_{l - 1}} \overline{\mu}_{\tilde{\tau}_{l - 1} : \tau_k}\left(\overline{\varepsilon}_{\tilde{\tau}_{l - 1} : \tau_k} - \overline{\varepsilon}_{\tau_k : \tilde{\tau}_{l}}\right) \right\vert
\leq \frac{6}{5} \tilde{C} \Delta_k^2,
\end{split}
\end{equation}
\begin{equation}\label{eq:boundMixedRight}
\begin{split}
\left\vert \frac{\tilde{\tau}_{l + 1} - \tilde{\tau}_{l}}{\tilde{\tau}_{l + 1} - \tau_k} \overline{\mu}_{\tilde{\tau}_{l} : \tilde{\tau}_{l + 1}}\left(\overline{\varepsilon}_{\tau_k : \tilde{\tau}_{l}} - \overline{\varepsilon}_{\tilde{\tau}_{l} : \tilde{\tau}_{l + 1}} \right)\right\vert
\leq \frac{1}{4} \tilde{C} \Delta_k^2,
\end{split}
\end{equation}
and that
\begin{equation}\label{eq:boundEpsilon}
\begin{split}
&\left\vert\frac{\left(\tilde{\tau}_{l} - \tau_k\right) \overline{\varepsilon}_{\tau_k :\tilde{\tau}_{l}}^2 + 2 \left(\tilde{\tau}_{l + 1} - \tilde{\tau}_{l}\right) \overline{\varepsilon}_{\tilde{\tau}_{l} :\tilde{\tau}_{l + 1}} \overline{\varepsilon}_{\tau_k :\tilde{\tau}_{l}} - \left(\tilde{\tau}_{l + 1} - \tilde{\tau}_{l}\right)\overline{\varepsilon}_{\tilde{\tau}_{l} :\tilde{\tau}_{l + 1}}^2}{\tilde{\tau}_{l + 1} - \tau_k}\right.\\
& \left. +  \frac{\left(\tau_k - \tilde{\tau}_{l - 1}\right) \overline{\varepsilon}_{\tilde{\tau}_{l - 1} :\tau_k}^2 - 2\left(\tau_k - \tilde{\tau}_{l - 1}\right) \overline{\varepsilon}_{\tilde{\tau}_{l - 1} :\tau_k}\overline{\varepsilon}_{\tau_k :\tilde{\tau}_{l}} - \left(\tilde{\tau}_{l} - \tau_k\right) \overline{\varepsilon}_{\tau_k :\tilde{\tau}_{l}}^2}{\tilde{\tau}_{l} - \tilde{\tau}_{l - 1}}\right\vert\\
\leq & 8 \tilde{C}^2 \Delta_k^2.
\end{split}
\end{equation}

If $\tau_k > \tilde{\tau}_{l}$, then we obtain similar bounds on the event
\begin{equation}\label{eq:conditionsOnMeanEpsilon>}
	\Omega_{>}(\mathcal{U}, k) := \left\{
	\vert\overline{\varepsilon}_{\tilde{\tau}_{l}:\tau_k}\vert \leq \tilde{C} \Delta_k,\
	\vert\overline{\varepsilon}_{\tau_k:\tilde{\tau}_{l + 1}}\vert \leq \tilde{C} \Delta_k \text{ and }
	\left\vert  \frac{\tilde{\tau}_{l} - \tilde{\tau}_{l - 1}}{\tau_k - \tilde{\tau}_{l - 1}} \overline{\varepsilon}_{\tilde{\tau}_{l-1} : \tilde{\tau}_{l}}^2\right\vert \leq \tilde{C}^2 \Delta_k^2
	\right\}.
\end{equation}

Thus, on the event
\begin{equation} \label{eq:Omega_U_def}
\Omega(\mathcal{U}, k):= \Omega_{<}(\mathcal{U}, k) \cap \Omega_{>}(\mathcal{U}, k)
\end{equation} with $\tilde{C} := \frac{1}{24}$, it follows from \eqref{eq:differenceCosts}, \eqref{eq:boundMu}, \eqref{eq:boundMixedLeft}, \eqref{eq:boundMixedRight}, and \eqref{eq:boundEpsilon} that
\begin{equation}\label{eq:positiveDifference}
\begin{split}
S_Y\left( \tilde{\mathcal{U}} \right)  - S_Y\left( (\tilde{\mathcal{U}} \setminus \tilde{\tau}_{j_{\tilde{k}}}) \cup \tau_{i_{\tilde{k}}} \right)
 > \left\vert\tilde{\tau}_{j_{\tilde{k}}} - \tau_{i_{\tilde{k}}}\right\vert\Delta_{i_{\tilde{k}}}^2\left(\frac{1}{4} - \frac{12}{5} \tilde{C} - \frac{1}{2} \tilde{C} - 8 \tilde{C}^2\right) > 0,
\end{split}
\end{equation}
with $\tilde{\mathcal{U}}$ as in \eqref{eq:definitionTildeU}.

%We have now showed that difference is positive on the events $\Omega(\mathcal{U}, k)$.
To complete the proof, we have to show that the $\Omega(\mathcal{U},k)$ hold simultaneously with high probability. To this end, let
\begin{equation}\label{eq:defOmega3n}
	\Omega_{3n} := \bigcap_{\substack{\mathcal{U} \in \mathfrak{T}_L(\mathcal{I}), \; k \in \mathcal{I}, \\ L \geq K, \; \emptyset \subsetneq \mathcal{I} \subseteq \{1,\ldots,K\}}} \Omega(\mathcal{U}, k).
\end{equation}
We have that $\Omega_{3n} \subseteq \Omega_{2n}$, so it suffices to show $\Pj(\Omega_{3n}) \to 1$. %To this end, let $\Omega_{3n, +}$ and $\Omega_{3n, -}$ be as $\Omega_{3n}$, but with the absolute values in the definition of $\Omega(\mathcal{U}, k)$ \eqref{eq:conditionsOnMeanEpsilon} replaced by positive and negative parts respectively. Then, $\Pj(\Omega_{3n}) \geq 1- (1 - \Pj(\Omega_{3n, +})) - (1 - \Pj(\Omega_{3n, -}))$. Furthermore, 
Furthermore, the definition of $\mathfrak{T}_L(\mathcal{I})$ \eqref{eq:definitionCounterExampleCpSet} implies that we only have to consider $\vert t - \tau_k\vert >  C_n c_{L,k}$. Then, it follows from a union bound and symmetry arguments that
\begin{equation*}
\begin{split}
\Pj(\Omega_{3n})
\geq  1 
& - \Pj\bigg(
%\max_{\substack{\mathcal{I} \subseteq \{1,\ldots,K\} \\ \mathcal{I} \neq \emptyset}}\max_{k \in \mathcal{I}}
\max_{k \in \{1,\ldots,K\}}
\max_{L \geq K}\max_{t - \tau_k > C_n c_{L, k}} (\overline{\varepsilon}_{\tau_k : t}  - \tilde{C} \Delta_k) > 0\bigg)\\
& - \Pj\bigg(
%\max_{\substack{\mathcal{I} \subseteq \{1,\ldots,K\} \\ \mathcal{I} \neq \emptyset}}\max_{k \in \mathcal{I}}
\max_{k \in \{1,\ldots,K\}}
\max_{L \geq K}\max_{t - \tau_k > C_n c_{L, k}} (\overline{\varepsilon}_{\tau_k : t}  + \tilde{C} \Delta_k) < 0\bigg)\\
& - \Pj\bigg(
%\max_{\substack{\mathcal{I} \subseteq \{1,\ldots,K\} \\ \mathcal{I} \neq \emptyset}}\max_{k \in \mathcal{I}}
\max_{k \in \{1,\ldots,K\}}
\max_{L \geq K}\max_{\tau_k - t > C_n c_{L, k}} (\overline{\varepsilon}_{t : \tau_k} - \tilde{C} \Delta_k) > 0\bigg)\\
& - \Pj\bigg(
%\max_{\substack{\mathcal{I} \subseteq \{1,\ldots,K\} \\ \mathcal{I} \neq \emptyset}}\max_{k \in \mathcal{I}}
\max_{k \in \{1,\ldots,K\}}
\max_{L \geq K}\max_{\tau_k - t > C_n c_{L, k}} (\overline{\varepsilon}_{t : \tau_k} + \tilde{C} \Delta_k) < 0\bigg)\\
& - \Pj\bigg(
%\max_{\substack{\mathcal{I} \subseteq \{1,\ldots,K\} \\ \mathcal{I} \neq \emptyset}}\max_{k \in \mathcal{I}}
\max_{k \in \{1,\ldots,K\}}
\max_{L > K}\max_{t_2 > t_1 > \tau_k + C_n c_{L, k}} \Big( \frac{t_2 - t_1}{t_2 - \tau_k} \overline{\varepsilon}_{t_1 : t_2}^2 - \tilde{C}^2 \Delta_k^2 \Big) > 0\bigg)\\
& - \Pj\bigg(
%\max_{\substack{\mathcal{I} \subseteq \{1,\ldots,K\} \\ \mathcal{I} \neq \emptyset}}\max_{k \in \mathcal{I}}
\max_{k \in \{1,\ldots,K\}}
\max_{L > K}\max_{t_2 < t_1 < \tau_k - C_n c_{L, k}} \Big( \frac{t_1 - t_2}{\tau_k - t_2} \overline{\varepsilon}_{t_2 : t_1}^2 - \tilde{C}^2 \Delta_k^2 \Big) > 0\bigg)\\
& - \Pj\bigg(
%\max_{\substack{\mathcal{I} \subseteq \{1,\ldots,K\} \\ \mathcal{I} \neq \emptyset}}\max_{k \in \mathcal{I}}
\max_{k \in \{1,\ldots,K\}}
\max_{\substack{t_2 > t_1 > \tau_k + C_n c_{K, k}, \\ \vert \tau_{k + 1} - t_2\vert \leq \underline{\lambda}/4}} \Big( \frac{t_2 - t_1}{t_2 - \tau_k} \overline{\varepsilon}_{t_1 : t_2}^2 - \tilde{C}^2 \Delta_k^2 \Big) > 0 \bigg)\\
& - \Pj\bigg(
%\max_{\substack{\mathcal{I} \subseteq \{1,\ldots,K\} \\ \mathcal{I} \neq \emptyset}}\max_{k \in \mathcal{I}}
\max_{k \in \{1,\ldots,K\}}
\max_{\substack{t_2 < t_1 < \tau_k - C_n c_{K, k}, \\ \vert \tau_{k - 1} - t_2\vert \leq \underline{\lambda}/4}} \Big( \frac{t_1 - t_2}{\tau_k - t_2} \overline{\varepsilon}_{t_2 : t_1}^2 - \tilde{C}^2 \Delta_k^2 \Big) > 0 \bigg)\\
\geq & 1 - 4\sum_{k = 1}^K\Pj\bigg(\max_{L \geq K}\max_{t - \tau_k > C_n c_{L, k}}\overline{\varepsilon}_{\tau_k : t} > \tilde{C} \Delta_k\bigg)\\
& - 2\sum_{k =1}^K\Pj\bigg(\max_{L > K}\max_{t_2 > t_1 > \tau_k + C_n c_{L, k}} \frac{t_2 - t_1}{t_2 - \tau_k} \overline{\varepsilon}_{t_1 : t_2}^2 > \tilde{C}^2 \Delta_k^2\bigg)\\
& - 2\Pj\bigg(\max_{k=1,\ldots,K} \max_{\substack{t_2 > t_1 > \tau_k + C_n c_{K, k}, \\ \vert \tau_{k + 1} - t_2\vert \leq \underline{\lambda}/4}} \frac{t_2 - t_1}{t_2 - \tau_k} \overline{\varepsilon}_{t_1 : t_2}^2 > \tilde{C}^2 \Delta_{(1)}^2\bigg).
\end{split}
\end{equation*}

Then, since $c_{L, k} \geq (\log(K) \vee 1)\sigma^2 / \Delta_k^2$ for all $L \geq K$, it follows from Lemma \ref{lemma:largeDeviationLargeLeftIndex} that
\begin{equation*}
\begin{split}
& \Pj\Big(\max_{L \geq K}\max_{t - \tau_k > C_n c_{L, k}}\overline{\varepsilon}_{\tau_k : t} > \tilde{C} \Delta_k\Big)\\
\leq & \Pj\Bigg( \max_{r \geq C_n\frac{(\log(K) \vee 1)\sigma^2}{\Delta_k^2}} \frac{\sum_{i=1}^{r}{\varepsilon_i / \sigma}}{r} \geq \tilde{C} \frac{\Delta_k}{\sigma} \Bigg)\\
\leq & 2\exp\left(-\frac{1}{4}C_n \tilde{C}^2 (\log(K) \vee 1)\right).
\end{split}
\end{equation*}

Next, since $c_{L, k} = c_{K + 1, k} \geq \log\left(eK \sigma^2 / \Delta_k^2 \right) \sigma^2 / \Delta_k^2$ for all $L>K$, it follows from Lemma \ref{lemma:largeDeviationShortInterval} that for $n$ large enough
\begin{equation*}
\begin{split}
& \Pj\left(\max_{L > K}\max_{t_2 > t_1 > \tau_k + C_n c_{L, k}} \frac{t_2 - t_1}{t_2 - \tau_k} \overline{\varepsilon}_{t_1 : t_2}^2 > \tilde{C}^2 \Delta_k^2\right)\\
\leq & \Pj\left( \max_{r \geq C_n c_{K + 1, k},\, s \geq 1} \frac{\sum_{j = r + 1}^{r + s}{\varepsilon_j / \sigma}}{\sqrt{s}\sqrt{r + s}} \geq \tilde{C}\frac{\Delta_k}{\sigma} \right)\\
\leq & 2 C_n^2 c_{K + 1, k}^2\exp\left(-\frac{1}{2}C_n c_{K + 1, k} \tilde{C}^2\frac{\Delta_k^2}{\sigma^2}\right)\\
\leq & 2 \exp\left(-\frac{1}{4}C_n \tilde{C}^2 (\log(K) \vee 1)\right).
\end{split}
\end{equation*}

Because of sub-Gaussianity and since $\vert \tau_{k + 1} - t_2\vert \leq \underline{\lambda}/4$ implies $\tau_k + \frac{3}{4}\underline{\lambda} \leq t_2 \leq \tau_k + \overline{\lambda} + \frac{1}{4}\underline{\lambda}$, we have that 
\begin{align*}
& \max_{k=1,\ldots, K} \max_{\substack{t_2 > t_1 > \tau_k + C_n c_{K, k},\\ \vert \tau_{k + 1} - t_2\vert \leq \underline{\lambda}/4}} \frac{t_2 - t_1}{t_2 - \tau_k} \overline{\varepsilon}_{t_1 : t_2}^2\\
\leq & \max_{k=1,\ldots, K}\max_{\tau_k \leq t_1 \leq t_2 \leq \tau_k + \overline{\lambda} + \frac{1}{4}\underline{\lambda}}
\frac{4}{3}\underline{\lambda}^{-1}{(t_2 - t_1) \overline{\varepsilon}_{t_1 :t_2}^2}
\leq  \mathcal{O}_\Pj \left(  \underline{\lambda}^{-1}\log( K \overline{\lambda}) \sigma^2 \right).
\end{align*}
Thus, it follows from \eqref{eq:detectableCPs} that
\begin{equation*}
\Pj\bigg(\max_{k=1,\ldots, K}\max_{\substack{t_2 > t_1 > \tau_k + C_n c_{K, k},\\ \vert \tau_{k + 1} - t_2\vert \leq \underline{\lambda}/4}} \frac{t_2 - t_1}{t_2 - \tau_k} \overline{\varepsilon}_{t_1 : t_2}^2 > \tilde{C}^2 \Delta_{(1)}^2\bigg) \to 0,\text{ as } n \to \infty.
\end{equation*}

Thus,
\begin{equation*}
\begin{split}
\Pj(\Omega_{3n})
\geq & 1 - 12 K \exp\left(-\frac{1}{4}C_n \tilde{C}^2 (\log(K) \vee 1)\right) + o(1)\\
\geq & 1  - 12\exp\left(-\frac{1}{4}C_n \tilde{C}^2\right) + o(1) \to 1,
\end{split}
\end{equation*}
since $C_n \to \infty$, as $n \to \infty$. This completes the proof of \eqref{detectionPrecision:L>=K}.

\subsubsection*{Proof of \eqref{detectionPrecision:L<K}}
If $\left\vert \hat{\tau} - \tau_{k} \right\vert > \frac{\underline{\lambda}}{4}\, \forall\, \hat{\tau} \in \hat{\mathcal{T}}_L$, then
\begin{equation*}
\sum_{i = \tau_{k} - \floor{\frac{\underline{\lambda}}{2}} + 1}^{\tau_{k} + \floor{\frac{\underline{\lambda}}{2}}}{(\mu_i - \overline{\mu}_{L, i})^2} \geq \min_{a \in \R} \sum_{i = \tau_{k} - \floor{\frac{\underline{\lambda}}{4}} + 1}^{\tau_{k} + \floor{\frac{\underline{\lambda}}{4}}}{(\mu_i - a)^2} \geq \Bfloor{\frac{\underline{\lambda}}{4}}\frac{\Delta_k^2}{4} \geq \frac{\underline{\lambda}\Delta_k^2}{200}.
\end{equation*}

Otherwise, let $l$ be such that $\hat{\tau}_{L,l}$ is the closest change-point to $\tau_{k}$. In case of a tie, we pick the change-point on the right-hand side. We consider the case where $\hat{\tau}_{L,l} \geq \tau_{k}$; if this is not true, we can instead consider the observations in reverse order. Additionally, since replacing every $\mu_i$ by $\mu_i - c$, with $c$ a global constant, does not change the costs, we can assume $\beta_k = \overline{\mu}_{\tau_k : \hat{\tau}_{L,l}} = 0$ as in our previous argument \eqref{eq:differenceCosts}. Finally, without loss of generality we may assume that the $k$th jump at $\tau_k$ is downwards, as otherwise we can consider $-Y_i$ instead of $Y_i$, which does not change the costs.

If $\overline{\mu}_{\hat{\tau}_{L,l - 1} : \tau_k} > \frac{6}{5}\Delta_k$ or $\overline{\mu}_{\hat{\tau}_{L,l - 1} : \tau_k} < \frac{4}{5} \Delta_k $, then $\hat{\tau}_{L,l - 1} < \tau_{k - 1}$, since otherwise $\overline{\mu}_{\hat{\tau}_{L,l - 1} : \tau_k} = \Delta_k$. Hence,
\begin{equation*}
\sum_{i = \tau_{k} - \floor{\frac{\underline{\lambda}}{2}} + 1}^{\tau_{k} + \floor{\frac{\underline{\lambda}}{2}}}{(\mu_i - \overline{\mu}_{L, i})^2} \geq \sum_{i = \tau_{k} - \floor{\frac{\underline{\lambda}}{4}} + 1}^{\tau_{k}}{(\mu_i - \overline{\mu}_{L, i})^2} \geq \Bfloor{\frac{\underline{\lambda}}{4}}\frac{\Delta_k^2}{25} \geq \frac{\underline{\lambda}\Delta_k^2}{200}.
\end{equation*}

If $\vert \overline{\mu}_{\hat{\tau}_{L,l} : \hat{\tau}_{L,l + 1}}\vert > \frac{1}{4} \Delta_k$, then $\hat{\tau}_{L,l + 1} > \tau_{k + 1}$, since otherwise $\overline{\mu}_{\hat{\tau}_{L,l} : \hat{\tau}_{L,l + 1}} = \beta_k = 0$. Hence,
\begin{equation*}
\sum_{i = \tau_{k} - \floor{\frac{\underline{\lambda}}{2}} + 1}^{\tau_{k} + \floor{\frac{\underline{\lambda}}{2}}}{(\mu_i - \overline{\mu}_{L, i})^2} \geq \sum_{i = \tau_{k} + \floor{\frac{\underline{\lambda}}{4}} + 1}^{\tau_{k} + \floor{\frac{\underline{\lambda}}{2}}}{(\mu_i - \overline{\mu}_{L, i})^2} \geq \Bfloor{\frac{\underline{\lambda}}{4}}\frac{\Delta_k^2}{16} \geq \frac{\underline{\lambda}\Delta_k^2}{200}.
\end{equation*}

But, if $\overline{\mu}_{\hat{\tau}_{L,l - 1} : \tau_k}\in \left[\frac{4}{5}\Delta_k, \frac{6}{5}\Delta_k\right]$ and $\vert \overline{\mu}_{\hat{\tau}_{L,l} : \hat{\tau}_{L,l + 1}}\vert \leq \frac{1}{4} \Delta_k$, then the statement in \eqref{detectionPrecision:L<K} follows from the same arguments as used to show \eqref{detectionPrecision:L>=K} but we only consider $k$ in the set
\begin{equation*}
\mathcal{K} := \Bigg\{1\leq k \leq K: \sum_{i = \tau_{k} - \floor{\frac{\underline{\lambda}}{2}} + 1}^{\tau_{k} + \floor{\frac{\underline{\lambda}}{2}}}{\big(\mu_i - \overline{\mu}_{L, i}\big)^2} < \frac{\underline{\lambda}\Delta_k^2}{200}\Bigg\}.
\end{equation*}
More precisely, we replace the definition of $\mathfrak{T}_L(\mathcal{I})$ in \eqref{eq:definitionCounterExampleCpSet} by the following definition. Let $\mathfrak{T}_L(\mathcal{I})$ be the class of change-point sets $\mathcal{U}$ of size $L$ such that the follow two properties hold:
\begin{gather*}
	\text{for each } k \in \mathcal{K}, \text{ there exists } \tau \in \mathcal{U} : |\tau - \tau_k| \leq \underline{\lambda}/4, \notag\\
\{k : |\tau_k - \tau| > C_nc_{L,k} \,\, \forall \tau \in \mathcal{U}\} = \mathcal{I}.
\end{gather*}
Then on the event $\Omega(\mathcal{U}, k)$  \eqref{eq:Omega_U_def}, we have that \eqref{eq:positiveDifference} holds, i.e. that
\begin{equation*}
S_Y\left( \tilde{\mathcal{U}} \right)  - S_Y\left( (\tilde{\mathcal{U}} \setminus \tilde{\tau}_{l}) \cup \tau_k \right)
 > 0,
\end{equation*}
since $\overline{\mu}_{\hat{\tau}_{L,l - 1} : \tau_k}\in \left[\frac{4}{5}\Delta_k, \frac{6}{5}\Delta_k\right]$ and $\vert \overline{\mu}_{\hat{\tau}_{L,l} : \hat{\tau}_{L,l + 1}}\vert \leq \frac{1}{4} \Delta_k$ and \eqref{eq:boundMixedLeft}-\eqref{eq:boundEpsilon} still hold. Finally, $\Pj(\Omega_{3n}) \to 1$, for $\Omega_{3n}$ as defined in \eqref{eq:defOmega3n}, follows from the same steps that we have outline after \eqref{eq:positiveDifference}. This completes the proof of \eqref{detectionPrecision:L<K}.
\end{proof}

\section{Proof of Theorem~\ref{theorem:underestimation}}\label{sec:proof:underestimation}
\begin{proof}[Proof of Theorem \ref{theorem:underestimation}]
%We will show that $\Pj\left(\operatorname{CV}_{(2)}(2) \leq \operatorname{CV}_{(2)}(1)\right) \to 0$, as $n\to\infty$, since this implies $\Pj\left(\hat{K} = 2\right) \to 0$. To this end, we first show with \eqref{eq:estimateCPunderestimationL1}~and~\eqref{eq:estimateCPunderestimationL2} statements how well we can estimate the change-points. Afterwards, we will bound $\operatorname{CV}_{(2)}(2) - \operatorname{CV}_{(2)}(1)$ condition on the intersection of events in these two statements to complete the proof. Note that having the change-points estimated accurately will simplify the cross-validation criteria and hence help us to bound the difference.

	We first show that, writing $\delta_0 := \lfloor \sqrt{\underline{\lambda}}\, \rfloor$,
\begin{equation}\label{eq:estimateCPunderestimationL2}
	\Pj\big( \vert \hat{\tau}^O_{2, 1} - \tau^O_1\vert \leq \delta_0,\ \vert \hat{\tau}^E_{2, 1} - \tau^E_1\vert \leq \delta_0,\ \hat{\tau}^O_{2, 2} = \tau^O_2,\ \hat{\tau}^E_{2, 2} = \tau^E_2\big) \to 1.
\end{equation}
and
\begin{equation}\label{eq:estimateCPunderestimationL1}
	\Pj\Big(\hat{\tau}^O_{1, 1} = \tau^O_2, \hat{\tau}^E_{1, 1} = \tau^E_2\Big) \to 1,
\end{equation}
so in particular, the second change-point is estimated precisely with high probability. Working on the intersection of events above will facilitate study of the cross-validation scores, from which we will conclude that $\Pj\left(\operatorname{CV}_{(2)}(2) \leq \operatorname{CV}_{(2)}(1)\right) \to 0$, implying the final result. The first of these statements \eqref{eq:estimateCPunderestimationL2} is given by Theorem~\ref{theorem:detectionPrecision}, since $\Delta_2 / \sigma \to \infty$, $\underline{\lambda}\Delta_1^2 / (\sigma^2 \log n) \to \infty$ and $\sqrt{\underline{\lambda}} \Delta_1^2 / \sigma^2 \to \infty$. For the second statement \eqref{eq:estimateCPunderestimationL1}, we proceed as follows. 

%We start by showing that 
%\begin{equation}\label{eq:estimateCPunderestimationL1}
%\Pj\Big(\hat{\tau}^O_{1, 1} = \tau^O_2, \hat{\tau}^E_{1, 1} = \tau^E_2\Big) \to 1.
%\end{equation}
\subsubsection*{Step 1:} We first show that
\begin{equation*}
	\big\{\hat{\tau}^O_{1, 1} = \tau^O_2\Big\}
	= \big\{S_Y\left( \left\{\tau^O_0, \tau, \tau^O_3 \right\} \right) > S_Y\left( \left\{\tau^O_0, \tau^O_2, \tau^O_3 \right\} \right)\ \forall\ \tau \neq \tau^O_2 \big\}=: \Gamma.
	\end{equation*}
satisfies $\Pj(\Gamma) \to 1$. Similar arguments (not shown) may be used to obtain $\Pj(\hat{\tau}^E_{1, 1} = \tau^E_2)\to 1$, which then together with the above gives \eqref{eq:estimateCPunderestimationL1}.
We have that
\begin{equation*}
\big\{\hat{\tau}^O_{1, 1} = \tau^O_2\Big\}
= \big\{S_Y\left( \left\{\tau^O_0, \tau, \tau^O_3 \right\} \right) > S_Y\left( \left\{\tau^O_0, \tau^O_2, \tau^O_3 \right\} \right)\ \forall\ \tau \neq \tau^O_2 \big\}.
\end{equation*}
We distinguish the cases $\tau > \tau^O_2$, $\tau \in [\tau^O_1,\tau^O_2)$ and $\tau < \tau^O_1$ and denote the corresponding events as $\Gamma_1$, $\Gamma_2$ and $\Gamma_3$, so $\Gamma = \cap_{j=1, 2, 3} \Gamma_j$. To this end, recall that $\tau^O_1 = \tau^E_1 = (n / 2 - \underline{\lambda}) / 2$, $\tau^O_2 = (n / 2 + 1) / 2$ and $\tau^E_2 = (n / 2 - 1) / 2$.
	
\emph{Case $\tau > \tau^O_2$}. We make use of Lemma~\ref{lemma:differenceCosts} in Section~\ref{sec:proof:detectionPrecision}; in order to use this we consider $-(Y_i - \Delta_2)$ instead of $Y_i$. Note that this does not change the costs. Hence, $\beta_0 = \Delta_2 - \Delta_1$, $\beta_1 = \Delta_2$ and $\beta_2 = 0$. Thus, $\overline{\mu}_{\tau^O_2 :\tau} = 0$, $\overline{\mu}_{\tau : \tau^O_3} = 0$ and $\overline{\mu}_{0:\tau^O_2} = \Delta_2 -\tau^O_1\Delta_1/\tau^O_2 = c \Delta_2$ for a $c\in(0,1)$, since $\Delta_1 < \Delta_2$. Consequently, Lemma~\ref{lemma:differenceCosts} gives us
\begin{equation*}
\begin{split}
& \big(\tau - \tau^O_2\big)^{-1} \Big[S_Y\Big( \big\{\tau^O_0, \tau, \tau^O_3 \big\} \Big) - S_Y\Big( \big\{\tau^O_0, \tau^O_2, \tau^O_3 \big\} \Big)\Big]\\
= & \frac{\tau^O_2}{\tau} \overline{\mu}_{0 : \tau^O_2}^2 + 2\frac{\tau^O_2}{\tau} \overline{\mu}_{0 : \tau^O_2}\big(\overline{\varepsilon}_{0 : \tau^O_2} - \overline{\varepsilon}_{\tau^O_2 : \tau}\big)\\
&+ \frac{\big(\tau - \tau^O_2\big) \overline{\varepsilon}_{\tau^O_2 :\tau}^2 + 2 \big(\tau^O_3 - \tau\big) \overline{\varepsilon}_{\tau :\tau^O_3} \overline{\varepsilon}_{\tau^O_2 :\tau} - \big(\tau^O_3 - \tau\big)\overline{\varepsilon}_{\tau :\tau^O_3}^2}{\tau^O_3 - \tau^O_2}\\
& + \frac{\tau^O_2 \overline{\varepsilon}_{0:\tau^O_2}^2 - 2\tau^O_2 \overline{\varepsilon}_{0:\tau^O_2}\overline{\varepsilon}_{\tau^O_2 :\tau} - \big(\tau - \tau^O_2\big) \overline{\varepsilon}_{\tau^O_2 :\tau}^2}{\tau}\\
\geq &\frac{1}{2} c^2 \Delta_2^2 + \mathcal{O}_\Pj\left(\Delta_2 \sigma \right) + \mathcal{O}_\Pj\big(\sigma^2\big),
\end{split}
\end{equation*}
since $\overline{\varepsilon}_{0 : \tau^O_2} = \mathcal{O}_\Pj\left(\sigma\right)$, $\overline{\varepsilon}_{\tau^O_2 :\tau} = \mathcal{O}_\Pj\left(\sigma\right)$, and $\overline{\varepsilon}_{\tau^O_2 :\tau} = \mathcal{O}_\Pj\left(\sigma\right)$ for all $\tau > \tau^O_2$ simultaneously. Then as \eqref{eq:minmax} and \eqref{eq:conditionUnderestimation} imply $\Delta_2^2 / \sigma^2 \to \infty$, we have that
\[
\frac{1}{\sigma^2}\Big(\tau - \tau^O_2\Big)^{-1} \Big[S_Y\Big( \big\{\tau^O_0, \tau, \tau^O_3 \big\} \Big) - S_Y\Big( \big\{\tau^O_0, \tau^O_2, \tau^O_3 \big\} \Big)\Big] \inprob \infty,
\]
showing in particular that $\Pj(\Gamma_1) \to 1$.

\emph{Case $\tau \in [\tau^O_1,\tau^O_2)$.}
Next consider $\tau \in [\tau^O_1,\tau^O_2)$. To use Lemma~\ref{lemma:differenceCosts}, we consider $Y^O_1,\ldots,Y^O_{ n / 2 }$ in reverse order. Note that this does not change the costs. Hence, $\beta_0 = \Delta_2$, $\beta_1 = 0$ and $\beta_2 = \Delta_1$. It follows $\overline{\mu}_{0:\tau^O_1} = \Delta_2$, $\overline{\mu}_{\tau^O_1 :\tau} = 0$ and $\overline{\mu}_{\tau:\tau^O_3} = \frac{\tau^O_3 - \tau^O_2}{\tau^O_3 - \tau}\Delta_1$. Consequently, Lemma~\ref{lemma:differenceCosts} gives us
\begin{equation}
\begin{split}
& \Big(\tau - \tau^O_1\Big)^{-1} \Big[S_Y\Big( \big\{\tau^O_0, \tau, \tau^O_3 \big\} \Big) - S_Y\Big( \big\{\tau^O_0, \tau^O_1, \tau^O_3 \big\} \big)\Big]\\
= & \frac{\tau^O_1}{\tau} \overline{\mu}_{0 : \tau^O_1}^2 - 
 \frac{\tau^O_3 - \tau}{\tau^O_3 - \tau^O_1} \overline{\mu}_{\tau : \tau^O_3}^2\\
 & + 2\frac{\tau^O_1}{\tau} \overline{\mu}_{0 : \tau^O_1}\Big(\overline{\varepsilon}_{0 : \tau^O_1} - \overline{\varepsilon}_{\tau^O_1 : \tau}\Big) + \frac{\tau^O_3 - \tau}{\tau^O_3 - \tau^O_1} \overline{\mu}_{\tau : \tau^O_3}\Big(\overline{\varepsilon}_{\tau^O_1 : \tau} - \overline{\varepsilon}_{\tau : \tau^O_3} \Big)\\
&+\frac{\big(\tau - \tau^O_1\big) \overline{\varepsilon}_{\tau^O_1 :\tau}^2 + 2 \big(\tau^O_3 - \tau\big) \overline{\varepsilon}_{\tau :\tau^O_3} \overline{\varepsilon}_{\tau^O_1 :\tau} - \big(\tau^O_3 - \tau\big)\overline{\varepsilon}_{\tau :\tau^O_3}^2}{\tau^O_3 - \tau^O_1}\\
& + \frac{\tau^O_1\overline{\varepsilon}_{0:\tau^O_1}^2 - 2\tau^O_1  \overline{\varepsilon}_{0:\tau^O_1}\overline{\varepsilon}_{\tau^O_1 :\tau} - \big(\tau - \tau^O_1\big) \overline{\varepsilon}_{\tau^O_1 :\tau}^2}{\tau}.
\end{split}
\end{equation}

We now return to our original notation, i.e.\ we consider $Y^O_1,\ldots,Y^O_{ n / 2 }$ in their original order. Then, using similar arguments to earlier, we see there exists a $c \in (0,1)$ such that
\begin{align*}
& \frac{1}{\sigma^2}\left(\tau - \tau^O_1\right)^{-1} \left[S_Y\left( \left\{\tau^O_0, \tau, \tau^O_3 \right\} \right) - S_Y\left( \left\{\tau^O_0, \tau^O_2, \tau^O_3 \right\} \right)\right]\\
\geq & c \frac{\Delta_2^2}{\sigma^2} + \mathcal{O}_\Pj\left(\frac{\Delta_2}{ \sigma}\right) + \mathcal{O}_\Pj\left(1\right) \inprob \infty.
\end{align*}
Thus $\Pj(\Gamma_2) \to 1$.
%Thus,
%\begin{equation*}
%\Pj\left(S_Y\left( \left\{\tau^O_0, \tau, \tau^O_3 \right\} \right) > S_Y\left( \left\{\tau^O_0, \tau^O_2, \tau^O_3 \right\} \right)\ \forall\ \tau \in [\tau^O_1,\tau^O_2) \right)\to 1.
%\end{equation*}

\emph{Case $\tau < \tau^O_1$.} This follows from similar arguments to the $\tau > \tau^O_2$ case.

We conclude that $\Pj\big(\hat{\tau}^O_{1, 1} = \tau^O_2\big)\to 1$. Together with $\Pj\big(\hat{\tau}^E_{1, 1} = \tau^E_2\big) \to 1$, which may be shown via similar arguments, this gives \eqref{eq:estimateCPunderestimationL1} as desired.
%The case $\tau < \tau^O_1$ follows from similar arguments. Consequently, $\Pj\Big(\hat{\tau}^O_{1, 1} = \tau^O_2\Big)\to 1$, as $n\to \infty$. Similarly we may show that $\Pj\Big(\hat{\tau}^E_{1, 1} = \tau^E_2\Big) \to 1$, as $n\to\infty$, giving \eqref{eq:estimateCPunderestimationL1} as desired.
%Next, since $\Delta_2 / \sigma \to \infty$, $\underline{\lambda}\Delta_1^2 / (\sigma^2 \log n) \to \infty$ and $\sqrt{\underline{\lambda}} \Delta_1^2 / \sigma^2 \to \infty$, as $n \to \infty$, it follows from Theorem~\ref{theorem:detectionPrecision} that writing $\delta_0 := \lfloor \sqrt{\underline{\lambda}} \rfloor$,
%\begin{equation}\label{eq:estimateCPunderestimationL2}
%\Pj\big( \vert \hat{\tau}^O_{2, 1} - \tau^O_1\vert \leq \delta_0,\ \vert \hat{\tau}^E_{2, 1} - \tau^E_1\vert \leq \delta_0,\ \hat{\tau}^O_{2, 2} = \tau^O_2,\ \hat{\tau}^E_{2, 2} = \tau^E_2\big) \to 1.
%\end{equation}
%, as $n\to\infty$. We therefore set $\delta_0 = \sqrt{\underline{\lambda}}$.

We denote by $\Omega_n$ the intersection of the events in \eqref{eq:estimateCPunderestimationL2} and \eqref{eq:estimateCPunderestimationL1}. We have shown $\Pj(\Omega_n)\to 1$.

\subsubsection*{Step 2:} In the following we work on the sequence $\Omega_n$ and bound $\operatorname{CV}_{(2)}(2) - \operatorname{CV}_{(2)}(1)$.
Observe that $\mu^E_i - \mu^O_i = \Delta_2$ if $i = (n / 2 + 1) / 2$, but $\mu^E_i - \mu^O_i = 0$ otherwise. Thus,
\begin{equation*}
\begin{split}
&\operatorname{CV}_{(2)}(L)
=\sum_{l = 0}^{L}\sum_{i = \hat{\tau}^O_{L,l} + 1}^{\hat{\tau}^O_{L,l + 1}}{ \left(Y^E_i - \overline{Y}^O_{\hat{\tau}^O_{L,l} :\hat{\tau}^O_{L,l + 1}}\right)^2 } + \sum_{l = 0}^{L}\sum_{i = \hat{\tau}^E_{L,l} + 1}^{\hat{\tau}^E_{L,l + 1}}{ \left(Y^O_i - \overline{Y}^E_{\hat{\tau}^E_{L,l} :\hat{\tau}^E_{L,l + 1}}\right)^2 }\\
=& \sum_{l = 0}^{L}\sum_{i = \hat{\tau}^O_{L,l} + 1}^{\hat{\tau}^O_{L,l + 1}}{ \left(\varepsilon^E_i + \mu^O_i - \overline{Y}^O_{\hat{\tau}^O_{L,l} :\hat{\tau}^O_{L,l + 1}}\right)^2 } + \Delta_2^2 + 2\Delta_2\left(\varepsilon^E_{\frac{\frac{n}{2} + 1}{2}} + \mu^O_{\frac{\frac{n}{2} + 1}{2}} - \overline{Y}^O_{\hat{\tau}^O_{L,L - 1} :\hat{\tau}^O_{L,L}}\right)\\
+ &\sum_{l = 0}^{L}\sum_{i = \hat{\tau}^E_{L,l} + 1}^{\hat{\tau}^E_{L,l + 1}}{ \left(\varepsilon^O_i + \mu^E_i - \overline{Y}^E_{\hat{\tau}^E_{L,l} :\hat{\tau}^E_{L,l + 1}}\right)^2 } + \Delta_2^2 - 2\Delta_2\left(\varepsilon^O_{\frac{\frac{n}{2} + 1}{2}} + \mu^E_{\frac{\frac{n}{2} + 1}{2}} - \overline{Y}^E_{\hat{\tau}^E_{L,L} :\hat{\tau}^E_{L,L + 1}}\right).
\end{split}
\end{equation*}

Consequently,
\begin{equation*}
\begin{split}
&\operatorname{CV}_{(2)}(2) - \operatorname{CV}_{(2)}(1)\\
= & \sum_{i = 1}^{\hat{\tau}^O_{2, 1}}{ \left(\varepsilon^E_i + \mu^O_i - \overline{Y}^O_{0:\hat{\tau}^O_{2,1}}\right)^2 }
 + \sum_{i = \hat{\tau}^O_{2, 1} + 1}^{\tau^O_{2}}{ \left(\varepsilon^E_i + \mu^O_i - \overline{Y}^O_{\hat{\tau}^O_{2, 1} :\tau^O_{2}}\right)^2 }
 - \sum_{i = 1}^{\tau^O_{2}}{ \left(\varepsilon^E_i + \mu^O_i - \overline{Y}^O_{0:\tau^O_{2}}\right)^2 }\\
& + \sum_{i = 1}^{\hat{\tau}^E_{2, 1}}{ \left(\varepsilon^O_i + \mu^E_i - \overline{Y}^E_{0:\hat{\tau}^E_{2, 1}}\right)^2 }
 + \sum_{i = \hat{\tau}^E_{2, 1} + 1}^{\tau^E_{2}}{ \left(\varepsilon^O_i + \mu^E_i - \overline{Y}^E_{\hat{\tau}^E_{2, 1} :\tau^E_{2}}\right)^2 }
 - \sum_{i = 1}^{\tau^E_{2}}{ \left(\varepsilon^O_i + \mu^E_i - \overline{Y}^E_{0:\tau^E_{2}}\right)^2 }\\
 & + 2\Delta_2\left(\overline{Y}^O_{0:\tau^O_{2}} - \overline{Y}^O_{\hat{\tau}^O_{2,1} :\tau^O_{2}}\right)\\
 = & 2\sum_{i = 1}^{\hat{\tau}^O_{2, 1}}{ \varepsilon^E_i\left( \overline{Y}^O_{0:\tau^O_{2}} - \overline{Y}^O_{0:\hat{\tau}^O_{2,1}}\right) }
 + 2\sum_{i = \hat{\tau}^O_{2, 1} + 1}^{\tau^O_{2}}{ \varepsilon^E_i\left( \overline{Y}^O_{0:\tau^O_{2}} - \overline{Y}^O_{\hat{\tau}^O_{2, 1} :\tau^O_{2}}\right) }\\
& + 2\sum_{i = 1}^{\hat{\tau}^E_{2, 1}}{ \varepsilon^O_i\left( \overline{Y}^E_{0:\tau^E_{2}} - \overline{Y}^E_{0:\hat{\tau}^E_{2,1}}\right) }
 + 2\sum_{i = \hat{\tau}^E_{2, 1} + 1}^{\tau^E_{2}}{ \varepsilon^O_i\left( \overline{Y}^E_{0:\tau^E_{2}} - \overline{Y}^E_{\hat{\tau}^E_{2, 1} :\tau^E_{2}}\right) }\\ 
& +\sum_{i = 1}^{\hat{\tau}^O_{2, 1}}{ \left(\mu^O_i - \overline{Y}^O_{0:\hat{\tau}^O_{2,1}}\right)^2 }
 + \sum_{i = \hat{\tau}^O_{2, 1} + 1}^{\tau^O_{2}}{ \left(\mu^O_i - \overline{Y}^O_{\hat{\tau}^O_{2, 1} :\tau^O_{2}}\right)^2 }
 - \sum_{i = 1}^{\tau^O_{2}}{ \left(\mu^O_i - \overline{Y}^O_{0:\tau^O_{2}}\right)^2 }\\ 
& + \sum_{i = 1}^{\hat{\tau}^E_{2, 1}}{ \left(\mu^E_i - \overline{Y}^E_{0:\hat{\tau}^E_{2, 1}}\right)^2 }
 + \sum_{i = \hat{\tau}^E_{2, 1} + 1}^{\tau^E_{2}}{ \left(\mu^E_i - \overline{Y}^E_{\hat{\tau}^E_{2, 1} :\tau^E_{2}}\right)^2 }
 - \sum_{i = 1}^{\tau^E_{2}}{ \left(\mu^E_i - \overline{Y}^E_{0:\tau^E_{2}}\right)^2 }\\
 & + 2\Delta_2\left(\overline{Y}^O_{0:\tau^O_{2}} - \overline{Y}^O_{\hat{\tau}^O_{2,1} :\tau^O_{2}}\right)=:A_n. 
\end{split}
\end{equation*}

We have shown $\{\operatorname{CV}_{(2)}(2) - \operatorname{CV}_{(2)}(1)\} \EINS_{\Omega_n} = A_n \EINS_{\Omega_n}$. Note that to complete the proof it suffices to show that $\Pj(A_n \EINS_{\Omega_n} > 0)\to 1$ as $n \to \infty$.

Now let $A_n = A_{1,n} + A_{2, n}$, with
\begin{equation*}
\begin{split}
A_{1,n} := & 2\sum_{i = 1}^{\hat{\tau}^O_{2, 1}}{ \varepsilon^E_i\left( \overline{Y}^O_{0:\tau^O_{2}} - \overline{Y}^O_{0:\hat{\tau}^O_{2,1}}\right) }
 + 2\sum_{i = \hat{\tau}^O_{2, 1} + 1}^{\tau^O_{2}}{ \varepsilon^E_i\left( \overline{Y}^O_{0:\tau^O_{2}} - \overline{Y}^O_{\hat{\tau}^O_{2, 1} :\tau^O_{2}}\right) }\\
& + 2\sum_{i = 1}^{\hat{\tau}^E_{2, 1}}{ \varepsilon^O_i\left( \overline{Y}^E_{0:\tau^E_{2}} - \overline{Y}^E_{0:\hat{\tau}^E_{2,1}}\right) }
 + 2\sum_{i = \hat{\tau}^E_{2, 1} + 1}^{\tau^E_{2}}{ \varepsilon^O_i\left( \overline{Y}^E_{0:\tau^E_{2}} - \overline{Y}^E_{\hat{\tau}^E_{2, 1} :\tau^E_{2}}\right) }
\end{split}
\end{equation*}
and $A_{2, n} := A_n - A_{1,n}$. We will show that $A_{1,n}$ is dominated by $A_{2,n}$.
\subsubsection*{Bounding $A_{1,n} \EINS_{\Omega_n}$.} Let us consider the first two terms of $A_{1,n}\EINS_{\Omega_n}$. It follows from the total law of probability that
\begin{align*}
&\Bigg(\sum_{i = 1}^{\hat{\tau}^O_{2, 1}}{\varepsilon^E_i\left( \overline{Y}^O_{0:\tau^O_{2}} - \overline{Y}^O_{0:\hat{\tau}^O_{2,1}}\right)} + \sum_{i = \hat{\tau}^O_{2, 1} + 1}^{\tau^O_{2}}{ \varepsilon^E_i\left( \overline{Y}^O_{0:\tau^O_{2}} - \overline{Y}^O_{\hat{\tau}^O_{2, 1} :\tau^O_{2}}\right)}\Bigg) \EINS_{\Omega_n}\\
= &\mathcal{O}_\Pj\left( \sigma \left[\hat{\tau}^O_{2, 1}\left( \overline{Y}^O_{0:\tau^O_{2}} - \overline{Y}^O_{0:\hat{\tau}^O_{2,1}}\right)^2 + \left(\tau^O_{2} - \hat{\tau}^O_{2, 1}\right)\left( \overline{Y}^O_{0:\tau^O_{2}} - \overline{Y}^O_{\hat{\tau}^O_{2, 1} :\tau^O_{2}}\right)^2\right]^{1/2} \EINS_{\Omega_n} \right),
\end{align*}
since conditional on $\varepsilon^O$, $\varepsilon^E_i\left( \overline{Y}^O_{0:\tau^O_{2}} - \overline{Y}^O_{0:\hat{\tau}^O_{2,1}}\right)$ is centred Gaussian with variance $\sigma^2 \left( \overline{Y}^O_{0:\tau^O_{2}} - \overline{Y}^O_{0:\hat{\tau}^O_{2,1}}\right)^2$. Next,
\begin{equation*}
\begin{split}
&\left[\hat{\tau}^O_{2, 1}\left( \overline{Y}^O_{0:\tau^O_{2}} - \overline{Y}^O_{0:\hat{\tau}^O_{2,1}}\right)^2 + \left(\tau^O_{2} - \hat{\tau}^O_{2, 1}\right)\left( \overline{Y}^O_{0:\tau^O_{2}} - \overline{Y}^O_{\hat{\tau}^O_{2, 1} :\tau^O_{2}}\right)^2\right]\EINS_{\Omega_n}\\
\leq & 2 \left[\hat{\tau}^O_{2, 1}\left( \overline{\mu}^O_{0:\tau^O_{2}} - \overline{\mu}^O_{0:\hat{\tau}^O_{2,1}}\right)^2 + \left(\tau^O_{2} - \hat{\tau}^O_{2, 1}\right)\left( \overline{\mu}^O_{0:\tau^O_{2}} - \overline{\mu}^O_{\hat{\tau}^O_{2, 1} :\tau^O_{2}}\right)^2\right]\EINS_{\Omega_n}\\
& + 2\left[\hat{\tau}^O_{2, 1}\left( \overline{\varepsilon}^O_{0:\tau^O_{2}} - \overline{\varepsilon}^O_{0:\hat{\tau}^O_{2,1}}\right)^2 + \left(\tau^O_{2} - \hat{\tau}^O_{2, 1}\right)\left( \overline{\varepsilon}^O_{0:\tau^O_{2}} - \overline{\varepsilon}^O_{\hat{\tau}^O_{2, 1} :\tau^O_{2}}\right)^2\right]\EINS_{\Omega_n}\\
\leq & 2\left[\max_{\tau \in \{1,\ldots,\tau^O_{2} - 1\}}{\left\{\tau\left( \overline{\mu}^O_{0:\tau^O_{2}} - \overline{\mu}^O_{0:\tau}\right)^2 + \left(\tau^O_{2} - \tau\right)\left( \overline{\mu}^O_{0:\tau^O_{2}} - \overline{\mu}^O_{\tau :\tau^O_{2}}\right)^2\right\}}\right]\\
& + 4\left[\hat{\tau}^O_{2, 1}\left( \overline{\varepsilon}^O_{0:\hat{\tau}^O_{2,1}}\right)^2 + \tau^O_{2}\left( \overline{\varepsilon}^O_{0:\tau^O_{2}}\right)^2 + \left(\tau^O_{2} - \hat{\tau}^O_{2, 1}\right)\left(\overline{\varepsilon}^O_{\hat{\tau}^O_{2, 1} :\tau^O_{2}}\right)^2\right]\EINS_{\Omega_n}\\
\leq & 2\left[\tau^O_{1}\left( \overline{\mu}^O_{0:\tau^O_{2}} - \overline{\mu}^O_{0:\tau^O_{1}}\right)^2 + \left(\tau^O_{2} - \tau^O_{1}\right)\left( \overline{\mu}^O_{0:\tau^O_{2}} - \overline{\mu}^O_{\tau^O_{1} :\tau^O_{2}}\right)^2\right]\\
& + 4\left[\hat{\tau}^O_{2, 1}\left( \overline{\varepsilon}^O_{0:\hat{\tau}^O_{2,1}}\right)^2 + \tau^O_{2}\left( \overline{\varepsilon}^O_{0:\tau^O_{2}}\right)^2 + \left(\tau^O_{2} - \hat{\tau}^O_{2, 1}\right)\left(\overline{\varepsilon}^O_{\hat{\tau}^O_{2, 1} :\tau^O_{2}}\right)^2\right]\EINS_{\Omega_n}.
\end{split}
\end{equation*}
We now bound the last two terms separately:
\begin{equation*}
\begin{split}
&\tau^O_{1}\left( \overline{\mu}^O_{0:\tau^O_{2}} - \overline{\mu}^O_{0:\tau^O_{1}}\right)^2 + \left(\tau^O_{2} - \tau^O_{1}\right)\left( \overline{\mu}^O_{0:\tau^O_{2}} - \overline{\mu}^O_{\tau^O_{1} :\tau^O_{2}}\right)^2\\
= &\tau^O_{1}\left( \frac{\tau^O_{1}}{\tau^O_{2}} \Delta_1 - \Delta_1\right)^2 + \left(\tau^O_{2} - \tau^O_{1}\right)\left(\frac{\tau^O_{1}}{\tau^O_{2}} \Delta_1\right)^2  \\
=& \Delta_1^2\frac{\tau^O_1\left(\tau^O_2 - \tau^O_1\right)}{\tau^O_2} = \mathcal{O}\left(\Delta_1^2 \underline{\lambda}(1 - 2\underline{\lambda}/n)\right).
\end{split}
\end{equation*}
Next, since $\tau^O_1 > \tau^O_2 - \tau^O_1 \geq (\underline{\lambda} - 1) / 2$,
\begin{equation*}
\begin{split}
& \left(\hat{\tau}^O_{2,1}\right)^{\frac{1}{2}} \left\vert \overline{\varepsilon}^O_{0:\hat{\tau}^O_{2,1}} \right\vert \EINS_{\Omega_n}
\leq \max_{\tau \in \{\tau^O_1 - \delta_0,\ldots,\tau^O_1 + \delta_0\}}{\tau^{-\frac{1}{2}}\left\vert \sum_{i = 1}^{\tau}{\varepsilon^O_i} \right\vert}\\
\leq &  \left(\tau^O_1 - \delta_0\right)^{-\frac{1}{2}} \left(\left\vert \sum_{i = 1}^{\tau^O_1}{\varepsilon^O_i} \right\vert + \max_{\tau \in \{\tau^O_1 - \delta_0,\ldots,\tau^O_1 - 1\}}{\left\vert \sum_{i = \tau}^{\tau^O_1 - 1}{\varepsilon^O_i} \right\vert} + \max_{\tau \in \{\tau^O_1 + 1,\ldots,\tau^O_1 + \delta_0\}}{\left\vert \sum_{i = \tau^O_1 + 1}^{\tau} \varepsilon^O_i \right\vert}\right)\\
\leq & \mathcal{O}_\Pj\left(\sigma \left(\tau^O_1 - \sqrt{\underline{\lambda}}\right)^{-\frac{1}{2}} \left(\sqrt{\tau^O_1}+\sqrt{\sqrt{\underline{\lambda}} \log \underline{\lambda}} \right)\right)
\leq \mathcal{O}_\Pj\left(\sigma\right).
\end{split}
\end{equation*}

Similarly,
\begin{equation}\label{eq:boundEpsilonA1n}
\left(\tau^O_2 - \hat{\tau}^O_{2,1}\right)^{\frac{1}{2}}\left\vert \overline{\varepsilon}^O_{\hat{\tau}^O_{2,1} :\tau^O_2} \right\vert \EINS_{\Omega_n} \leq \mathcal{O}_\Pj\left(\sigma\right).
\end{equation}
Finally,
\begin{equation*}
\tau^O_{2}\left( \overline{\varepsilon}^O_{0:\tau^O_{2}}\right)^2 \EINS_{\Omega_n} \leq \tau^O_{2}\left( \overline{\varepsilon}^O_{0:\tau^O_{2}}\right)^2 = \mathcal{O}_\Pj\left(\sigma^2\right).
\end{equation*}
The same bounds can be achieved for the last two terms in $A_{1,n} \EINS_{\Omega_n}$ by interchanging $O$'s and $E$'s. Thus, we obtain
\begin{equation}\label{eq:BoundA1}
	A_{1,n} \EINS_{\Omega_n} = \mathcal{O}_{\Pj} \left( \sigma \sqrt{\Delta_1^2 \underline{\lambda}(1 - 2\underline{\lambda}/n) + \sigma^2} \right) = \mathcal{O}_{\Pj}\left(\sigma \sqrt{\Delta_1^2 \underline{\lambda}(1 - 2\underline{\lambda}/n)}\right),
\end{equation}
the last equality following from \eqref{eq:minmax}.
%\begin{equation}\label{eq:BoundA1}
%A_{1,n} \EINS_{\Omega_n} = \mathcal{O}_{\Pj}\left(\sigma \sqrt{\Delta_1^2\left(\frac{\tau^E_1\left(\tau^E_2 - \tau^E_1\right)}{\tau^E_2} + \frac{\tau^O_1\left(\tau^O_2 - \tau^O_1\right)}{\tau^O_2}\right) + \sigma^2} \right).
%\end{equation}
%A_{1,n} \EINS_{\Omega_n} = \mathcal{O}_{\Pj}\left(\Delta_1 \sigma \sqrt{\frac{\tau^E_1\left(\tau^E_2 - \tau^E_1\right)}{\tau^E_2} + \frac{\tau^O_1\left(\tau^O_2 - \tau^O_1\right)}{\tau^O_2}} + \sigma^2 \right).

\subsubsection*{Bounding $A_{2,n}$:}
\begin{equation*}
\begin{split}
A_{2,n} = &\sum_{i = 1}^{\hat{\tau}^O_{2, 1}}{ \left(\mu^O_i - \overline{Y}^O_{0:\hat{\tau}^O_{2,1}}\right)^2 }
 + \sum_{i = \hat{\tau}^O_{2, 1} + 1}^{\tau^O_{2}}{ \left(\mu^O_i - \overline{Y}^O_{\hat{\tau}^O_{2, 1} :\tau^O_{2}}\right)^2 }
 - \sum_{i = 1}^{\tau^O_{2}}{ \left(\mu^O_i - \overline{Y}^O_{0:\tau^O_{2}}\right)^2 }\\ 
& + \sum_{i = 1}^{\hat{\tau}^E_{2, 1}}{ \left(\mu^E_i - \overline{Y}^E_{0:\hat{\tau}^E_{2, 1}}\right)^2 }
 + \sum_{i = \hat{\tau}^E_{2, 1} + 1}^{\tau^E_{2}}{ \left(\mu^E_i - \overline{Y}^E_{\hat{\tau}^E_{2, 1} :\tau^E_{2}}\right)^2 }
 - \sum_{i = 1}^{\tau^E_{2}}{ \left(\mu^E_i - \overline{Y}^E_{0:\tau^E_{2}}\right)^2 }\\
 & + 2\Delta_2\left(\overline{Y}^O_{0:\tau^O_{2}} - \overline{Y}^O_{\hat{\tau}^O_{2,1} :\tau^O_{2}}\right)\\
\geq & 2\Delta_2\left(\overline{Y}^O_{0:\tau^O_{2}} - \overline{Y}^O_{\hat{\tau}^O_{2,1} :\tau^O_{2}}\right) 
- \sum_{i = 1}^{\tau^O_{2}}{ \left(\mu^O_i - \overline{Y}^O_{0:\tau^O_{2}}\right)^2 } 
- \sum_{i = 1}^{\tau^E_{2}}{ \left(\mu^E_i - \overline{Y}^E_{0:\tau^E_{2}}\right)^2 }\\
\geq & 2\Delta_2\left(\overline{\mu}^O_{0:\tau^O_{2}} - \overline{\mu}^O_{\hat{\tau}^O_{2,1} :\tau^O_{2}} - \left\vert\overline{\varepsilon}^O_{0:\tau^O_{2}}\right\vert - \left\vert\overline{\varepsilon}^O_{\hat{\tau}^O_{2,1} :\tau^O_{2}}\right\vert\right)\\ 
- & \sum_{i = 1}^{\tau^O_{2}}{ \left(\mu^O_i - \overline{\mu}^O_{0:\tau^O_{2}}\right)^2 }
 - 2\sqrt{\sum_{i = 1}^{\tau^O_{2}}{ \left(\mu^O_i - \overline{\mu}^O_{0:\tau^O_{2}}\right)^2 }}\sqrt{\tau^O_{2}\left(\overline{\varepsilon}^O_{0:\tau^O_{2}}\right)^2} 
 - \tau^O_{2}\left(\overline{\varepsilon}^O_{0:\tau^O_{2}}\right)^2\\
& - \sum_{i = 1}^{\tau^E_{2}}{ \left(\mu^E_i - \overline{\mu}^E_{0:\tau^E_{2}}\right)^2 }
 - 2\sqrt{\sum_{i = 1}^{\tau^E_{2}}{ \left(\mu^E_i - \overline{\mu}^E_{0:\tau^E_{2}}\right)^2 }}\sqrt{\tau^E_{2}\left(\overline{\varepsilon}^E_{0:\tau^E_{2}}\right)^2}
  - \tau^E_{2}\left(\overline{\varepsilon}^E_{0:\tau^E_{2}}\right)^2.
\end{split}
\end{equation*}
We bound the terms separately. Using \eqref{eq:boundEpsilonA1n} gives us
\begin{align*}
& 2\Delta_2\left(\overline{\mu}^O_{0:\tau^O_{2}} - \overline{\mu}^O_{\hat{\tau}^O_{2,1} :\tau^O_{2}} - \left\vert\overline{\varepsilon}^O_{0:\tau^O_{2}}\right\vert - \left\vert\overline{\varepsilon}^O_{\hat{\tau}^O_{2,1} :\tau^O_{2}}\right\vert\right) \EINS_{\Omega_n} \notag\\
\geq & 2\Delta_2\left(\overline{\mu}^O_{0:\tau^O_{2}} - \overline{\mu}^O_{\tau^O_1 - \sqrt{\underline{\lambda}} :\tau^O_{2}} - \left\vert\overline{\varepsilon}^O_{0:\tau^O_{2}}\right\vert - \left(\tau^O_2 - \tau^O_1 - \delta_0\right)^{-\frac{1}{2}}\left(\tau^O_2 - \hat{\tau}^O_{2,1}\right)^{\frac{1}{2}}\left\vert \overline{\varepsilon}^O_{\hat{\tau}^O_{2,1} :\tau^O_2} \right\vert \EINS_{\Omega_n}\right) \notag\\ 
\geq & 2\Delta_2\left[\frac{\tau^O_1}{\tau^O_2}\Delta_1 - \frac{\sqrt{\underline{\lambda}}}{\tau^O_2 - \tau^O_1 + \sqrt{\underline{\lambda}}}\Delta_1 + \mathcal{O}_\Pj\left(\sigma \left(\tau^O_2\right)^{-\frac{1}{2}}\right) + \mathcal{O}_\Pj\left(\sigma \left(\tau^O_2 - \tau^O_1 - \sqrt{\underline{\lambda}}\right)^{-\frac{1}{2}}\right)\right].   \label{eq:term1}
\end{align*}

Next we have,
\begin{equation*}
\sum_{i = 1}^{\tau^O_{2}}{ \left(\mu^O_i - \overline{\mu}^O_{0:\tau^O_{2}}\right)^2 } \EINS_{\Omega_n} \leq \sum_{i = 1}^{\tau^O_{2}}{ \left(\mu^O_i - \overline{\mu}^O_{0:\tau^O_{2}}\right)^2 } = \frac{\tau^O_1\left(\tau^O_2 - \tau^O_1\right)}{\tau^O_2}\Delta_1^2 =\mathcal{O}\left(\Delta_1^2 \underline{\lambda}(1 - 2\underline{\lambda}/n)\right)
\end{equation*}
and
\begin{equation*}
\tau^O_{2}\left(\overline{\varepsilon}^O_{0:\tau^O_{2}}\right)^2\EINS_{\Omega_n}
\leq \tau^O_{2}\left(\overline{\varepsilon}^O_{0:\tau^O_{2}}\right)^2
=\mathcal{O}_\Pj\left(\sigma^2 \right).
\end{equation*}
The same bounds can be achieved when interchanging $O$'s and $E$'s.
Now as
\begin{equation*}
\frac{\tau^O_1}{\tau^O_2} 
= \frac{n / 2 - \underline{\lambda}}{n / 2 + 1}
= \left(1 - \frac{2\underline{\lambda}}{n}\right)\left(1 - \frac{1}{n / 2 + 1}\right) 
= \left(1 - \frac{2\underline{\lambda}}{n}\right) \left(1 + o(1)\right),
\end{equation*}
%and, since $\tau^E_1 = \tau^O_1$, $\tau^E_2 = \tau^O_2 - 1$,
%\begin{equation*}
%\begin{split}
%&\frac{\tau^E_1\left(\tau^E_2 - \tau^E_1\right)}{\tau^E_2} + \frac{\tau^O_1\left(\tau^O_2 - \tau^O_1\right)}{\tau^O_2}
%= \frac{\tau^O_1}{\tau^O_2} \left( 2 (\tau^O_2 - \tau^O_1) - \frac{\tau^O_1}{\tau^O_2 - 1} \right)\\
%= &\frac{\tau^O_1}{\tau^O_2} \underline{\lambda} \left(1 + \frac{1 - \underline{\lambda}^{-1}}{n / 2 - 1}\right)
%= \underline{\lambda} \left(1 - \frac{2\underline{\lambda}}{n}\right) \left(1 + o(1)\right).
%\end{split}
%\end{equation*}
$\sqrt{\underline{\lambda}} = \mathcal{O}\left(\sqrt{\tau^O_2 - \tau^O_1}\right)$, $\tau^O_1 > \tau^O_2 - \tau^O_1$, and by assumption $\underline{\lambda}\Delta_1^2/ \sigma^2 \to \infty$, it follows that
\begin{equation*}
\begin{split}
A_{2,n}\EINS_{\Omega_n} \geq & 2\Delta_2\left[\frac{\tau^O_1}{\tau^O_2}\Delta_1 - \frac{\sqrt{\underline{\lambda}}}{\tau^O_2 - \tau^O_1 + \sqrt{\underline{\lambda}}}\Delta_1 + \mathcal{O}_\Pj\left(\sigma \left(\tau^O_2\right)^{-\frac{1}{2}}\right) + \mathcal{O}_\Pj\left(\sigma \left(\tau^O_2 - \tau^O_1 - \sqrt{\underline{\lambda}}\right)^{-\frac{1}{2}}\right)\right]\\ 
& + \mathcal{O}\left(\Delta_1^2 \underline{\lambda}(1 - 2\underline{\lambda}/n)\right) +  \mathcal{O}_\Pj\left(\Delta_1^2 \sqrt{\underline{\lambda}(1 - 2\underline{\lambda}/n)}  \sigma \right)
+ \mathcal{O}_\Pj\left(\sigma^2 \right)\\
= &\Delta_1^2 \left(1 - \frac{2\underline{\lambda}}{n}\right) \left[ 2\frac{\Delta_2}{\Delta_1} - \underline{\lambda} \right] \left(1 + o_\Pj(1)\right).
\end{split}
\end{equation*}
This and \eqref{eq:BoundA1} then imply
%and
%\begin{equation*}
%\begin{split}
%A_{1,n}\EINS_{\Omega_n} = & \mathcal{O}_{\Pj}\left(\sigma \sqrt{\Delta_1^2\left(\frac{\tau^E_1\left(\tau^E_2 - \tau^E_1\right)}{\tau^E_2} + \frac{\tau^O_1\left(\tau^O_2 - \tau^O_1\right)}{\tau^O_2}\right) + \sigma^2} \right)
%=  \mathcal{O}_{\Pj}\left(\Delta_1 \sigma \sqrt{\underline{\lambda} \left(1 - \frac{2\underline{\lambda}}{n}\right)}\right).
%\end{split}
%\end{equation*}
%Hence,
\begin{equation*}
\lim_{n \to \infty}{\Pj\left(A_{1,n}\EINS_{\Omega_n} + A_{2,n}\EINS_{\Omega_n} > 0  \right)} = 1,
\end{equation*}
since the assumptions of the theorem imply
\begin{equation*}
\begin{split}
\frac{\Delta_1^2 \left(1 - 2\underline{\lambda}/n\right) \left[ 2\Delta_2/\Delta_1 - \underline{\lambda} \right]}{\left( \Delta_1^2 \sigma^2 \underline{\lambda} \left(1 - 2\underline{\lambda}/n\right) \right)^{\frac{1}{2}}} 
= \frac{\Delta_1}{\sigma} \sqrt{\underline{\lambda}} \sqrt{1 - \frac{2\underline{\lambda}}{n}}\left(\frac{2}{\underline{\lambda}}\frac{\Delta_2}{\Delta_1} - 1 \right) \to \infty.
\end{split}
\end{equation*}
This completes the proof of $\Pj(\hat{K} = 2) \to 0$.

\subsubsection*{Derivation of \eqref{eq:L2riskSquaredErrorLoss}:}
%It remains to show \eqref{eq:L2riskSquaredErrorLoss}.
We define $\hat{f}_L:\ [0,1] \to \R,\ t \mapsto  \sum_{l = 0}^{L}{\overline{Y}_{\hat{\tau}^O_{L,l} :\hat{\tau}^O_{L,l + 1}} \EINS_{(\hat{\tau}^O_{L,l} / n, \hat{\tau}^O_{L,l + 1} / n]}(t)}$. Given $K_{\max} = 2$, it is straightforward to see that $\Pj\big( \hat{K} = K\big)\to 1$, as $n\to \infty$. Secondly, with the same arguments as before, we see that $\Pj\left(\hat{\tau}_{1, 1} = \tau_2\right) \to 1$. Moreover, a straightforward calculation shows that the $L_2$-loss minimizer and hence $\hat{f}_{1}$ satisfies
\begin{equation*}
\int_0^1{ \left(\hat{f}_{1}(t) - f(t)\right)^2 dt} \EINS_{\{\hat{\tau}_{1, 1} = \tau_2\}} \geq \left(1 - \frac{\underline{\lambda}}{n}\right)\frac{\underline{\lambda} \Delta_1^2}{n} \EINS_{\hat{\tau}_{1, 1} = \tau_2} \geq \frac{\underline{\lambda} \Delta_1^2}{2n}\EINS_{\{\hat{\tau}_{1, 1} = \tau_2\}}\quad \text{a.s.}
\end{equation*}
Thus, \eqref{eq:L2riskSquaredErrorLoss} follows from the following calculation
\begin{equation*}
\begin{split}
&\Pj\left( \left[\int_0^1{ \left(\hat{f}_{\hat{K}}(t) - f(t)\right)^2 dt}\right]^{-1} \leq \frac{2n}{\underline{\lambda} \Delta_1^2}\right)\\
= & \Pj\left( \int_0^1{ \left(\hat{f}_{\hat{K}}(t) - f(t)\right)^2 dt} \geq \frac{\underline{\lambda} \Delta_1^2}{2n}\right)\\
= & \Pj\left( \int_0^1{ \left(\hat{f}_{1}(t) - f(t)\right)^2 dt} \geq \frac{\underline{\lambda} \Delta_1^2}{2n}, \hat{\tau}_{1, 1} = \tau_2\right) - \Pj(\hat{K} \neq 1) - \Pj(\hat{\tau}_{1, 1} \neq \tau_2) \to 1,
\end{split}
\end{equation*}
as $n\to\infty$.
\end{proof}

\section{Proof of Theorem~\ref{theorem:overestimation}}\label{sec:proof:overestimation}
\begin{proof}[Proof of Theorem~\ref{theorem:overestimation}]
We focus to begin with on the first term of the cross-validation criterion: 
\begin{equation*}
\operatorname{CV}_{(2)}^{O}(L)
:= \sum_{l = 0}^{L}\sum_{i = \hat{\tau}^O_{L,l} + 1}^{\hat{\tau}^O_{L,l + 1}}{ \left(Y^E_i - \overline{Y}^O_{\hat{\tau}^O_{L,l} :\hat{\tau}^O_{L,l + 1}}\right)^2 }.
\end{equation*}
We now briefly outline how our arguments proceed. We begin by showing
\begin{equation} \label{eq:AnEqualCV}
\{\operatorname{CV}_{(2)}^{O}(1) - \operatorname{CV}_{(2)}^{O}(2)\}\EINS_{\Omega^O_1 \cap \Omega^O_2} = A^O_n \EINS_{\Omega^O_1 \cap \Omega^O_2}
\end{equation}
for some events $\Omega^O_1$ and $\Omega^O_2$ and a random variable $A^O_n$. We then lower bound the probability of the event $\{A^O_n > 0\} \cap \Omega^O_1 \cap \Omega^O_2$. To this end, we split the event $\{A^O_n > 0\}$ into further events $\Omega^O_3,\ldots,\Omega^O_{10}$. The proof concludes by symmetry arguments to include the second term of the cross-validation criterion, which we denote with a superscript $E$.

\subsubsection*{Deriving \eqref{eq:AnEqualCV}}
We now define the events
\begin{equation*}
\begin{split}
\Omega^O_1 := \{\hat{\tau}^O_{1, 1} = \tau^O_1\} \text{ and } 
\Omega^O_2 := \{\hat{\tau}^O_{2, 2} = \tau^O_1\}.
\end{split}
\end{equation*}
In the following, we work on $\Omega^O_1 \cap \Omega^O_2$. Furthermore, we observe that $\tau^O_1 = (n / 2 + 1) / 2$ and $\mu^E_i - \mu^O_i = \Delta_1$ if $i = (n / 2 + 1) / 2$, but $\mu^E_i - \mu^O_i = 0$ otherwise. We have,
\begin{equation*}
\begin{split}
\operatorname{CV}_{(2)}^{O}(1)
=& \sum_{i = 1}^{\tau^O_1}{ \left(Y^E_i - \overline{Y}^O_{0:\tau^O_1}\right)^2 }
+ \sum_{i = \tau^O_1 + 1}^{\tau^O_2}{ \left(Y^E_i - \overline{Y}^O_{\tau^O_1 :\tau^O_2}\right)^2 }\\
=& \sum_{i = 1}^{\tau^O_1}{ \left(\varepsilon^E_i - \overline{\varepsilon}^O_{0:\tau^O_1}\right)^2 } + \Delta_1^2 + 2\Delta_1\left(\varepsilon^E_{\tau^O_1} - \overline{\varepsilon}^O_{0:\tau^O_1}\right)
 + \sum_{i = \tau^O_1 + 1}^{\tau^O_2}{ \left(Y^E_i - \overline{Y}^O_{\tau^O_1 :\tau^O_2}\right)^2 }
\end{split}
\end{equation*}
and
\begin{equation*}
\begin{split}
&\operatorname{CV}_{(2)}^{O}(2)
= \sum_{i = 1}^{\tilde{\tau}^O_{2,1}}{ \left(Y^E_i - \overline{Y}^O_{0:\tilde{\tau}^O_{2,1}}\right)^2 }
+ \sum_{i = \tilde{\tau}^O_{2,1} + 1}^{\tau^O_1}{ \left(Y^E_i - \overline{Y}^O_{\tilde{\tau}^O_{2,1} :\tau^O_1}\right)^2 }
+ \sum_{i = \tau^O_1 + 1}^{\tau^O_2}{ \left(Y^E_i - \overline{Y}^O_{\tau^O_1 :\tau^O_2}\right)^2 }\\
=& \sum_{i = 1}^{\tilde{\tau}^O_{2,1}}{ \left(\varepsilon^E_i - \overline{\varepsilon}^O_{0:\tilde{\tau}^O_{2,1}}\right)^2 }
+ \sum_{i = \tilde{\tau}^O_{2,1} + 1}^{\tau^O_1}{ \left(\varepsilon^E_i - \overline{\varepsilon}^O_{\tilde{\tau}^O_{2,1} :\tau^O_1}\right)^2 }\\
& + \Delta_1^2 + 2\Delta_1\left(\varepsilon^E_{\tau^O_1} - \overline{\varepsilon}^O_{\tilde{\tau}^O_{2,1} :\tau^O_1}\right)
 + \sum_{i = \tau^O_1 + 1}^{\tau^O_2}{ \left(Y^E_i - \overline{Y}^O_{\tau^O_1 :\tau^O_2}\right)^2 },
\end{split}
\end{equation*}
with
\begin{equation}\label{eq:definitonTildeTau1}
\begin{split}
\tilde{\tau}^O_{2,1}
:= & \argmin_{\tau = 1,\ldots, \tau^O_1 - 1}{\left\{ \sum_{i = 1}^{\tau}{\left( \varepsilon^O_i - \overline{\varepsilon}^O_{0:\tau}\right)^2} + \sum_{\tau + 1}^{\tau^O_1}{\left(\varepsilon^O_i - \overline{\varepsilon}^O_{\tau :\tau^O_1}\right)^2}\right\}}\\
= & \argmin_{\tau = 1,\ldots, \tau^O_1 - 1}{\left\{ \sum_{i = 1}^{\tau^O_1}{\left(\varepsilon^O_i\right)^2} - \tau\left(\overline{\varepsilon}^O_{0:\tau}\right)^2 - \left(\tau^O_1 - \tau\right)\left(\overline{\varepsilon}^O_{\tau :\tau^O_1}\right)^2\right\}}\\
= & \argmax_{\tau = 1,\ldots, \tau^O_1 - 1}{\left\{\tau\left(\overline{\varepsilon}^O_{0:\tau}\right)^2 + \left(\tau^O_1 - \tau\right)\left(\overline{\varepsilon}^O_{\tau :\tau^O_1}\right)^2\right\}}.
\end{split}
\end{equation}
Hence,
\begin{equation*}
\begin{split}
\operatorname{CV}_{(2)}^{O}(1) - \operatorname{CV}_{(2)}^{O}(2)
= & \sum_{i = 1}^{\tau^O_1}{ \left(\varepsilon^E_i - \overline{\varepsilon}^O_{0:\tau^O_1}\right)^2 }
- \sum_{i = 1}^{\tilde{\tau}^O_{2,1}}{ \left(\varepsilon^E_i - \overline{\varepsilon}^O_{0:\tilde{\tau}^O_{2,1}}\right)^2 }
- \sum_{i = \tilde{\tau}^O_{2,1} + 1}^{\tau^O_1}{ \left(\varepsilon^E_i - \overline{\varepsilon}^O_{\tilde{\tau}^O_{2,1} :\tau^O_1}\right)^2 } \\
& + 2\Delta_1\left(\overline{\varepsilon}^O_{\tilde{\tau}^O_{2,1} :\tau^O_1} - \overline{\varepsilon}^O_{0:\tau^O_1}\right) 
=: A^O_n.
\end{split}
\end{equation*}
To summarise, we have shown \eqref{eq:AnEqualCV}.
%\begin{equation}\label{eq:AnEqualCV}
%\{\operatorname{CV}_{(2)}^{O}(1) - \operatorname{CV}_{(2)}^{O}(2)\}\EINS_{\Omega^O_1 \cap \Omega^O_2} = A^O_n \EINS_{\Omega^O_1 \cap \Omega^O_2}.
%\end{equation}

\subsubsection*{The event $\{A^O_n > 0\} \cap \Omega^O_1 \cap \Omega^O_2$.}
From \eqref{eq:conditionOverestimation} we have that
\begin{equation*}
r_n:=\frac{\Delta_1}{\sigma \sqrt{n\log\log{n}}} \to \infty,\quad \text{as } n\to \infty.
\end{equation*}
Then, let
\begin{equation*}
\begin{split}
&\Omega^O_3 := \left\{\overline{\varepsilon}^O_{\tilde{\tau}^O_{2,1} :\tau^O_1} > s_1 \sigma\sqrt{\frac{\log\log{n}}{n}}\right\},\\
&\Omega^O_4 := \left\{\left(\tau^O_1\right)^{1/2}\overline{\varepsilon}^O_{0:\tau^O_1} \leq s_2 \sigma\right\},\\
&\Omega^O_5 := \left\{ \tilde{\tau}^O_{2,1} \left(\overline{\varepsilon}^O_{0:\tilde{\tau}^O_{2,1}}\right)^2 + \left(\tau^O_1 - \tilde{\tau}^O_{2,1}\right) \left(\overline{\varepsilon}^O_{\tilde{\tau}^O_{2,1} :\tau^O_1}\right)^2 \leq s_3 \sigma^2 \log\log{n} \right\},\\
&\Omega^O_6 := \left\{\left(\tau^O_1\right)^{1/2}\overline{\varepsilon}^E_{0:\tau^O_1} \leq s_4 \sigma,\
\left(\tilde{\tau}^O_{2,1}\right)^{1/2}\overline{\varepsilon}^E_{0:\tilde{\tau}^O_{2,1}} \leq s_4 \sigma, \text{ and }
\left(\tau^O_1 - \tilde{\tau}^O_{2,1}\right)^{1/2}\overline{\varepsilon}^E_{\tilde{\tau}^O_{2,1} :\tau^O_1} \leq s_4 \sigma\right\},
\end{split}
\end{equation*}
with $s_1 := \max(r_n^{-1/2}, (\log \log n)^{-1/4})),\ s_2:=(\log\log n)^{1/8},\ s_3 := (r_n s_1)^{1/2} \text{ and } s_4:= (\log\log n)^{1/4}$. Hence, $s_1 \to 0,\ s_2 \to \infty,\ s_3 \to \infty,\ s_4 \to \infty, \text{ as } n \to \infty$. Moreover, we have that
\begin{equation*}
\begin{split}
& \sum_{i = 1}^{\tau^O_1}{ \left(\varepsilon^E_i - \overline{\varepsilon}^O_{0:\tau^O_1}\right)^2 }
- \sum_{i = 1}^{\tilde{\tau}^O_{2,1}}{ \left(\varepsilon^E_i - \overline{\varepsilon}^O_{0:\tilde{\tau}^O_{2,1}}\right)^2 }
- \sum_{i = \tilde{\tau}^O_{2,1} + 1}^{\tau^O_1}{ \left(\varepsilon^E_i - \overline{\varepsilon}^O_{\tilde{\tau}^O_{2,1} :\tau^O_1}\right)^2 }\\
= & \tau^O_1 \left(\overline{\varepsilon}^O_{0:\tau^O_1}\right)^2 - \tilde{\tau}^O_{2,1} \left(\overline{\varepsilon}^O_{0:\tilde{\tau}^O_{2,1}}\right)^2 - \left(\tau^O_1 - \tilde{\tau}^O_{2,1}\right) \left(\overline{\varepsilon}^O_{\tilde{\tau}^O_{2,1} :\tau^O_1}\right)^2\\
& -2\left[\tau^O_1 \overline{\varepsilon}^E_{0:\tau^O_1}\overline{\varepsilon}^O_{0:\tau^O_1} - \tilde{\tau}^O_{2,1} \overline{\varepsilon}^E_{0:\tilde{\tau}^O_{2,1}}\overline{\varepsilon}^O_{0:\tilde{\tau}^O_{2,1}} - \left(\tau^O_1 - \tilde{\tau}^O_{2,1}\right) \overline{\varepsilon}^E_{\tilde{\tau}^O_{2,1} :\tau^O_1}\overline{\varepsilon}^O_{\tilde{\tau}^O_{2,1} :\tau^O_1}
\right].
\end{split}
\end{equation*}
Thus, for $n$ large enough,
\begin{equation}\label{eq:AOn>0}
\bigcap_{j=3,4,5,6}\Omega^O_j \subseteq \{A^O_n > 0\}.
\end{equation}

Next, let
\begin{equation*}
\begin{split}
\Omega^O_7 := \bigg\{&\max_{t < \tau^O_1}\Big\{t \big(\overline{\varepsilon}^O_{0:t}\big)^2 + (\tau^O_1  - t) \big(\overline{\varepsilon}^O_{t :\tau^O_1}\big)^2\Big\} > \\
& \max_{\tau^O_1 <t  < n/2}\Big\{(t - \tau^O_1) \big(\overline{\varepsilon}^O_{\tau^O_1 :t}\big)^2 + (n/2  - t) \big(\overline{\varepsilon}^O_{t :n/2}\big)^2\Big\}\bigg\}
\end{split}
\end{equation*}
and
\begin{equation*}
\Omega^O_8:= \{\hat{\tau}^O_{2, 1} = \tau^O_1 \text{ or } \hat{\tau}^O_{2, 2} = \tau^O_1\}.
\end{equation*}
Then, it follows from the fact that $\sum_{i = a}^b \big(\varepsilon^O_i -  \overline{\varepsilon}^O_{a:b}\big)^2 = \sum_{i = a}^b \big(\varepsilon^O_i\big)^2 - (b - a + 1)\big(\overline{\varepsilon}^O_{a:b}\big)^2$ that $\Omega^O_2 = \Omega^O_7\cap \Omega^O_8$.

\subsubsection*{Symmetry arguments.}
We define $\operatorname{CV}_{(2)}^{E}(L)$ and $\Omega^E_1,\ldots,\Omega^E_8$ in the same way but with $O$ and $E$ interchanged and the observations considered in reverse order. Note that considering the observations in reverse order does not change the value of $\operatorname{CV}_{(2)}^{E}(L)$. Hence, we conclude that the same statements hold as for their counterparts denoted by $O$.

Note that $\operatorname{CV}_{(2)}(L) = \operatorname{CV}_{(2)}^{O}(L) + \operatorname{CV}_{(2)}^{E}(L)$. Hence,
\begin{equation*}
\left\{\operatorname{CV}_{(2)}(1) > \operatorname{CV}_{(2)}(2)\right\} \supseteq
\left\{\operatorname{CV}_{(2)}^{O}(1) > \operatorname{CV}_{(2)}^{O}(2), \operatorname{CV}_{(2)}^{E}(1) > \operatorname{CV}_{(2)}^{E}(2)\right\}.
\end{equation*}
Then, it follows from \eqref{eq:AnEqualCV}, $\Omega^O_2 = \Omega^O_7\cap \Omega^O_8$ and \eqref{eq:AOn>0}, that if $n$ is large enough for \eqref{eq:AOn>0} to hold, then
\begin{equation*}
\begin{split}
& \Pj\left(\operatorname{CV}_{(2)}^{O}(1) > \operatorname{CV}_{(2)}^{O}(2), \operatorname{CV}_{(2)}^{E}(1) > \operatorname{CV}_{(2)}^{E}(2)\right)\\
\geq & \Pj\left( A^O_n > 0,\ \Omega^O_1 \cap \Omega^O_2,\ A^E_n > 0,\ \Omega^E_1 \cap \Omega^E_2\right)\\
=& \Pj\left( A^O_n > 0,\ \Omega^O_1 \cap \Omega^O_7\cap \Omega^O_8,\ A^E_n > 0,\ \Omega^E_1 \cap \Omega^E_7 \cap \Omega^E_8\right)\\
\geq & \Pj\left( \Omega^O_1 \cap \bigcap_{j = 3}^{8}\Omega^O_j,\ \Omega^E_1  \cap \bigcap_{j = 3}^{8}\Omega^E_j\right).
\end{split}
\end{equation*}
Moreover, it follows from the symmetry of events and from the fact that $\Omega^O_3 \cap \Omega^O_7$ and $\Omega^E_3 \cap \Omega^E_7$ are independent since they depend on different $\varepsilon_i$'s, that 
\begin{equation*}
\begin{split}
 & \Pj\left( \Omega^O_1 \cap \bigcap_{j = 3}^{8}\Omega^O_j,\ \Omega^E_1  \cap \bigcap_{j = 3}^{8}\Omega^E_j\right)\\
\geq &\Pj\left( \Omega^O_3 \cap \Omega^O_7,\ \Omega^E_3 \cap \Omega^E_7\right) - \Pj\left(\big(\Omega^O_1\big)^C \cup \big(\Omega^O_8\big)^C\right) - \Pj\left(\big(\Omega^E_1\big)^C \cup \big(\Omega^E_8\big)^C\right)\\
& - \sum_{j = 4}^{6}\Pj\left(\big(\Omega^O_j\big)^C\right) - \sum_{j = 4}^{6}\Pj\left(\big(\Omega^E_j\big)^C\right)\\
= & \Pj\left( \Omega^O_3 \cap \Omega^O_7\right)^2 - 2 \Pj\left(\big(\Omega^O_1\big)^C \cup \big(\Omega^O_2\big)^C\right) - 2 \sum_{j = 4}^{6}\Pj\left(\big(\Omega^O_j\big)^C\right).
\end{split}
\end{equation*}
By combining all the inequalities, we obtain, for $n$ sufficiently large that \eqref{eq:AOn>0} holds, 
\begin{equation}\label{eq:probCVoverestimation}
\begin{split}
& \Pj\left(\operatorname{CV}_{(2)}(1) > \operatorname{CV}_{(2)}(2)\right)\\
\geq & \Pj\left( \Omega^O_3 \cap \Omega^O_7\right)^2 - 2 \Pj\left(\big(\Omega^O_1\big)^C \cup \big(\Omega^O_2\big)^C\right) - 2 \sum_{j = 4}^{6}\Pj\left(\big(\Omega^O_j\big)^C\right).
\end{split}
\end{equation}

\subsubsection*{Bounding probability of events $\Omega_j$.}
As $\Delta_1 / \sigma \to \infty$ due to \eqref{eq:conditionOverestimation}, it follows from Theorem \ref{theorem:detectionPrecision} that $\Pj(\Omega^O_1 \cap \Omega^O_2) \to 1$, as $n\to\infty$. Since $\overline{\varepsilon}^O_{a:b} = \mathcal{O}_{\Pj}\left( \sigma (b - a)^{-\frac{1}{2}} \right)$ and $\overline{\varepsilon}^E_{a:b} = \mathcal{O}_{\Pj}\left( \sigma (b - a)^{-\frac{1}{2}} \right)$ for all $a \leq b$ we have that $\Pj\left((\Omega^O_4)^C\right) \to 0$, as $n \to \infty$. Additionally,  $\Pj\left((\Omega^O_6)^C | \tilde{\tau}^O_{2,1}\right) \to 0$ almost surely, and so by dominated convergence, $\Pj\left((\Omega^O_6)^C \right) \to 0$ as well. Furthermore, it follows from standard extreme value theory that %(or from \citep[Lemma 2]{zou2020consistent})
\begin{equation*}
\begin{split}
&\tilde{\tau}^O_{2,1} \left(\overline{\varepsilon}^O_{0:\tilde{\tau}^O_{2,1}}\right)^2 + \left(\tau^O_1 - \tilde{\tau}^O_{2,1}\right) \left(\overline{\varepsilon}^O_{\tilde{\tau}^O_{2,1} :\tau^O_1}\right)^2
=  \max_{\tau < \tau^O_1}\Big\{\tau \big(\varepsilon^O_{0:\tau}\big)^2 + (\tau^O_1  - \tau) \big(\varepsilon^O_{\tau :\tau^O_1}\big)^2\Big\}\\
\leq & \max_{\tau = 1,\ldots, \tau^O_1 - 1}{\left\{\tau \left(\overline{\varepsilon}^O_{0:\tau}\right)^2\right\}}
+ \max_{\tau = 1,\ldots, \tau^O_1 - 1}{\left\{\left(\tau^O_1 - \tau\right) \left(\overline{\varepsilon}^O_{\tau :\tau^O_1}\right)^2\right\}}
 = \mathcal{O}_\Pj\left(\sigma^2 \log\log{n}\right)
\end{split}
\end{equation*}
and hence $\Pj\left(\big(\Omega^O_5\big)^C\right) \to 0$, as $n \to \infty$. Finally, let 
\begin{equation*}
\begin{split}
\Omega^O_9 := & \left\{ \tilde{\tau}^O_{2,1}\left(\overline{\varepsilon}^O_{0:\tilde{\tau}^O_{2,1}}\right)^2 + \left(\tau^O_1 - \tilde{\tau}^O_{2,1}\right)\left(\overline{\varepsilon}^O_{\tilde{\tau}^O_{2,1} :\tau^O_1}\right)^2 > 2 s_1^2 \sigma^2 \log\log n \right\} \text{ and }\\
\Omega^O_{10} := & \left\{ \left(\tau^O_1 - \tilde{\tau}^O_{2,1}\right)\left(\overline{\varepsilon}^O_{\tilde{\tau}^O_{2,1} :\tau^O_1}\right)^2 > \tilde{\tau}^O_{2,1}\left(\overline{\varepsilon}^O_{0:\tilde{\tau}^O_{2,1}}\right)^2 \right\}.
\end{split}
\end{equation*}
Note that on $\Omega^O_9 \cap \Omega^O_{10}$, $\left(\tau^O_1 - \tilde{\tau}^O_{2,1}\right)\left(\overline{\varepsilon}^O_{\tilde{\tau}^O_{2,1} :\tau^O_1}\right)^2 >  s_1^2 \sigma^2 \log\log n$, so
\begin{equation*}
\begin{split}
 \Pj\left( \Omega^O_3 \cap \Omega^O_7\right)
= & \Pj\left( \Omega^O_7 \right)  \Pj\bigg( \overline{\varepsilon}^O_{\tilde{\tau}^O_{2,1} :\tau^O_1} > s_1 \sigma\sqrt{\frac{\log\log{n}}{n}}\, \bigg\vert \Omega^O_7\bigg)\\
\geq & \Pj\left( \Omega^O_7 \right)
\Pj\Big( \overline{\varepsilon}^O_{\tilde{\tau}^O_{2,1} :\tau^O_1} > 0\, \big| \Omega^O_7 \cap \Omega^O_9 \cap \Omega^O_{10} \Big)
\Pj\Big( \Omega^O_{10}\, \big| \Omega^O_7 \cap \Omega^O_9 \Big)
\Pj\big( \Omega^O_9\, \vert \Omega^O_7 \big).
\end{split}
\end{equation*}

Because of symmetry and since the left segment is one observation longer, it follows that $\Pj\left( \Omega^O_7 \right) \geq \frac{1}{2}$. Additionally, since the $\varepsilon^O_i$'s are independent and symmetric around zero as centred Gaussian distributed errors we have that
\begin{equation*}
\Pj\Big( \overline{\varepsilon}^O_{\tilde{\tau}^O_{2,1} :\tau^O_1} > 0\, \big|\, \Omega^O_7 \cap \Omega^O_9 \cap \Omega^O_{10} \Big) = \Pj\left(\overline{\varepsilon}^O_{\tilde{\tau}^O_{2,1} :\tau^O_1} > 0\right) = \frac{1}{2}.
\end{equation*}
Moreover, $\tilde{\tau}^O_{2,1}\big(\overline{\varepsilon}^O_{0:\tilde{\tau}^O_{2,1}}\big)^2$ and $(\tau^O_1 - \tilde{\tau}^O_{2,1})\big(\overline{\varepsilon}^O_{\tilde{\tau}^O_{2,1} :\tau^O_1}\big)^2$  have the same distribution. This still holds when we condition on $\Omega^O_7 \cap \Omega^O_9$. Thus, $\Pj\big( \Omega^O_{10}\, \vert \Omega^O_7 \cap \Omega^O_9 \big) \geq \frac{1}{2}$.

Finally, from the definition of $\tilde{\tau}^O_{2,1}$ \eqref{eq:definitonTildeTau1} and standard extreme value theory, it follows that
\begin{align*}
	\tilde{\tau}^O_{2,1}\left(\overline{\varepsilon}^O_{0:\tilde{\tau}^O_{2,1}}\right)^2 + \left(\tau^O_1 - \tilde{\tau}^O_{2,1}\right)\left(\overline{\varepsilon}^O_{\tilde{\tau}^O_{2,1} :\tau^O_1}\right)^2 \geq \max_{\tau = 1,\ldots, \tau^O_1 - 1}\tau\left(\overline{\varepsilon}^O_{0:\tau}\right)^2 
\end{align*}
and
\begin{equation} \label{eq:maxtau}
	 \left\{\max_{\tau = 1,\ldots, \tau^O_1 - 1} \tau\left(\overline{\varepsilon}^O_{0:\tau}\right)^2 \right\}^{-1} = \mathcal{O}_\Pj\{(\sigma^2 \log\log n)^{-1}\}.
\end{equation}
%\begin{align}
%&\mathcal{O}_\Pj\left(\tilde{\tau}^O_{2,1}\left(\overline{\varepsilon}^O_{0:\tilde{\tau}^O_{2,1}}\right)^2 + \left(\tau^O_1 - \tilde{\tau}^O_{2,1}\right)\left(\overline{\varepsilon}^O_{\tilde{\tau}^O_{2,1} :\tau^O_1}\right)^2\right) \notag\\
%= & \mathcal{O}_\Pj\left(\max_{\tau = 1,\ldots, \tau^O_1 - 1}{\left\{\tau\left(\overline{\varepsilon}^O_{0:\tau}\right)^2 + \left(\tau^O_1 - \tau\right)\left(\overline{\varepsilon}^O_{\tau :\tau^O_1}\right)^2\right\}}\right) \label{eq:maxtau}\\
%\geq & \mathcal{O}_\Pj\left(\max_{\tau = 1,\ldots, \tau^O_1 - 1}{\left\{\tau\left(\overline{\varepsilon}^O_{0:\tau}\right)^2 \right\}} \right)
%=\sigma^2 \log\log(n). \notag
%\end{align}
Moreover, the conditional distribution of the maximum on the l.h.s.\ of \eqref{eq:maxtau} given $\Omega^O_7$ is stochastically larger than its unconditional distribution.
Hence we may conclude that $\liminf_{n \to \infty}{\Pj\left( \Omega^O_3 \cap \Omega^O_7\right)} \geq 1/8$.
Then, combining \eqref{eq:probCVoverestimation} and the probabilities above gives us
\[
	\liminf_{n \to \infty}{\Pj(\operatorname{CV}_{(2)}(1) > \operatorname{CV}_{(2)}(2))} > 1/64.
\]
The proof concludes by noting that $\Pj(\operatorname{CV}_{(2)}(0) > \operatorname{CV}_{(2)}(1))\to 1$, as $n\to \infty$, and hence
\begin{equation*}
\liminf_{n \to \infty} \Pj\left( \hat{K} > 1 \right) \geq \liminf_{n \to \infty}{\Pj\left(\operatorname{CV}_{(2)}(1) > \operatorname{CV}_{(2)}(2)\right)} \geq 1/64 > 0. \qedhere
\end{equation*}
\end{proof}

\section{Proof of Theorem~\ref{theorem:positiveResultRescaledCV}}\label{sec:proof:RescaledCV}
We will show that
\begin{equation}\label{eq:sketchMainIdea}
\Pj\left( \hat{K} = K \right) \geq \Pj\Bigg( \min_{\substack{L = 0,\ldots,K_{\max},\\ L \neq K}} \operatorname{CV}_{\mathrm{mod}}(L) - \operatorname{CV}_{\mathrm{mod}}(K) > 0\Bigg) \to 1, \text{ as } n\to \infty.
\end{equation}
Note that the two terms in $\operatorname{CV}_{\mathrm{mod}}(L)$ are symmetric. Hence, we can focus most of the time on the first term only. %We will begin by treating the case where $K \geq 1\ \forall n$ for all sufficiently large $n$. Next we will consider the case where $K=0$ for all $n$.

To this end, we use the notations $E_i := \varepsilon^E_i + \mu_i^O,\ i = 1,\ldots,n / 2$,
\begin{equation}\label{eq:CVmodO}
\operatorname{CV}^O_{\mathrm{mod}}(L) := 
\sum_{l = 0}^{L}\sum_{i = \hat{\tau}^O_{L,l} + 1}^{\hat{\tau}^O_{L,l + 1} - 1}{ \frac{\hat{n}^O_l}{\hat{n}^O_l - 1} \left(Y^E_i - \overline{Y}^O_{\hat{\tau}^O_{L,l} :\hat{\tau}^O_{L,l + 1}}\right)^2 },
\end{equation}
and
\begin{equation}\label{eq:simplifiedCVterm}
\widetilde{\operatorname{CV}}^O_{\mathrm{mod}}(L) := 
\sum_{l = 0}^{L}\sum_{i = \hat{\tau}^O_{L,l} + 1}^{\hat{\tau}^O_{L,l + 1}} \left(E_i - \overline{Y}^O_{\hat{\tau}^O_{L,l} :\hat{\tau}^O_{L,l + 1}}\right)^2.
\end{equation}
We later also use the same definitions with all instances of $O$'s and $E$'s interchanged to define $O_i$, $\operatorname{CV}^E_{\mathrm{mod}}$ and $\widetilde{\operatorname{CV}}^E_{\mathrm{mod}}(L)$.

Then, our main argument proceeds as follows. In Lemma~\ref{lemma:differenceCVLtoK} we lower bound
\begin{equation} \label{eq:CV_tilde}
\min_{L \neq K} \widetilde{\operatorname{CV}}^O_{\mathrm{mod}}(L) - \widetilde{\operatorname{CV}}^O_{\mathrm{mod}}(K).
\end{equation}
Next, in Lemma~\ref{lemma:differenceCVs} we upper bound
\[
\max_{L=0,\ldots,K_{\max}} |\operatorname{CV}^O_{\mathrm{mod}}(L) - \widetilde{\operatorname{CV}}^O_{\mathrm{mod}}(L)|.
\]
Putting these together, along with corresponding results for $\operatorname{CV}^E_{\mathrm{mod}}$ and $\widetilde{\operatorname{CV}}^E_{\mathrm{mod}}(L)$, which follow from identical arguments, gives the final result.
%the maximal difference for any $L$ between $\operatorname{CV}^O_{\mathrm{mod}}(L)$ and $\widetilde{\operatorname{CV}}^O_{\mathrm{mod}}(L)$.
To show Lemma~\ref{lemma:differenceCVLtoK}, we follow the strategy in \citet{zou2020consistent} and split $\widetilde{\operatorname{CV}}^O_{\mathrm{mod}}(L)$ as follows.

For any set of time points $\mathcal{U}=\{t_0 < t_1 < \cdots <t_K < t_{K + 1}\}$ and any collection of vectors $X = (X_1,\ldots,X_{t_{K + 1}}) \in \R^{d \times t_{K+1}}$ and $Y = (Y_1,\ldots,Y_{t_{K + 1}}) \in \R^{d \times t_{K+1}}$, define
\begin{equation}\label{eq:Sxy_notation}
	S_{X, Y}(\mathcal{U}) := \sum_{k = 0}^{K}\sum_{i = t_k + 1}^{t_{k+1}} \left(X_i - \overline{X}_{t_k :t_{k+1}}\right)\left(Y_i - \overline{Y}_{t_k :t_{k+1}}\right).
\end{equation}
Then for any $L = 0,\ldots,K_{\max}$, it may be shown that
\begin{align*}
& \widetilde{\operatorname{CV}}^O_{\mathrm{mod}}(L)
=  S_{E}\left( \hat{\mathcal{T}}^O_L \right) - S_{\varepsilon^O}\left( \hat{\mathcal{T}}^O_L \right) - S_{\varepsilon^E}\left( \hat{\mathcal{T}}^O_L \right)
 + 2 S_{\varepsilon^O, \varepsilon^E}\left( \hat{\mathcal{T}}^O_L \right) + \sum_{i = 1}^{n / 2} \left(\varepsilon^O_i - \varepsilon^E_i\right)^2,
\end{align*}
with $S_{X}\left( \hat{\mathcal{T}}^O_L \right)$ defined as in \eqref{eq:S_notation}.
%Furthermore, for any set of time points $\mathcal{U}=\{t_0 < t_1 < \cdots <t_K < t_{K + 1}\}$ and any collection of vectors $X = (X_1,\ldots,X_{t_{K + 1}}) \in \R^{d \times t_{K+1}}$ and $Y = (Y_1,\ldots,Y_{t_{K + 1}}) \in \R^{d \times t_{K+1}}$ we used the notation
%\begin{equation}\label{eq:Sxy_notation}
%S_{X, Y}(\mathcal{U}) := \sum_{k = 0}^{K}\sum_{i = t_k + 1}^{t_{k+1}} \left(X_i - \overline{X}_{t_k :t_{k+1}}\right)\left(Y_i - \overline{Y}_{t_k :t_{k+1}}\right).
%\end{equation}

Consequently, for any $L \neq K$,
\begin{equation}\label{eq:splitCostDifferenceModCost}
\begin{split}
& \widetilde{\operatorname{CV}}^O_{\mathrm{mod}}(L)  - \widetilde{\operatorname{CV}}^O_{\mathrm{mod}}(K)\\
= & \left\{ S_{E}\left( \hat{\mathcal{T}}^O_L \right) - S_{E}\left( \hat{\mathcal{T}}^O_K \right)\right\}
- \left\{ S_{\varepsilon^O}\left( \hat{\mathcal{T}}^O_L \right) - S_{\varepsilon^O}\left( \hat{\mathcal{T}}^O_K \right)\right\}
- \left\{ S_{\varepsilon^E}\left( \hat{\mathcal{T}}^O_L \right) - S_{\varepsilon^E}\left( \hat{\mathcal{T}}^O_K \right)\right\}\\
& + 2 \left\{ S_{\varepsilon^O, \varepsilon^E}\left( \hat{\mathcal{T}}^O_L \right) - S_{\varepsilon^O, \varepsilon^E}\left( \hat{\mathcal{T}}^O_K \right)\right\}.
\end{split}
\end{equation}
The main argument needed to bound the terms in \eqref{eq:splitCostDifferenceModCost} uniformly will be derived in Lemmas~\ref{lemma:boundsForEpsilonE}--\ref{lemma:boundsForEpsilonO}. Preliminary results needed for this are given in Lemmas~\ref{lemma:maxYao}--\ref{lemma:maxEpsilonEsquared}, which bound the maximum of rescaled local sums, and in Lemmas~\ref{lemma:differenceCostsNestedChangePointSets}~and~\ref{lemma:differenceCostsMixedTerm}, where we bound $S_{X}(\mathcal{U})-S_{X}(\mathcal{V})$ and $S_{X, Y}(\mathcal{U})-S_{X, Y}(\mathcal{V})$, respectively, when $\mathcal{U} \subset \mathcal{V}$. Lemma~\ref{lemma:differenceCVs} is shown by splitting up terms in a similar fashion.

A lower bound of $\widetilde{\operatorname{CV}}^O_{\mathrm{mod}}(L) - \widetilde{\operatorname{CV}}^O_{\mathrm{mod}}(K)$ is developed in \citet{zou2020consistent} (see \eqref{eq:whatZouetalHaveShown} below) where a simplifying assumption is made that
change-points do not occur at odd locations.
This additional assumption ensures that the first term in $\operatorname{CV}_{(2)}(L)$ is identical to $\widetilde{\operatorname{CV}}^O_{\mathrm{mod}}(L)$ as claimed by their second displayed formula on pg 433, where it is assumed that $O_i$ and $E_i$ have the same expectation.
As discussed at the beginning of the proof of Theorem~1 on pg.~8 in the supplementary material of \citet{wang2021data}, such an assumption may be justified when the noise variance is bounded away from zero and the signal magnitude is bounded; we however do not make these assumptions in our results in order to maintain greater fidelity to phenomena observed in finite sample settings.

%lower bounds of
%\begin{equation} \label{eq:CV_tilde}
%	\min_{L \neq K} \widetilde{\operatorname{CV}}^O_{\mathrm{mod}}(L) - \widetilde{\operatorname{CV}}^O_{\mathrm{mod}}(K).
%\end{equation}
In \citet{zou2020consistent}, results of the form
\begin{equation}\label{eq:whatZouetalHaveShown}
	\Pj\left( \operatorname{CV}_{(2)}(L) - \operatorname{CV}_{(2)}(K) > a_n \right) \to 1, \text{ as } n\to \infty,
\end{equation}
 are shown,
for each fixed $L = 0,\ldots,K_{\max},\ L \neq K$, where $a_n$ is a positive sequence. However while these immediately give consistency when $K$ and $K_{\max}$ are finite, it is not clear to us how the arguments may be extended directly to allow for diverging $K$ and $K_{\max}$ as this would require delicate control of the rates at which the probabilities above approach $1$.
Several of the lemmas below, which we require in our proof of Lemma~\ref{lemma:differenceCVLtoK}, are therefore uniform versions of Lemmas~1--5 in \citet{zou2020consistent}.

For notational convenience we use the following convention: whenever $u$ and $v$ are vectors of the same dimension, $u v$ should be understood to mean $u^{\top} v$, and similarly $|u|$ and $u^2$ should be taken to mean $\|u\|_2$ and $\|u\|_2^2$ respectively. Furthermore, we use the notation $\hat{n}^O_l := \hat{\tau}^O_{L,l + 1} - \hat{\tau}^O_{L,l}$.

The following two lemmas are uniform versions of Lemmas~1~and~2 in \citet{zou2020consistent}.
\begin{Lemma}\label{lemma:maxYao}
Let $(Y_{nkj})_{n=1,2,\ldots, k = 1,\ldots,K_n, j = 1,\ldots,n}$ be independent (potentially multivariate) mean-zero random variables with $\E[\| Y_{nkj}\|_2^2] \leq 1$ and $\E \|Y_{nkj}\|_2^{q} \leq \frac{q!}{2} c^{q - 2}$ for all $q\geq 3$ and a constant $c > 0$, $n \in \N$, $k = 1,\ldots,K_n$, $j = 1,\ldots,n$. Then,
\begin{equation*}
\max_{k = 1,\ldots,K_n} \max_{0 \leq i < j \leq n} \frac{1}{j - i} \left\|\sum_{l = i + 1}^{j} Y_{nkl}\right\|_2^2 = \mathcal{O}_\Pj\big( \left(\log n\right)^2 + \left(\log K_n\right)^2 \big).
\end{equation*}
\end{Lemma}
\begin{proof}
Without loss of generality, we may assume that $Y_{nkj}$ is one-dimensional. If it is multidimensional, we can bound each coordinate individually. To this end, note that the conditions of Bernstein's inequality (see Theorem~2.10 in \citet{boucheron2013concentration}) are still satisfied. Hence, we observe that the bound only increases by a constant depending solely on the dimension.

It follows from union bounds and Bernstein's inequality that for $\xi > 0$ sufficiently large,
\begin{equation*}
\begin{split}
&\Pj\left(\max_{k = 1,\ldots,K_n} \max_{0 \leq i < j \leq n} \frac{1}{j - i} \left(\sum_{l = i + 1}^{j} Y_{nkl}\right)^2 > \big((\sqrt{2}+c)\xi (\log n + \log K_n )\big)^2\right)\\
\leq & \sum_{k = 1}^{K_n}\sum_{0 \leq i < j \leq n} \Pj\left(\bigg| \sum_{l = i + 1}^{j} Y_{nkl} \bigg| > (\sqrt{2}+c)\xi (\log n + \log K_n) \sqrt{j - i} \right)\\
\leq & 2K_n \sum_{0 \leq i < j \leq n} \exp\left( - \xi (\log n + \log K_n )\right)\\
\leq & 2K_n^{-\xi + 1} n^{-\xi + 2}, 
\end{split}
\end{equation*}
Hence, the r.h.s.\ can be made arbitrarily small by choosing $\xi$ large enough.
\end{proof}

\begin{Lemma}\label{lemma:maxOneSided}
%Let $K_n$ be a sequence and $n_k \leq n,\ \forall\ k = 1,\ldots,K_n, n \in \N$. For all $n \in \N$ suppose further that $(Y_{k, j})_{k = 1,\ldots,K_n, j = 1,\ldots,n_k}$ are sequences of independent, potentially multivariate random variables with $\max_{k = 1,\ldots,K_n}\max_{i = 1,\ldots,n_k}\E \|Y_{k,i}\|_2^{2m} \leq C$ for a constant $0< C <\infty$ and some fixed even $m \geq 2$.
%Moreover, let
Consider the setup of Lemma~\ref{lemma:maxYao}. Then,
\begin{equation*}
\max_{k = 1,\ldots,K_n} \max_{1 \leq j \leq n} \frac{1}{j} \left\|\sum_{i = 1}^{j} Y_{nki}\right\|_2^2 = \mathcal{O}_\Pj\Big( \log\log n + (\log K_n)^2 \Big).
\end{equation*}
\end{Lemma}
\begin{proof}
Without loss of generality, we may assume $K_n \to \infty$ as $n \to \infty$.
Similarly to the proof of Lemma~\ref{lemma:maxYao}, we may assume without loss of generality that  $Y_{nkj}$ is one-dimensional.

Let $n$ and $k \in \{1,\ldots,K_n\}$ be fixed. Furthermore, let $M_t := \sum_{i = 1}^{t} Y_{nki}$, $\xi_t := M_t - M_{t - 1} = Y_{nki}$, and $V_t := \sum_{i = 1}^{t} \E\left[ \xi_i^2 \right] = \sum_{i = 1}^{t}\E\left[Y_{nki}^2\right] \leq t$. Then $M_t$ is a martingale and the Bernstein condition in Lemma~23 in \citet{balsubramani2014sharp} is satisfied. Hence, it follows from Theorem~5 in \citet{balsubramani2014sharp} (see the discussion after their Theorem~5 which explains that the Bernstein condition may replace the interval condition in Theorem~5) that for all $\delta \in (0,1)$,
\begin{equation*}
\Pj\left( \vert M_t \vert \leq \sqrt{6(\sqrt{2c} - 2) t \left(2 \log\log\left(\frac{3(\sqrt{2c} - 2) t}{\vert M_t\vert}\right) + \log\left(\frac{2}{\delta}\right)\right)}\quad \forall\ t \geq \tau_0 \right) \geq 1 - \delta,
\end{equation*}
with $\tau_0 := \lceil 173 \log(4/\delta) / \{2(\sqrt{2c}-2)\} \rceil$. (Note that we may assume without loss of generality that $c > 2$.)

Therefore for all $\delta \in (0, 1)$ and all $n$ sufficiently large, we have
\begin{equation*}
\Pj\left( \max_{t = \tau_0,\ldots,n} M_t^2 / t > 6(\sqrt{2c} - 2) \big\{2 \log\log(3(\sqrt{2c} - 2) n) + \log\left(2 / \delta\right)\big\}\right) < \delta.
\end{equation*}
Note here we have used the fact that for those $t$ for which $\vert M_t \vert < 1$, the inequality is trivially satisfied for $n$ sufficiently large.

Now let $\delta := 2 \exp\left(-\frac{x}{6(\sqrt{2c} - 2)}\right)$. Then,
\begin{equation*}
\Pj\left( \max_{t = \tau_0,\ldots,n} M_t^2 / t > 12(\sqrt{2c} - 2)\log\log\left(3(\sqrt{2c} - 2) n\right) + x\right) < 2 \exp\left(-\frac{x}{6(\sqrt{2c} - 2)}\right)
\end{equation*}
and
\[
\tau_0 = \left\lceil\frac{173}{12(\sqrt{2c} - 2)^2} x + \frac{173\log(2)}{2(\sqrt{2c}-2)}\right\rceil.
\]
Let us consider $x = A \log K_n$, with $A > 0$ a constant. Then a union bound yields
\begin{equation}\label{eq:maxOneSideLargeScales}
\begin{split}
& \Pj\left( \max_{k = 1,\ldots,K_n} \max_{t = \tau_0,\ldots,n} \left(\sum_{l = 1}^{t} Y_{nkt}\right)^2 / t > 12(\sqrt{2c} - 2)\log\log\left(3(\sqrt{2c} - 2) n\right) + A \log K_n\right)\\
< & 2 K_n \exp\left(-\frac{A \log K_n}{6(\sqrt{2c} - 2)}\right) = 2 K_n^{-A / (6(\sqrt{2c} - 2)) +1},
\end{split}
\end{equation}
with $\tau_0 = B \log K_n$, where $B > 0$ is a constant only depending on $c$ and $A$. The r.h.s.\ in \eqref{eq:maxOneSideLargeScales} can be made arbitrarily small by choosing $A$ large enough.

It remains to consider $t < \tau_0 = B \log K_n$. It follows from a union bound and Bernstein's inequality, see Corollary~2.11 in \citet{boucheron2013concentration}, that for any constant $C > 0$,
\begin{equation*}
\begin{split}
&\Pj\left(\max_{k = 1,\ldots,K_n} \max_{t = 1,\ldots,\tau_0-1} \frac{1}{t}\left(\sum_{l = 1}^{t} Y_{nkt}\right)^2 > C (\log K_n)^2\right)\\
\leq & \sum_{k = 1}^{K_n}\sum_{t = 1}^{\tau_0 - 1} \left\{ \Pj\left(\sum_{l = 1}^{t} Y_{nkt} > \sqrt{C t} \log K_n \right) + \Pj\left(\sum_{l = 1}^{t} Y_{nkt} < - \sqrt{C t} \log K_n \right) \right\}\\
\leq & 2 K_n B \log K_n \exp\left( - \sqrt{C} / (\sqrt{2} + c) \log K_n  \right)
= 2 B K_n^{- \sqrt{C} / (\sqrt{2} + c) - 1} \log K_n. 
\end{split}
\end{equation*}
The r.h.s.\ can be made arbitrarily small by choosing $C$ large enough, which completes the proof.
\end{proof}

\begin{Lemma}\label{lemma:specialCaseMaxOneSided}
Consider the setup of Lemma~\ref{lemma:maxYao} but where $K_n \geq 2$ eventually. Then, for any constant $C>2$ and any sequence $(\delta_n)_{n=1}^\infty$, with $1<\delta_n<n/C$, we have that
\[
\max_{k = 1,\ldots,K_n} \delta_n \max_{\lceil C \delta_n \rceil \leq j \leq n} \left\|\frac{1}{j}\sum_{i = 1}^{j} Y_{nki}\right\|_2^2 = \mathcal{O}_\Pj\big( (\log K_n)^2 \big).
\]
\end{Lemma}
\begin{proof}
A union bound and Bernstein's inequality, see Corollary~2.11 in \citet{boucheron2013concentration}, yield that for any $A > 0$,
\begin{equation*}
\begin{split}
&\Pj\left( \max_{k = 1,\ldots,K_n} \delta_n \max_{\lceil C \delta_n \rceil \leq j \leq n} \left(\frac{\sum_{i = 1}^{j} Y_{nki}}{j}\right)^2 > A^2 (\log K_n)^2 \right)\\
\leq & \sum_{k = 1}^{K_n} \sum_{j = \lceil C\delta_n \rceil,\ldots,n} \left\{\Pj\left( \sum_{i = 1}^{j} Y_{nki} > \frac{j}{\sqrt{\delta_n}} A \log K_n\right) + \Pj\left( \sum_{i = 1}^{j} Y_{nki} < -\frac{j}{\sqrt{\delta_n}} A \log K_n\right)\right\}\\
\leq & 2K_n \sum_{j = \lceil C\delta_n \rceil,\ldots,n} \exp\left(-\frac{j^2 A^2(\log K_n)^2/ \delta_n}{2\big(j + c A \log (K_n) j / \sqrt{\delta_n}\big)}\right).
\end{split}
\end{equation*}
Moreover, there exist constants $B > 0$ and $D'>0$ not depending on $A$, such that for $D:=D'A$,
\begin{equation*}
\begin{split}
& K_n \sum_{j = \lceil C\delta_n \rceil,\ldots,n} \exp\left(-\frac{j^2 A^2(\log K_n)^2 / \delta_n}{2\big(j + c A \log (K_n) j / \sqrt{\delta_n}\big)}\right)
\leq K_n \sum_{j = \lceil C\delta_n \rceil,\ldots,n} \exp\left(- \frac{j A \log K_n}{B \sqrt{\delta_n}}\right)\\
\leq & K_n \sum_{j = \lceil C\delta_n \rceil,\ldots,n} \exp\left(-\sqrt{j} A \log (K_n) / B\right)
= \sum_{j = \lceil C\delta_n \rceil,\ldots,n} K_n^{-\sqrt{j} A / B + 1}\\
\leq & \int_{1}^{\infty} \left(K_n^{-D}\right)^{\sqrt{x}} dx
=  \left[ \frac{2 \left(K_n^{-D}\right)^{\sqrt{x}}\left(-\sqrt{x} D \log K_n - 1\right) }{D^2 (\log K_n)^2} \right]_1^\infty\\
= & \frac{2 K_n^{-D} (D \log K_n + 1)}{D^2 (\log K_n)^2}.
\end{split}
\end{equation*}
The r.h.s.\ can be made arbitrarily small by choosing $A$ and hence $D$ large enough.
\end{proof}

\begin{Lemma}\label{lemma:maxEpsilonEsquared}
Let $(Y_{nj})_{n=1,2,\ldots, j = 1,\ldots,n}$ be independent (potentially multivariate) mean-zero random variables with $\E[\| Y_{nj}\|_2^2] \leq 1$ and $\E \|Y_{nj}\|_2^{q} \leq \frac{q!}{2} c^{q - 2}$ for all $q\geq 3$ and a constant $c > 0$, $n \in \N$, $j = 1,\ldots,n$. For any $L = 1,\ldots,L_{\max}$, where $L_{\max} \geq 2$, let $0\leq \tau_{L,1,s} < \tau_{L,1,e} \leq \tau_{L,2,s} < \tau_{L,2,e} \leq \cdots \tau_{L,L_{\max},s} < \tau_{L,L_{\max},e} \leq n$ be sequences of (random) time points that are independent of $(Y_{nj})_{n=1,2,\ldots, j = 1,\ldots,n}$. Then,
\begin{equation} \label{eq:bd1}
\max_{L = 1,\ldots,L_{\max}} \sum_{l = 1}^{L_{\max}} \left(\tau_{L,l,e} - \tau_{L,l,s} \right)^{-1} \left\|\sum_{j = \tau_{L,l,s} + 1}^{\tau_{L,l,e}} Y_{nj}\right\|_2^2
= \mathcal{O}_\Pj\big( (L_{\max}\log L_{\max})^{1/2} \big) 
\end{equation}
and
\begin{equation*}
\begin{split}
&\max_{L = 1,\ldots,L_{\max}} \left\vert \sum_{l = 1}^{L_{\max}} \left\{ \left(\tau_{L,l,e} - \tau_{L,l,s} \right)^{-1} \sum_{j = \tau_{L,l,s} + 1}^{\tau_{L,l,e}} \left\| Y_{nj}\right\|_2^2 - \E\left[ \left(\tau_{L,l,e} - \tau_{L,l,s} \right)^{-1} \sum_{j = \tau_{L,l,s} + 1}^{\tau_{L,l,e}} \left\| Y_{nj}\right\|_2^2 \right]\right\} \right\vert\\
= & \mathcal{O}_\Pj\big( (L_{\max}\log L_{\max})^{1/2} \big).
\end{split}
\end{equation*}
\end{Lemma}
\begin{proof}
Similarly to the proof of Lemma~\ref{lemma:maxYao}, we may assume without loss of generality that  $Y_{nj}$ is one-dimensional.

We consider the first bound to begin with. For any $L = 1,\ldots,L_{\max}$, let $\mathcal{T}_L$ denote the $\sigma$-algebra generated by $\tau_{L,1,s}$, $\tau_{L,1,e}$, $\tau_{L,2,s}$, $\tau_{L,2,e}, \ldots ,\tau_{L,L_{\max},s}$, $\tau_{L,L_{\max},e}$. Then, it follows from a union bound and the law of total probability that for any $\xi > 0$,
\begin{align*}
& \Pj\left( \max_{L = 1,\ldots,L_{\max}} \sum_{l = 1}^{L_{\max}} \left(\tau_{L,l,e} - \tau_{L,l,s} \right)^{-1} \left(\sum_{j = \tau_{L,l,s} + 1}^{\tau_{L,l,e}} Y_{nj}\right)^2  > \xi \right)\\
\leq & \sum_{L = 1}^{L_{\max}} \Pj\left(\sum_{l = 1}^{L_{\max}} \left(\tau_{L,l,e} - \tau_{L,l,s} \right)^{-1} \left(\sum_{j = \tau_{L,l,s} + 1}^{\tau_{L,l,e}} Y_{nj}\right)^2  > \xi \right)\\
= & \sum_{L = 1}^{L_{\max}} \E\left[ \Pj\left(\left. \sum_{l = 1}^{L_{\max}} \left(\tau_{L,l,e} - \tau_{L,l,s} \right)^{-1} \left(\sum_{j = \tau_{L,l,s} + 1}^{\tau_{L,l,e}} Y_{nj}\right)^2  > \xi \right\vert \mathcal{T}_L \right)\right].
\end{align*}
In the following we focus on the conditional probability inside the expectation. Conditional on $\mathcal{T}_L$,  $\left(\tau_{L,l,e} - \tau_{L,l,s} \right)^{-1} \left(\sum_{j = \tau_{L,l,s} + 1}^{\tau_{L,l,e}} Y_{nj}\right)^2,\ l = 1,\ldots,L_{\max}$ are independent.  Bernstein's inequality, see Theorem~2.10 in \citet{boucheron2013concentration}, yields that $\left(\tau_{L,l,e} - \tau_{L,l,s} \right)^{-1/2} \sum_{j = \tau_{L,l,s} + 1}^{\tau_{L,l,e}} Y_{nj}$ is sub-exponentially distributed with Orlicz norm bounded uniformly over $l = 1,\ldots,L_{\max},\ L = 1,\ldots,L_{\max}$, i.e.\ there exist a constant $\psi^{(1)} < \infty$ such that
\[
\inf\left\{C \in (0,\infty)\,:\, \E\left[\exp\left( C^{-1}\left(\tau_{L,l,e} - \tau_{L,l,s} \right)^{-1} \left(\sum_{j = \tau_{L,l,s} + 1}^{\tau_{L,l,e}} Y_{nj}\right)^2\right)\right] \leq 2 \right\} \leq \psi^{(1)}
\]
for all $l = 1,\ldots,L_{\max},\ L = 1,\ldots,L_{\max}$. Consequently, it follows from Corollary~3 in \citet{zhang2021sharper} that $\left(\tau_{L,l,e} - \tau_{L,l,s} \right)^{-1} \left(\sum_{j = \tau_{L,l,s} + 1}^{\tau_{L,l,e}} Y_{nj}\right)^2$ are sub-Weibull distributed with parameter $\theta = 1/2$ and uniformly bounded sub-Weibull norms, i.e.\ there exist a constant $\psi^{(1/2)} < \infty$ such that
\[
\inf\left\{C \in (0,\infty)\,:\, \E\left[\exp\left( C^{-1/2}\left(\left(\tau_{L,l,e} - \tau_{L,l,s} \right)^{-1} \left(\sum_{j = \tau_{L,l,s} + 1}^{\tau_{L,l,e}} Y_{nj}\right)^2\right)^{1/2}\right)\right] \leq 2 \right\} \leq \psi^{(1/2)}
\]
for all $l = 1,\ldots,L_{\max},\ L =L 1,\ldots,L_{\max}$. Hence, it follows from Corollary~6.4 in \citet{zhang2020concentration} (see also Theorem~1 in \citet{zhang2021sharper}) that there exist constants $C_1,C_2 > 0$ (not depending on $L$) such that for any $t > 0$,
\begin{equation}\label{eq:boundForSubWeibull}
\Pj\left( \left. \sum_{l = 1}^{L_{\max}} \left(\tau_{L,l,e} - \tau_{L,l,s} \right)^{-1} \left(\sum_{j = \tau_{L,l,s} + 1}^{\tau_{L,l,e}} Y_{nj}\right)^2 \geq C_1 \sqrt{L_{\max} t} + C_2 t^2 \right\vert \mathcal{T}_L \right) \leq 2 e^{-t}
\end{equation}
for all $L = 1,\ldots,L_{\max}$. Thus, for any $a > 1$,
\begin{align*}
&\Pj\left( \max_{L = 1,\ldots,L_{\max}} \sum_{l = 1}^{L_{\max}} \left(\tau_{L,l,e} - \tau_{L,l,s} \right)^{-1} \left(\sum_{j = \tau_{L,l,s} + 1}^{\tau_{L,l,e}} Y_{nj}\right)^2  > a \left[ C_1 \left(L_{\max}\log L_{\max}\right)^{1/2} + C_2 \left(\log L_{\max}\right)^{2}\right] \right)\\
\leq &\sum_{L = 1}^{L_{\max}} 2 \exp\left(-\sqrt{a} \log L_{\max}\right) = 2 L_{\max}^{-\sqrt{a} + 1}.
\end{align*}
The r.h.s.\ can be made arbitrarily small by choosing $a$ large enough, thus showing $\eqref{eq:bd1}$.

We will now show the second bound. From following the same steps as before we see that it suffices to show that
\[\sum_{l = 1}^{L_{\max}} \left\{ \left(\tau_{L,l,e} - \tau_{L,l,s} \right)^{-1} \sum_{j = \tau_{L,l,s} + 1}^{\tau_{L,l,e}} \left\| Y_{nj}\right\|_2^2 - \E\left[ \left(\tau_{L,l,e} - \tau_{L,l,s} \right)^{-1} \sum_{j = \tau_{L,l,s} + 1}^{\tau_{L,l,e}} \left\| Y_{nj}\right\|_2^2 \right]\right\}\] conditional on $\mathcal{T}_L$ is sub-Weibull distributed with parameter $\theta = 1/2$ and uniformly bounded sub-Weibull norms.

It follows from the assumptions that the $\left\| Y_{nj}\right\|_2$'s are independent and sub-exponentially distributed with uniformly bounded Orlicz norms. Consequently, it follows from Corollary~3 in \citet{zhang2021sharper} that the $\left\| Y_{nj}\right\|_2^2$'s are sub-Weibull distributed with parameter $\theta = 1/2$ and uniformly bounded sub-Weibull norms. The same applies for the centred random variables and the mean of those centred random variables. Hence, the bound follows from the same arguments as used to derive the bound for the first term.
\end{proof}

\begin{Lemma}\label{lemma:differenceCostsNestedChangePointSets}
For any $L \in \N$ and any sequence of time points $1\leq a =: t_0 < t_1 < \cdots < t_L < t_{L + 1} := b \leq n$ and any vectors $X_1,\ldots,X_n \in \R^d$ we have that
\begin{align*}
S_{X}\left(\left\{ a, b \right\}\right) - S_{X}\left(\left\{ t_0,t_1,\ldots,t_L, t_{L + 1} \right\}\right)
\leq & 2\sum_{l = 0}^{L} \frac{(b - a) - (t_{l + 1} - t_l)}{b - a} (t_{l + 1} - t_l) \left(\overline{X}_{t_{l} :t_{l + 1}}\right)^2\\
\leq & 2\sum_{l = 0}^{L} (t_{l + 1} - t_l) \left(\overline{X}_{t_{l} :t_{l + 1}}\right)^2.
\end{align*}
\end{Lemma}
\begin{proof}
%To show positivity, let $P_0$ and $P_1$ be projection matrices that project vectors onto the constant segments given by the change-point sets $\left\{ a, b \right\}$ and $\left\{ t_0,t_1,\ldots,t_L, t_{L + 1} \right\}$. Then,
%\begin{align*}
%& S_{X}\left(\left\{ a, b \right\}\right) - S_{X}\left(\left\{ t_0,t_1,\ldots,t_L, t_{L + 1} \right\}\right)\\
%= & X^t (I - P_0) X - X^t (I - P_1) X = X^t (P_1 - P_0) X = \| (P_1 - P_0) X \|_2^2 \geq 0.
%\end{align*}
%The first inequality is clear. For the second, we argue as follows.
%For the second part, it follows
From \eqref{eq:costSplit}, we have that
\begin{align*}
& S_{X}\left(\left\{ a, b \right\}\right) - S_{X}\left(\left\{ t_0,t_1,\ldots,t_L, t_{L + 1} \right\}\right)\\
= & \sum_{l = 0}^{L} (t_{l + 1} - t_l) \left(\overline{X}_{t_{l} :t_{l + 1}}\right)^2 - (b - a) \left(\overline{X}_{a :b}\right)^2\\
= & \sum_{l = 0}^{L} (t_{l + 1} - t_l) \left(\overline{X}_{t_{l} :t_{l + 1}}\right)^2 - \sum_{l_1 = 0}^{L}\sum_{l_2 = 0}^{L} \frac{(t_{l_1 + 1} - t_{l_1})(t_{l_2 + 1} - t_{l_2})}{b - a}\overline{X}_{t_{l_1} :t_{l_1 + 1}}\overline{X}_{t_{l_2} :t_{l_2 + 1}}\\
= & \sum_{l = 0}^{L} \left(1 - \frac{t_{l + 1} - t_l}{b - a}\right) (t_{l + 1} - t_l) \left(\overline{X}_{t_{l} :t_{l + 1}}\right)^2\\
& -\sum_{l = 0}^{L} (b-a)^{-1}\sum_{l^{'} = 0,\ldots,L;\ l^{'} \neq l} (t_{l^{'} + 1} - t_{l^{'}})  (t_{l + 1} - t_l) \overline{X}_{t_{l} :t_{l + 1}}\overline{X}_{t_{l^{'}} :t_{l^{'} + 1}}.
\end{align*}
Using the inequality $2xy \leq x^2 + y^2$ yields
\begin{align*}
& \left\vert \sum_{l = 0}^{L} (b-a)^{-1}\sum_{l^{'} = 0,\ldots,L;\ l^{'} \neq l} (t_{l^{'} + 1} - t_{l^{'}})  (t_{l + 1} - t_l) \overline{X}_{t_{l} :t_{l + 1}}\overline{X}_{t_{l^{'}} :t_{l^{'} + 1}} \right\vert\\
\leq & \frac{1}{2} \sum_{l = 0}^{L} (b-a)^{-1}\sum_{l^{'} = 0,\ldots,L;\ l^{'} \neq l} (t_{l^{'} + 1} - t_{l^{'}})  (t_{l + 1} - t_l)  
\left\{\left(\overline{X}_{t_{l} :t_{l + 1}}\right)^2 + \left(\overline{X}_{t_{l^{'}} :t_{l^{'} + 1}}\right)^2\right\}\\
= & \sum_{l = 0}^{L} (b-a)^{-1} (t_{l + 1} - t_l) \left(\overline{X}_{t_{l} :t_{l + 1}}\right)^2 \sum_{l^{'} = 0,\ldots,L;\ l^{'} \neq l} (t_{l^{'} + 1} - t_{l^{'}}).
\end{align*}
The proof is completed by noting that $\sum_{l^{'} = 0,\ldots,L;\ l^{'} \neq l} (t_{l^{'} + 1} - t_{l^{'}}) = (b - a) - (t_{l + 1} - t_l) < b - a$.
\end{proof}

\begin{Lemma}\label{lemma:differenceCostsMixedTerm}
For any $L \in \N$ and any sequence of time points $1\leq a =: t_0 < t_1 < \cdots < t_L < t_{L + 1} := b \leq n$ and any vectors $X_1,\ldots,X_n \in \R^d$ and $Y_1,\ldots,Y_n \in \R^d$ we have that
\begin{align*}
&2\left\vert S_{X, Y}\left(\left\{ t_0,t_1,\ldots,t_L, t_{L + 1} \right\}\right) - S_{X, Y}\left(\left\{ a, b \right\}\right) \right\vert\\
\leq & \left\{ S_{X}\left(\left\{ a, b \right\}\right) - S_{X}\left(\left\{ t_0,t_1,\ldots,t_L, t_{L + 1} \right\}\right) \right\} + \left\{ S_{Y}\left(\left\{ a, b \right\}\right) - S_{Y}\left(\left\{ t_0,t_1,\ldots,t_L, t_{L + 1} \right\}\right) \right\}.
\end{align*}
\end{Lemma}
\begin{proof}
Let $P_0 \in \R^{n \times n}$ and $P_1 \in \R^{n \times n}$ be the orthogonal projection matrices that project vectors onto the constant segments given by the change-point sets $\left\{ a, b \right\}$ and $\left\{ t_0,t_1,\ldots,t_L, t_{L + 1} \right\}$ respectively.
 Then, $X^{\top} (I - P_0) Y = S_{X, Y}\left(\left\{ a, b \right\}\right)$ and $X^{\top} (I - P_1) Y = S_{X, Y}\left(\left\{ t_0,t_1,\ldots,t_L, t_{L + 1} \right\}\right)$. Hence,
\begin{align*}
&2\left\vert S_{X, Y}\left(\left\{ a, b \right\}- S_{X, Y}\left(\left\{ t_0,t_1,\ldots,t_L, t_{L + 1} \right\}\right)\right) \right\vert\\
= & 2\left\vert X^{\top} (P_1 - P_0) Y \right\vert \leq X^{\top} (P_1 - P_0) X + Y^{\top} (P_1 - P_0) Y\\
= & \left\{ S_{X}\left(\left\{ a, b \right\}\right) - S_{X}\left(\left\{ t_0,t_1,\ldots,t_L, t_{L + 1} \right\}\right) \right\} + \left\{ S_{Y}\left(\left\{ a, b \right\}\right) - S_{Y}\left(\left\{ t_0,t_1,\ldots,t_L, t_{L + 1} \right\}\right) \right\}. 
\end{align*}
\end{proof}

\begin{Lemma}\label{lemma:boundsForEpsilonE}
Suppose that Assumptions~\ref{assumption:cpNumberMultivariate}--\ref{assumption:minimumSignalMultivariate} hold in the case where $K\geq 1$ eventually, or only Assumptions~\ref{assumption:cpNumberMultivariate}, \ref{assumption:NoiseBernstein}~and~\ref{assumption:overestimation} in the case $K = 0\ \forall\ n$. Then,
\begin{enumerate}[label=(\roman*)]
\item \label{item1:lemma:boundsForEpsilonE} $\max_{L = 0,\ldots,K_{\max}} \left\{ S_{\varepsilon^E}\left(\hat{\mathcal{T}}^O_L\right) - S_{\varepsilon^E}\left(\hat{\mathcal{T}}^O_L \cup \mathcal{T}^{O}_{K} \right) \right\} = o_{\Pj}(\overline{\sigma}^2 \log\log \overline{\lambda})$.
\item \label{item2:lemma:boundsForEpsilonE} $\max_{L = 0,\ldots,K_{\max}} \left\{ S_{\varepsilon^E}\left(\mathcal{T}^{O}_{K}\right) - S_{\varepsilon^E}\left(\hat{\mathcal{T}}^O_L \cup \mathcal{T}^{O}_{K} \right) \right\} = o_{\Pj}(\overline{\sigma}^2 \log\log \overline{\lambda})$.
\end{enumerate}
\end{Lemma}
\begin{proof}
To show \ref{item1:lemma:boundsForEpsilonE}, first let $L \in \{0,\ldots, K_{\max}\}$ be fixed. Let $\hat{\tau}^O_{L,l} =: \hat{\tau}^O_{L,l,0} < \hat{\tau}^O_{L,l, 1} < \cdots < \hat{\tau}^O_{L,l,\hat{K}_l} < \hat{\tau}^O_{L,l,\hat{K}_l + 1} := \hat{\tau}^O_{L,l + 1}$ be the true change-points between $\hat{\tau}^O_{L,l}$ and $\hat{\tau}^O_{L,l + 1}$, so
$\bigcup_{l = 0}^{L}\bigcup_{k = 0}^{\hat{K}_l}{\hat{\tau}^O_{L,l,k}} = \bigcup_{l = 0}^{L}{\hat{\tau}^O_{L,l}} \cup \bigcup_{k = 0}^{K}{\tau_{k}^O}$.
It follows from Lemma~\ref{lemma:differenceCostsNestedChangePointSets} that
\begin{align*}
 0 \leq S_{\varepsilon^E}\left(\hat{\mathcal{T}}^O_L\right) - S_{\varepsilon^E}\left(\hat{\mathcal{T}}^O_L \cup \mathcal{T}^{O}_{K} \right)
\leq 2\sum_{\substack{l = 0,\ldots,L,\\ \hat{K}_l > 0}} \sum_{k = 0}^{\hat{K}_l} \left(\hat{\tau}^O_{L,l,k + 1} - \hat{\tau}^O_{L,l,k} \right)\left(\overline{\varepsilon}^E_{\hat{\tau}^O_{L,l,k} :\hat{\tau}^O_{L,l,k + 1}}\right)^2.
\end{align*}

Then, the fact that $\sum_{\substack{l = 0,\ldots, L,\\ \hat{K}_l > 0}} (\hat{K}_l + 1) \leq 2 K$ and Lemma~\ref{lemma:maxEpsilonEsquared} gives us that
\begin{equation*}
\max_{L = 0,\ldots,K_{\max}} \sum_{\substack{l = 0,\ldots,L,\\ \hat{K}_l > 0}} \sum_{k = 0}^{\hat{K}_l} \left(\hat{\tau}^O_{L,l,k + 1} - \hat{\tau}^O_{L,l,k} \right)\left(\overline{\varepsilon}^E_{\hat{\tau}^O_{L,l,k} :\hat{\tau}^O_{L,l,k + 1}}\right)^2
=  \mathcal{O}_\Pj\left(\overline{\sigma}^2 (K_{\max} \log K_{\max})^{1/2}\right).
\end{equation*}
Finally, Assumption~\ref{assumption:cpNumberMultivariate}\ref{assumption:cpNumberKmaxBound} yields $(K_{\max} \log K_{\max})^{1/2} = o(\log\log \overline{\lambda})$.

We now show \ref{item2:lemma:boundsForEpsilonE}. Let $L \in \{0,\ldots, K_{\max}\}$ be fixed. Let $\tau_{k}^O =: \tilde{\tau}^O_{L,k,0} \leq \tilde{\tau}^O_{L,k, 1} < \cdots < \tilde{\tau}^O_{L,k,\tilde{L}_k} \leq \tilde{\tau}^O_{L,k,\tilde{L}_k + 1} := \tau_{k + 1}^O$ be the estimated change-points in $\hat{\mathcal{T}}^O_L$ between $\tau_{k}^O$ and $\tau_{k + 1}^O$, so $\bigcup_{k = 0}^{K}\bigcup_{l = 0}^{\tilde{L}_k}{\tilde{\tau}^O_{L,k,l}} = \bigcup_{k = 0}^{K}{\tau_{k}^O} \cup \bigcup_{l = 0}^{L}{\hat{\tau}^O_{L,l}}$ and $\sum_{k = 0}^{K} L_{k} = L$.
It follows from Lemma~\ref{lemma:differenceCostsNestedChangePointSets} that
\begin{align*}
0 \leq S_{\varepsilon^E}\left(\mathcal{T}^{O}_{K}\right) - S_{\varepsilon^E}\left(\hat{\mathcal{T}}^O_L \cup \mathcal{T}^{O}_{K} \right)
\leq 2\sum_{\substack{k = 0,\ldots, K,\\ \tilde{L}_k > 0}} \sum_{l = 0}^{\tilde{L}_k} \left(\tilde{\tau}^O_{L,k,l + 1} - \tilde{\tau}^O_{L,k,l} \right)\left(\overline{\varepsilon}^E_{\tilde{\tau}^O_{L,k,l} :\tilde{\tau}^O_{L,k,l + 1}}\right)^2.
\end{align*}

Then, the fact that $\sum_{\substack{k = 0,\ldots, K,\ \tilde{L}_k > 0}} (\tilde{L}_k + 1) \leq 2 L$ and Lemma~\ref{lemma:maxEpsilonEsquared} gives us that
\begin{align*}
\max_{L = 0,\ldots,K_{\max}} \bigg\{\sum_{\substack{k = 0,\ldots, K,\\ \tilde{L}_k > 0}} \sum_{l = 0}^{\tilde{L}_k} \left(\tilde{\tau}^O_{L,k,l + 1} - \tilde{\tau}^O_{L,k,l} \right)\left(\overline{\varepsilon}^E_{\tilde{\tau}^O_{L,k,l} :\tilde{\tau}^O_{L,k,l + 1}}\right)^2 \bigg\}
= \mathcal{O}_{\Pj}\left( \overline{\sigma}^2 (K_{\max} \log K_{\max})^{1/2} \right),
\end{align*}
and so the result follows similarly to (i).
%Finally, Assumption~\ref{assumption:cpNumberMultivariate}\ref{assumption:cpNumberKmaxBound} yields that $(K_{\max} \log K_{\max})^{1/2} = o(\log\log \overline{\lambda})$.
\end{proof}

The following lemma is a uniform version of Lemma~4 in \citet{zou2020consistent}. Note that \ref{item1:lemma:boundsForYE} below extends \citet[Lemma~4 (i)]{zou2020consistent} to allow for sequences of index sets $\mathcal{I}_L$ of missing change-points rather than a fixed single change-point.
\begin{Lemma}\label{lemma:boundsForYE}
Let $E_i := \mu^O_i + \varepsilon^E_i,\ i = 1,\ldots, n / 2 $. Suppose that Assumptions~\ref{assumption:cpNumberMultivariate}--\ref{assumption:minimumSignalMultivariate} hold in the case where $K\geq 1$ eventually, or only Assumptions~\ref{assumption:cpNumberMultivariate}, \ref{assumption:NoiseBernstein}~and~\ref{assumption:overestimation} in the case $K = 0\ \forall\ n$. Then we have the following.
\begin{enumerate}[label=(\roman*)]
\item \label{item1:lemma:boundsForYE} 
Suppose that, in addition, for a constant $A>0$, there exist sequences of non-empty sets $\mathcal{I}_L\subseteq \{1,\ldots,K\}$ such that
\begin{equation}\label{eq:assumptionMeanO}
\Pj\Bigg(\forall\ L < K,\  \forall\ k\in \mathcal{I}_L,\ \sum_{i = \tau_{k}^O - \floor{\frac{\underline{\lambda}}{4}} + 1}^{\tau_{k}^O + \floor{\frac{\underline{\lambda}}{4}}}{\big(\mu^O_i - \overline{\mu}^O_{L, i}\big)^2} \geq  A\underline{\lambda}\Delta_k^2 \Bigg) \to 1,
\end{equation}
with $\overline{\mu}^O_{L, i} := \sum_{l = 0}^{L}{\EINS_{\{\hat{\tau}^O_{L,l} + 1 \leq i \leq \hat{\tau}^O_{L,l + 1}\}} \overline{\mu}^O_{\hat{\tau}^O_{L,l} :\hat{\tau}^O_{L,l + 1}}}$. Then
\begin{equation*}
\begin{split}
\min_{L = 0,\ldots,K-1} \left(\sum_{k\in \mathcal{I}_L}{\Delta_k^2}\right)^{-1} \left\{S_{E}\left(\hat{\mathcal{T}}^O_L\right) - S_{E}\left(\hat{\mathcal{T}}^O_{K}\right)\right\}\geq \underline{\lambda} \left(A + o_\Pj(1)\right).
\end{split}
\end{equation*}
\item \label{item3:lemma:boundsForYE} $\max_{L = 0,\ldots,K_{\max}} \left\{ S_{E}\left(\mathcal{T}^{O}_{K}\right) - S_{E}\left(\hat{\mathcal{T}}^O_L \cup \mathcal{T}^{O}_{K} \right) \right\} = o_{\Pj}(\overline{\sigma}^2 \log\log \overline{\lambda})$.
\item \label{item4:lemma:boundsForYE} $S_{E}\left(\hat{\mathcal{T}}^O_K \right) - S_{E}\left(\mathcal{T}^{O}_{K}\right) = o_{\Pj}(\overline{\sigma}^2 \log\log \overline{\lambda})$.
\end{enumerate}
\end{Lemma}
\begin{proof}
We will show \ref{item1:lemma:boundsForYE} after \ref{item3:lemma:boundsForYE}~and~\ref{item4:lemma:boundsForYE}, since \ref{item1:lemma:boundsForYE} is more complex and uses other results.

Turning to \ref{item3:lemma:boundsForYE}, note that
\begin{align*}
 S_{E}\left(\mathcal{T}^{O}_{K}\right) - S_{E}\left(\hat{\mathcal{T}}^O_L \cup \mathcal{T}^{O}_{K} \right) = S_{\varepsilon^E}\left(\mathcal{T}^{O}_{K}\right) - S_{\varepsilon^E}\left(\hat{\mathcal{T}}^O_L \cup \mathcal{T}^{O}_{K} \right).
\end{align*}
Then, \ref{item3:lemma:boundsForYE} follows from Lemma~\ref{lemma:boundsForEpsilonE}\ref{item1:lemma:boundsForEpsilonE}.

We now show \ref{item4:lemma:boundsForYE}.
It follows from \eqref{eq:costSplit} that
\begin{equation*}
\begin{split}
& S_{E}\left(\hat{\mathcal{T}}^O_K \right) - S_{E}\left(\mathcal{T}^{O}_{K}\right)\\
= & \sum_{k = 0}^{K}\sum_{i = \hat{\tau}^O_{K, k} + 1}^{\hat{\tau}^O_{K,k + 1}}\left(\mu^O_i - \overline{\mu}^O_{\hat{\tau}^O_{K, k} :\hat{\tau}^O_{K,k + 1}}\right)^2
+ 2\sum_{k = 0}^{K}\sum_{i = \hat{\tau}^O_{K,k} + 1}^{\hat{\tau}^O_{K,k + 1}}\left(\mu^O_i - \overline{\mu}^O_{\hat{\tau}^O_{K,k} :\hat{\tau}^O_{K,k + 1}} \right) \varepsilon^E_i\\
& - \sum_{k = 0}^{K}\left(\hat{\tau}^O_{K,k + 1} - \hat{\tau}^O_{K,k}\right)\left(\overline{\varepsilon}^E_{\hat{\tau}^O_{K,k} :\hat{\tau}^O_{K,k + 1}}\right)^2 + \sum_{k = 0}^{K}\left(\tau^O_{k + 1} - \tau^O_{k}\right)\left(\overline{\varepsilon}^E_{\tau^O_{k} :\tau^O_{k + 1}}\right)^2\\
=: & A_1 + 2A_2 + A_3 + A_4.
\end{split}
\end{equation*}
By Markov's inequality, $A_3 = \mathcal{O}_{\Pj}\big( K \overline{\sigma}^2 \big)$ and $A_4 = \mathcal{O}_{\Pj}\big( K \overline{\sigma}^2 \big)$.

In the following we will bound $A_1$ and $A_2$. To this end, let $\hat{\tau}^{\min}_{k} := \min \big\{\hat{\tau}^O_{K,k}, \tau^O_{k}\big\}$ and $\hat{\tau}^{\max}_k := \max \big\{\hat{\tau}^O_{K,k}, \tau^O_{k}\big\}$, $k = 0\ldots,K + 1$. Furthermore, we define $\Omega_n$ to be the event in Assumption~\ref{assumption:detectionPrecision}\ref{assumption:detectionPrecision:L>=K}. Note that $\Pj(\Omega_n) \to 1$. In the following we work on $\Omega_n$.

From Assumptions~\ref{assumption:detectionPrecision}\ref{assumption:detectionPrecision:Infq0}~and~\ref{assumption:minimumSignalMultivariate}, we have that $\hat{\tau}^{\max}_k  - \hat{\tau}^{\min}_{k} \leq \delta_{0,k} \leq \underline{\lambda} / 4$. Consequently, $\hat{\tau}^O_{K,k + 1} - \hat{\tau}^O_{K, k} \geq \underline{\lambda} / 2$, $k = 0,\ldots,K$.

Now for $i \in [\hat{\tau}^{\max}_k, \hat{\tau}^{\min}_{k+1}]$ we know that $\mu^O_i = \beta_k$, and
\begin{equation*}
\begin{split}
\left\vert \overline{\mu}^O_{\hat{\tau}^O_{K, k} :\hat{\tau}^O_{K,k + 1}} - \beta_k\right\vert
\leq &\frac{(\hat{\tau}^{\max}_{k} - \hat{\tau}^{\min}_{k})\Delta_{k} + (\hat{\tau}^{\max}_{k + 1} - \hat{\tau}^{\min}_{k + 1})\Delta_{k + 1}}{\hat{\tau}^O_{K,k + 1} - \hat{\tau}^O_{K, k}}\\
\leq & \frac{\delta_{0, k} \Delta_k + \delta_{0, k + 1} \Delta_{k + 1}}{\max\left\{\underline{\lambda} / 2, \hat{\tau}^{\min}_{k + 1} - \hat{\tau}^{\max}_{k}\right\}},
\end{split}
\end{equation*}
for every $k = 0,\ldots,K$, where we have used the notation $\Delta_{0} = \Delta_{K + 1} = 0$. Thus it follows from the fact that $(x + y)^2\leq 2x^2 + 2y^2$ and $\delta_{0, k}\leq \underline{\lambda} / 2$, $k = 1,\ldots,K$, that
\begin{equation*}
\begin{split}
A_1 \leq & \sum_{k = 0}^{K} \left(\hat{\tau}^{\min}_{k + 1} - \hat{\tau}^{\max}_{k}\right) \left(\beta_k - \overline{\mu}^O_{\hat{\tau}^O_{K, k} :\hat{\tau}^O_{K,k + 1}}\right)^2\\
& + \sum_{k = 1}^{K} \left(\hat{\tau}^{\max}_{k} - \hat{\tau}^{\min}_{k}\right) \max\left\{\left(\beta_k - \overline{\mu}^O_{\hat{\tau}^O_{K, k - 1} :\hat{\tau}^O_{K,k}}\right)^2 , \left(\beta_{k - 1} - \overline{\mu}^O_{\hat{\tau}^O_{K, k} :\hat{\tau}^O_{K,k + 1}}\right)^2\right\}\\
= & \mathcal{O}\left( \sum_{k = 1}^{K} \delta_{0, k} \Delta_k^2 \right).
\end{split}
\end{equation*}
Turning to $A_2$, for the following equality we use that, by definition of $\hat{\tau}^{\min}_{k}$ and $\hat{\tau}^{\max}_{k}$, $\mu^O_i - \overline{\mu}^O_{\hat{\tau}^O_{K,k} :\hat{\tau}^O_{K,k + 1}}$ is constant between $\hat{\tau}^{\max}_{k} + 1$ and $\hat{\tau}^{\min}_{k + 1}$ as well as between $\hat{\tau}^{\min}_{k}$ and $\hat{\tau}^{\max}_{k}$.

Next, the Cauchy--Schwarz inequality and $2\vert xy\vert \leq x^2 + y^2$ yield
\begin{equation*}
\begin{split}
\left\vert A_2 \right\vert
= & \left\vert\sum_{k = 0}^{K}\sum_{i = \hat{\tau}^{\max}_{k} + 1}^{\hat{\tau}^{\min}_{k + 1}} \left(\mu^O_i - \overline{\mu}^O_{\hat{\tau}^O_{K,k} :\hat{\tau}^O_{K,k + 1}} \right)\overline{\varepsilon}^E_{\hat{\tau}^{\max}_{k} :\hat{\tau}^{\min}_{k + 1}}\right.\\
& + \left. \sum_{k = 1}^{K} \sum_{i = \hat{\tau}^{\min}_{k}}^{\hat{\tau}^{\max}_{k}}  \left(\mu^O_i - \overline{\mu}^O_{\hat{\tau}^O_{K,k} :\hat{\tau}^O_{K,k + 1}} \right)\overline{\varepsilon}^E_{\hat{\tau}^{\min}_{k} :\hat{\tau}^{\max}_{k}} \right\vert\\
\leq & \sum_{k = 0}^{K}\sum_{i = \hat{\tau}^O_{K, k} + 1}^{\hat{\tau}^O_{K,k + 1}}\left(\mu^O_i - \overline{\mu}^O_{\hat{\tau}^O_{K, k} :\hat{\tau}^O_{K,k + 1}}\right)^2\\
& + \sum_{k = 0}^{K}\left(\hat{\tau}^{\min}_{k + 1} - \hat{\tau}^{\max}_{k}\right)\left(\overline{\varepsilon}^E_{\hat{\tau}^{\max}_{k} :\hat{\tau}^{\min}_{k + 1}}\right)^2 + \sum_{k = 1}^{K} \left(\hat{\tau}^{\max}_{k} - \hat{\tau}^{\min}_{k}\right) \left(\overline{\varepsilon}^E_{\hat{\tau}^{\min}_{k} :\hat{\tau}^{\max}_{k}}\right)^2 .
\end{split}
\end{equation*}

Finally, since $\varepsilon^O_1,\ldots,\varepsilon^O_{n / 2}$ and $\varepsilon^E_1,\ldots,\varepsilon^E_{n / 2}$ are independent, we obtain from Markov's inequality that
\begin{equation*}
\begin{split}
& \sum_{k = 0}^{K}\left(\hat{\tau}^{\min}_{k + 1} - \hat{\tau}^{\max}_{k}\right)\left(\overline{\varepsilon}^E_{\hat{\tau}^{\max}_{k} :\hat{\tau}^{\min}_{k + 1}}\right)^2 = \mathcal{O}_{\Pj}\left( K \overline{\sigma}^2 \right),\\
&\sum_{k = 1}^{K} \left(\hat{\tau}^{\max}_{k} - \hat{\tau}^{\min}_{k}\right) \left(\overline{\varepsilon}^E_{\hat{\tau}^{\min}_{k} :\hat{\tau}^{\max}_{k}}\right)^2 = \mathcal{O}_{\Pj}\left( K \overline{\sigma}^2 \right).
%,\\
%-A_3 = &\sum_{k = 0}^{K}\left(\hat{\tau}^O_{K,k + 1} - \hat{\tau}^O_{K,k}\right)\left(\overline{\varepsilon}^E_{\hat{\tau}^O_{K,k} :\hat{\tau}^O_{K,k + 1}}\right)^2 = \mathcal{O}_{\Pj}\left( K \overline{\sigma}^2\right),\\
%A_4 = &\sum_{k = 0}^{K}\left(\tau^O_{k + 1} - \tau^O_{k}\right)\left(\overline{\varepsilon}^E_{\tau^O_{k} :\tau^O_{k + 1}}\right)^2 = \mathcal{O}_{\Pj}\left( K \overline{\sigma}^2\right). 
\end{split}
\end{equation*}

Recall that $\Pj(\Omega_n) \to 1$. Thus, in total we obtain 
\begin{equation*}
\begin{split}
& S_{E}\left(\hat{\mathcal{T}}^O_K \right) - S_{E}\left(\mathcal{T}^{O}_{K}\right) = A_1 + A_2 + A_3 + A_4
= \mathcal{O}_{\Pj}\left( \sum_{k = 1}^{K} \delta_{0, k} \Delta_k^2 \right) + \mathcal{O}_{\Pj}\left( K \overline{\sigma}^2\right).
\end{split}
\end{equation*}
Moreover, Assumption~\ref{assumption:cpNumberMultivariate}\ref{assumption:cpNumberBound} yields $K = o(\log\log \overline{\lambda})$. Hence, it follows from Assumption~\ref{assumption:detectionPrecision}\ref{assumption:detectionPrecision:Infq>0} that
\[
S_{E}\left(\hat{\mathcal{T}}^O_K \right) - S_{E}\left(\mathcal{T}^{O}_{K}\right) = o_{\Pj}\left( \overline{\sigma}^2\log\log \overline{\lambda}\right).
\]

We now show \ref{item1:lemma:boundsForYE}.
We have that
\begin{equation*}
S_{E}\left(\hat{\mathcal{T}}^O_L\right) - S_{E}\left(\hat{\mathcal{T}}^O_{K}\right)
= \left\{ S_{E}\left(\hat{\mathcal{T}}^O_L\right) - S_{E}\left(\mathcal{T}^{O}_{K}\right) \right\} - \left\{ S_{E}\left(\hat{\mathcal{T}}^O_{K}\right) - S_{E}\left(\mathcal{T}^{O}_{K}\right) \right\}.
\end{equation*}

From \ref{item4:lemma:boundsForYE} we have that $S_{E}\left(\hat{\mathcal{T}}^O_{K}\right) - S_{E}\left(\mathcal{T}^{O}_{K}\right) = o_\Pj\left(\overline{\sigma}^2 \log\log \overline{\lambda} \right)$.

Furthermore, it follows from \eqref{eq:costSplit} that
\begin{equation*}
\begin{split}
& S_{E}\left(\hat{\mathcal{T}}^O_L\right) - S_{E}\left(\mathcal{T}^{O}_{K}\right)\\
= & \sum_{l = 0}^{L}\sum_{i = \hat{\tau}^O_{L,l} + 1}^{\hat{\tau}^O_{L,l + 1}}{\left(\mu^O_i - \overline{\mu}^O_{\hat{\tau}^O_{L,l} :\hat{\tau}^O_{L,l + 1}}\right)^2}
+ 2\sum_{l = 0}^{L}\sum_{i = \hat{\tau}^O_{L,l} + 1}^{\hat{\tau}^O_{L,l + 1}}{\left(\mu^O_i - \overline{\mu}^O_{\hat{\tau}^O_{L,l} :\hat{\tau}^O_{L,l + 1}} \right) \varepsilon^E_i}
 + S_{\varepsilon^E}\left(\hat{\mathcal{T}}^O_L\right) - S_{\varepsilon^E}\left(\mathcal{T}^{O}_{K}\right).
\end{split}
\end{equation*}
% \sum_{l = 0}^{L}\sum_{k = 0}^{\hat{K}_l}{\left(\hat{\tau}^O_{L,l,k + 1} - \hat{\tau}^O_{L,l,k} \right)\left[\left(\mu^O_{\hat{\tau}^O_{L,l,k} + 1} - \overline{\mu}^O_{\hat{\tau}^O_{L,l} :\hat{\tau}^O_{L,l + 1}}\right)^2 + 2\left(\mu^O_{\hat{\tau}^O_{L,l,k} + 1} - \overline{\mu}^O_{\hat{\tau}^O_{L,l} :\hat{\tau}^O_{L,l + 1}}\right) \overline{\varepsilon}^E_{\hat{\tau}^O_{L,l,k} :\hat{\tau}^O_{L,l,k + 1}}\right]}\\

In addition, Lemma~\ref{lemma:boundsForEpsilonE} yields
\begin{equation}\label{eq:alsoDifferenceCVL<KthirdTerm}
\begin{split}
& \max_{L = 0,\ldots,K - 1}\left\vert S_{\varepsilon^E}\left( \hat{\mathcal{T}}^O_L \right) - S_{\varepsilon^E}\left( \hat{\mathcal{T}}^O_K \right)\right\vert\\
\leq & \max_{L = 0,\ldots,K - 1} \left\{ S_{\varepsilon^E}\left( \hat{\mathcal{T}}^O_L \right) - S_{\varepsilon^E}\left(\hat{\mathcal{T}}^O_L \cup \mathcal{T}^{O}_{K} \right)\right\}
 + \max_{L = 0,\ldots,K - 1} \left\{ S_{\varepsilon^E}\left(\mathcal{T}^{O}_{K}\right) - S_{\varepsilon^E}\left(\hat{\mathcal{T}}^O_L \cup \mathcal{T}^{O}_{K} \right) \right\}\\
& + \left\{ S_{\varepsilon^E}\left(\mathcal{T}^{O}_{K}\right) -  S_{\varepsilon^E}\left(\hat{\mathcal{T}}^O_K \cup \mathcal{T}^{O}_{K} \right)\right\} + \left\{ S_{\varepsilon^E}\left( \hat{\mathcal{T}}^O_K \right) - S_{\varepsilon^E}\left(\hat{\mathcal{T}}^O_K \cup \mathcal{T}^{O}_{K} \right)\right\}\\
= & o_{\Pj}(\overline{\sigma}^2 \log\log\overline{\lambda}).
\end{split}
\end{equation}

Let $\hat{\tau}^O_{L,l,k}$, $l = 0,\ldots,L,\ k = 0,\ldots,\hat{K}_l$, be defined as in the proof of Lemma~\ref{lemma:boundsForEpsilonE}\ref{item1:lemma:boundsForEpsilonE}. Then,
\begin{align*}
&\min_{L = 0,\ldots,K - 1} \left\{S_{E}\left(\hat{\mathcal{T}}^O_L\right) - S_{E}\left(\hat{\mathcal{T}}^O_{K}\right)\right\}\\
= &\min_{L = 0,\ldots,K - 1}\sum_{l = 0}^{L}\sum_{k = 0}^{\hat{K}_l}\left(\hat{\tau}^O_{L,l,k + 1} - \hat{\tau}^O_{L,l,k} \right)\bigg[\left(\mu^O_{\hat{\tau}^O_{L,l,k} + 1} - \overline{\mu}^O_{\hat{\tau}^O_{L,l} :\hat{\tau}^O_{L,l + 1}}\right)^2\\
& \hspace*{150pt} + 2\left(\mu^O_{\hat{\tau}^O_{L,l,k} + 1} - \overline{\mu}^O_{\hat{\tau}^O_{L,l} :\hat{\tau}^O_{L,l + 1}}\right) \overline{\varepsilon}^E_{\hat{\tau}^O_{L,l,k} :\hat{\tau}^O_{L,l,k + 1}}\bigg]\\
& + o_\Pj\left(\overline{\sigma}^2 \log\log(\overline{\lambda})\right).
\end{align*}

Next, it follows from the Cauchy--Schwarz inequality that
\begin{equation*}
\begin{split}
& \sum_{l = 0}^{L}\sum_{k = 0}^{\hat{K}_l}{\left(\hat{\tau}^O_{L,l,k + 1} - \hat{\tau}^O_{L,l,k} \right)\left(\mu^O_{\hat{\tau}^O_{L,l,k} + 1} - \overline{\mu}^O_{\hat{\tau}^O_{L,l} :\hat{\tau}^O_{L,l + 1}}\right) \overline{\varepsilon}^E_{\hat{\tau}^O_{L,l,k} :\hat{\tau}^O_{L,l,k + 1}}}\\
\geq & -\left(\sum_{l = 0}^{L}\sum_{k = 0}^{\hat{K}_l} \left(\hat{\tau}^O_{L,l,k + 1} - \hat{\tau}^O_{L,l,k} \right)\left(\mu^O_{\hat{\tau}^O_{L,l,k} + 1} - \overline{\mu}^O_{\hat{\tau}^O_{L,l} :\hat{\tau}^O_{L,l + 1}}\right)^2 \right)^{1/2}\\
& \left(\sum_{l = 0}^{L}\sum_{k = 0}^{\hat{K}_l} \left(\hat{\tau}^O_{L,l,k + 1} - \hat{\tau}^O_{L,l,k} \right)\left(\overline{\varepsilon}^E_{\hat{\tau}^O_{L,l,k} :\hat{\tau}^O_{L,l,k + 1}}\right)^2 \right)^{1/2}.
\end{split}
\end{equation*}

Then, Lemma~\ref{lemma:maxEpsilonEsquared} yields
\begin{equation*}
\begin{split}
\max_{L = 0,\ldots,K - 1}\sum_{l = 0}^{L}\sum_{k = 0}^{\hat{K}_l} \left(\hat{\tau}^O_{L,l,k + 1} - \hat{\tau}^O_{L,l,k} \right)\left(\overline{\varepsilon}^E_{\hat{\tau}^O_{L,l,k} :\hat{\tau}^O_{L,l,k + 1}}\right)^2
= \mathcal{O}_\Pj\left(\overline{\sigma}^2 (K \log K)^{1/2}\right),
\end{split}
\end{equation*}
since $\sum_{l = 0}^{L} (\hat{K}_l + 1) \leq K + L \leq 2K$. In addition, Assumption~\ref{assumption:cpNumberMultivariate}\ref{assumption:cpNumberKmaxBound} yields that 
\begin{equation*}
(K \log K)^{1/2} \leq (K_{\max} \log K_{\max})^{1/2} = o(\log\log \overline{\lambda}).
\end{equation*}
Moreover, $\vert\mathcal{I}_L\vert \geq 1$ and Assumption \ref{assumption:minimumSignalMultivariate} yield
\begin{align*}
&\min_{L = 0,\ldots,K - 1} \Big(\underline{\lambda} \sum_{k\in \mathcal{I}_L}{\Delta_k^2}\Big)^{-1} \overline{\sigma}^2 \log\log \overline{\lambda}
\leq \left(\underline{\lambda} \Delta_{(1)}^2\right)^{-1} \overline{\sigma}^2 \log\log \overline{\lambda} \to 0.
\end{align*}

In addition, it follows from the definition of $\overline{\mu}^O_{L, i}$ and \eqref{eq:assumptionMeanO} that
\begin{equation*}
\begin{split}
& \min_{L = 0,\ldots,K - 1} \left\{\Big(\underline{\lambda} \sum_{k\in \mathcal{I}_L}{\Delta_k^2}\Big)^{-1} \sum_{l = 0}^{L}\sum_{k = 0}^{\hat{K}_l}{\left(\hat{\tau}^O_{L,l,k + 1} - \hat{\tau}^O_{L,l,k} \right)\left(\mu^O_{\hat{\tau}^O_{L,l,k} + 1} - \overline{\mu}^O_{\hat{\tau}^O_{L,l} :\hat{\tau}^O_{L,l + 1}}\right)^2}\right\}\\
= & \min_{L = 0,\ldots,K - 1} \left\{\Big(\underline{\lambda} \sum_{k\in \mathcal{I}_L}{\Delta_k^2}\Big)^{-1} \sum_{k = 0}^{K}\sum_{i = \tau_{k}^O + 1}^{\tau_{k + 1}^O}(\mu^O_i - \overline{\mu}^O_{L, i})^2\right\}\\
\geq & \min_{L = 0,\ldots,K - 1} \left\{\Big(\underline{\lambda} \sum_{k\in \mathcal{I}_L}{\Delta_k^2}\Big)^{-1} \sum_{k \in \mathcal{I}_L}\sum_{i = \tau_{k}^O - \floor{\frac{\underline{\lambda}}{4}} + 1}^{\tau_{k}^O + \floor{\frac{\underline{\lambda}}{4}}}(\mu^O_i - \overline{\mu}^O_{L, i})^2\right\}
\geq A,
\end{split}
\end{equation*}
with probability approaching 1.
Thus,
\begin{equation*}
 \min_{L = 0,\ldots,K - 1} \left\{\Big(\underline{\lambda} \sum_{k\in \mathcal{I}_L}{\Delta_k^2}\Big)^{-1}\left( S_{E}\left(\hat{\mathcal{T}}^O_L\right) - S_{E}\left(\hat{\mathcal{T}}^O_{K}\right)\right)\right\}
\geq A + o_\Pj(1).
\end{equation*}
\end{proof}

The following lemma is a uniform version of Lemma~5 in \citet{zou2020consistent}.
\begin{Lemma}\label{lemma:boundsForEpsilonO}
Suppose that Assumptions~\ref{assumption:cpNumberMultivariate}--\ref{assumption:minimumSignalMultivariate} hold in the case where $K\geq 1$ eventually, or only Assumptions~\ref{assumption:cpNumberMultivariate}, \ref{assumption:NoiseBernstein}~and~\ref{assumption:overestimation} in the case $K = 0\ \forall\ n$. Then,
\begin{enumerate}[label=(\roman*)]
\item \label{item1:boundsForEpsilonO} $\max_{L = 0,\ldots,K - 1} \left\{ S_{\varepsilon^O}\left(\mathcal{T}^{O}_{K}\right) - S_{\varepsilon^O}\left(\hat{\mathcal{T}}^O_L \cup \mathcal{T}^{O}_{K} \right)\right\} = \mathcal{O}_{\Pj}(K \overline{\sigma}^2(\log \overline{\lambda})^2)$ and\\
$\max_{L = 0,\ldots,K - 1} \left\{ S_{\varepsilon^O}\left( \hat{\mathcal{T}}^O_L \right) - S_{\varepsilon^O}\left(\hat{\mathcal{T}}^O_L \cup \mathcal{T}^{O}_{K} \right)\right\} = \mathcal{O}_{\Pj}(K \overline{\sigma}^2(\log \overline{\lambda})^2)$,

\item \label{item2:boundsForEpsilonO} $S_{\varepsilon^O}\left(\mathcal{T}^{O}_{K}\right) - S_{\varepsilon^O}\left(\hat{\mathcal{T}}^O_K \cup \mathcal{T}^{O}_{K} \right) = o_{\Pj}(\overline{\sigma}^2 \log\log\overline{\lambda})$,

\item \label{item3:boundsForEpsilonO} $\max_{L = K,\ldots,K_{\max}} \left\{ S_{\varepsilon^O}\left( \hat{\mathcal{T}}^O_L \right) - S_{\varepsilon^O}\left(\hat{\mathcal{T}}^O_L \cup \mathcal{T}^{O}_{K} \right)\right\} = o_{\Pj}(\overline{\sigma}^2 \log\log\overline{\lambda})$.
\end{enumerate}
\end{Lemma}
\begin{proof} 
Without loss of generality we may assume $K > 0$, since the statements are trivial for $K = 0$.

\ref{item1:boundsForEpsilonO} Let $L \in \{0,\ldots, K - 1\}$ be fixed. Recall that $\tau_k^O =: \tilde{\tau}^O_{L,k,0} \leq \tilde{\tau}^O_{L,k, 1} < \cdots < \tilde{\tau}^O_{L,k,\tilde{L}_k} \leq \tilde{\tau}^O_{L,k,\tilde{L}_k + 1} := \tau_{k + 1}^O$ denote the estimated change-points in $\hat{\mathcal{T}}^O_L$ between $\tau_k^O$ and $\tau_{k + 1}^O$, so $\bigcup_{k = 0}^{K}\bigcup_{l = 0}^{\tilde{L}_k}{\tilde{\tau}^O_{L,k,l}} = \bigcup_{k = 0}^{K}{\tau_k^O} \cup \bigcup_{l = 0}^{L}{\hat{\tau}^O_{L,l}}$ and $\sum_{k = 0}^{K} L_{k} = L$.

Then, Lemma~\ref{lemma:differenceCostsNestedChangePointSets} yields
\begin{align*}
0 \leq & S_{\varepsilon^O}\left(\mathcal{T}^{O}_{K}\right) - S_{\varepsilon^O}\left(\hat{\mathcal{T}}^O_L \cup \mathcal{T}^{O}_{K} \right)
\leq  2\sum_{\substack{k = 0,\ldots, K,\\ \tilde{L}_k > 0}} \sum_{l = 0}^{\tilde{L}_k} \left(\tilde{\tau}^O_{L,k,l + 1} - \tilde{\tau}^O_{L,k,l} \right)\left(\overline{\varepsilon}^O_{\tilde{\tau}^O_{L,k,l} :\tilde{\tau}^O_{L,k,l + 1}}\right)^2.
\end{align*}

It follows from $\sum_{\substack{k = 0,\ldots, K,\\ \tilde{L}_k > 0}} (\tilde{L}_k + 1) \leq 2 L$ and Lemma~\ref{lemma:maxYao} that
\begin{align*}
&\max_{L = 0,\ldots,K - 1} \sum_{\substack{k = 0,\ldots, K,\\ \tilde{L}_k > 0}} \sum_{l = 0}^{\tilde{L}_k} \left(\tilde{\tau}^O_{L,k,l + 1} - \tilde{\tau}^O_{L,k,l} \right)\left(\overline{\varepsilon}^O_{\tilde{\tau}^O_{L,k,l} :\tilde{\tau}^O_{L,k,l + 1}}\right)^2\\
\leq & \max_{L = 0,\ldots,K - 1} \sum_{\substack{k = 0,\ldots, K,\\ \tilde{L}_k > 0}} (\tilde{L}_k + 1) \max_{\tau_k^O \leq i < j \leq \tau_{k + 1}^O} (j - i) \left(\overline{\varepsilon}^O_{i :j}\right)^2\\
\leq & \max_{L = 0,\ldots,K - 1} 2L \max_{k = 1,\ldots,K}\max_{\tau_k^O \leq i < j \leq \tau_{k + 1}^O} (j - i) \left(\overline{\varepsilon}^O_{i :j}\right)^2\\
= & \mathcal{O}_{\Pj}\left(K \overline{\sigma}^2 \left((\log \overline{\lambda})^2 + (\log K)^2\right)\right).
\end{align*}
Finally, Assumption~\ref{assumption:cpNumberMultivariate}\ref{assumption:cpNumberBound} yields that $\log K \leq \log \overline{\lambda}$ and the first statement follows.

For the second statement, let $L \in \{0,\ldots, K - 1\}$ be fixed. Recall that $\hat{\tau}^O_{L,l} =: \hat{\tau}^O_{L,l,0} < \hat{\tau}^O_{L,l, 1} < \cdots < \hat{\tau}^O_{L,l,\hat{K}_l} < \hat{\tau}^O_{L,l,\hat{K}_l + 1} := \hat{\tau}^O_{L,l + 1}$ denote the true change-points between $\hat{\tau}^O_{L,l}$ and $\hat{\tau}^O_{L,l + 1}$, so
\begin{equation*}
\bigcup_{l = 0}^{L}\bigcup_{k = 0}^{\hat{K}_l}{\hat{\tau}^O_{L,l,k}} = \bigcup_{l = 0}^{L}{\hat{\tau}^O_{L,l}} \cup \bigcup_{k = 0}^{K}{\tau_k^O}.
\end{equation*}
Then, Lemma~\ref{lemma:differenceCostsNestedChangePointSets} yields
\begin{align*}
0 \leq & S_{\varepsilon^O}\left( \hat{\mathcal{T}}^O_L \right) - S_{\varepsilon^O}\left(\hat{\mathcal{T}}^O_L \cup \mathcal{T}^{O}_{K} \right)
\leq  2\sum_{\substack{l = 0,\ldots, L,\\ \hat{K}_l > 0}}\sum_{k = 0}^{\hat{K}_l} \left(\hat{\tau}^O_{L,l,k + 1} - \hat{\tau}^O_{L,l,k} \right)\left(\overline{\varepsilon}^O_{\hat{\tau}^O_{L,l,k} :\hat{\tau}^O_{L,l,k + 1}}\right)^2.
\end{align*}

It follows from Lemma~\ref{lemma:maxOneSided} that
\begin{equation}\label{eq:argumentsOneSidedBound}
\begin{split}
&\max_{L = 0,\ldots,K - 1} \sum_{\substack{l = 0,\ldots, L,\\ \hat{K}_l > 0}}\sum_{k = 0}^{\hat{K}_l} \left(\hat{\tau}^O_{L,l,k + 1} - \hat{\tau}^O_{L,l,k} \right)\left(\overline{\varepsilon}^O_{\hat{\tau}^O_{L,l,k} :\hat{\tau}^O_{L,l,k + 1}}\right)^2\\
\leq & \sum_{k = 0}^{K} (\tau_{k + 1}^O - \tau_k^O) \left(\overline{\varepsilon}^O_{\tau_k^O :\tau_{k + 1}^O}\right)^2\\
& + \sum_{k = 0}^{K} \max_{i = \tau_k^O + 1,\ldots,\tau_{k + 1}^O - 1} (\tau_{k + 1}^O - i) \left(\overline{\varepsilon}^O_{i :\tau_{k + 1}^O}\right)^2
+ \sum_{k = 0}^{K} \max_{i = \tau_k^O + 1,\ldots,\tau_{k + 1}^O - 1} (i - \tau_k^O) \left(\overline{\varepsilon}^O_{\tau_k^O :i}\right)^2\\
\leq &\mathcal{O}_{\Pj}\left(\sqrt{K}\overline{\sigma}^2\right)
 + (K + 1) \max_{k = 0,\ldots,K} \max_{i = \tau_k^O + 1,\ldots,\tau_{k + 1}^O - 1} \left\{ (\tau_{k + 1}^O - i) \left(\overline{\varepsilon}^O_{i :\tau_{k + 1}^O}\right)^2
+  (i - \tau_k^O) \left(\overline{\varepsilon}^O_{\tau_k^O :i}\right)^2\right\}\\
= &\mathcal{O}_{\Pj}\left(K \overline{\sigma}^2 \left(\log\log \overline{\lambda} + (\log K)^2\right)\right).
\end{split}
\end{equation}
Finally, Assumption~\ref{assumption:cpNumberMultivariate}\ref{assumption:cpNumberBound} yields that $\log K \leq \log \overline{\lambda}$ and the statement follows.

\ref{item2:boundsForEpsilonO} Let $\mathcal{T}^{\delta}_{K} := \{\tau_k^O \pm \floor{\delta_{0, k}},\ k = 1,\ldots,K\}$ and $\mathcal{T}^{\delta}_{L} := \{\tau_k^O \pm \floor{\delta_{1, k}},\ k = 1,\ldots,K\}$, $L = K + 1,\ldots,K_{\max}$. Then,
\begin{align*}
& S_{\varepsilon^O}\left(\mathcal{T}^{O}_{K}\right) - S_{\varepsilon^O}\left(\hat{\mathcal{T}}^O_K \cup \mathcal{T}^{O}_{K} \right)\\
= &\left\{S_{\varepsilon^O}\left(\mathcal{T}^{O}_{K}\right) - S_{\varepsilon^O}\left(\hat{\mathcal{T}}^O_K \cup \mathcal{T}^{O}_{K} \cup \mathcal{T}^{\delta}_{K} \right)\right\} - 
\left\{S_{\varepsilon^O}\left(\hat{\mathcal{T}}^O_K \cup \mathcal{T}^{O}_{K} \right) - S_{\varepsilon^O}\left(\hat{\mathcal{T}}^O_K \cup \mathcal{T}^{O}_{K} \cup \mathcal{T}^{\delta}_{K} \right)\right\}.
\end{align*}
We are working on the event in Assumption~\ref{assumption:detectionPrecision}\ref{assumption:detectionPrecision:L>=K}, but Assumption~\ref{assumption:detectionPrecision}\ref{assumption:detectionPrecision:L>=K} assumes that the probability of that event converges to one. Using similar arguments as in \ref{item1:boundsForEpsilonO} gives us
\begin{align*}
&\left\{S_{\varepsilon^O}\left(\mathcal{T}^{O}_{K}\right) - S_{\varepsilon^O}\left(\hat{\mathcal{T}}^O_K \cup \mathcal{T}^{O}_{K} \cup \mathcal{T}^{\delta}_{K} \right)\right\} - 
\left\{S_{\varepsilon^O}\left(\hat{\mathcal{T}}^O_K \cup \mathcal{T}^{O}_{K} \right) - S_{\varepsilon^O}\left(\hat{\mathcal{T}}^O_K \cup \mathcal{T}^{O}_{K} \cup \mathcal{T}^{\delta}_{K} \right)\right\}\\
= & \mathcal{O}_{\Pj}\left(K \overline{\sigma}^2 \left(\log\log \Big(\max_{k = 1,\ldots,K} \delta_{0, k} \vee e\Big) + (\log K)^2\right)\right).
\end{align*}

Finally, Assumption~\ref{assumption:detectionPrecision}\ref{assumption:detectionPrecision:Rate} gives $K \log\log \Big(\max_{k = 1,\ldots,K} \delta_{0, k} \vee e\Big) = o(\log\log\overline{\lambda})$ and Assumption~\ref{assumption:cpNumberMultivariate}\ref{assumption:cpNumberBound} yields $K(\log K)^2 = o(\log\log\overline{\lambda})$.

\ref{item3:boundsForEpsilonO} For $k = 1,\ldots,K$ let $\delta_k := \max_{q = 0,1} \floor{\delta_{q, k}}$. Moreover, let $\mathcal{T}^{\delta}_{L} := \{\tau_k^O \pm \delta_k,\ k = 1,\ldots,K\}$, $L = K,\ldots,K_{\max}$. Then,
\begin{align*}
& S_{\varepsilon^O}\left(\hat{\mathcal{T}}^O_L\right) - S_{\varepsilon^O}\left(\hat{\mathcal{T}}^O_L \cup \mathcal{T}^{O}_{K} \right)\\
= &\left\{S_{\varepsilon^O}\left(\hat{\mathcal{T}}^O_L\right) - S_{\varepsilon^O}\left(\hat{\mathcal{T}}^O_L \cup \mathcal{T}^{O}_{K} \cup \mathcal{T}^{\delta}_{L} \right)\right\} - 
\left\{S_{\varepsilon^O}\left(\hat{\mathcal{T}}^O_L \cup \mathcal{T}^{O}_{K} \right) - S_{\varepsilon^O}\left(\hat{\mathcal{T}}^O_L \cup \mathcal{T}^{O}_{K} \cup \mathcal{T}^{\delta}_{L} \right)\right\}.
\end{align*}
We will focus on how to bound $\max_{L = K,\ldots,K_{\max}} \left\{S_{\varepsilon^O}\left(\hat{\mathcal{T}}^O_L\right) - S_{\varepsilon^O}\left(\hat{\mathcal{T}}^O_L \cup \mathcal{T}^{O}_{K} \cup \mathcal{T}^{\delta}_{L} \right)\right\}$, but the same bound can be obtained for $\max_{L = K,\ldots,K_{\max}}\left\{S_{\varepsilon^O}\left(\hat{\mathcal{T}}^O_L \cup \mathcal{T}^{O}_{K} \right) - S_{\varepsilon^O}\left(\hat{\mathcal{T}}^O_L \cup \mathcal{T}^{O}_{K} \cup \mathcal{T}^{\delta}_{L} \right)\right\}$ by using identical arguments. 

Let $L \in \{K,\ldots,K_{\max}\}$ be fixed. Let $\hat{\tau}^O_{L,l} =: \check{\tau}^O_{L,l,0} < \check{\tau}^O_{L,l, 1} < \cdots < \check{\tau}^O_{L,l,\check{K}_l} < \check{\tau}^O_{L,l,\check{K}_l + 1} := \hat{\tau}^O_{L,l + 1}$ denote all elements of the set $\mathcal{T}^{O}_{K} \cup \mathcal{T}^{\delta}_{L}$ that are between $\hat{\tau}^O_{L,l}$ and $\hat{\tau}^O_{L,l + 1}$, so
$\bigcup_{l = 0}^{L}\bigcup_{k = 0}^{\check{K}_l}{\check{\tau}^O_{L,l,k}} = \bigcup_{l = 0}^{L}{\hat{\tau}^O_{L,l}} \cup \mathcal{T}^{O}_{K} \cup \mathcal{T}^{\delta}_{L}$. Then, Lemma~\ref{lemma:differenceCostsNestedChangePointSets} yields
\begin{align*}
0 \leq & S_{\varepsilon^O}\left(\hat{\mathcal{T}}^O_L\right) - S_{\varepsilon^O}\left(\hat{\mathcal{T}}^O_L \cup \mathcal{T}^{O}_{K} \cup \mathcal{T}^{\delta}_{L} \right)\\
\leq & 2\sum_{\substack{l = 0,\ldots, L,\\ \check{K}_l > 0}}\sum_{k = 0}^{\check{K}_l} \frac{\left(\hat{\tau}^O_{L,l + 1} - \hat{\tau}^O_{L,l}\right) - \left(\check{\tau}^O_{L,l,k + 1} - \check{\tau}^O_{L,l,k} \right)}{\hat{\tau}^O_{L,l + 1} - \hat{\tau}^O_{L,l}}\left(\check{\tau}^O_{L,l,k + 1} - \check{\tau}^O_{L,l,k} \right)\left(\overline{\varepsilon}^O_{\check{\tau}^O_{L,l,k} :\check{\tau}^O_{L,l,k + 1}}\right)^2\\
\leq & 2\sum_{\substack{l = 0,\ldots, L,\\ \check{K}_l > 0}}\sum_{k = 0}^{\check{K}_l} \left(\check{\tau}^O_{L,l,k + 1} - \check{\tau}^O_{L,l,k} \right)\left(\overline{\varepsilon}^O_{\check{\tau}^O_{L,l,k} :\check{\tau}^O_{L,l,k + 1}}\right)^2.
\end{align*}

We are working on the event in Assumption~\ref{assumption:detectionPrecision}\ref{assumption:detectionPrecision:L>=K}, but Assumption~\ref{assumption:detectionPrecision}\ref{assumption:detectionPrecision:L>=K} assumes that the probability of that event converges to one. Hence, for every $k = 1,\ldots,K$ there exists an $l = 1,\ldots,L$ (not necessarily unique) such that $\tau_k^O - \delta_k \leq \hat{\tau}^O_{L,l} \leq \tau_k^O$ or $\tau_k^O \leq \hat{\tau}^O_{L,l} \leq \tau_k^O + \delta_k$. Thus, 
\begin{align*}
& \frac{1}{2}\max_{L = K,\ldots,K_{\max}} \left\{S_{\varepsilon^O}\left(\hat{\mathcal{T}}^O_L\right) - S_{\varepsilon^O}\left(\hat{\mathcal{T}}^O_L \cup \mathcal{T}^{O}_{K} \cup \mathcal{T}^{\delta}_{L} \right)\right\}\\
\leq & \sum_{k = 0}^{K} (\tau_{k + 1}^O - \delta_{k + 1} - \tau_k^O - \delta_k) \left(\overline{\varepsilon}^O_{\tau_k^O + \delta_k :\tau_{k + 1}^O - \delta_{k + 1}}\right)^2
+ \sum_{k = 1}^{K} \delta_k\left(\overline{\varepsilon}^O_{\tau_k^O - \delta_k :\tau_k^O}\right)^2
+ \sum_{k = 1}^{K} \delta_k\left(\overline{\varepsilon}^O_{\tau_k^O :\tau_k^O + \delta_k}\right)^2\\
& + \sum_{k = 1}^{K} \max_{t = \tau_k^O - \delta_k + 1,\ldots,\tau_k^O - 1}  (\tau_k^O - t)\left(\overline{\varepsilon}^O_{\tau_k^O - \delta_k :t}\right)^2
+ \sum_{k = 1}^{K} \max_{t = \tau_k^O - \delta_k + 1,\ldots,\tau_k^O - 1} (t - \tau_k^O + \delta_k)\left(\overline{\varepsilon}^O_{t :\tau_k^O}\right)^2\\
& + \sum_{k = 1}^{K} \max_{t = \tau_k^O + 1,\ldots,\tau_k^O + \delta_k - 1}  (t - \tau_k^O)\left(\overline{\varepsilon}^O_{\tau_k^O :t}\right)^2
+ \sum_{k = 1}^{K} \max_{t = \tau_k^O + 1,\ldots,\tau_k^O + \delta_k - 1} (\tau_k^O + \delta_k - t)\left(\overline{\varepsilon}^O_{t :\tau_k^O + \delta_k}\right)^2\\
& + \sum_{k = 1}^{K} \max_{t = \tau_k^O + \delta_k + 1,\ldots,\tau_{k + 1}^O - \delta_{k + 1} - 1} \frac{(t - \tau_k^O + \delta_k) - (t - \tau_k^O - \delta_k)}{t - \tau_k^O + \delta_k} (t - \tau_k^O - \delta_k)\left(\overline{\varepsilon}^O_{\tau_k^O + \delta_{k} :t}\right)^2\\
& + \sum_{k = 1}^{K} \max_{t = \tau_{k - 1}^O + \delta_{k - 1} + 1,\ldots,\tau_k^O - \delta_{k} - 1} \frac{(\tau_k^O + \delta_{k} - t) - (\tau_k^O - \delta_{k} - t)}{\tau_k^O + \delta_{k} - t} (\tau_k^O - \delta_{k} - t)\left(\overline{\varepsilon}^O_{t :\tau_k^O - \delta_{k}}\right)^2,
\end{align*}
where we have used the notation $\delta_0 := \delta_{K + 1} := 0$.

The following calculation follows from a central limit theorem for triangular arrays. Note that Lyapunov's condition is met, since higher moments are uniformly bounded because of Assumption~\ref{assumption:NoiseBernstein}. We obtain
\begin{align*}
&\sum_{k = 0}^{K} (\tau_{k + 1}^O - \delta_{k + 1} - \tau_k^O - \delta_k) \left(\overline{\varepsilon}^O_{\tau_k^O + \delta_k :\tau_{k + 1}^O - \delta_{k + 1}}\right)^2
+ \sum_{k = 1}^{K} \delta_k\left(\overline{\varepsilon}^O_{\tau_k^O - \delta_k :\tau_k^O}\right)^2
+ \sum_{k = 1}^{K} \delta_k\left(\overline{\varepsilon}^O_{\tau_k^O :\tau_k^O + \delta_k}\right)^2\\
= & \mathcal{O}_{\Pj}\left(\sqrt{K} \overline{\sigma}^2\right).
\end{align*}
Furthermore, Assumption~\ref{assumption:cpNumberMultivariate}\ref{assumption:cpNumberBound} yields $\sqrt{K} = o(\log\log\overline{\lambda})$.

It follows from Lemma~\ref{lemma:maxOneSided} that
\begin{align*}
&\sum_{k = 1}^{K} \max_{t = \tau_k^O - \delta_k + 1,\ldots,\tau_k^O - 1}  (\tau_k^O - t)\left(\overline{\varepsilon}^O_{\tau_k^O - \delta_k :t}\right)^2\\
\leq  &K \max_{k = 1,\ldots,K}\max_{t = \tau_k^O - \max_{k = 1,\ldots,K} \delta_k + 1,\ldots,\tau_k^O - 1}  (\tau_k^O - t)\left(\overline{\varepsilon}^O_{\tau_k^O - \max_{k = 1,\ldots,K}\delta_k :t}\right)^2\\
= &\mathcal{O}_{\Pj}\left(K \overline{\sigma}^2\left(\log\log \Big(\max_{k = 1,\ldots,K}\delta_k  \vee e\Big) + (\log K)^2\right)\right).
\end{align*}
Furthermore, Assumption~\ref{assumption:cpNumberMultivariate}\ref{assumption:cpNumberBound} yields $K(\log K)^2  = o(\log\log\overline{\lambda})$ and Assumption~\ref{assumption:detectionPrecision}\ref{assumption:detectionPrecision:Rate} yields 
\begin{equation*}
K \log\log \Big(\max_{k = 1,\ldots,K}\delta_k  \vee e\Big) = o(\log\log\overline{\lambda}).
\end{equation*}
The same bound applies to the other terms and hence
\begin{align*}
&\sum_{k = 1}^{K} \max_{t = \tau_k^O - \delta_k + 1,\ldots,\tau_k^O - 1}  (\tau_k^O - t)\left(\overline{\varepsilon}^O_{\tau_k^O - \delta_k :t}\right)^2
+ \sum_{k = 1}^{K} \max_{t = \tau_k^O - \delta_k + 1,\ldots,\tau_k^O - 1} (t - \tau_k^O + \delta_k)\left(\overline{\varepsilon}^O_{t :\tau_k^O}\right)^2\\
& + \sum_{k = 1}^{K} \max_{t = \tau_k^O + 1,\ldots,\tau_k^O + \delta_k - 1}  (t - \tau_k^O)\left(\overline{\varepsilon}^O_{\tau_k^O :t}\right)^2
+ \sum_{k = 1}^{K} \max_{t = \tau_k^O + 1,\ldots,\tau_k^O + \delta_k - 1} (\tau_k^O + \delta_k - t)\left(\overline{\varepsilon}^O_{t :\tau_k^O + \delta_k}\right)^2\\
= &o_{\Pj}\left(\overline{\sigma}^2 \log\log\overline{\lambda}\right).
\end{align*}

Finally, let $C > 2$ be a constant. Then,
\begin{align*}
&\sum_{k = 1}^{K} \max_{t = \tau_k^O + \delta_k + 1,\ldots,\tau_{k + 1}^O - \delta_{k + 1} - 1} \frac{(t - \tau_k^O + \delta_k) - (t - \tau_k^O - \delta_k)}{t - \tau_k^O + \delta_k} (t - \tau_k^O - \delta_k)\left(\overline{\varepsilon}^O_{\tau_k^O + \delta_{k} :t}\right)^2\\
\leq &\sum_{k = 1}^{K} \max_{t = \tau_k^O + \delta_k + 1,\ldots,\tau_{k + 1}^O - \tau_k^O + C \{\max_{l = 1,\ldots,K}\delta_l\}}  (t - \tau_k^O - \delta_k)\left(\overline{\varepsilon}^O_{\tau_k^O + \delta_{k} :t}\right)^2\\
& + \sum_{k = 1}^{K} \max_{t = \tau_k^O + C \{\max_{l = 1,\ldots,K}\delta_l\} + 1,\ldots,\tau_{k + 1}^O - \delta_{k + 1} - 1} 2\delta_k\left(\overline{\varepsilon}^O_{\tau_k^O + \delta_{k} :t}\right)^2.
\end{align*}
It follows from Lemma~\ref{lemma:maxOneSided} that
\begin{align*}
& \sum_{k = 1}^{K} \max_{t = \tau_k^O + \delta_k + 1,\ldots,\tau_{k + 1}^O - \tau_k^O + C \{\max_{l = 1,\ldots,K}\delta_l\}}  (t - \tau_k^O - \delta_k)\left(\overline{\varepsilon}^O_{\tau_k^O + \delta_{k} :t}\right)^2\\
\leq&  K \max_{k = 1,\ldots,K} \max_{t = \tau_k^O + \delta_k + 1,\ldots,\tau_{k + 1}^O - \tau_k^O + C \{\max_{l = 1,\ldots,K}\delta_l\}}  (t - \tau_k^O - \delta_k)\left(\overline{\varepsilon}^O_{\tau_k^O + \delta_{k} :t}\right)^2\\
= & \mathcal{O}_{\Pj}\left(K \overline{\sigma}^2\left(\log\log \Big(\max_{k = 1,\ldots,K}\delta_k  \vee e\Big) + (\log K)^2\right)\right).
\end{align*}

Moreover, it follows from Lemma~\ref{lemma:specialCaseMaxOneSided} that
\begin{align*}
&\sum_{k = 1}^{K} \max_{t = \tau_k^O + C \{\max_{l = 1,\ldots,K}\delta_l \big\} + 1,\ldots,\tau_{k + 1}^O - \delta_{k + 1} - 1} 2\delta_k\left(\overline{\varepsilon}^O_{\tau_k^O + \delta_{k} :t}\right)^2\\
\leq & 2\sum_{k = 1}^{K} \big\{\max_{l = 1,\ldots,K}\delta_l \big\} \delta_k \max_{t = \tau_k^O + C \{\max_{l = 1,\ldots,K}\delta_l\} + 1,\ldots,\tau_{k + 1}^O - \delta_{k + 1} - 1} \left(\overline{\varepsilon}^O_{\tau_k^O + \delta_{k} :t}\right)^2\\
\leq & \mathcal{O}_{\Pj}\left(K \overline{\sigma}^2  (\log K)^2\right).
\end{align*}

Thus, similarly to our earlier argument, it follows from Assumptions~\ref{assumption:detectionPrecision}\ref{assumption:detectionPrecision:Rate}~and~Assumption~\ref{assumption:cpNumberMultivariate}\ref{assumption:cpNumberBound} that
\begin{align*}
&\sum_{k = 1}^{K} \max_{t = \tau_k^O + \delta_k + 1,\ldots,\tau_{k + 1}^O - \delta_{k + 1} - 1} \frac{(t - \tau_k^O + \delta_k) - (t - \tau_k^O - \delta_k)}{t - \tau_k^O + \delta_k} (t - \tau_k^O - \delta_k)\left(\overline{\varepsilon}^O_{\tau_k^O + \delta_{k} :t}\right)^2\\
& + \sum_{k = 1}^{K} \max_{t = \tau_{k - 1}^O + \delta_{k - 1} + 1,\ldots,\tau_k^O - \delta_{k} - 1} \frac{(\tau_k^O + \delta_{k} - t) - (\tau_k^O - \delta_{k} - t)}{\tau_k^O + \delta_{k} - t} (\tau_k^O - \delta_{k} - t)\left(\overline{\varepsilon}^O_{t :\tau_k^O - \delta_{k}}\right)^2\\
= &o_{\Pj}\left(\overline{\sigma}^2 \log\log\overline{\lambda}\right).
\end{align*}
This completes the proof.
\end{proof}

\begin{Lemma}\label{lemma:differenceCVLtoK}
Recall $E_i := \mu^O_i + \varepsilon^E_i,\ i = 1,\ldots, n / 2$ and $\widetilde{\operatorname{CV}}^O_{\mathrm{mod}}(L)$ as defined in \eqref{eq:simplifiedCVterm}. Suppose that Assumptions~\ref{assumption:cpNumberMultivariate}--\ref{assumption:minimumSignalMultivariate} hold in the case where $K\geq 1$ eventually, or only Assumptions~\ref{assumption:cpNumberMultivariate}, \ref{assumption:NoiseBernstein}~and~\ref{assumption:overestimation} in the case $K = 0\ \forall\ n$. Then we have the following.
\begin{enumerate}[label=(\roman*)]
\item  \label{item:differenceCVL<K} Suppose that, in addition, $K\geq 1$.
%Let us write  $\overline{\mu}^O_{L, i} := \sum_{l = 0}^{L}{\EINS_{\{\hat{\tau}^O_{L,l} + 1 \leq i \leq \hat{\tau}^O_{L,l + 1}\}} \overline{\mu}^O_{\hat{\tau}^O_{L,l} :\hat{\tau}^O_{L,l + 1}}}$.
If for some $A>0$, $\mathcal{I}_L\subseteq \{1,\ldots,K\}$ for $L=0,\ldots,K-1$ is a sequence of non-empty sets satisfying
	\begin{equation*}
		\Pj\Bigg(\forall\ L < K,\  \forall\ k\in \mathcal{I}_L,\ \sum_{i = \tau_{k}^O - \floor{\frac{\underline{\lambda}}{4}} + 1}^{\tau_{k}^O + \floor{\frac{\underline{\lambda}}{4}}}{\big(\mu^O_i - \overline{\mu}^O_{L, i}\big)^2} \geq  A\underline{\lambda}\Delta_k^2 \Bigg) \to 1,
	\end{equation*}
	then
\begin{align*}
	& \min_{L = 0,\ldots,K-1} \Bigg(\sum_{k\in \mathcal{I}_L}{\Delta_k^2}\Bigg)^{-1} \left\{\widetilde{\operatorname{CV}}^O_{\mathrm{mod}}(L) - \widetilde{\operatorname{CV}}^O_{\mathrm{mod}}(K) \right\} \geq \underline{\lambda} \left(A + o_\Pj(1)\right).
\end{align*}	

%Suppose that, in addition, $K\geq 1$ and for a constant $A>0$ there exist sequences of non-empty sets $\mathcal{I}_L\subseteq \{1,\ldots,K\}$ such that
%\begin{equation*}
%\Pj\Bigg(\forall\ L < K,\  \forall\ k\in \mathcal{I}_L,\ \sum_{i = \tau_{k}^O - \frac{\underline{\lambda}}{4} + 1}^{\tau_{k}^O + \frac{\underline{\lambda}}{4}}{(\mu^O_i - \overline{\mu}^O_{L, i})^2} \geq  A\underline{\lambda}\Delta_k^2 \Bigg) \to 1,
%\end{equation*}
%with $\overline{\mu}^O_{L, i} := \sum_{l = 0}^{L}{\EINS_{\{\hat{\tau}^O_{L,l} + 1 \leq i \leq \hat{\tau}^O_{L,l + 1}\}} \overline{\mu}^O_{\hat{\tau}^O_{L,l} :\hat{\tau}^O_{L,l + 1}}}$. Then,
%\begin{align*}
%& \min_{L = 0,\ldots,K-1} \left(\sum_{k\in \mathcal{I}_L}{\Delta_k^2}\right)^{-1} \left\{\widetilde{\operatorname{CV}}^O_{\mathrm{mod}}(L) - \widetilde{\operatorname{CV}}^O_{\mathrm{mod}}(K) \right\} \geq \underline{\lambda} \left(A + o_\Pj(1)\right).
%\end{align*}
\item  \label{item:differenceCVL>K} 
\begin{align*}
&\min_{L = K + 1,\ldots,K_{\max}} \left\{\widetilde{\operatorname{CV}}^O_{\mathrm{mod}}(L) - \widetilde{\operatorname{CV}}^O_{\mathrm{mod}}(K) \right\}\\
=&  \min_{L = K + 1,\ldots,K_{\max}} \left\{ S_{\varepsilon^O}\left(\mathcal{T}^{O}_{K}\right) - S_{\varepsilon^O}\left(\hat{\mathcal{T}}^O_L \cup \mathcal{T}^{O}_{K} \right) \right\} \left(1 + o_{\Pj}(1)\right).
\end{align*}
\end{enumerate}
\end{Lemma}
\begin{proof}
%For any $L = 0,\ldots,K_{\max}$ we have that
%\begin{align*}
%& \widetilde{\operatorname{CV}}^O_{\mathrm{mod}}(L)
%=  S_{E}\left( \hat{\mathcal{T}}^O_L \right) - S_{\varepsilon^O}\left( \hat{\mathcal{T}}^O_L \right) - S_{\varepsilon^E}\left( \hat{\mathcal{T}}^O_L \right)\\
%& + 2 \sum_{l = 0}^{L}\sum_{i = \hat{\tau}^O_{L,l} + 1}^{\hat{\tau}^O_{L,l + 1}} \left(\varepsilon^O_i - \overline{\varepsilon}^O_{\hat{\tau}^O_{L,l} :\hat{\tau}^O_{L,l + 1}}\right)\left(\varepsilon^E_i - \overline{\varepsilon}^E_{\hat{\tau}^O_{L,l} :\hat{\tau}^O_{L,l + 1}}\right) + \sum_{i = 1}^{n / 2} \left(\varepsilon^O_i - \varepsilon^E_i\right)^2.
%\end{align*}
%Consequently, for any $L \neq K$,
%\begin{align*}
%& \widetilde{\operatorname{CV}}^O_{\mathrm{mod}}(L)  - \widetilde{\operatorname{CV}}^O_{\mathrm{mod}}(K)\\
%= & \left\{ S_{E}\left( \hat{\mathcal{T}}^O_L \right) - S_{E}\left( \hat{\mathcal{T}}^O_K \right)\right\}
%- \left\{ S_{\varepsilon^O}\left( \hat{\mathcal{T}}^O_L \right) - S_{\varepsilon^O}\left( \hat{\mathcal{T}}^O_K \right)\right\}
%- \left\{ S_{\varepsilon^E}\left( \hat{\mathcal{T}}^O_L \right) - S_{\varepsilon^E}\left( \hat{\mathcal{T}}^O_K \right)\right\}\\
%& + 2 \left\{ S_{\varepsilon^O, \varepsilon^E}\left( \hat{\mathcal{T}}^O_L \right) - S_{\varepsilon^O, \varepsilon^E}\left( \hat{\mathcal{T}}^O_K \right)\right\},
%\end{align*}
%with $S_{X, Y}\left( \mathcal{T} \right)$ as defined before Lemma~\ref{lemma:differenceCostsMixedTerm}. In the following we will derive uniform bounds for each of the four terms separately.

As seen in \eqref{eq:splitCostDifferenceModCost}, we have for any $L \neq K$,
\begin{align*}
& \widetilde{\operatorname{CV}}^O_{\mathrm{mod}}(L)  - \widetilde{\operatorname{CV}}^O_{\mathrm{mod}}(K)\\
= & \left\{ S_{E}\left( \hat{\mathcal{T}}^O_L \right) - S_{E}\left( \hat{\mathcal{T}}^O_K \right)\right\}
- \left\{ S_{\varepsilon^O}\left( \hat{\mathcal{T}}^O_L \right) - S_{\varepsilon^O}\left( \hat{\mathcal{T}}^O_K \right)\right\}
- \left\{ S_{\varepsilon^E}\left( \hat{\mathcal{T}}^O_L \right) - S_{\varepsilon^E}\left( \hat{\mathcal{T}}^O_K \right)\right\}\\
& + 2 \left\{ S_{\varepsilon^O, \varepsilon^E}\left( \hat{\mathcal{T}}^O_L \right) - S_{\varepsilon^O, \varepsilon^E}\left( \hat{\mathcal{T}}^O_K \right)\right\} \\
&=: A_1 + A_2 + A_3 + 2A_4.
\end{align*}
In the following we will derive uniform bounds for each of the $A_j$ separately.

\subsubsection*{Showing \ref{item:differenceCVL<K}:}
Let $K\geq 1$. We have that
\begin{align*}
S_{E}\left( \hat{\mathcal{T}}^O_L \right) - S_{E}\left( \hat{\mathcal{T}}^O_K \right)
= \left\{ S_{E}\left( \hat{\mathcal{T}}^O_L \right) - S_{E}\left( \mathcal{T}^{O}_{K} \right)\right\}
- \left\{ S_{E}\left( \hat{\mathcal{T}}^O_K \right) - S_{E}\left( \mathcal{T}^{O}_{K} \right)\right\}.
\end{align*}
\emph{Bounding $A_1$:} It follows from Lemma~\ref{lemma:boundsForYE}, \ref{item1:lemma:boundsForYE}, \ref{item4:lemma:boundsForYE}, the fact that $\vert \mathcal{I}_L \vert \geq 1$ and Assumption~\ref{assumption:minimumSignalMultivariate}, that
\begin{align}\label{eq:differenceCVL<KfirstTerm}
\min_{L = 0,\ldots,K-1} \Bigg(\sum_{k\in \mathcal{I}_L}{\Delta_k^2}\Bigg)^{-1} \left\{ S_{E}\left( \hat{\mathcal{T}}^O_L \right) - S_{E}\left( \hat{\mathcal{T}}^O_K \right)\right\} \geq \underline{\lambda} \left(A + o_\Pj(1)\right).
\end{align}
\emph{Bounding $A_2$:} It follows from Lemma~\ref{lemma:boundsForEpsilonO} that
\begin{equation}\label{eq:differenceCVL<KsecondTerm}
\begin{split}
& \max_{L = 0,\ldots,K - 1}\left\vert S_{\varepsilon^O}\left( \hat{\mathcal{T}}^O_L \right) - S_{\varepsilon^O}\left( \hat{\mathcal{T}}^O_K \right)\right\vert\\
\leq & \max_{L = 0,\ldots,K - 1} \left\{ S_{\varepsilon^O}\left( \hat{\mathcal{T}}^O_L \right) - S_{\varepsilon^O}\left(\hat{\mathcal{T}}^O_L \cup \mathcal{T}^{O}_{K} \right)\right\}
 + \max_{L = 0,\ldots,K - 1} \left\{ S_{\varepsilon^O}\left(\mathcal{T}^{O}_{K}\right) - S_{\varepsilon^O}\left(\hat{\mathcal{T}}^O_L \cup \mathcal{T}^{O}_{K} \right) \right\}\\
& + \left\{ S_{\varepsilon^O}\left(\mathcal{T}^{O}_{K}\right) - S_{\varepsilon^O}\left(\hat{\mathcal{T}}^O_K \cup \mathcal{T}^{O}_{K} \right)\right\}
 + \left\{ S_{\varepsilon^O}\left( \hat{\mathcal{T}}^O_K \right) - S_{\varepsilon^O}\left(\hat{\mathcal{T}}^O_K \cup \mathcal{T}^{O}_{K} \right) \right\}\\
= & \mathcal{O}_{\Pj}(K \overline{\sigma}^2(\log \overline{\lambda})^2).
\end{split}
\end{equation}

\emph{Bounding $A_3$:} A bound for the third term is given in \eqref{eq:alsoDifferenceCVL<KthirdTerm}. We have that
\begin{equation}\label{eq:differenceCVL<KthirdTerm}
\begin{split}
& \max_{L = 0,\ldots,K - 1}\left\{ S_{\varepsilon^E}\left( \hat{\mathcal{T}}^O_L \right) - S_{\varepsilon^E}\left( \hat{\mathcal{T}}^O_K \right)\right\}
= o_{\Pj}(\overline{\sigma}^2 \log\log\overline{\lambda}).
\end{split}
\end{equation}

\emph{Bounding $A_4$:} It follows from Lemma~\ref{lemma:differenceCostsMixedTerm} that
\begin{equation}\label{eq:differenceCVL<KfourthTerm}
\begin{split}
& 2\max_{L = 0,\ldots,K - 1}\left\vert S_{\varepsilon^O, \varepsilon^E}\left( \hat{\mathcal{T}}^O_L \right) - S_{\varepsilon^O, \varepsilon^E}\left( \hat{\mathcal{T}}^O_K \right)\right\vert\\
\leq & 2\max_{L = 0,\ldots,K - 1} \left\vert S_{\varepsilon^O, \varepsilon^E}\left( \hat{\mathcal{T}}^O_L \right) - S_{\varepsilon^O, \varepsilon^E}\left(\hat{\mathcal{T}}^O_L \cup \mathcal{T}^{O}_{K} \right)\right\vert\\
& + 2\max_{L = 0,\ldots,K - 1} \left\vert S_{\varepsilon^O, \varepsilon^E}\left(\hat{\mathcal{T}}^O_L \cup \mathcal{T}^{O}_{K} \right) - S_{\varepsilon^O, \varepsilon^E}\left(\mathcal{T}^{O}_{K}\right)\right\vert\\
& + 2\left\vert S_{\varepsilon^O, \varepsilon^E}\left(\mathcal{T}^{O}_{K}\right) - S_{\varepsilon^O, \varepsilon^E}\left(\hat{\mathcal{T}}^O_K \cup \mathcal{T}^{O}_{K} \right)\right\vert
 + 2\left\vert S_{\varepsilon^O, \varepsilon^E}\left(\hat{\mathcal{T}}^O_K \cup \mathcal{T}^{O}_{K} \right) - S_{\varepsilon^O, \varepsilon^E}\left( \hat{\mathcal{T}}^O_K \right)\right\vert\\
\leq &\max_{L = 0,\ldots,K - 1} \left\{ S_{\varepsilon^O}\left( \hat{\mathcal{T}}^O_L \right) - S_{\varepsilon^O}\left(\hat{\mathcal{T}}^O_L \cup \mathcal{T}^{O}_{K} \right)\right\} + \max_{L = 0,\ldots,K - 1} \left\{ S_{\varepsilon^E}\left( \hat{\mathcal{T}}^O_L \right) - S_{\varepsilon^E}\left(\hat{\mathcal{T}}^O_L \cup \mathcal{T}^{O}_{K} \right)\right\}\\
& + \max_{L = 0,\ldots,K - 1} \left\{ S_{\varepsilon^O}\left(\mathcal{T}^{O}_{K}\right) - S_{\varepsilon^O}\left(\hat{\mathcal{T}}^O_L \cup \mathcal{T}^{O}_{K} \right)\right\} + \max_{L = 0,\ldots,K - 1} \left\{ S_{\varepsilon^E}\left(\mathcal{T}^{O}_{K}\right) - S_{\varepsilon^E}\left(\hat{\mathcal{T}}^O_L \cup \mathcal{T}^{O}_{K} \right)\right\}\\
& + \left\{ S_{\varepsilon^O}\left(\mathcal{T}^{O}_{K}\right) - S_{\varepsilon^O}\left(\hat{\mathcal{T}}^O_K \cup \mathcal{T}^{O}_{K} \right)\right\} + \left\{ S_{\varepsilon^E}\left(\mathcal{T}^{O}_{K}\right) - S_{\varepsilon^E}\left(\hat{\mathcal{T}}^O_K \cup \mathcal{T}^{O}_{K} \right)\right\}\\
& + \left\{ S_{\varepsilon^O}\left( \hat{\mathcal{T}}^O_K \right) - S_{\varepsilon^O}\left(\hat{\mathcal{T}}^O_K \cup \mathcal{T}^{O}_{K} \right) \right\} + \left\{ S_{\varepsilon^E}\left( \hat{\mathcal{T}}^O_K \right) - S_{\varepsilon^E}\left(\hat{\mathcal{T}}^O_K \cup \mathcal{T}^{O}_{K} \right) \right\}\\
=& \mathcal{O}_{\Pj}\left(K \overline{\sigma}^2 (\log \overline{\lambda})^2\right).
\end{split}
\end{equation}
where the last bound is a consequence of Lemmas~\ref{lemma:differenceCostsNestedChangePointSets},~\ref{lemma:boundsForEpsilonE}~and
\ref{lemma:boundsForEpsilonO} as seen before.

Thus, by combining \eqref{eq:differenceCVL<KfirstTerm}--\eqref{eq:differenceCVL<KfourthTerm} we have that
\begin{align*}
& \min_{L = 0,\ldots,K-1} \Bigg(\underline{\lambda}\sum_{k\in \mathcal{I}_L}{\Delta_k^2}\Bigg)^{-1} \left\{\widetilde{\operatorname{CV}}^O_{\mathrm{mod}}(L) - \widetilde{\operatorname{CV}}^O_{\mathrm{mod}}(K) \right\}\\
\geq & A + o_\Pj(1) + \mathcal{O}_{\Pj}\Bigg(\min_{L = 0,\ldots,K-1} \Bigg(\underline{\lambda}\sum_{k\in \mathcal{I}_L}{\Delta_k^2}\Bigg)^{-1} K \overline{\sigma}^2 (\log \overline{\lambda})^2\Bigg).
\end{align*}
Hence, \ref{item:differenceCVL<K} follows from the fact that $\vert \mathcal{I}_L \vert \geq 1$ and Assumption~\ref{assumption:minimumSignalMultivariate}.

\subsubsection*{Showing \ref{item:differenceCVL>K}:}
\emph{Bounding $A_1$:} It follows from Lemma~\ref{lemma:boundsForYE}\ref{item3:lemma:boundsForYE}~and~\ref{item4:lemma:boundsForYE} that
\begin{equation}\label{eq:differenceCVL>KfirstTerm}
\begin{split}
&\min_{L = K + 1,\ldots,K_{\max}} \left\{ S_{E}\left( \hat{\mathcal{T}}^O_L \right) - S_{E}\left( \hat{\mathcal{T}}^O_K \right)\right\}\\
\geq &\min_{L = K + 1,\ldots,K_{\max}} \left\{ S_{E}\left(\hat{\mathcal{T}}^O_L\right) - S_{E}\left(\hat{\mathcal{T}}^O_L \cup \mathcal{T}^{O}_{K} \right)\right\}\\
& - \max_{L = K + 1,\ldots,K_{\max}} \left\{ S_{E}\left(\mathcal{T}^{O}_{K}\right) - S_{E}\left(\hat{\mathcal{T}}^O_L \cup \mathcal{T}^{O}_{K} \right) \right\}
 - \left\vert S_{E}\left(\hat{\mathcal{T}}^O_K \right) - S_{E}\left(\mathcal{T}^{O}_{K}\right)\right\vert\\
\geq & - \max_{L = K + 1,\ldots,K_{\max}} \left\{ S_{E}\left(\mathcal{T}^{O}_{K}\right) - S_{E}\left(\hat{\mathcal{T}}^O_L \cup \mathcal{T}^{O}_{K} \right) \right\}
 - \left\vert S_{E}\left(\hat{\mathcal{T}}^O_K \right) - S_{E}\left(\mathcal{T}^{O}_{K}\right)\right\vert\\
= & -\left\vert o_{\Pj}(\overline{\sigma}^2 \log\log \overline{\lambda})\right\vert.
\end{split}
\end{equation}

\emph{Bounding $A_2$:} It follows from Lemma~\ref{lemma:boundsForEpsilonO} that
\begin{equation}\label{eq:differenceCVL>KsecondTerm}
\begin{split}
& \min_{L = K + 1,\ldots,K_{\max}}\left\{ S_{\varepsilon^O}\left( \hat{\mathcal{T}}^O_K \right) - S_{\varepsilon^O}\left( \hat{\mathcal{T}}^O_L \right)\right\}\\
\geq & -\max_{L = K + 1,\ldots,K_{\max}} \left\{ S_{\varepsilon^O}\left( \hat{\mathcal{T}}^O_L \right) - S_{\varepsilon^O}\left(\hat{\mathcal{T}}^O_L \cup \mathcal{T}^{O}_{K} \right)\right\}
 + \min_{L = K + 1,\ldots,K_{\max}} \left\{ S_{\varepsilon^O}\left(\mathcal{T}^{O}_{K}\right) - S_{\varepsilon^O}\left(\hat{\mathcal{T}}^O_L \cup \mathcal{T}^{O}_{K} \right) \right\}\\
& - \left\{ S_{\varepsilon^O}\left(\mathcal{T}^{O}_{K}\right) - S_{\varepsilon^O}\left(\hat{\mathcal{T}}^O_K \cup \mathcal{T}^{O}_{K} \right)\right\}
 - \left\{ S_{\varepsilon^O}\left( \hat{\mathcal{T}}^O_K \right) - S_{\varepsilon^O}\left(\hat{\mathcal{T}}^O_K \cup \mathcal{T}^{O}_{K} \right) \right\}\\
= & \min_{L = K + 1,\ldots,K_{\max}} \left\{ S_{\varepsilon^O}\left(\mathcal{T}^{O}_{K}\right) - S_{\varepsilon^O}\left(\hat{\mathcal{T}}^O_L \cup \mathcal{T}^{O}_{K} \right) \right\} + o_{\Pj}(\overline{\sigma}^2 \log\log\overline{\lambda}).
\end{split}
\end{equation}

\emph{Bounding $A_3$:} It follows from Lemma~\ref{lemma:boundsForEpsilonE} that
\begin{equation}\label{eq:differenceCVL>KthirdTerm}
\begin{split}
& -\min_{L = K + 1,\ldots,K_{\max}}\left\{ S_{\varepsilon^E}\left( \hat{\mathcal{T}}^O_L \right) - S_{\varepsilon^E}\left( \hat{\mathcal{T}}^O_K \right)\right\} = \max_{L = K + 1,\ldots,K_{\max}}\left\{ S_{\varepsilon^E}\left( \hat{\mathcal{T}}^O_L \right) - S_{\varepsilon^E}\left( \hat{\mathcal{T}}^O_K \right)\right\}\\
\leq & \max_{L = K + 1,\ldots,K_{\max}} \left\{ S_{\varepsilon^E}\left( \hat{\mathcal{T}}^O_L \right) - S_{\varepsilon^E}\left(\hat{\mathcal{T}}^O_L \cup \mathcal{T}^{O}_{K} \right)\right\}
 + \max_{L = K + 1,\ldots,K_{\max}} \left\{ S_{\varepsilon^E}\left(\hat{\mathcal{T}}^O_L \cup \mathcal{T}^{O}_{K} \right) - S_{\varepsilon^E}\left(\mathcal{T}^{O}_{K}\right)\right\}\\
& + \left\{ S_{\varepsilon^E}\left(\mathcal{T}^{O}_{K}\right) -  S_{\varepsilon^E}\left(\hat{\mathcal{T}}^O_K \cup \mathcal{T}^{O}_{K} \right) \right\} + \left\{ S_{\varepsilon^E}\left(\hat{\mathcal{T}}^O_K \cup \mathcal{T}^{O}_{K} \right) -  S_{\varepsilon^E}\left( \hat{\mathcal{T}}^O_K \right) \right\}\\
= & o_{\Pj}(\overline{\sigma}^2 \log\log\overline{\lambda}).
\end{split}
\end{equation}

\emph{Bounding $A_4$:} We have that
\begin{align*}
& 2\left\vert S_{\varepsilon^O, \varepsilon^E}\left( \hat{\mathcal{T}}^O_L \right) - S_{\varepsilon^O, \varepsilon^E}\left( \hat{\mathcal{T}}^O_K \right)\right\vert\\
\leq & 2\left\vert S_{\varepsilon^O, \varepsilon^E}\left( \hat{\mathcal{T}}^O_L \right) - S_{\varepsilon^O, \varepsilon^E}\left(\hat{\mathcal{T}}^O_L \cup \mathcal{T}^{O}_{K} \right)\right\vert
 + 2\left\vert S_{\varepsilon^O, \varepsilon^E}\left(\hat{\mathcal{T}}^O_L \cup \mathcal{T}^{O}_{K} \right) - S_{\varepsilon^O, \varepsilon^E}\left(\mathcal{T}^{O}_{K}\right)\right\vert\\
& + 2\left\vert S_{\varepsilon^O, \varepsilon^E}\left(\mathcal{T}^{O}_{K}\right) - S_{\varepsilon^O, \varepsilon^E}\left(\hat{\mathcal{T}}^O_K \cup \mathcal{T}^{O}_{K} \right)\right\vert
 + 2\left\vert S_{\varepsilon^O, \varepsilon^E}\left(\hat{\mathcal{T}}^O_K \cup \mathcal{T}^{O}_{K} \right) - S_{\varepsilon^O, \varepsilon^E}\left( \hat{\mathcal{T}}^O_K \right)\right\vert.
\end{align*}
Moreover, it follows from the Cauchy--Schwarz inequality that
\begin{align*}
& \left\vert S_{\varepsilon^O, \varepsilon^E}\left(\hat{\mathcal{T}}^O_L \cup \mathcal{T}^{O}_{K} \right) - S_{\varepsilon^O, \varepsilon^E}\left(\mathcal{T}^{O}_{K}\right)\right\vert\\
\leq & \left(S_{\varepsilon^O}\left(\mathcal{T}^{O}_{K}\right) - S_{\varepsilon^O}\left(\hat{\mathcal{T}}^O_L \cup \mathcal{T}^{O}_{K} \right) \right)^{1/2}
\left(S_{\varepsilon^E}\left(\mathcal{T}^{O}_{K}\right) - S_{\varepsilon^E}\left(\hat{\mathcal{T}}^O_L \cup \mathcal{T}^{O}_{K} \right) \right)^{1/2}
\end{align*}
Thus, it follows from Lemmas~\ref{lemma:differenceCostsMixedTerm}, \ref{lemma:differenceCostsNestedChangePointSets},~\ref{lemma:boundsForEpsilonE}~and~\ref{lemma:boundsForEpsilonO} that
\begin{equation}\label{eq:differenceCVL>KfourthTerm}
\begin{split}
& \min_{L = K + 1,\ldots,K_{\max}} \left\vert S_{\varepsilon^O, \varepsilon^E}\left( \hat{\mathcal{T}}^O_L \right) - S_{\varepsilon^O, \varepsilon^E}\left( \hat{\mathcal{T}}^O_K \right)\right\vert\\
= & o_{\Pj}\left( \min_{L = K + 1,\ldots,K_{\max}} \left\{ S_{\varepsilon^O}\left(\mathcal{T}^{O}_{K}\right) - S_{\varepsilon^O}\left(\hat{\mathcal{T}}^O_L \cup \mathcal{T}^{O}_{K} \right) \right\} \right) + o_{\Pj}(\overline{\sigma}^2 \log\log\overline{\lambda}).
\end{split}
\end{equation}

Thus, by combining \eqref{eq:differenceCVL>KfirstTerm}--\eqref{eq:differenceCVL>KfourthTerm} we have that
\begin{align*}
& \min_{L = K + 1,\ldots,K_{\max}} \left\{ \widetilde{\operatorname{CV}}^O_{\mathrm{mod}}(L) - \widetilde{\operatorname{CV}}^O_{\mathrm{mod}}(K) \right\}\\
= & \min_{L = K + 1,\ldots,K_{\max}} \left\{ S_{\varepsilon^O}\left(\mathcal{T}^{O}_{K}\right) - S_{\varepsilon^O}\left(\hat{\mathcal{T}}^O_L \cup \mathcal{T}^{O}_{K} \right) \right\}\\
& + o_{\Pj}\left( \min_{L = K + 1,\ldots,K_{\max}} \left\{ S_{\varepsilon^O}\left(\mathcal{T}^{O}_{K}\right) - S_{\varepsilon^O}\left(\hat{\mathcal{T}}^O_L \cup \mathcal{T}^{O}_{K} \right) \right\} \right) + o_{\Pj}(\overline{\sigma}^2 \log\log\overline{\lambda}).
\end{align*}
Hence, \ref{item:differenceCVL>K} follows from Assumption~\ref{assumption:overestimation}.
\end{proof}

\begin{Lemma}\label{lemma:preliminaryCalculationsDifferenceCVs}
Suppose that Assumptions~\ref{assumption:cpNumberMultivariate}--\ref{assumption:minimumSignalMultivariate} hold in the case where $K\geq 1$ eventually.
%, or only Assumptions~\ref{assumption:cpNumberMultivariate}, \ref{assumption:NoiseBernstein}~and~\ref{assumption:overestimation} in the case $K = 0\ \forall\ n$. 
Let
\begin{align*}
B_{1, L, l}^{(n)} := & \left\vert \overline{\mu}^E_{\hat{\tau}^O_{L,l} :\hat{\tau}^O_{L,l + 1} - 1} - \overline{\mu}^O_{\hat{\tau}^O_{L,l} :\hat{\tau}^O_{L,l + 1}}\right\vert,\\
B_{2, L, l}^{(n)} := & \sum_{i = \hat{\tau}^O_{L,l} + 1}^{\hat{\tau}^O_{L,l + 1} - 1}{ \left\vert\mu^E_i - \mu^O_i\right\vert},\\
B_{3, L, l}^{(n)} := & \max_{i = \hat{\tau}^O_{L,l} + 1,\ldots,\hat{\tau}^O_{L,l + 1} - 1}{ \left\vert\mu^E_i - \overline{\mu}^E_{\hat{\tau}^O_{L,l} :\hat{\tau}^O_{L,l + 1} - 1}\right\vert},\\
B_{4, L, l}^{(n)} := & \max_{i = \hat{\tau}^O_{L,l} + 1,\ldots,\hat{\tau}^O_{L,l + 1}}{ \left\vert \mu^O_i - \overline{\mu}^O_{\hat{\tau}^O_{L,l} :\hat{\tau}^O_{L,l + 1}} \right\vert },
\end{align*}
and
\begin{align*}
B_{L}^{(n)} := \sum_{l = 0}^L \max\left(B_{1, L, l}^{(n)}, B_{2, L, l}^{(n)}, B_{3, L, l}^{(n)}, B_{4, L, l}^{(n)}\right)^2.
\end{align*}
Then we have the following.
\begin{enumerate}[label=(\roman*)]
	 \item \begin{equation*}
	 	\max_{L = K + 1,\ldots,K_{\max}} \left(S_{\varepsilon^O}\left(\mathcal{T}^{O}_{K}\right) - S_{\varepsilon^O}\left(\hat{\mathcal{T}}^O_L\cup \mathcal{T}^{O}_{K}\right)\right)^{-1} \max\Big(B_K^{(n)},B_L^{(n)}\Big) = o_{\Pj}(1).
	 \end{equation*}
	 \item When $K \geq 1$, there exists a constant $A>0$ such that defining
	 \[
	 \overline{\mu}^O_{L, i} := \sum_{l = 0}^{L}{\EINS_{\{\hat{\tau}^O_{L,l} + 1 \leq i \leq \hat{\tau}^O_{L,l + 1}\}} \overline{\mu}^O_{\hat{\tau}^O_{L,l} :\hat{\tau}^O_{L,l + 1}}}
	 \] and
	 \begin{equation*}
	 \mathcal{I}_L := \Bigg\{ k : \sum_{i = \tau_{k}^O - \floor{\frac{\underline{\lambda}}{4}} + 1}^{\tau_{k}^O + \floor{\frac{\underline{\lambda}}{4}}}{\big(\mu^O_i - \overline{\mu}^O_{L, i}\big)^2} \geq  A \underline{\lambda}\Delta_k^2  \Bigg\},
	 \end{equation*}
	 we have that $\mathcal{I}_L \neq \emptyset$ for all $L \leq K-1$ and all $n$, and
%	 \Rajen{we have $\Pj(\forall L \leq K-1, \; \mathcal{I}_L \text{ is non-empty}) \to 1$, where }
%	 is non-empty and
	 \begin{equation*}
	 	\max_{L = 0,\ldots,K - 1} \Bigg(\sum_{k\in \mathcal{I}_L}{\Delta_k^2}\Bigg)^{-1} \max\Big(B_K^{(n)},B_L^{(n)}\Big) = o_\Pj\left(\underline{\lambda}\right).
	 \end{equation*}
\end{enumerate}
%and there exists a constant $A>0$ such that
%\begin{equation*}
%\max_{L = 0,\ldots,K - 1} \left(\sum_{k\in \mathcal{I}_L}{\Delta_k^2}\right)^{-1} \max\Big(B_K^{(n)},B_L^{(n)}\Big) = o_\Pj\left(\underline{\lambda}\right),
%\end{equation*}
%where
%\[
%\mathcal{I}_L := \left\{ k : \sum_{i = \tau_{k}^O - \frac{\underline{\lambda}}{4} + 1}^{\tau_{k}^O + \frac{\underline{\lambda}}{4}}{(\mu^O_i - \overline{\mu}^O_{L, i})^2} \geq  A \underline{\lambda}\Delta_k^2  \right\},
%\]
%and $\overline{\mu}^O_{L, i} := \sum_{l = 0}^{L}{\EINS_{\{\hat{\tau}^O_{L,l} + 1 \leq i \leq \hat{\tau}^O_{L,l + 1}\}} \overline{\mu}^O_{\hat{\tau}^O_{L,l} :\hat{\tau}^O_{L,l + 1}}}$.
\end{Lemma}
\begin{proof}
	First note that if $K= 0$ for all $n$, then $B^{(n)}_L = 0$ for all $L$ and so the claims hold trivially. Now suppose $K \geq 1 $ eventually and let the event $\Omega_n$ be the intersection of the events in \ref{assumption:detectionPrecision:L>=K} and \ref{assumption:detectionPrecision:L<K} of Assumption~\ref{assumption:detectionPrecision}. We note that $\Omega_n$ must have probability converging to $1$. In the following we work on $\Omega_n$.

\subsubsection*{Proof of (i).}
Let us fix $L \geq K$. For each $l = 0,\ldots,L$, let $k = k(l) \in \{1,\ldots,K+1\}$ be such that $\tau_{k - 1}^O \leq \hat{\tau}^O_{L,l} < \tau_{k}^O$. 
Further, when $k \neq K+1$, due to the fact that we are on $\Omega_n$, we have that there exists some $\tilde{l}$ such that 
$\tau_{k}^O \leq \hat{\tau}^O_{L,\tilde{l}} \leq \tau_{k+2}^O$ as $ \min(\tau_{k+2}^O - \tau_{k+1}^O, \tau_{k+1}^O  - \tau_{k}^O) \geq \underline{\lambda}/2-1$.

%The other way around, for each $\tilde{k}$ there exists an $\tilde{l}$ such that 
%$\tau_{\tilde{k} - 1}^O < \hat{\tau}^O_{L,\tilde{l}} < \tau_{\tilde{k} + 1}^O$. In other words, there must be an estimated change-point close to $\tau_{\tilde{k} + 1}^O$. 
Therefore, exactly one of the following scenarios must occur
\begin{enumerate}[label=(\arabic*)]
\item $\tau_{k - 1}^O \leq \hat{\tau}^O_{L,l} < \hat{\tau}^O_{L,l + 1} \leq \tau_{k}^O$,
\item $\tau_{k - 1}^O \leq \hat{\tau}^O_{L,l} < \tau_{k}^O < \hat{\tau}^O_{L,l + 1} \leq \tau_{k + 1}^O$,
\item $\tau_{k - 1}^O \leq \hat{\tau}^O_{L,l} < \tau_{k}^O < \tau_{k + 1}^O < \hat{\tau}^O_{L,l + 1} \leq \tau_{k + 2}^O$.
\end{enumerate}
We observe that $\mu^E_i - \mu^O_i = \mu_{2i} - \mu_{2i - 1} = \beta_k - \beta_{k - 1}$ if $\tau_k = 2 i - 1$ and zero if such a $k$ does not exist. Hence, $\vert \mu^E_i - \mu^O_i \vert \leq \Delta_k$ if $\tau_{k}^O = i$ and zero if such a $k$ does not exist. Therefore,

\noindent if (1) occurs, then
\begin{equation*}
\max\left(B_{1, L, l}^{(n)}, B_{2, L, l}^{(n)}, B_{3, L, l}^{(n)}, B_{4, L, l}^{(n)}\right) = 0;
\end{equation*}
if (2) occurs, then
\begin{equation*}
\max\left(B_{1, L, l}^{(n)}, B_{2, L, l}^{(n)}, B_{3, L, l}^{(n)}, B_{4, L, l}^{(n)}\right) \leq \Delta_{k(l)};
\end{equation*}
if (3) occurs, then
\begin{equation*}
\max\left(B_{1, L, l}^{(n)}, B_{2, L, l}^{(n)}, B_{3, L, l}^{(n)}, B_{4, L, l}^{(n)}\right) \leq \Delta_{k(l)} + \Delta_{k(l) + 1}.
\end{equation*}
Also, because of Assumption~\ref{assumption:detectionPrecision}\ref{assumption:detectionPrecision:Infq>0}, if $\overline{\sigma}^2 \log\log \overline{\lambda} / (K \Delta_{k(l)}^2) \leq C_n$, then there exists an $\tilde{l}$ such that $\tau_{k(l)}^O = \hat{\tau}^O_{L,\tilde{l}}$. Consequently, scenario (1) must occur. From the same reasoning, we have that scenario (3) cannot occur if $\overline{\sigma}^2 \log\log \overline{\lambda} / (K \Delta_{k(l) + 1}^2) \leq C_n$. Moreover, by the structure of the scenarios, if $k(l_1) = k(l_2)$ for any $l_1 \neq l_2$, then scenario (1) must occur for $l_1$ or $l_2$ (or for both). Thus, %almost surely,
\begin{equation}\label{eq:preliminaryCalculationForK}
\max_{L = K, \ldots, K_{\max}} B_{L}^{(n)} \EINS_{\Omega_n} \leq 4 K \ \left(C_n^{-1}\overline{\sigma}^2 \log\log \overline{\lambda} / K\right) = o\left(\overline{\sigma}^2\log\log\overline{\lambda}\right).
\end{equation}
Thus we obtain from Assumption~\ref{assumption:overestimation} and $\Pj(\Omega_n) \to 1$ that
\begin{equation*}
\max_{L = K + 1,\ldots,K_{\max}} \frac{\max\Big(B_L^{(n)},B_K^{(n)}\Big)}{S_{\varepsilon^O}\left(\mathcal{T}^{O}_{K}\right) - S_{\varepsilon^O}\left(\hat{\mathcal{T}}^O_L\cup \mathcal{T}^{O}_{K}\right)}
= o_\Pj(1)
\end{equation*}
as required.
%Note that this also holds when we replace the maximum by $L = K$.

\subsubsection*{Proof of (ii).} Let us set $A = \min(1/200, C)$. Fix $L < K$ and $k \in \{1,\ldots,K\}$. Note that
\[
\min_{c \in \R} \sum_{i = \tau_{k}^O - \floor{\frac{\underline{\lambda}}{4}} + 1}^{\tau_{k}^O + \floor{\frac{\underline{\lambda}}{4}}}{(\mu^O_i - c)^2} \geq  A \underline{\lambda}\Delta_k^2.
\]
Consequently, since $L < K$, it follows that $\mathcal{I}_L$ is non-empty. Furthermore, note that Assumption~\ref{assumption:detectionPrecision}\ref{assumption:detectionPrecision:L<K} is also true when $C$ is replaced by $A$.

Let $0 = \hat{k}_{0} \leq \cdots \leq \hat{k}_{L + 1} = K + 1$ be such that $\tau_{\hat{k}_{l}}^O \leq \hat{\tau}^O_{L,l} < \tau_{\hat{k}_{l} + 1}^O$. Then for each $l = 0,\ldots, L$ we have that $\tau_{\hat{k}_{l}}^O \leq \hat{\tau}^O_{L,l} < \hat{\tau}^O_{L,l + 1} < \tau_{\hat{k}_{l + 1} + 1}^O$. Recall that $\vert \mu^E_i - \mu^O_i \vert \leq \Delta_k$ if $\tau_{k}^O = i$ and zero if such a $k$ does not exist. Therefore,
\begin{equation*}
\max\left(B_{1, L, l}^{(n)}, B_{2, L, l}^{(n)}, B_{3, L, l}^{(n)}, B_{4, L, l}^{(n)}\right) \leq \sum_{k = \hat{k}_l + 1}^{\hat{k}_{l + 1}} \Delta_k.
\end{equation*}
Recall that we are working on $\Omega_n$. Due to Assumption~\ref{assumption:detectionPrecision}\ref{assumption:detectionPrecision:L<K}, for each $k \not\in \mathcal{I}_L$ we have that either $k \notin \mathcal{K}_n$ so $\Delta_k^2 < \overline{\sigma}^2 \log\log \overline{\lambda} / (K C_n)$ or there exists an $l = 1,\ldots,L$ such that $\tau_{k}^O = \hat{\tau}^O_{L,l}$. Hence,
\begin{equation*}
\max\left(B_{1, L, l}^{(n)}, B_{2, L, l}^{(n)}, B_{3, L, l}^{(n)}, B_{4, L, l}^{(n)}\right) \leq \sum_{\substack{k = \hat{k}_l + 1, \ldots,\hat{k}_{l + 1},\\ k \in \mathcal{I}_L}}{\Delta_k} + \big(\hat{k}_{l + 1} - \hat{k}_l\big) \sqrt{\overline{\sigma}^2 \log\log \overline{\lambda} / (K C_n)}.
\end{equation*}
Consequently, by  using $(\sum_{i = 1}^{n} a_i)^2 \leq \sum_{i = 1}^{n} n a_i^2$ for any $a_1,\ldots,a_n \in \R$,%  $(x+y)^2\leq 2x^2 + 2y^2$ multiple times
we obtain
\begin{equation*}
\begin{split}
B_{L}^{(n)}
\leq & \sum_{l = 0}^{L}{\Bigg(\sum_{\substack{k = \hat{k}_l + 1, \ldots,\hat{k}_{l + 1},\\ k \in \mathcal{I}_L}}{\Delta_k} + \big(\hat{k}_{l + 1} - \hat{k}_l\big) \sqrt{\overline{\sigma}^2 \log\log \overline{\lambda} / (K C_n)}\Bigg)^2}\\
\leq & 2\sum_{l = 0}^{L}{\Bigg(\sum_{\substack{k = \hat{k}_l + 1, \ldots,\hat{k}_{l + 1},\\ k \in \mathcal{I}_L}}{\Delta_k}\Bigg)^2} + \sum_{l = 0}^{L}{2 \big(\hat{k}_{l + 1} - \hat{k}_l\big)^2 \overline{\sigma}^2 \log\log \overline{\lambda} / (K C_n)}\\
\leq & 2 K  \sum_{k \in \mathcal{I}_L}{\Delta_k^2} + 2 K \overline{\sigma}^2 \log\log \overline{\lambda}.
\end{split}
\end{equation*}
Recall that $\mathcal{I}_L$ is non-empty. Hence, $\sum_{k \in \mathcal{I}_L}{\Delta_k^2} \geq   \Delta_{(1)}^2$. 
Thus, Assumption~\ref{assumption:minimumSignalMultivariate} yields
\[K \overline{\sigma}^2 \log\log \overline{\lambda} \leq K \overline{\sigma}^2 (\log\overline{\lambda})^2  = o_\Pj\bigg(\min_{L = 0,\ldots,K - 1} \underline{\lambda}  \sum_{k \in \mathcal{I}_L}{\Delta_k^2}\bigg)\] 
and hence because of Assumption~\ref{assumption:cpNumberMultivariate}\ref{assumption:cpNumberBound},
\begin{equation*}
\max_{L = 0,\ldots,K - 1}\Bigg(\sum_{k \in \mathcal{I}_L}{\Delta_k^2}\Bigg)^{-1} B_{L}^{(n)} = o_\Pj\left(\underline{\lambda}\right).
\end{equation*}
From \eqref{eq:preliminaryCalculationForK} and Assumption~\ref{assumption:minimumSignalMultivariate} we also obtain that 
\begin{equation*}
\max_{L = 0,\ldots,K - 1}\Bigg(\sum_{k \in \mathcal{I}_L}{\Delta_k^2}\Bigg)^{-1} \max\Big(B_L^{(n)},B_L^{(n)}\Big) = o_\Pj\left(\underline{\lambda}\right). \qedhere
\end{equation*}
\end{proof}

\begin{Lemma}\label{lemma:differenceCVs}
Recall $\operatorname{CV}^O_{\mathrm{mod}}(L)$ and $\widetilde{\operatorname{CV}}^O_{\mathrm{mod}}(L)$ from \eqref{eq:CVmodO} and \eqref{eq:simplifiedCVterm}, respectively. Suppose that Assumptions~\ref{assumption:cpNumberMultivariate}--\ref{assumption:minimumSignalMultivariate} hold in the case where $K\geq 1$ eventually, or only Assumptions~\ref{assumption:cpNumberMultivariate}, \ref{assumption:NoiseBernstein}~and~\ref{assumption:overestimation} in the case $K = 0\ \forall\ n$. Let
\begin{equation*}
A_L^{(n)} := \left\vert \operatorname{CV}^O_{\mathrm{mod}}(L) - \widetilde{\operatorname{CV}}^O_{\mathrm{mod}}(L) \right\vert.
\end{equation*}
Then,
\begin{equation*}
\max_{L = K + 1,\ldots,K_{\max}} \left(S_{\varepsilon^O}\left(\mathcal{T}^{O}_{K}\right) - S_{\varepsilon^O}\left(\hat{\mathcal{T}}^O_L\cup \mathcal{T}^{O}_{K}\right)\right)^{-1} \max\big(A_L^{(n)},\, A_K^{(n)}\big) = o_{\Pj}(1).
\end{equation*}
Furthermore, there exists a constant $A>0$ such that writing
\[
\mathcal{I}_L := \Bigg\{ k : \sum_{i = \tau_{k}^O - \floor{\frac{\underline{\lambda}}{4}} + 1}^{\tau_{k}^O + \floor{\frac{\underline{\lambda}}{4}}}{\big(\mu^O_i - \overline{\mu}^O_{L, i}\big)^2} \geq  A \underline{\lambda}\Delta_k^2  \Bigg\},
\]
where
$\overline{\mu}^O_{L, i} := \sum_{l = 0}^{L}{\EINS_{\{\hat{\tau}^O_{L,l} + 1 \leq i \leq \hat{\tau}^O_{L,l + 1}\}} \overline{\mu}^O_{\hat{\tau}^O_{L,l} :\hat{\tau}^O_{L,l + 1}}}$, we have that $\mathcal{I}_L \neq \emptyset$ for all $L \leq K-1$ and all $n$, and
\begin{equation} \label{eq:A_L_0K_bd}
\max_{L = 0\ldots,K - 1} \Bigg(\sum_{k\in \mathcal{I}_L}{\Delta_k^2}\Bigg)^{-1} \max\big(A_L^{(n)},\, A_K^{(n)}\big) = o_\Pj\left(\underline{\lambda}\right).
\end{equation}
%where
%\[
%\mathcal{I}_L := \left\{ k : \sum_{i = \tau_{k}^O - \floor{\frac{\underline{\lambda}}{4}} + 1}^{\tau_{k}^O + \floor{\frac{\underline{\lambda}}{4}}}{(\mu^O_i - \overline{\mu}^O_{L, i})^2} \geq  A \underline{\lambda}\Delta_k^2  \right\},
%\]
%$\overline{\mu}^O_{L, i} := \sum_{l = 0}^{L}{\EINS_{\{\hat{\tau}^O_{L,l} + 1 \leq i \leq \hat{\tau}^O_{L,l + 1}\}} \overline{\mu}^O_{\hat{\tau}^O_{L,l} :\hat{\tau}^O_{L,l + 1}}}$, and \Rajen{for a constant} $A>0$ is chosen such that $\mathcal{I}_L$ is non-empty.
\end{Lemma}
\begin{proof}
We will split $A_L^{(n)}$ into a sum of five terms and show the statement for each of them separately, thus giving the result. Note that the existence of a constant $A>0$ such that $\mathcal{I}_L$ is non-empty is guaranteed by Lemma~\ref{lemma:preliminaryCalculationsDifferenceCVs}.

\subsubsection*{Splitting $A_L^{(n)}$}
We have
\begin{equation*}
\begin{split}
A_L^{(n)} \leq & \left\vert \sum_{l = 0}^{L}\sum_{i = \hat{\tau}^O_{L,l} + 1}^{\hat{\tau}^O_{L,l + 1} - 1}{ \frac{\hat{n}^O_l}{\hat{n}^O_l - 1} \left(\varepsilon^E_i - \overline{\varepsilon}^O_{\hat{\tau}^O_{L,l} :\hat{\tau}^O_{L,l + 1}}\right)^2}
- \sum_{l = 0}^{L}\sum_{i = \hat{\tau}^O_{L,l} + 1}^{\hat{\tau}^O_{L,l + 1}}{\left(\varepsilon^E_i - \overline{\varepsilon}^O_{\hat{\tau}^O_{L,l} :\hat{\tau}^O_{L,l + 1}}\right)^2} \right\vert\\
& + \left\vert \sum_{l = 0}^{L}\sum_{i = \hat{\tau}^O_{L,l} + 1}^{\hat{\tau}^O_{L,l + 1} - 1}{ \frac{\hat{n}^O_l}{\hat{n}^O_l - 1} \left[   2 \left(\varepsilon^E_i - \overline{\varepsilon}^O_{\hat{\tau}^O_{L,l} :\hat{\tau}^O_{L,l + 1}}\right) \left(\mu^E_i - \overline{\mu}^O_{\hat{\tau}^O_{L,l} :\hat{\tau}^O_{L,l + 1}}\right) + \left(\mu^E_i - \overline{\mu}^O_{\hat{\tau}^O_{L,l} :\hat{\tau}^O_{L,l + 1}}\right)^2 \right] }\right.\\
& - \left.\sum_{l = 0}^{L}\sum_{i = \hat{\tau}^O_{L,l} + 1}^{\hat{\tau}^O_{L,l + 1}}{ \left[   2 \left(\varepsilon^E_i - \overline{\varepsilon}^O_{\hat{\tau}^O_{L,l} :\hat{\tau}^O_{L,l + 1}}\right) \left(\mu^O_i - \overline{\mu}^O_{\hat{\tau}^O_{L,l} :\hat{\tau}^O_{L,l + 1}}\right) + \left(\mu^O_i - \overline{\mu}^O_{\hat{\tau}^O_{L,l} :\hat{\tau}^O_{L,l + 1}}\right)^2 \right] } \right\vert.
\end{split}
\end{equation*}
Also,
\begin{equation*}
\begin{split}
&\sum_{l = 0}^{L}\sum_{i = \hat{\tau}^O_{L,l} + 1}^{\hat{\tau}^O_{L,l + 1} - 1}{ \frac{\hat{n}^O_l}{\hat{n}^O_l - 1} \left[   2 \left(\varepsilon^E_i - \overline{\varepsilon}^O_{\hat{\tau}^O_{L,l} :\hat{\tau}^O_{L,l + 1}}\right) \left(\mu^E_i - \overline{\mu}^O_{\hat{\tau}^O_{L,l} :\hat{\tau}^O_{L,l + 1}}\right) + \left(\mu^E_i - \overline{\mu}^O_{\hat{\tau}^O_{L,l} :\hat{\tau}^O_{L,l + 1}}\right)^2 \right] }\\
=& \sum_{l = 0}^{L}\sum_{i = \hat{\tau}^O_{L,l} + 1}^{\hat{\tau}^O_{L,l + 1} - 1}{ \frac{\hat{n}^O_l}{\hat{n}^O_l - 1} \left[   2 \left(\varepsilon^E_i - \overline{\varepsilon}^O_{\hat{\tau}^O_{L,l} :\hat{\tau}^O_{L,l + 1}}\right) \left(\mu^E_i - \overline{\mu}^E_{\hat{\tau}^O_{L,l} :\hat{\tau}^O_{L,l + 1} - 1}\right) + \left(\mu^E_i - \overline{\mu}^E_{\hat{\tau}^O_{L,l} :\hat{\tau}^O_{L,l + 1} - 1}\right)^2 \right] }\\
&+ 2\sum_{l = 0}^{L}\sum_{i = \hat{\tau}^O_{L,l} + 1}^{\hat{\tau}^O_{L,l + 1} - 1}{ \frac{\hat{n}^O_l}{\hat{n}^O_l - 1}  \left(\varepsilon^E_i - \overline{\varepsilon}^O_{\hat{\tau}^O_{L,l} :\hat{\tau}^O_{L,l + 1}}\right) \left(\overline{\mu}^E_{\hat{\tau}^O_{L,l} :\hat{\tau}^O_{L,l + 1} - 1} - \overline{\mu}^O_{\hat{\tau}^O_{L,l} :\hat{\tau}^O_{L,l + 1}}\right)}\\
& + \sum_{l = 0}^{L}\sum_{i = \hat{\tau}^O_{L,l} + 1}^{\hat{\tau}^O_{L,l + 1} - 1}{ \frac{\hat{n}^O_l}{\hat{n}^O_l - 1}\left(\overline{\mu}^E_{\hat{\tau}^O_{L,l} :\hat{\tau}^O_{L,l + 1} - 1} - \overline{\mu}^O_{\hat{\tau}^O_{L,l} :\hat{\tau}^O_{L,l + 1}} \right)^2  }\\
=& \sum_{l = 0}^{L}\sum_{i = \hat{\tau}^O_{L,l} + 1}^{\hat{\tau}^O_{L,l + 1} - 1}{ \frac{\hat{n}^O_l}{\hat{n}^O_l - 1} \left[   2 \left(\varepsilon^E_i - \overline{\varepsilon}^E_{\hat{\tau}^O_{L,l} :\hat{\tau}^O_{L,l + 1} - 1}\right) \left(\mu^E_i - \overline{\mu}^E_{\hat{\tau}^O_{L,l} :\hat{\tau}^O_{L,l + 1} - 1}\right) + \left(\mu^E_i - \overline{\mu}^E_{\hat{\tau}^O_{L,l} :\hat{\tau}^O_{L,l + 1} - 1}\right)^2 \right] }\\
&+ 2\sum_{l = 0}^{L}\sum_{i = \hat{\tau}^O_{L,l} + 1}^{\hat{\tau}^O_{L,l + 1} - 1}{ \frac{\hat{n}^O_l}{\hat{n}^O_l - 1}  \left(\varepsilon^E_i - \overline{\varepsilon}^O_{\hat{\tau}^O_{L,l} :\hat{\tau}^O_{L,l + 1}}\right) \left(\overline{\mu}^E_{\hat{\tau}^O_{L,l} :\hat{\tau}^O_{L,l + 1} - 1} - \overline{\mu}^O_{\hat{\tau}^O_{L,l} :\hat{\tau}^O_{L,l + 1}}\right)}\\
& + \sum_{l = 0}^{L} \hat{n}^O_l \left(\overline{\mu}^E_{\hat{\tau}^O_{L,l} :\hat{\tau}^O_{L,l + 1} - 1} - \overline{\mu}^O_{\hat{\tau}^O_{L,l} :\hat{\tau}^O_{L,l + 1}} \right)^2.
\end{split}
\end{equation*}
To see the last equation, note that
\begin{equation*}
\sum_{i = \hat{\tau}^O_{L,l} + 1}^{\hat{\tau}^O_{L,l + 1} - 1}{\overline{\varepsilon}^O_{\hat{\tau}^O_{L,l} :\hat{\tau}^O_{L,l + 1}} \left(\mu^E_i - \overline{\mu}^E_{\hat{\tau}^O_{L,l} :\hat{\tau}^O_{L,l + 1} - 1}\right)} = 0 = \sum_{i = \hat{\tau}^O_{L,l} + 1}^{\hat{\tau}^O_{L,l + 1} - 1}{\overline{\varepsilon}^E_{\hat{\tau}^O_{L,l} :\hat{\tau}^O_{L,l + 1} - 1} \left(\mu^E_i - \overline{\mu}^E_{\hat{\tau}^O_{L,l} :\hat{\tau}^O_{L,l + 1} - 1}\right)}.
\end{equation*}

Hence, 
\begin{equation*}
\begin{split}
A_L^{(n)}
\leq & \left\vert \sum_{l = 0}^{L}\sum_{i = \hat{\tau}^O_{L,l} + 1}^{\hat{\tau}^O_{L,l + 1} - 1}{ \frac{\hat{n}^O_l}{\hat{n}^O_l - 1} \left(\varepsilon^E_i - \overline{\varepsilon}^O_{\hat{\tau}^O_{L,l} :\hat{\tau}^O_{L,l + 1}}\right)^2}
- \sum_{l = 0}^{L}\sum_{i = \hat{\tau}^O_{L,l} + 1}^{\hat{\tau}^O_{L,l + 1}}{\left(\varepsilon^E_i - \overline{\varepsilon}^O_{\hat{\tau}^O_{L,l} :\hat{\tau}^O_{L,l + 1}}\right)^2} \right\vert\\
& + \left\vert\sum_{l = 0}^{L}\sum_{i = \hat{\tau}^O_{L,l} + 1}^{\hat{\tau}^O_{L,l + 1} - 1}{ \frac{\hat{n}^O_l}{\hat{n}^O_l - 1}\left(\mu^E_i - \overline{\mu}^E_{\hat{\tau}^O_{L,l} :\hat{\tau}^O_{L,l + 1} - 1}\right)^2 } 
- \sum_{l = 0}^{L}\sum_{i = \hat{\tau}^O_{L,l} + 1}^{\hat{\tau}^O_{L,l + 1}}{\left(\mu^O_i - \overline{\mu}^O_{\hat{\tau}^O_{L,l} :\hat{\tau}^O_{L,l + 1}}\right)^2 }\right\vert\\
& + 2\left\vert\sum_{l = 0}^{L}\sum_{i = \hat{\tau}^O_{L,l} + 1}^{\hat{\tau}^O_{L,l + 1} - 1}{ \frac{\hat{n}^O_l}{\hat{n}^O_l - 1} \left(\varepsilon^E_i - \overline{\varepsilon}^E_{\hat{\tau}^O_{L,l} :\hat{\tau}^O_{L,l + 1} - 1}\right) \left(\mu^E_i - \overline{\mu}^E_{\hat{\tau}^O_{L,l} :\hat{\tau}^O_{L,l + 1} - 1}\right)}\right.\\
& - \left.\sum_{l = 0}^{L}\sum_{i = \hat{\tau}^O_{L,l} + 1}^{\hat{\tau}^O_{L,l + 1}}{ \left(\varepsilon^E_i - \overline{\varepsilon}^E_{\hat{\tau}^O_{L,l} :\hat{\tau}^O_{L,l + 1}}\right) \left(\mu^O_i - \overline{\mu}^O_{\hat{\tau}^O_{L,l} :\hat{\tau}^O_{L,l + 1}}\right)} \right\vert\\
& + 2\left\vert \sum_{l = 0}^{L}\sum_{i = \hat{\tau}^O_{L,l} + 1}^{\hat{\tau}^O_{L,l + 1} - 1}{ \frac{\hat{n}^O_l}{\hat{n}^O_l - 1}  \left(\varepsilon^E_i - \overline{\varepsilon}^O_{\hat{\tau}^O_{L,l} :\hat{\tau}^O_{L,l + 1}}\right) \left(\overline{\mu}^E_{\hat{\tau}^O_{L,l} :\hat{\tau}^O_{L,l + 1} - 1} - \overline{\mu}^O_{\hat{\tau}^O_{L,l} :\hat{\tau}^O_{L,l + 1}}\right)}\right\vert\\
& + \left\vert\sum_{l = 0}^{L} \hat{n}^O_l \left(\overline{\mu}^E_{\hat{\tau}^O_{L,l} :\hat{\tau}^O_{L,l + 1} - 1} - \overline{\mu}^O_{\hat{\tau}^O_{L,l} :\hat{\tau}^O_{L,l + 1}} \right)^2\right\vert\\
=: & A_{1, L}^{(n)} + A_{2, L}^{(n)} + 2A_{3, L}^{(n)} + 2A_{4, L}^{(n)} + A_{5, L}^{(n)},
\end{split}
\end{equation*}
where we have used
\begin{equation*}
\sum_{i = \hat{\tau}^O_{L,l} + 1}^{\hat{\tau}^O_{L,l + 1} - 1}{\overline{\varepsilon}^E_{\hat{\tau}^O_{L,l} :\hat{\tau}^O_{L,l + 1}} \left(\mu^O_i - \overline{\mu}^O_{\hat{\tau}^O_{L,l} :\hat{\tau}^O_{L,l + 1}}\right)} = 0 = \sum_{i = \hat{\tau}^O_{L,l} + 1}^{\hat{\tau}^O_{L,l + 1}}{\overline{\varepsilon}^O_{\hat{\tau}^O_{L,l} :\hat{\tau}^O_{L,l + 1}} \left(\mu^O_i - \overline{\mu}^O_{\hat{\tau}^O_{L,l} :\hat{\tau}^O_{L,l + 1}}\right)}.
\end{equation*}

We will bound each of the five terms in the following.

\subsubsection*{Bounding $A_{1, L}^{(n)}$}

\begin{equation}\label{eq:splitFirstTerm}
\begin{split}
A_{1, L}^{(n)} \leq & \left\vert\sum_{l = 0}^{L} \left\{{\frac{1}{\hat{n}^O_l - 1}\sum_{i = \hat{\tau}^O_{L,l} + 1}^{\hat{\tau}^O_{L,l + 1} - 1}{\left(\varepsilon^E_i\right)^2} - \left(\varepsilon^E_{\hat{\tau}^O_{L,l + 1}}\right)^2} \right\}\right\vert\\
& + 2\left\vert\sum_{l = 0}^{L}{\overline{\varepsilon}^O_{\hat{\tau}^O_{L,l} :\hat{\tau}^O_{L,l + 1}} \left(   \frac{1}{\hat{n}^O_l - 1}\sum_{i = \hat{\tau}^O_{L,l} + 1}^{\hat{\tau}^O_{L,l + 1} - 1}{\varepsilon^E_i} - \varepsilon^E_{\hat{\tau}^O_{L,l + 1}}\right)}\right\vert =: A_{11, L}^{(n)} + 2 A_{12, L}^{(n)}.
\end{split}
\end{equation}

We will now work on both terms separately and then bound the sum of them at the end.

\subsubsection*{Bounding $A_{11, L}^{(n)}$}
Let $\mathcal{E}^O$ be the sigma algebra generated by $\varepsilon^O_1,\ldots,\varepsilon^O_{ n / 2 }$. Then,
\begin{equation}\label{eq:splitFirstTermFirstTerm}
\begin{split}
A_{11, L}^{(n)}=& \left\vert\sum_{l = 0}^{L}{\left\{\frac{1}{\hat{n}^O_l - 1}\sum_{i = \hat{\tau}^O_{L,l} + 1}^{\hat{\tau}^O_{L,l + 1} - 1}{\left(\varepsilon^E_i\right)^2} - \left(\varepsilon^E_{\hat{\tau}^O_{L,l + 1}}\right)^2\right\}}\right\vert \\
\leq &\left\vert\sum_{l = 0}^{L}{\left\{\frac{1}{\hat{n}^O_l - 1}\sum_{i = \hat{\tau}^O_{L,l} + 1}^{\hat{\tau}^O_{L,l + 1} - 1}{\left(\varepsilon^E_i\right)^2} - \E\left[\left.\frac{1}{\hat{n}^O_l - 1}\sum_{i = \hat{\tau}^O_{L,l} + 1}^{\hat{\tau}^O_{L,l + 1} - 1}{\left(\varepsilon^E_i\right)^2}\right\vert \mathcal{E}^O\right]\right\}}\right\vert\\
& + \left\vert\sum_{l = 0}^{L}{\left\{\left(\varepsilon^E_{\hat{\tau}^O_{L,l + 1}}\right)^2 - \E\left[\left.\left(\varepsilon^E_{\hat{\tau}^O_{L,l + 1}}\right)^2\right\vert \mathcal{E}^O\right]\right\}}\right\vert\\
& + \sum_{l = 0}^{L}{\left\vert \E\left[\left.\left(\varepsilon^E_{\hat{\tau}^O_{L,l + 1}}\right)^2\right\vert \mathcal{E}^O\right] - \E\left[\left.\frac{1}{\hat{n}^O_l - 1}\sum_{i = \hat{\tau}^O_{L,l} + 1}^{\hat{\tau}^O_{L,l + 1} - 1}{\left(\varepsilon^E_i\right)^2}\right\vert \mathcal{E}^O\right] \right\vert}.
\end{split}
\end{equation}

In the following we bound the r.h.s.\ in \eqref{eq:splitFirstTermFirstTerm}. It follows from Lemma~\ref{lemma:maxEpsilonEsquared} that
\begin{equation*}
\begin{split}
&\max_{L = 0,\ldots,K_{\max}}\left\vert\sum_{l = 0}^{L}{\left\{\frac{1}{\hat{n}^O_l - 1}\sum_{i = \hat{\tau}^O_{L,l} + 1}^{\hat{\tau}^O_{L,l + 1} - 1}{\left(\varepsilon^E_i\right)^2} - \E\left[\left.\frac{1}{\hat{n}^O_l - 1}\sum_{i = \hat{\tau}^O_{L,l} + 1}^{\hat{\tau}^O_{L,l + 1} - 1}{\left(\varepsilon^E_i\right)^2}\right\vert \mathcal{E}^O\right]\right\}}\right\vert\\
 = &\mathcal{O}_\Pj \left( \left(K_{\max} \log K_{\max}\right)^{1/2}\,\overline{\sigma}^2\right)
\end{split}
\end{equation*}
and
\begin{equation*}
\begin{split}
\max_{L = 0,\ldots,K_{\max}}\left\vert\sum_{l = 0}^{L}{\left\{\left(\varepsilon^E_{\hat{\tau}^O_{L,l + 1}}\right)^2 - \E\left[\left.\left(\varepsilon^E_{\hat{\tau}^O_{L,l + 1}}\right)^2\right\vert \mathcal{E}^O\right]\right\}}\right\vert = \mathcal{O}_\Pj \left( \left(K_{\max} \log K_{\max}\right)^{1/2}\,\overline{\sigma}^2\right).
\end{split}
\end{equation*}
Consider the last term on the r.h.s.\ in \eqref{eq:splitFirstTermFirstTerm}. Each conditional expectation is bounded by $\overline{\sigma}^2$ almost surely. Moreover, the conditional expectations can only differ if there is a true change-point between $\hat{\tau}^O_{L,l}$ and $\hat{\tau}^O_{L,l + 1}$ as observations on the same segment have the same variance-covariance matrix. Hence, at most $K$ terms can be non-zero. Thus almost surely
\begin{equation*}
\max_{L = 0,\ldots,K_{\max}}\sum_{l = 0}^{L}{\left\vert \E\left[\left.\left(\varepsilon^E_{\hat{\tau}^O_{L,l + 1}}\right)^2\right\vert \mathcal{E}^O\right] - \E\left[\left.\frac{1}{\hat{n}^O_l - 1}\sum_{i = \hat{\tau}^O_{L,l} + 1}^{\hat{\tau}^O_{L,l + 1} - 1}{\left(\varepsilon^E_i\right)^2}\right\vert \mathcal{E}^O\right] \right\vert}
\leq K \overline{\sigma}^2.
\end{equation*}

Hence,
\begin{equation*}
%\max_{L = 0,\ldots,K_{\max}}\left\vert\sum_{l = 0}^{L}{\left\{\frac{1}{\hat{n}^O_l - 1}\sum_{i = \hat{\tau}^O_{L,l} + 1}^{\hat{\tau}^O_{L,l + 1} - 1}{\left(\varepsilon^E_i\right)^2} - \left(\varepsilon^E_{\hat{\tau}^O_{L,l + 1}}\right)^2\right\}}\right\vert
\max_{L = 0,\ldots,K_{\max}} A_{11, L}^{(n)} = \mathcal{O}_\Pj \left(\left(K_{\max} \log K_{\max}\right)^{1/2}\,\overline{\sigma}^2\right).
\end{equation*}
Recall that we have assumed $\left(K_{\max} \log K_{\max}\right)^{1/2} = o(\log\log\overline{\lambda})$, see Assumption~\ref{assumption:cpNumberMultivariate}\ref{assumption:cpNumberKmaxBound}. Thus putting things together, we see that %the first term in \eqref{eq:splitFirstTerm} satisfies
\begin{equation}\label{eq:bd_A_11}
%\max_{L = 0,\ldots,K_{\max}}\left\vert\sum_{l = 0}^{L} \left\{{\frac{1}{\hat{n}^O_l - 1}\sum_{i = \hat{\tau}^O_{L,l} + 1}^{\hat{\tau}^O_{L,l + 1} - 1}{\left(\varepsilon^E_i\right)^2} - \left(\varepsilon^E_{\hat{\tau}^O_{L,l + 1}}\right)^2} \right\}\right\vert
\max_{L = 0,\ldots,K_{\max}} A_{11, L}^{(n)}= o_{\Pj}(\overline{\sigma}^2 \log\log \overline{\lambda}).
\end{equation}

\subsubsection*{Bounding $A_{12, L}^{(n)}$}
We now bound the second term in the r.h.s.\ of \eqref{eq:splitFirstTerm}:
\begin{align}
A_{12, L}^{(n)}=&\left\vert\sum_{l = 0}^{L}{\overline{\varepsilon}^O_{\hat{\tau}^O_{L,l} :\hat{\tau}^O_{L,l + 1}} \left(   \frac{1}{\hat{n}^O_l - 1}\sum_{i = \hat{\tau}^O_{L,l} + 1}^{\hat{\tau}^O_{L,l + 1} - 1}{\varepsilon^E_i} - \varepsilon^E_{\hat{\tau}^O_{L,l + 1}}\right)}\right\vert \notag\\
\leq &\left\vert\sum_{l = 0}^{L}{\overline{\varepsilon}^O_{\hat{\tau}^O_{L,l} :\hat{\tau}^O_{L,l + 1}} \frac{1}{\hat{n}^O_l - 1}\sum_{i = \hat{\tau}^O_{L,l} + 1}^{\hat{\tau}^O_{L,l + 1} - 1}{\varepsilon^E_i}}\right\vert + \left\vert\sum_{l = 0}^{L}{\overline{\varepsilon}^O_{\hat{\tau}^O_{L,l} :\hat{\tau}^O_{L,l + 1}} \varepsilon^E_{\hat{\tau}^O_{L,l + 1}}}\right\vert. \label{eq:A_12_decomp}
\end{align}
We focus on the second term in the display above, but the same steps lead to the same bound for the first term. It follows from the law of total probability that for any $\mathcal{E}^0$-measurable random variable $X$ that is positive almost surely,
\begin{equation*}
\begin{split}
& \Pj\left( \left\vert\sum_{l = 0}^{L}{\overline{\varepsilon}^O_{\hat{\tau}^O_{L,l} :\hat{\tau}^O_{L,l + 1}} \varepsilon^E_{\hat{\tau}^O_{L,l + 1}}}\right\vert \geq X\right)\\
=& \Pj\left( \sum_{l = 0}^{L}{\overline{\varepsilon}^O_{\hat{\tau}^O_{L,l} :\hat{\tau}^O_{L,l + 1}} \varepsilon^E_{\hat{\tau}^O_{L,l + 1}}} \geq X\right)
 + \Pj\left( \sum_{l = 0}^{L}{\overline{\varepsilon}^O_{\hat{\tau}^O_{L,l} :\hat{\tau}^O_{L,l + 1}} \varepsilon^E_{\hat{\tau}^O_{L,l + 1}}} \leq  - X\right)\\
= & \E\left[ \Pj\left(\left. \sum_{l = 0}^{L}{\overline{\varepsilon}^O_{\hat{\tau}^O_{L,l} :\hat{\tau}^O_{L,l + 1}} \varepsilon^E_{\hat{\tau}^O_{L,l + 1}}} \geq X\right\vert \mathcal{E}^O\right) \right]
 + \E\left[ \Pj\left(\left. \sum_{l = 0}^{L}{\overline{\varepsilon}^O_{\hat{\tau}^O_{L,l} :\hat{\tau}^O_{L,l + 1}} \varepsilon^E_{\hat{\tau}^O_{L,l + 1}}} \leq - X\right\vert \mathcal{E}^O\right) \right].
\end{split}
\end{equation*}
We focus on the first term and the probability inside the expectation. The second term can be bounded in the same way. Let $L = 0,\ldots,K_{\max}$ be fixed. We will use Bernstein's inequality conditional on $\mathcal{E}^O$. To this end, note that $\overline{\varepsilon}^O_{\hat{\tau}^O_{L,l} :\hat{\tau}^O_{L,l + 1}} \varepsilon^E_{\hat{\tau}^O_{L,l + 1}},\ l = 0,\ldots,L$, are independent conditional on $\mathcal{E}^O$
with
\[
\E\Big[\left(\overline{\varepsilon}^O_{\hat{\tau}^O_{L,l} :\hat{\tau}^O_{L,l + 1}}\right)^{\top} \varepsilon^E_{\hat{\tau}^O_{L,l + 1}}\Big\vert \mathcal{E}^O\Big] = 0,\quad \forall\ l = 0,\ldots,L.
\]
Moreover, it follows from the Cauchy--Schwarz inequality and Assumption~\ref{assumption:NoiseBernstein} that
\begin{align*}
%\E\Big[\left(\overline{\varepsilon}^O_{\hat{\tau}^O_{L,l} :\hat{\tau}^O_{L,l + 1}}\right)^{\top} \varepsilon^E_{\hat{\tau}^O_{L,l + 1}}\Big\vert \mathcal{E}^O\Big] =& 0,\quad \forall\ l = 0,\ldots,L,\\
\sum_{l = 0}^{L}\E\bigg[\bigg(\left(\overline{\varepsilon}^O_{\hat{\tau}^O_{L,l} :\hat{\tau}^O_{L,l + 1}}\right)^{\top} \varepsilon^E_{\hat{\tau}^O_{L,l + 1}}\bigg)^2\bigg\vert \mathcal{E}^O\bigg]
\leq &\sum_{l = 0}^{L}\E\bigg[\left\|\overline{\varepsilon}^O_{\hat{\tau}^O_{L,l} :\hat{\tau}^O_{L,l + 1}}\right\|_2^2 \left\|\varepsilon^E_{\hat{\tau}^O_{L,l + 1}}\right\|_2^2\bigg\vert \mathcal{E}^O\bigg]\\
\leq &\overline{\sigma}^2 \sum_{l = 0}^{L} \left\|\overline{\varepsilon}^O_{\hat{\tau}^O_{L,l} :\hat{\tau}^O_{L,l + 1}}\right\|_2^2,\\
\sum_{l = 0}^{L}\E\bigg[\bigg(\left(\overline{\varepsilon}^O_{\hat{\tau}^O_{L,l} :\hat{\tau}^O_{L,l + 1}}\right)^{\top} \varepsilon^E_{\hat{\tau}^O_{L,l + 1}}\bigg)^q\bigg\vert \mathcal{E}^O\bigg]
\leq &\sum_{l = 0}^{L}\E\bigg[\left\|\overline{\varepsilon}^O_{\hat{\tau}^O_{L,l} :\hat{\tau}^O_{L,l + 1}}\right\|_2^q \left\|\varepsilon^E_{\hat{\tau}^O_{L,l + 1}}\right\|_2^q\bigg\vert \mathcal{E}^O\bigg]\\
\leq \frac{q!}{2} c^{q - 2} \overline{\sigma}^q \sum_{l = 0}^{L} \left\|\overline{\varepsilon}^O_{\hat{\tau}^O_{L,l} :\hat{\tau}^O_{L,l + 1}}\right\|_2^q
\leq & \frac{q!}{2} \overline{\sigma}^2 \sum_{l = 0}^{L} \left\|\overline{\varepsilon}^O_{\hat{\tau}^O_{L,l} :\hat{\tau}^O_{L,l + 1}}\right\|_2^2 \left(c\, \overline{\sigma} \left(\sum_{l = 0}^{L} \left\|\overline{\varepsilon}^O_{\hat{\tau}^O_{L,l} :\hat{\tau}^O_{L,l + 1}}\right\|_2^2\right)^{1/2}\right)^{q - 2},
\end{align*}
for $q \geq 3$. Thus, it follows from Bernstein's inequality \citep[Cor.~2.11]{boucheron2013concentration} that for any $a \geq 1$
\begin{equation*}
\begin{split}
&  \Pj\left(\left. \sum_{l = 0}^{L}{\overline{\varepsilon}^O_{\hat{\tau}^O_{L,l} :\hat{\tau}^O_{L,l + 1}} \varepsilon^E_{\hat{\tau}^O_{L,l + 1}}} \geq a \log\left(K_{\max}\right) \overline{\sigma}\left(\sum_{l = 0}^{L} \left(\overline{\varepsilon}^O_{\hat{\tau}^O_{L,l} :\hat{\tau}^O_{L,l + 1}}\right)^2\right)^{1/2} (\sqrt{2} + c) \right\vert \mathcal{E}^O\right)\\
\leq & \exp(-a \log\left(K_{\max}\right)) = K_{\max}^{-a}.
\end{split}
\end{equation*}
Hence, by a similar argument involving the negative, we have
\begin{equation*}
\begin{split}
\Pj\left( \left\vert\sum_{l = 0}^{L}{\overline{\varepsilon}^O_{\hat{\tau}^O_{L,l} :\hat{\tau}^O_{L,l + 1}} \varepsilon^E_{\hat{\tau}^O_{L,l + 1}}}\right\vert \geq a \log\left(K_{\max}\right) \overline{\sigma}\left(\sum_{l = 0}^{L} \left(\overline{\varepsilon}^O_{\hat{\tau}^O_{L,l} :\hat{\tau}^O_{L,l + 1}}\right)^2\right)^{1/2} (\sqrt{2} + c)\right)
\leq 2 K_{\max}^{-a}
\end{split}
\end{equation*}
and so
\begin{equation*}
\begin{split}
& \max_{L = 0,\ldots,K_{\max}} \left(\sum_{l = 0}^{L} \left(\overline{\varepsilon}^O_{\hat{\tau}^O_{L,l} :\hat{\tau}^O_{L,l + 1}}\right)^2\right)^{-1/2}  \left\vert\sum_{l = 0}^{L}{\overline{\varepsilon}^O_{\hat{\tau}^O_{L,l} :\hat{\tau}^O_{L,l + 1}} \varepsilon^E_{\hat{\tau}^O_{L,l + 1}}}\right\vert
= \mathcal{O}_\Pj\left(\log\left(K_{\max}\right) \overline{\sigma}\right).
\end{split}
\end{equation*}
As remarked, the same bound applies to the first term in \eqref{eq:A_12_decomp}. Hence, we obtain
\begin{equation}\label{eq:firstboundA12L}
\begin{split}
& \max_{L = 0,\ldots,K_{\max}} \left(\sum_{l = 0}^{L} \left(\overline{\varepsilon}^O_{\hat{\tau}^O_{L,l} :\hat{\tau}^O_{L,l + 1}}\right)^2\right)^{-1/2} A_{12, L}^{(n)} = \mathcal{O}_\Pj\left(\log\left(K_{\max}\right) \overline{\sigma}\right).
\end{split}
\end{equation}
We now simplify this term. To this end, we use the notation
\begin{equation*}
\left\{\tilde{\tau}^O_{L,0},\ldots, \tilde{\tau}^O_{L,\tilde{L} + 1}\right\} := \left\{\tau_{0}^O,\ldots,\tau_{K + 1}^O\right\}\, \cup\, \left\{\hat{\tau}^O_{L,1},\ldots,\hat{\tau}^O_{L,L}\right\},
\end{equation*}
with $\tilde{L} \leq K + L$. Further, we use the fact that $(\sum_{i = 1}^{n} a_i)^2 \leq \sum_{i = 1}^{n} n a_i^2$ for any $a_1,\ldots,a_n \in \R$. Using this and that the number of change-points in the set above between $\hat{\tau}^O_{L,l}$ and $\hat{\tau}^O_{L,l + 1}$ can be upper bounded by the number of observations $\tilde{\tau}^O_{L,l + 1} - \tilde{\tau}^O_{L,l}$, we have
\begin{equation}\label{eq:boundMeanEpsilonOsq}
\begin{split}
\sum_{l = 0}^{L}{\left( \overline{\varepsilon}^O_{\hat{\tau}^O_{L,l} :\hat{\tau}^O_{L,l + 1}}\right)^2}
&\leq \sum_{l = 0}^{\tilde{L}}{\left(\tilde{\tau}^O_{L,l + 1} - \tilde{\tau}^O_{L,l}\right) \left( \overline{\varepsilon}^O_{\tilde{\tau}^O_{L,l} :\tilde{\tau}^O_{L,l + 1}} \right)^2}\\
=&\sum_{k = 0}^{K}{\sum_{i = \tau^O_{k} + 1}^{\tau^O_{k + 1}}{\left(\varepsilon^O_i - \overline{\varepsilon}^O_{\tau^O_{k} :\tau^O_{k + 1}} \right)^2}} - \sum_{l = 0}^{\tilde{L}}{\sum_{i = \tilde{\tau}^O_{L,l} + 1}^{\tilde{\tau}^O_{L,l + 1}}{\left(\varepsilon^O_i - \overline{\varepsilon}^O_{\tilde{\tau}^O_{L,l} :\tilde{\tau}^O_{L,l + 1}} \right)^2}}\\
& + \sum_{k = 0}^{K}{\left(\tau^O_{k + 1} - \tau^O_{k}\right) \left( \overline{\varepsilon}^O_{\tau^O_{k} :\tau^O_{k + 1}} \right)^2} \\
= & S_{\varepsilon^O}\left(\mathcal{T}^{O}_{K}\right) - S_{\varepsilon^O}\left(\hat{\mathcal{T}}^O_L\cup \mathcal{T}^{O}_{K}\right) + \sum_{k = 0}^{K}{\left(\tau^O_{k + 1} - \tau^O_{k}\right) \left( \overline{\varepsilon}^O_{\tau^O_{k} :\tau^O_{k + 1}} \right)^2}.
\end{split}
\end{equation}
Note that by Markov's inequality,
\begin{equation}\label{eq:boundMeanEpsilonOsqSecondPart}
\sum_{k = 0}^{K}{\left(\tau^O_{k + 1} - \tau^O_{k}\right) \left( \overline{\varepsilon}^O_{\tau^O_{k} :\tau^O_{k + 1}} \right)^2} = \mathcal{O}_\Pj\left((K \vee 1)\overline{\sigma}^2\right).
\end{equation}
To bound the terms further we distinguish the cases $L = K, L > K,$ and $L < K$. We will also combine the obtained bounds with \eqref{eq:bd_A_11} to obtain the final statement for $A_{1, L}^{(n)}$.

\emph{Case $L=K$.}
From Lemma~\ref{lemma:boundsForEpsilonO}\ref{item2:boundsForEpsilonO} we have that
\begin{equation*}
S_{\varepsilon^O}\left(\mathcal{T}^{O}_{K}\right) - S_{\varepsilon^O}\left(\hat{\mathcal{T}}^O_K\cup \mathcal{T}^{O}_{K}\right)
= o_\Pj\left(\overline{\sigma}^2 \log\log\overline{\lambda}\right).
\end{equation*}
Moreover, Assumption~\ref{assumption:cpNumberMultivariate}\ref{assumption:cpNumberBound} yields $K = o(\log\log\overline{\lambda})$. Thus,
\begin{equation}\label{eq:boundMeanEpsilonOSqK}
\sum_{k = 0}^{K}{\left( \overline{\varepsilon}^O_{\hat{\tau}^O_{K,k} :\hat{\tau}^O_{K,k + 1}}\right)^2} = o_{\Pj}(\overline{\sigma}^2 \log\log\overline{\lambda}),
\end{equation}
and this also holds when $K=0$. Furthermore, it follows from Assumption~\ref{assumption:cpNumberMultivariate}\ref{assumption:cpNumberKmaxBound} that $\log K_{\max} = o\Big((\log\log\overline{\lambda})^{1/2}\Big)$. To this end, note that this is trivially satisfied if $K_{\max} = \mathcal{O}(1)$ and otherwise since $(\log K_{\max})^3 / K_{\max} \to 0$.
Hence,
\begin{equation*}
\begin{split}
\left\vert\sum_{k = 0}^{K}{\overline{\varepsilon}^O_{\hat{\tau}^O_{K,k} :\hat{\tau}^O_{K,k + 1}} \varepsilon^E_{\hat{\tau}^O_{K,k + 1}}}\right\vert = o_{\Pj}(\overline{\sigma}^2 \log\log\overline{\lambda})
\end{split}
\end{equation*}
and thus using \eqref{eq:bd_A_11} we have
\begin{equation}\label{eq:boundA1K}
A_{1, K}^{(n)} = o_{\Pj}(\overline{\sigma}^2 \log\log\overline{\lambda}).
\end{equation}
We will use this bound in the following to bound $\max\big( A_{1, L}^{(n)},\, A_{1, K}^{(n)}\big)$ for $L > K$ and $L < K$.

%Next, we consider $L > K$.
\emph{Case $L>K$.}
Using \eqref{eq:boundMeanEpsilonOsq}, \eqref{eq:boundMeanEpsilonOsqSecondPart}, Assumptions~\ref{assumption:cpNumberMultivariate}\ref{assumption:cpNumberBound}~and~\ref{assumption:cpNumberKmaxBound} similarly to the case where $L = K$ as well as Assumption~\ref{assumption:overestimation} gives us
\begin{equation}\label{eq:boundMeanEpsilonOSqL>K}
\max_{L = K + 1,\ldots,K_{\max}} \left(S_{\varepsilon^O}\left(\mathcal{T}^{O}_{K}\right) - S_{\varepsilon^O}\left(\hat{\mathcal{T}}^O_L\cup \mathcal{T}^{O}_{K}\right)\right)^{-1}\sum_{l = 0}^{L}{\left( \overline{\varepsilon}^O_{\hat{\tau}^O_{L,l} :\hat{\tau}^O_{L,l + 1}}\right)^2} = \mathcal{O}_{\Pj}(1).
\end{equation}
Combining this with \eqref{eq:bd_A_11} and \eqref{eq:firstboundA12L} as well as \eqref{eq:boundA1K} and Assumption~\ref{assumption:overestimation} yields
\begin{equation*}
\begin{split}
& \max_{L = K + 1,\ldots,K_{\max}} \left(S_{\varepsilon^O}\left(\mathcal{T}^{O}_{K}\right) - S_{\varepsilon^O}\left(\hat{\mathcal{T}}^O_L\cup \mathcal{T}^{O}_{K}\right)\right)^{-1} \max\big( A_{1, L}^{(n)},\, A_{1, K}^{(n)}\big) = o_{\Pj}(1).
\end{split}
\end{equation*}
Note that this completes the proof in the case where $K = 0\ \forall\ n$ as we only have to consider $L \geq K$ and since all $\mu_i$'s are the same, we have that  $A_L^{(n)} \leq A_{1, L}^{(n)}$. Moreover, note that the arguments above in the case where $K=0$ do not use Assumptions~\ref{assumption:detectionPrecision}~and~\ref{assumption:minimumSignalMultivariate}. In the following we may therefore assume that $K \geq 1$.

%We now consider the final case where $L < K$.
\emph{Case $L<K$.} From Lemma~\ref{lemma:boundsForEpsilonO}\ref{item1:boundsForEpsilonO} we have that
\begin{equation*}
\max_{L = 0\ldots,K - 1} \left\{S_{\varepsilon^O}\left(\mathcal{T}^{O}_{K}\right) - S_{\varepsilon^O}\left(\hat{\mathcal{T}}^O_L\cup \mathcal{T}^{O}_{K}\right)\right\}
= \mathcal{O}_{\Pj}\big(K \overline{\sigma}^2(\log \overline{\lambda})^2\big).
\end{equation*}
Thus, \eqref{eq:boundMeanEpsilonOsq} and \eqref{eq:boundMeanEpsilonOsqSecondPart} yields
\begin{equation}\label{eq:boundMeanEpsilonOSqL<K}
\max_{L = 0\ldots,K - 1}\sum_{l = 0}^{L}{\left( \overline{\varepsilon}^O_{\hat{\tau}^O_{L,l} :\hat{\tau}^O_{L,l + 1}}\right)^2} = \mathcal{O}_{\Pj}\big(K \overline{\sigma}^2(\log \overline{\lambda})^2\big).
\end{equation}
Consequently, \eqref{eq:firstboundA12L} and \eqref{eq:boundA1K} give us
\begin{equation*}
\begin{split}
\max_{L = 0\ldots,K - 1} \max\big( A_{12, L}^{(n)},\, A_{12, K}^{(n)}\big) = \mathcal{O}_{\Pj}\big(K \overline{\sigma}^2(\log \overline{\lambda})^2\big)
\end{split}
\end{equation*}
and hence from \eqref{eq:bd_A_11} we have that
\begin{equation*}
\begin{split}
\max_{L = 0\ldots,K - 1} \max\big( A_{1, L}^{(n)},\, A_{1, K}^{(n)}\big) = \mathcal{O}_{\Pj}\big(K \overline{\sigma}^2(\log \overline{\lambda})^2\big).
\end{split}
\end{equation*}
Finally, note that $\mathcal{I}_L$ is non-empty. Hence, $ \sum_{k \in \mathcal{I}_L}{\Delta_k^2} \geq \Delta_{(1)}^2$. Thus, Assumption~\ref{assumption:minimumSignalMultivariate} yields 
\begin{equation*}
K \overline{\sigma}^2(\log \overline{\lambda})^2 = o_\Pj\left( \min_L \underline{\lambda}  \sum_{k \in \mathcal{I}_L}{\Delta_k^2}\right).
\end{equation*}
In summary, it follows that $A_{1, L}^{(n)}$ satisfies the same bounds as $A_{L}^{(n)}$ in the statement of the lemma.

\subsubsection*{Bounding $A_{2, L}^{(n)}$}
\begin{equation*}
\begin{split}
A_{2, L}^{(n)}
\leq & \left\vert\sum_{l = 0}^{L}\sum_{i = \hat{\tau}^O_{L,l} + 1}^{\hat{\tau}^O_{L,l + 1} - 1}{ \frac{\hat{n}^O_l}{\hat{n}^O_l - 1}\left(\mu^E_i - \overline{\mu}^E_{\hat{\tau}^O_{L,l} :\hat{\tau}^O_{L,l + 1} - 1}\right)^2 }
- \sum_{l = 0}^{L}\sum_{i = \hat{\tau}^O_{L,l} + 1}^{\hat{\tau}^O_{L,l + 1} - 1}{ \frac{\hat{n}^O_l}{\hat{n}^O_l - 1}\left(\mu^E_i - \overline{\mu}^O_{\hat{\tau}^O_{L,l} :\hat{\tau}^O_{L,l + 1}}\right)^2 }\right\vert\\
&+ \left\vert\sum_{l = 0}^{L}\sum_{i = \hat{\tau}^O_{L,l} + 1}^{\hat{\tau}^O_{L,l + 1} - 1}{ \frac{\hat{n}^O_l}{\hat{n}^O_l - 1}\left(\mu^E_i - \overline{\mu}^O_{\hat{\tau}^O_{L,l} :\hat{\tau}^O_{L,l + 1}}\right)^2 } 
- \sum_{l = 0}^{L}\sum_{i = \hat{\tau}^O_{L,l} + 1}^{\hat{\tau}^O_{L,l + 1} - 1}{ \frac{\hat{n}^O_l}{\hat{n}^O_l - 1}\left(\mu^O_i - \overline{\mu}^O_{\hat{\tau}^O_{L,l} :\hat{\tau}^O_{L,l + 1}}\right)^2 }\right\vert \\
&+ \left\vert\sum_{l = 0}^{L}\sum_{i = \hat{\tau}^O_{L,l} + 1}^{\hat{\tau}^O_{L,l + 1} - 1}{ \frac{\hat{n}^O_l}{\hat{n}^O_l - 1}\left(\mu^O_i - \overline{\mu}^O_{\hat{\tau}^O_{L,l} :\hat{\tau}^O_{L,l + 1}}\right)^2 }
 - \sum_{l = 0}^{L}\sum_{i = \hat{\tau}^O_{L,l} + 1}^{\hat{\tau}^O_{L,l + 1}}{\left(\mu^O_i - \overline{\mu}^O_{\hat{\tau}^O_{L,l} :\hat{\tau}^O_{L,l + 1}}\right)^2 }\right\vert.
\end{split}
\end{equation*}
The r.h.s.\ is bounded from above by
\begin{equation*}
	\begin{split}
 & 2\left\vert\sum_{l = 0}^{L}\sum_{i = \hat{\tau}^O_{L,l} + 1}^{\hat{\tau}^O_{L,l + 1} - 1}{\frac{\hat{n}^O_l}{\hat{n}^O_l - 1}  \left(\mu^E_i - \overline{\mu}^E_{\hat{\tau}^O_{L,l} :\hat{\tau}^O_{L,l + 1} - 1}\right)  \left(\overline{\mu}^E_{\hat{\tau}^O_{L,l} :\hat{\tau}^O_{L,l + 1} - 1} - \overline{\mu}^O_{\hat{\tau}^O_{L,l} :\hat{\tau}^O_{L,l + 1}}\right)}\right\vert\\
 & + \left\vert\sum_{l = 0}^{L}\sum_{i = \hat{\tau}^O_{L,l} + 1}^{\hat{\tau}^O_{L,l + 1} - 1}{\frac{\hat{n}^O_l}{\hat{n}^O_l - 1} \left(\overline{\mu}^E_{\hat{\tau}^O_{L,l} :\hat{\tau}^O_{L,l + 1} - 1} - \overline{\mu}^O_{\hat{\tau}^O_{L,l} :\hat{\tau}^O_{L,l + 1}}\right)^2}\right\vert\\
 & + \left\vert\sum_{l = 0}^{L}\sum_{i = \hat{\tau}^O_{L,l} + 1}^{\hat{\tau}^O_{L,l + 1} - 1}{\frac{\hat{n}^O_l}{\hat{n}^O_l - 1} \left(\mu^E_i - \mu^O_i\right)^2}\right\vert\\
 & + 2\left\vert\sum_{l = 0}^{L}\sum_{i = \hat{\tau}^O_{L,l} + 1}^{\hat{\tau}^O_{L,l + 1} - 1}{\frac{\hat{n}^O_l}{\hat{n}^O_l - 1} \left(\mu^E_i - \mu^O_i\right) \left(\mu^O_i - \overline{\mu}^O_{\hat{\tau}^O_{L,l} :\hat{\tau}^O_{L,l + 1}}\right)}\right\vert\\
 & + \left\vert\sum_{l = 0}^{L}\sum_{i = \hat{\tau}^O_{L,l} + 1}^{\hat{\tau}^O_{L,l + 1} - 1}{ \frac{1}{\hat{n}^O_l - 1}\left(\mu^O_i - \overline{\mu}^O_{\hat{\tau}^O_{L,l} :\hat{\tau}^O_{L,l + 1}}\right)^2 }
 - \sum_{l = 0}^{L}{\left(\mu^O_{\hat{\tau}^O_{L,l + 1}} - \overline{\mu}^O_{\hat{\tau}^O_{L,l} :\hat{\tau}^O_{L,l + 1}}\right)^2 }\right\vert,
\end{split}
\end{equation*}
which is in turn bounded from above by
 \begin{equation*}
 	\begin{split}
& 2\sum_{l = 0}^{L}\left(\max_{i = \hat{\tau}^O_{L,l} + 1,\ldots,\hat{\tau}^O_{L,l + 1} - 1}{ \left\vert\mu^E_i - \overline{\mu}^E_{\hat{\tau}^O_{L,l} :\hat{\tau}^O_{L,l + 1} - 1}\right\vert}\right)\left( \hat{n}^O_l\left\vert\overline{\mu}^E_{\hat{\tau}^O_{L,l} :\hat{\tau}^O_{L,l + 1} - 1} - \overline{\mu}^O_{\hat{\tau}^O_{L,l} :\hat{\tau}^O_{L,l + 1}}\right\vert\right)\\
 & + \sum_{l = 0}^{L}\left( \hat{n}^O_l\left\vert\overline{\mu}^E_{\hat{\tau}^O_{L,l} :\hat{\tau}^O_{L,l + 1} - 1} - \overline{\mu}^O_{\hat{\tau}^O_{L,l} :\hat{\tau}^O_{L,l + 1}}\right\vert\right)^2\\
 & + 2\sum_{l = 0}^{L}\left(\sum_{i = \hat{\tau}^O_{L,l} + 1}^{\hat{\tau}^O_{L,l + 1} - 1}{ \left\vert\mu^E_i - \mu^O_i\right\vert}\right)^2\\
 & + 4\sum_{l = 0}^{L}\left(\max_{i = \hat{\tau}^O_{L,l} + 1,\ldots,\hat{\tau}^O_{L,l + 1} - 1}{ \left\vert \mu^O_i - \overline{\mu}^O_{\hat{\tau}^O_{L,l} :\hat{\tau}^O_{L,l + 1}} \right\vert } \right)\left(\sum_{i = \hat{\tau}^O_{L,l} + 1}^{\hat{\tau}^O_{L,l + 1} - 1}{\left\vert\mu^E_i - \mu^O_i\right\vert}\right)\\
 & + \sum_{l = 0}^{L}\left(\max_{i = \hat{\tau}^O_{L,l} + 1,\ldots,\hat{\tau}^O_{L,l + 1} - 1}{ \left\vert \mu^O_i - \overline{\mu}^O_{\hat{\tau}^O_{L,l} :\hat{\tau}^O_{L,l + 1}} \right\vert } \right)^2\\
\leq & 16 B_{L}^{(n)}, 
\end{split}
\end{equation*}
where we have used the Cauchy-Schwarz inequality and $B_{L}^{(n)}$ is defined as in Lemma~\ref{lemma:preliminaryCalculationsDifferenceCVs}. Then, it follows from Lemma~\ref{lemma:preliminaryCalculationsDifferenceCVs} that $A_{2, L}^{(n)}$ satisfies the same bounds as $A_{L}^{(n)}$ in the statement of the lemma.

\subsubsection*{Bounding $A_{3, L}^{(n)}$}
Since
\begin{equation*}
\overline{\varepsilon}^E_{\hat{\tau}^O_{L,l} :\hat{\tau}^O_{L,l + 1}}\left(\mu^O_i - \overline{\mu}^O_{\hat{\tau}^O_{L,l} :\hat{\tau}^O_{L,l + 1}}\right) = 0 = \overline{\varepsilon}^E_{\hat{\tau}^O_{L,l} :\hat{\tau}^O_{L,l + 1} - 1}\left(\mu^O_i - \overline{\mu}^O_{\hat{\tau}^O_{L,l} :\hat{\tau}^O_{L,l + 1}}\right),
\end{equation*}
we have
\begin{equation*}
\begin{split}
A_{3, L}^{(n)}
\leq & \Bigg\vert\sum_{l = 0}^{L}\sum_{i = \hat{\tau}^O_{L,l} + 1}^{\hat{\tau}^O_{L,l + 1} - 1}{ \frac{\hat{n}^O_l}{\hat{n}^O_l - 1} \left(\varepsilon^E_i - \overline{\varepsilon}^E_{\hat{\tau}^O_{L,l} :\hat{\tau}^O_{L,l + 1} - 1}\right) \left(\mu^E_i - \overline{\mu}^E_{\hat{\tau}^O_{L,l} :\hat{\tau}^O_{L,l + 1} - 1}\right)}\\
& - \sum_{l = 0}^{L}\sum_{i = \hat{\tau}^O_{L,l} + 1}^{\hat{\tau}^O_{L,l + 1} - 1}{\frac{\hat{n}^O_l}{\hat{n}^O_l - 1} \left(\varepsilon^E_i - \overline{\varepsilon}^E_{\hat{\tau}^O_{L,l} :\hat{\tau}^O_{L,l + 1} - 1}\right) \left(\mu^O_i - \overline{\mu}^O_{\hat{\tau}^O_{L,l} :\hat{\tau}^O_{L,l + 1}}\right)}\Bigg\vert\\
& + \Bigg\vert\sum_{l = 0}^{L}\sum_{i = \hat{\tau}^O_{L,l} + 1}^{\hat{\tau}^O_{L,l + 1} - 1}{\frac{\hat{n}^O_l}{\hat{n}^O_l - 1} \left(\varepsilon^E_i - \overline{\varepsilon}^E_{\hat{\tau}^O_{L,l} :\hat{\tau}^O_{L,l + 1} - 1}\right) \left(\mu^O_i - \overline{\mu}^O_{\hat{\tau}^O_{L,l} :\hat{\tau}^O_{L,l + 1}}\right)}\\
& - \sum_{l = 0}^{L}\sum_{i = \hat{\tau}^O_{L,l} + 1}^{\hat{\tau}^O_{L,l + 1}}{ \left(\varepsilon^E_i - \overline{\varepsilon}^E_{\hat{\tau}^O_{L,l} :\hat{\tau}^O_{L,l + 1} - 1}\right) \left(\mu^O_i - \overline{\mu}^O_{\hat{\tau}^O_{L,l} :\hat{\tau}^O_{L,l + 1}}\right)} \Bigg\vert.
\end{split}
\end{equation*}

Using the same arguments as for the first term gives us
\begin{equation*}
\begin{split}
&A_{3, L}^{(n)} \Bigg(\sum_{l = 0}^{L}\Bigg( \sum_{i = \hat{\tau}^O_{L,l} + 1}^{\hat{\tau}^O_{L,l + 1} - 1}{\left\vert\mu^E_i - \mu^O_i\right\vert} + \hat{n}^O_l \left\vert \overline{\mu}^O_{\hat{\tau}^O_{L,l} :\hat{\tau}^O_{L,l + 1}} - \overline{\mu}^E_{\hat{\tau}^O_{L,l} :\hat{\tau}^O_{L,l + 1} - 1}\right\vert\\
& \hspace*{70pt} + \max_{i = \hat{\tau}^O_{L,l} + 1,\ldots, \hat{\tau}^O_{L,l + 1}}{ \left\vert \mu^O_i - \overline{\mu}^O_{\hat{\tau}^O_{L,l} :\hat{\tau}^O_{L,l + 1}} \right\vert }\Bigg)\Bigg)^{-1}\\
= & \mathcal{O}_\Pj\left( \log(K_{\max}) \overline{\sigma} \right).
\end{split}
\end{equation*}
Hence, from Lemma~\ref{lemma:preliminaryCalculationsDifferenceCVs} and the same simplifications as seen before it follows that $A_{3, L}^{(n)}$ satisfies the same bounds as $A_{L}^{(n)}$ in the statement of the lemma.

\subsubsection*{Bounding $A_{4, L}^{(n)}$}
Applying the triangle and Cauchy--Schwarz inequalities gives
\begin{equation*}
\begin{split}
A_{4, L}^{(n)}
\leq & \left\vert \sum_{l = 0}^{L}\sum_{i = \hat{\tau}^O_{L,l} + 1}^{\hat{\tau}^O_{L,l + 1} - 1}{ \frac{\hat{n}^O_l}{\hat{n}^O_l - 1} \varepsilon^E_i \left(\overline{\mu}^E_{\hat{\tau}^O_{L,l} :\hat{\tau}^O_{L,l + 1} - 1} - \overline{\mu}^O_{\hat{\tau}^O_{L,l} :\hat{\tau}^O_{L,l + 1}}\right)}\right \vert\\
& + \left\vert \sum_{l = 0}^{L}{\hat{n}^O_l\overline{\varepsilon}^O_{\hat{\tau}^O_{L,l} :\hat{\tau}^O_{L,l + 1}} \left(\overline{\mu}^E_{\hat{\tau}^O_{L,l} :\hat{\tau}^O_{L,l + 1} - 1} - \overline{\mu}^O_{\hat{\tau}^O_{L,l} :\hat{\tau}^O_{L,l + 1}}\right)}\right \vert\\
\leq & \left\vert \sum_{l = 0}^{L}\sum_{i = \hat{\tau}^O_{L,l} + 1}^{\hat{\tau}^O_{L,l + 1} - 1}{ \frac{\hat{n}^O_l}{\hat{n}^O_l - 1} \varepsilon^E_i \left(\overline{\mu}^E_{\hat{\tau}^O_{L,l} :\hat{\tau}^O_{L,l + 1} - 1} - \overline{\mu}^O_{\hat{\tau}^O_{L,l} :\hat{\tau}^O_{L,l + 1}}\right)}\right \vert\\
& + \left( \sum_{l = 0}^{L}{\Big(\overline{\varepsilon}^O_{\hat{\tau}^O_{L,l} :\hat{\tau}^O_{L,l + 1}}\Big)^2}\right)^{1/2}\left( \sum_{l = 0}^{L}{\Big(\hat{n}^O_l\Big)^2\left(\overline{\mu}^E_{\hat{\tau}^O_{L,l} :\hat{\tau}^O_{L,l + 1} - 1} - \overline{\mu}^O_{\hat{\tau}^O_{L,l} :\hat{\tau}^O_{L,l + 1}}\right)^2}\right)^{1/2}.
\end{split}
\end{equation*}

Then it follows from Lemma~\ref{lemma:preliminaryCalculationsDifferenceCVs} and the same calculations as for bounding $A_{1, L}^{(n)}$ that $A_{4, L}^{(n)}$ satisfies the same bounds as $A_{L}^{(n)}$ in the statement of the lemma.

\subsubsection*{Bounding $A_{5, L}^{(n)}$}
From Lemma~\ref{lemma:preliminaryCalculationsDifferenceCVs} it follows that $A_{5, L}^{(n)}$ satisfies the same bounds as $A_{L}^{(n)}$ in the statement of the lemma.

\subsubsection*{Proving the statement}
Combining the bounds for $A_L^{(n)}$, $A_{1, L}^{(n)}$, $A_{2, L}^{(n)}$, $A_{3, L}^{(n)}$, $A_{4, L}^{(n)}$ and $A_{5, L}^{(n)}$ gives that $A_{L}^{(n)}$ satisfies the bounds given in the statement of the lemma.
\end{proof}

\begin{proof}[Proof of Theorem \ref{theorem:positiveResultRescaledCV}]
Recall $\operatorname{CV}^O_{\mathrm{mod}}(L)$, $\operatorname{CV}^E_{\mathrm{mod}}(L)$, $\widetilde{\operatorname{CV}}^O_{\mathrm{mod}}(L)$ and $\widetilde{\operatorname{CV}}^E_{\mathrm{mod}}(L)$ from \eqref{eq:CVmodO} and \eqref{eq:simplifiedCVterm}, respectively.

Then, $\operatorname{CV}_{\mathrm{mod}}(L) = \operatorname{CV}^O_{\mathrm{mod}}(L) + \operatorname{CV}^E_{\mathrm{mod}}(L)$. Thus,
\begin{align*}
& \Pj\left( \hat{K} = K \right)\\
= &\Pj\left( \min_{\substack{L = 0,\ldots,K_{\max},\\ L \neq K}} \operatorname{CV}_{\mathrm{mod}}(L) - \operatorname{CV}_{\mathrm{mod}}(K) > 0\right)\\
\geq &\Pj\left( \min_{\substack{L = 0,\ldots,K_{\max},\\ L \neq K}} \operatorname{CV}^O_{\mathrm{mod}}(L) - \operatorname{CV}^O_{\mathrm{mod}}(K) > 0, \min_{\substack{L = 0,\ldots,K_{\max},\\ L \neq K}} \operatorname{CV}^E_{\mathrm{mod}}(L) - \operatorname{CV}^E_{\mathrm{mod}}(K) > 0\right)\\
= & \Pj\left( \min_{\substack{L = 0,\ldots,K_{\max},\\ L \neq K}} \operatorname{CV}^O_{\mathrm{mod}}(L) - \operatorname{CV}^O_{\mathrm{mod}}(K) > 0\right) + \Pj\left( \min_{\substack{L = 0,\ldots,K_{\max},\\ L \neq K}} \operatorname{CV}^E_{\mathrm{mod}}(L) - \operatorname{CV}^E_{\mathrm{mod}}(K) > 0\right)\\
& - \Pj\left( \min_{\substack{L = 0,\ldots,K_{\max},\\ L \neq K}} \operatorname{CV}^O_{\mathrm{mod}}(L) - \operatorname{CV}^O_{\mathrm{mod}}(K) > 0 \text{ or } \min_{\substack{L = 0,\ldots,K_{\max},\\ L \neq K}} \operatorname{CV}^E_{\mathrm{mod}}(L) - \operatorname{CV}^E_{\mathrm{mod}}(K) > 0\right).
\end{align*}
In the following we will show that
\begin{equation}\label{eq:CVOpositive}
\Pj\Bigg( \min_{\substack{L = 0,\ldots,K_{\max},\\ L \neq K}} \operatorname{CV}^O_{\mathrm{mod}}(L) - \operatorname{CV}^O_{\mathrm{mod}}(K) > 0\Bigg) \to 1, \text{ as } n\to \infty.
\end{equation}
Note that this completes the proof as it implies
\begin{equation*}
\Pj\Bigg( \min_{\substack{L = 0,\ldots,K_{\max},\\ L \neq K}} \operatorname{CV}^E_{\mathrm{mod}}(L) - \operatorname{CV}^E_{\mathrm{mod}}(K) > 0\Bigg) \to 1, \text{ as } n\to \infty,
\end{equation*}
since $\operatorname{CV}^O_{\mathrm{mod}}(L)$ and $\operatorname{CV}^E_{\mathrm{mod}}(L)$ are symmetric and for both terms the same properties follow from Assumptions~\ref{assumption:cpNumberMultivariate}--\ref{assumption:minimumSignalMultivariate}. In fact, one way to see this is to consider the observations in reverse order, which interchanges $Y^O_i$ and $Y^E_i$, but Assumptions~\ref{assumption:cpNumberMultivariate}--\ref{assumption:minimumSignalMultivariate} are not altered. Hence, \eqref{eq:CVOpositive} implies $\Pj\left( \hat{K} = K \right)  \to 1, \text{ as } n\to \infty$.

It remains to show \eqref{eq:CVOpositive}. We will consider $L > K$ and $L < K$ separately. Let $L > K$. Recall the notation $A_L^{(n)} = \left\vert \operatorname{CV}^O_{\mathrm{mod}}(L) - \widetilde{\operatorname{CV}}^O_{\mathrm{mod}}(L) \right\vert$
from Lemma~\ref{lemma:differenceCVs}. It follows from Lemmas~\ref{lemma:differenceCVLtoK}~and~\ref{lemma:differenceCVs} that
\begin{align*}
&\min_{L = K + 1,\ldots,K_{\max}} \operatorname{CV}^O_{\mathrm{mod}}(L) - \operatorname{CV}^O_{\mathrm{mod}}(K)\\
= &\min_{L = K + 1,\ldots,K_{\max}}\left\{\widetilde{\operatorname{CV}}^O_{\mathrm{mod}}(L)  - A_L^{(n)}\right\} - \widetilde{\operatorname{CV}}^O_{\mathrm{mod}}(K) + A_K^{(n)}\\
\geq & \min_{L = K + 1,\ldots,K_{\max}} \left\{ S_{\varepsilon^O}\left(\mathcal{T}^{O}_{K}\right) - S_{\varepsilon^O}\left(\hat{\mathcal{T}}^O_L \cup \mathcal{T}^{O}_{K} \right) \right\} \left(1 + o_{\Pj}(1)\right).
\end{align*}
It follows from Assumption~\ref{assumption:overestimation} that the r.h.s.\ is positive with probability converging to $1$ as $n \to \infty$. Hence, \eqref{eq:CVOpositive} holds when $L > K$. Note that thus far, Assumptions~\ref{assumption:detectionPrecision}~and~\ref{assumption:minimumSignalMultivariate} have not been used directly, and moreover they are not required for the referenced Lemmas when $K=0$ for all $n$.

Let $L < K$ now. Lemma~\ref{lemma:preliminaryCalculationsDifferenceCVs} yields that there exist a constant $A > 0$ and sequences of stochastic non-empty sets $\mathcal{I}_L\subseteq \{1,\ldots,K\}$ such that
\begin{equation*}
\Pj\Bigg( \forall \,L<K, \,\,  \sum_{i = \tau_{k}^O - \floor{\frac{\underline{\lambda}}{4}} + 1}^{\tau_{k}^O + \floor{\frac{\underline{\lambda}}{4}}}{\big(\mu^O_i - \overline{\mu}^O_{L, i}\big)^2} \geq  A \underline{\lambda}\Delta_k^2\ \forall\ k\in \mathcal{I}_L \Bigg) \to 1,
\end{equation*}
where $\overline{\mu}^O_{L, i} := \sum_{l = 0}^{L}{\EINS_{\{\hat{\tau}^O_{L,l} + 1 \leq i \leq \hat{\tau}^O_{L,l + 1}\}} \overline{\mu}^O_{\hat{\tau}^O_{L,l} :\hat{\tau}^O_{L,l + 1}}}$. Thus, it follows from Lemmas~\ref{lemma:differenceCVLtoK}~and~\ref{lemma:differenceCVs} that
\begin{align*}
&\min_{L = 0\ldots,K - 1} \operatorname{CV}^O_{\mathrm{mod}}(L) - \operatorname{CV}^O_{\mathrm{mod}}(K)\\
= &\min_{L = 0\ldots,K - 1}\left\{\widetilde{\operatorname{CV}}^O_{\mathrm{mod}}(L)  - A_L^{(n)}\right\} - \widetilde{\operatorname{CV}}^O_{\mathrm{mod}}(K) + A_K^{(n)}\\
\geq & \min_{L = 0\ldots,K - 1} \Bigg\{ \underline{\lambda}\sum_{k\in \mathcal{I}_L}{\Delta_k^2} \Bigg\} \big(A + o_{\Pj}(1)\big).
\end{align*}
The r.h.s.\ is positive, since $\vert \mathcal{I}_L \vert \geq 1,\ \forall L = L = 0\ldots,K - 1$ and $A > 0$. Hence, \eqref{eq:CVOpositive} holds also when $L < K$. This completes the proof as noted before.
\end{proof}

\section{Proof of Theorem~\ref{thm:consistency}}\label{sec:proof:consistency}
The first statement in Theorem~\ref{thm:consistency} follows directly from Theorem~\ref{theorem:positiveResultRescaledCV} if its assumptions are satisfied by least squares estimation under the given setting. To show this, we will use Theorem~\ref{theorem:detectionPrecision} and Lemma~\ref{lemma:overestimation}. The latter is a shorter version of Theorem~2 in \citet{zou2020consistent}. It shows that Assumption~\ref{assumption:overestimation} is indeed satisfied in the given set-up.
\begin{Lemma}\label{lemma:overestimation}
Suppose the noise assumption in Section~\ref{sec:cvL2}, i.e. sub-Gaussian noise, constant variance on each segment and ratio between the smallest variance and largest variance proxy is bounded from below. Then, for all $M>0$,
\begin{equation*}
	\Pj \left(\frac{\min_{L = K + 1,\ldots, K_{\max}}\left\{S_{\varepsilon^O}\left(\mathcal{T}^{O}_{K}\right) - S_{\varepsilon^O}\left(\hat{\mathcal{T}}^O_L\cup \mathcal{T}^{O}_{K}\right)\right\}}{\overline{\sigma}^2 \log\log\overline{\lambda}} < M \right)  \to 0
\end{equation*}
and as above but with all instances of $O$ replaced by $E$.
\end{Lemma}
\begin{proof}
Let ${k^*}$ be the index of the longest segment, so $\overline{\lambda}/2 -1 \leq \tau^O_{{k^*} + 1} - \tau^O_{k^*} \leq \overline{\lambda}/2+1$. Then, it follows from the fact that adding change-points only decreases the costs and the definition of least squares estimation that for every $L = K + 1,\ldots, K_{\max}$, 
\begin{align*}
S_{\varepsilon^O}\left(\hat{\mathcal{T}}^O_L\cup \mathcal{T}^{O}_{K}\right)
= & S_{Y^O}\left(\hat{\mathcal{T}}^O_L\cup \mathcal{T}^{O}_{K}\right)
\leq S_{Y^O}\left(\hat{\mathcal{T}}^O_{K + 1}\right)\\
\leq & \min_{\tau^O_{k^*} < t < \tau^O_{{k^*} + 1}} S_{Y^O}\left(\mathcal{T}^{O}_{K} \cup \{ t\} \right)\\
= & \min_{\tau^O_{k^*} < t < \tau^O_{{k^*} + 1}} S_{\varepsilon^O}\left(\mathcal{T}^{O}_{K} \cup \{ t\} \right).
\end{align*}

Thus,
\begin{align*}
& S_{\varepsilon^O}\left(\mathcal{T}^{O}_{K}\right) - S_{\varepsilon^O}\left(\hat{\mathcal{T}}^O_L\cup \mathcal{T}^{O}_{K}\right)\\
\geq & S_{\varepsilon^O}\left(\mathcal{T}^{O}_{K}\right) - \min_{\tau^O_{k^*} < t < \tau^O_{{k^*} + 1}} S_{\varepsilon^O}\left(\mathcal{T}^{O}_{K} \cup \{ t\} \right)\\
\geq & \max_{\tau^O_{k^*} < t < \tau^O_{{k^*} + 1}} \left\{ (t - \tau^O_{k^*}) \left(\overline{\varepsilon}^O_{\tau^O_{k^*} :t}\right)^2 + ( \tau^O_{{k^*} + 1} - t) \left(\overline{\varepsilon}^O_{t :\tau^O_{{k^*} + 1}}\right)^2 \right\} - (\tau^O_{{k^*} + 1} - \tau^O_{k^*}) \left(\overline{\varepsilon}^O_{\tau^O_{k^*} :\tau^O_{{k^*} + 1}}\right)^2.
\end{align*}

It follows from \citet[Lemma~2]{zou2020consistent} (see also \citet[Lemma~2.1]{horvath1993maximum}) that there exists a constant $c > 0$ such that
\begin{equation*}
\Pj\left(\max_{\tau^O_{k^*} < t < \tau^O_{{k^*} + 1}} \left\{ (t - \tau^O_{k^*}) \left(\overline{\varepsilon}^O_{\tau^O_{k^*} :t}\right)^2 + ( \tau^O_{{k^*} + 1} - t) \left(\overline{\varepsilon}^O_{t :\tau^O_{{k^*} + 1}}\right)^2 \right\} > c\, \underline{\sigma}^2 \log\log \overline{\lambda}\right) \to 1.
\end{equation*}
Moreover, $(\tau^O_{{k^*} + 1} - \tau^O_{k^*}) \Big(\overline{\varepsilon}^O_{\tau^O_{k^*} :\tau^O_{{k^*} + 1}}\Big)^2 = \mathcal{O}_{\Pj}(\sigma)$. Hence, the stated formula follows by using that $\limsup_{n \to \infty}\sigma / \underline{\sigma} < \infty$.

The same argument holds with all instances of $O$ replaced by $E$.
\end{proof}

\begin{proof}[Proof of Theorem~\ref{thm:consistency}]
For the first part we will use Theorem~\ref{theorem:positiveResultRescaledCV} and hence, we first verify its assumptions in the following. Note that Assumptions~\ref{assumption:cpNumberMultivariate}~and~\ref{assumption:minimumSignalMultivariate} are assumed in Theorem~\ref{thm:consistency} as well. Since $\varepsilon_1,\ldots,\varepsilon_n$ are sub-Gaussian with uniformly bounded variance proxy $\sigma^2$, Assumption~\ref{assumption:NoiseBernstein} follows. Moreover, Lemma~\ref{lemma:overestimation} shows that Assumption~\ref{assumption:overestimation} holds as well. It remains to show Assumption~\ref{assumption:detectionPrecision}. To this end, we will apply Theorem~\ref{theorem:detectionPrecision} to $Y_1^O,\ldots,Y_{n/2}^O$. Note that this sequence has $K$ change-points with minimal jump size $\Delta_{(1)}$, minimal distance between change-points of at least $\floor{\,\underline{\lambda}/2}$, and maximal distance between change-points of at most $\overline{\lambda}/2$. As noted in the discussion following Theorem~\ref{theorem:positiveResultRescaledCV}, Assumption~\ref{assumption:detectionPrecision} is not required for the conclusion when $K=0$. Hence, we can assume that $K > 0$. 

Now let $C_n$ be a sequence with $C_n \to \infty$ but such that
\begin{align}
C_n^2 K\log\log\bigg((\log(K) \vee 1) \frac{\sigma^2}{\Delta_{(1)}^2} \vee e \bigg) &= o(\log\log \overline{\lambda}), \label{eq:a_n_1}\\
C_n K (\log K \vee 1) &= o(\log\log \overline{\lambda}). \label{eq:a_n_2}
	\end{align}
The existence of such a sequence is guaranteed by \eqref{eq:conditionConsistency} and Assumption~\ref{assumption:cpNumberMultivariate}\ref{assumption:cpNumberBound}. Now we set
\begin{align*}
\gamma_{0k} &:= C_n \left(\log\left(K \right)\vee 1\right) \frac{\sigma^2}{\Delta_k^2}, \\
\gamma_{1k} &:=  C_n \left(\log\left(K \frac{\sigma^2}{\Delta_k^2} \right)\vee  1\right) \frac{\sigma^2}{\Delta_k^2}.
\end{align*}
This choice satisfies \eqref{eq:precisionRate}. We now claim that taking $\delta_{0k}=\floor{\gamma_{0k}}$ and $\delta_{1k} = \floor{\gamma_{1k}}$ satisfies Assumption~\ref{assumption:detectionPrecision}. It follows from \eqref{detectionPrecision:L>=K} in Theorem~\ref{theorem:detectionPrecision} that Assumption~\ref{assumption:detectionPrecision}\ref{assumption:detectionPrecision:L>=K} is satisfied, if we show that $\gamma_{0,k} \vee \gamma_{1,k} \leq \underline{\lambda}/2-1$. It follows from \eqref{eq:conditionConsistency} that
\begin{align*}
K \frac{\sigma^2}{\Delta_k^2} = o(\underline{\lambda}).
\end{align*}
Thus, \eqref{eq:a_n_2} and \eqref{eq:conditionConsistency} imply
\begin{align*}
\frac{\gamma_{0,k} \vee \gamma_{1,k}}{\underline{\lambda}}
\leq o\left(\frac{\sigma^2}{\underline{\lambda}\Delta_k^2}\log(\underline{\lambda})\log\log \overline{\lambda}\right)\to 0.
\end{align*}
From \eqref{eq:a_n_1}, we see that Assumption~\ref{assumption:detectionPrecision}\ref{assumption:detectionPrecision:Rate} is satisfied, since 
\begin{align*}
& \max_{q=0,1;\ 1\leq k \leq K }  K \log\log{ (\delta_{q,k} \vee e)}\\
\leq & K \log\log\left(C_n \left(\log\left(K\frac{\sigma^2}{\Delta_{k}^2} \right) \vee (\log K) \vee 1 \right) \frac{\sigma^2}{\Delta_{k}^2}\vee e\right)\\
\leq & \mathcal{O}\left(C_n K \log\log\left(\left(\log(K)\vee 1 \right)\frac{\sigma^2}{\Delta_{(1)}^2}\vee e\right)\right)\\
= & o(\log\log\overline{\lambda}).
\end{align*}

Next, from \eqref{eq:a_n_2} we see that Assumption~\ref{assumption:detectionPrecision}\ref{assumption:detectionPrecision:Infq0} is satisfied, since
\[
\sum_{k=1}^K \delta_{0,k}\Delta_k^2 \leq C_n K (\log (K) \vee 1) \sigma^2 = o(\sigma^2 \log\log \overline{\lambda}).
\]
Furthermore,  Assumption~\ref{assumption:detectionPrecision}\ref{assumption:detectionPrecision:Infq>0} is satisfied as well, since if $\frac{\sigma^2 \log \log \overline{\lambda}}{K \Delta_k^2} \leq C_n$, then
\begin{align*}
\gamma_{0,k} \vee \gamma_{1,k} \leq \frac{C_n^2 K}{\log \log \overline{\lambda}}  \left\{ \log K \vee \log\left(\frac{C_n}{\log \log \overline{\lambda}} \right)  \vee 1 \right\} \vee  1 \to 0,
\end{align*}
due to \eqref{eq:a_n_2}. Thus, $\delta_{0k} = \floor{\gamma_{0,k}} = \delta_{1k} = \floor{\gamma_{1,k}} = 0$.
Finally, it follows from $\delta_{1k} = 0$ if $\frac{\sigma^2 \log \log \overline{\lambda}}{K \Delta_k^2} \leq C_n$ and from \eqref{detectionPrecision:L<K} in Theorem~\ref{theorem:detectionPrecision} that Assumption~\ref{assumption:detectionPrecision}\ref{assumption:detectionPrecision:L<K} is satisfied as well. To this end, recall that the minimal distance between two change-points is at least $\floor{\,\underline{\lambda}/2}$ and note that 
$$
\bigg\lfloor\frac{\bfloor{\frac{\underline{\lambda}}{2}}}{2}\bigg\rfloor = \Bfloor{\frac{\underline{\lambda}}{4}}.
$$
Hence, Assumption~\ref{assumption:detectionPrecision} holds. The same is true when we replace O's by E's. Consequently, it follows from Theorem~\ref{theorem:positiveResultRescaledCV} that $\Pj\left( \hat{K} = K \right) \to 1, \text{ as } n\to \infty$.

It remains to show \eqref{eq:L2riskModifiedCriterion}. Let us write $\hat{\tau}_k := \hat{\tau}_{K, k}$ for $k = 0,\ldots,K+1$ and $\hat{\delta}_k := \vert \hat{\tau}_{k} - \tau_k \vert$. Recall that $\hat{f}_K:\ [0,1] \to \R,\ t \mapsto  \sum_{k = 0}^{K}{\overline{Y}_{\hat{\tau}_{k} :\hat{\tau}_{k + 1}} \EINS_{(\hat{\tau}_{k} / n, \hat{\tau}_{k + 1} / n]}(t)}$. We have that
\begin{equation*}
\begin{split}
&n \int_0^1{ \left(\hat{f}_{K}(t) - f(t)\right)^2 dt}\\
\leq & \sum_{k = 0}^{K}(\hat{\tau}_{k + 1} - \hat{\tau}_{k}) \vert \overline{Y}_{\hat{\tau}_{k} :\hat{\tau}_{k + 1}} - \beta_k \vert^2 + \sum_{k = 1}^{K}\hat{\delta}_k \left( \Delta_k + \max\big(\vert \overline{Y}_{\hat{\tau}_{k - 1} :\hat{\tau}_{k}} - \beta_{k - 1} \vert,\, \vert \overline{Y}_{\hat{\tau}_{k} :\hat{\tau}_{k + 1}} - \beta_k \vert\big) \right)^2.
\end{split}
\end{equation*}
In the following we will bound these terms. It follows from  Theorem~\ref{theorem:detectionPrecision} and the fact that $\Pj\left( \hat{K} = K \right) \to 1$, that for $\gamma_k := c_n(\log(K) \vee 1) \sigma^2 / \Delta_k^2$, where $c_n$ can be any sequence such that $c_n \to \infty$, as $n\to\infty$,
\begin{equation}\label{eq:deltakK}
\Pj\left(\hat{\delta}_k \leq \gamma_k \ \forall\ k = 1,\ldots,K  \text{ and } \hat{K} = K\right) \to 1.
\end{equation}
Furthermore, \eqref{eq:CPdetection} implies $\gamma_k < \underline{\lambda} / 2$ for a suitable chosen $c_n$. In the following, we work on the sequence of events in \eqref{eq:deltakK}. Consequently, $\hat{\tau}_{k + 1} - \hat{\tau}_{k} \geq \underline{\lambda} / 2$.

We also have that
\begin{equation*}
\vert \overline{Y}_{\hat{\tau}_{k} :\hat{\tau}_{k + 1}} - \beta_k \vert
\leq \frac{\hat{\delta}_k \Delta_k + \hat{\delta}_{k + 1} \Delta_{k + 1}}{\hat{\tau}_{k + 1} - \hat{\tau}_{k}} + \vert\overline{\varepsilon}_{\hat{\tau}_{k} :\hat{\tau}_{k + 1}}\vert.
\end{equation*}

For the following calculation we assume w.l.o.g.\ that $\hat{\tau}_{k} < \tau_{k} < \tau_{k + 1} < \hat{\tau}_{k + 1}$, since other cases lead to the same bound. Then,
\begin{equation*}
\begin{split}
&(\hat{\tau}_{k + 1} - \hat{\tau}_{k})\vert\overline{\varepsilon}_{\hat{\tau}_{k} :\hat{\tau}_{k + 1}}\vert^2\\
\leq & 3\left((\tau_{k} - \hat{\tau}_{k})\vert\overline{\varepsilon}_{\hat{\tau}_{k} :\tau_{k}}\vert^2 + (\tau_{k + 1} - \tau_{k})\vert\overline{\varepsilon}_{\tau_{k} :\tau_{k + 1}}\vert^2
+ (\hat{\tau}_{k + 1} - \tau_{k + 1})\vert\overline{\varepsilon}_{\tau_{k + 1} :\hat{\tau}_{k + 1}}\vert^2
\right)\\
\leq &\max_{t = \tau_{k} - \gamma_k,\ldots,\tau_{k} - 1}\{(\tau_{k} - t)\vert\overline{\varepsilon}_{t :\tau_{k}}\vert^2\} + (\tau_{k + 1} - \tau_{k})\vert\overline{\varepsilon}_{\tau_{k} :\tau_{k + 1}}\vert^2\\
& \hspace{15pt} + \max_{t = \tau_{k + 1} + 1,\ldots,\tau_{k + 1} + \gamma_{k+1}}\{(t - \tau_{k + 1})\vert\overline{\varepsilon}_{\tau_{k + 1} :t}\vert^2\}.
\end{split}
\end{equation*}
Hence, due to sub-Gaussianity and independence,
\begin{equation*}
\max_{k = 1,\ldots,K} \big(\log\log \gamma_k + 1 + \log\log \gamma_{k+1}\big)^{-1} (\hat{\tau}_{k + 1} - \hat{\tau}_{k})\vert\overline{\varepsilon}_{\hat{\tau}_{k} :\hat{\tau}_{k + 1}}\vert^2 = \mathcal{O}_\Pj(\log(K)\sigma^2).
\end{equation*}
Note that here and in the following we have written $\log\log \gamma_k$ instead of $\log\log(\gamma_k \vee e)$ to improve readability.

Thus, using $(x + y)^2 \leq 2x^2 + 2y^2$,
\begin{equation*}
\begin{split}
& \sum_{k = 0}^{K}(\hat{\tau}_{k + 1} - \hat{\tau}_{k}) \vert \overline{Y}_{\hat{\tau}_{k} :\hat{\tau}_{k + 1}} - \beta_k \vert^2
= \mathcal{O}_\Pj\left( \sum_{k = 1}^{K} \frac{\gamma_k \Delta_k^2}{\underline{\lambda}} \right)
+ \mathcal{O}_\Pj\left( \sum_{k = 1}^{K} (1 + \log\log \gamma_k) \log(K)\sigma^2 \right),
\end{split}
\end{equation*}
and
\begin{equation*}
\begin{split}
& \sum_{k = 1}^{K}\hat{\delta}_k \left( \Delta_k + \max\big(\vert \overline{Y}_{\hat{\tau}_{k - 1} :\hat{\tau}_{k}} - \beta_{k - 1} \vert,\, \vert \overline{Y}_{\hat{\tau}_{k} :\hat{\tau}_{k + 1}} - \beta_k \vert\big) \right)^2\\
= & \mathcal{O}_\Pj\left(\sum_{k = 1}^{K} \gamma_k \Delta_k^2\right) + \mathcal{O}_\Pj\left(\sum_{k = 1}^{K} \gamma_k \frac{\left(\gamma_{k-1}^2 \Delta_{k - 1}^2 + \gamma_k^2 \Delta_k^2 + \gamma_{k+1}^2 \Delta_{k + 1}^2\right)}{\underline{\lambda}^2} \right)\\
& + \mathcal{O}_\Pj\left( \sum_{k = 1}^{K} \gamma_k (1 + \log\log \gamma_{k-1} + \log\log \gamma_k + \log\log \gamma_{k+1}) \log(K)\sigma^2 / \underline{\lambda} \right),
\end{split}
\end{equation*}
where we have used the notation $\Delta_{0} = \Delta_{K + 1} = 0$ and $\gamma_{0} = \gamma_{K + 1} = e$.

Since $\gamma_k \leq \underline{\lambda}$, it follows that
\begin{equation*}
\begin{split}
&n \int_0^1{ \left(\hat{f}_{K}(t) - f(t)\right)^2 dt}\\
= & \mathcal{O}_\Pj\left( \sum_{k = 1}^{K} \gamma_k \Delta_k^2 \right)
+ \mathcal{O}_\Pj\left( \sum_{k = 1}^{K} (1 + \log\log \gamma_k) \log(K)\sigma^2 \right)\\
= & \mathcal{O}_\Pj\left( \sum_{k = 1}^{K} (c_n + \log\log \gamma_k) (\log(K) \vee 1) \sigma^2\right)\\
= & \mathcal{O}_\Pj\left( \left(c_n + \log\log\left( c_n(\log(K) \vee 1) \sigma^2 / \Delta_{(1)}^2 \right) \right) K (\log(K) \vee 1) \sigma^2\right).
\end{split}
\end{equation*}
%We have that
%\begin{equation*}
%\begin{split}
%&\Pj\left( n \int_0^1{ \left(\hat{f}_{\hat{K}}(t) - f(t)\right)^2 dt} > x \right)\\
%= &\Pj\left( n \int_0^1{ \left(\hat{f}_{K}(t) - f(t)\right)^2 dt} > x, \hat{K} = K, \hat{\delta}_k \leq \gamma_k\ \forall\ k = 1,\ldots,K \right)\\
%& + \Pj\left( n \int_0^1{ \left(\hat{f}_{K}(t) - f(t)\right)^2 dt} > x, \hat{K} \neq K \text{ or } \exists\ k = 1,\ldots,K\,:\, \hat{\delta}_k > \gamma_k\right)
%\end{split}
%\end{equation*}
%and that the last probability converges to zero, due to $\Pj\big( \hat{K} = K \big) \to 1$ and \eqref{eq:deltakK}. Consequently, we obtain the bound
%\begin{equation*}
%\mathcal{O}_\Pj\left( n^{-1}\left(c_n + \log\log\left( c_n(\log(eK) \sigma^2 / \Delta_{(1)}^2 \right) \right) K \log(eK) \sigma^2\right)
%\end{equation*}
By using the second part of \eqref{eq:conditionConsistency} and that $\log\log\overline{\lambda} \to \infty$, as $n\to\infty$, as well as by choosing a $c_n$ that increases slowly enough, this simplifies to the claimed bound.
\end{proof}

%\bibliographystyle{imsart-nameyear}
%\bibliography{Literature}

\begin{thebibliography}{46}
	% BibTex style file: imsart-nameyear.bst, 2017-11-03
	% Default style options (sort=1,type=nameyear).
	% Used options (sort=1,type=nameyear).
	
	\bibitem[\protect\citeauthoryear{Antoch, Hu{\v{s}}kov{\'a} and
		Veraverbeke}{1995}]{antoch1995change}
	\begin{barticle}[author]
		\bauthor{\bsnm{Antoch},~\bfnm{Jaromir}\binits{J.}},
		\bauthor{\bsnm{Hu{\v{s}}kov{\'a}},~\bfnm{Marie}\binits{M.}} \AND
		\bauthor{\bsnm{Veraverbeke},~\bfnm{No{\"e}l}\binits{N.}}
		(\byear{1995}).
		\btitle{Change-point problem and bootstrap}.
		\bjournal{J. Nonparametr. Stat.}
		\bvolume{5}
		\bpages{123--144}.
	\end{barticle}
	\endbibitem
	
	\bibitem[\protect\citeauthoryear{Arlot and Celisse}{2010}]{arlot2010survey}
	\begin{barticle}[author]
		\bauthor{\bsnm{Arlot},~\bfnm{S.}\binits{S.}} \AND
		\bauthor{\bsnm{Celisse},~\bfnm{A.}\binits{A.}}
		(\byear{2010}).
		\btitle{A survey of cross-validation procedures for model selection}.
		\bjournal{Stat. Surv.}
		\bvolume{4}
		\bpages{40--79}.
	\end{barticle}
	\endbibitem
	
	\bibitem[\protect\citeauthoryear{Arlot and
		Celisse}{2011}]{arlot2011segmentation}
	\begin{barticle}[author]
		\bauthor{\bsnm{Arlot},~\bfnm{S.}\binits{S.}} \AND
		\bauthor{\bsnm{Celisse},~\bfnm{A.}\binits{A.}}
		(\byear{2011}).
		\btitle{Segmentation of the mean of heteroscedastic data via cross-validation}.
		\bjournal{Stat. Comput.}
		\bvolume{21}
		\bpages{613--632}.
	\end{barticle}
	\endbibitem
	
	\bibitem[\protect\citeauthoryear{Auger and
		Lawrence}{1989}]{auger1989algorithms}
	\begin{barticle}[author]
		\bauthor{\bsnm{Auger},~\bfnm{I.~E.}\binits{I.~E.}} \AND
		\bauthor{\bsnm{Lawrence},~\bfnm{C.~E.}\binits{C.~E.}}
		(\byear{1989}).
		\btitle{Algorithms for the optimal identification of segment neighborhoods}.
		\bjournal{Bull. Math. Biol.}
		\bvolume{51}
		\bpages{39--54}.
	\end{barticle}
	\endbibitem
	
	\bibitem[\protect\citeauthoryear{Bai and Perron}{2003}]{bai2003computation}
	\begin{barticle}[author]
		\bauthor{\bsnm{Bai},~\bfnm{J.}\binits{J.}} \AND
		\bauthor{\bsnm{Perron},~\bfnm{P.}\binits{P.}}
		(\byear{2003}).
		\btitle{Computation and analysis of multiple structural change models}.
		\bjournal{J. Appl. Econometrics}
		\bvolume{18}
		\bpages{1--22}.
	\end{barticle}
	\endbibitem
	
	\bibitem[\protect\citeauthoryear{Chan and Walther}{2013}]{chan2013detection}
	\begin{barticle}[author]
		\bauthor{\bsnm{Chan},~\bfnm{H.~P.}\binits{H.~P.}} \AND
		\bauthor{\bsnm{Walther},~\bfnm{G.}\binits{G.}}
		(\byear{2013}).
		\btitle{Detection with the scan and the average likelihood ratio}.
		\bjournal{Statist. Sinica}
		\bpages{409--428}.
	\end{barticle}
	\endbibitem
	
	\bibitem[\protect\citeauthoryear{Chetverikov, Liao and
		Chernozhukov}{2021}]{chetverikov2021cross}
	\begin{barticle}[author]
		\bauthor{\bsnm{Chetverikov},~\bfnm{D.}\binits{D.}},
		\bauthor{\bsnm{Liao},~\bfnm{Z.}\binits{Z.}} \AND
		\bauthor{\bsnm{Chernozhukov},~\bfnm{V.}\binits{V.}}
		(\byear{2021}).
		\btitle{On cross-validated lasso in high dimensions}.
		\bjournal{Ann. Statist.}
		\bvolume{49}
		\bpages{1300--1317}.
	\end{barticle}
	\endbibitem
	
	\bibitem[\protect\citeauthoryear{D’Angelo et~al.}{2011}]{d2011incipient}
	\begin{barticle}[author]
		\bauthor{\bsnm{D’Angelo},~\bfnm{M.}\binits{M.}},
		\bauthor{\bsnm{Palhares},~\bfnm{R.~M.}\binits{R.~M.}},
		\bauthor{\bsnm{Takahashi},~\bfnm{R.}\binits{R.}},
		\bauthor{\bsnm{Loschi},~\bfnm{R.~H.}\binits{R.~H.}},
		\bauthor{\bsnm{Baccarini},~\bfnm{L.}\binits{L.}} \AND
		\bauthor{\bsnm{Caminhas},~\bfnm{W.}\binits{W.}}
		(\byear{2011}).
		\btitle{Incipient fault detection in induction machine stator-winding using a
			fuzzy-Bayesian change point detection approach}.
		\bjournal{Appl. Soft. Comput.}
		\bvolume{11}
		\bpages{179--192}.
	\end{barticle}
	\endbibitem
	
	\bibitem[\protect\citeauthoryear{Donoho and Johnstone}{1994}]{donoho1994ideal}
	\begin{barticle}[author]
		\bauthor{\bsnm{Donoho},~\bfnm{D.~L.}\binits{D.~L.}} \AND
		\bauthor{\bsnm{Johnstone},~\bfnm{J.~M.}\binits{J.~M.}}
		(\byear{1994}).
		\btitle{Ideal spatial adaptation by wavelet shrinkage}.
		\bjournal{Biometrika}
		\bvolume{81}
		\bpages{425--455}.
	\end{barticle}
	\endbibitem
	
	\bibitem[\protect\citeauthoryear{Du, Kao and Kou}{2016}]{du2016stepwise}
	\begin{barticle}[author]
		\bauthor{\bsnm{Du},~\bfnm{C.}\binits{C.}},
		\bauthor{\bsnm{Kao},~\bfnm{C.~L.~M.}\binits{C.~L.~M.}} \AND
		\bauthor{\bsnm{Kou},~\bfnm{S.~C.}\binits{S.~C.}}
		(\byear{2016}).
		\btitle{Stepwise signal extraction via marginal likelihood}.
		\bjournal{J. Amer. Statist. Assoc.}
		\bvolume{111}
		\bpages{314--330}.
	\end{barticle}
	\endbibitem
	
	\bibitem[\protect\citeauthoryear{D{\"u}mbgen and
		Spokoiny}{2001}]{dumbgen2001multiscale}
	\begin{barticle}[author]
		\bauthor{\bsnm{D{\"u}mbgen},~\bfnm{L.}\binits{L.}} \AND
		\bauthor{\bsnm{Spokoiny},~\bfnm{V.~G.}\binits{V.~G.}}
		(\byear{2001}).
		\btitle{Multiscale testing of qualitative hypotheses}.
		\bjournal{Ann. Statist.}
		\bpages{124--152}.
	\end{barticle}
	\endbibitem
	
	\bibitem[\protect\citeauthoryear{Fearnhead}{2006}]{fearnhead2006exact}
	\begin{barticle}[author]
		\bauthor{\bsnm{Fearnhead},~\bfnm{P.}\binits{P.}}
		(\byear{2006}).
		\btitle{Exact and efficient Bayesian inference for multiple changepoint
			problems}.
		\bjournal{Stat. Comput.}
		\bvolume{16}
		\bpages{203--213}.
	\end{barticle}
	\endbibitem
	
	\bibitem[\protect\citeauthoryear{Fearnhead and
		Rigaill}{2019}]{fearnhead2019changepoint}
	\begin{barticle}[author]
		\bauthor{\bsnm{Fearnhead},~\bfnm{P.}\binits{P.}} \AND
		\bauthor{\bsnm{Rigaill},~\bfnm{G.}\binits{G.}}
		(\byear{2019}).
		\btitle{Changepoint detection in the presence of outliers}.
		\bjournal{J. Amer. Statist. Assoc.}
		\bvolume{114}
		\bpages{169--183}.
	\end{barticle}
	\endbibitem
	
	\bibitem[\protect\citeauthoryear{Fearnhead and
		Rigaill}{2020}]{fearnhead2020relating}
	\begin{barticle}[author]
		\bauthor{\bsnm{Fearnhead},~\bfnm{P.}\binits{P.}} \AND
		\bauthor{\bsnm{Rigaill},~\bfnm{G.}\binits{G.}}
		(\byear{2020}).
		\btitle{Relating and comparing methods for detecting changes in mean}.
		\bjournal{Stat}
		\bvolume{9}
		\bpages{e291}.
	\end{barticle}
	\endbibitem
	
	\bibitem[\protect\citeauthoryear{Frick, Munk and
		Sieling}{2014}]{frick2014multiscale}
	\begin{barticle}[author]
		\bauthor{\bsnm{Frick},~\bfnm{K.}\binits{K.}},
		\bauthor{\bsnm{Munk},~\bfnm{A.}\binits{A.}} \AND
		\bauthor{\bsnm{Sieling},~\bfnm{H.}\binits{H.}}
		(\byear{2014}).
		\btitle{Multiscale change point inference}.
		\bjournal{J. R. Stat. Soc. Ser. B. Stat. Methodol.}
		\bpages{495--580}.
	\end{barticle}
	\endbibitem
	
	\bibitem[\protect\citeauthoryear{Fryzlewicz}{2014}]{fryzlewicz2014wild}
	\begin{barticle}[author]
		\bauthor{\bsnm{Fryzlewicz},~\bfnm{P.}\binits{P.}}
		(\byear{2014}).
		\btitle{Wild binary segmentation for multiple change-point detection}.
		\bjournal{Ann. Statist.}
		\bvolume{42}
		\bpages{2243--2281}.
	\end{barticle}
	\endbibitem
	
	\bibitem[\protect\citeauthoryear{Fryzlewicz}{2020}]{fryzlewicz2020detecting}
	\begin{barticle}[author]
		\bauthor{\bsnm{Fryzlewicz},~\bfnm{P.}\binits{P.}}
		(\byear{2020}).
		\btitle{Detecting possibly frequent change-points: {W}ild Binary Segmentation 2
			and steepest-drop model selection}.
		\bjournal{J. Korean Statist. Soc.}
		\bvolume{49}
		\bpages{1027--1070}.
	\end{barticle}
	\endbibitem
	
	\bibitem[\protect\citeauthoryear{Garreau and
		Arlot}{2018}]{Garreau2018Consistent}
	\begin{barticle}[author]
		\bauthor{\bsnm{Garreau},~\bfnm{Damien}\binits{D.}} \AND
		\bauthor{\bsnm{Arlot},~\bfnm{Sylvain}\binits{S.}}
		(\byear{2018}).
		\btitle{{Consistent change-point detection with kernels}}.
		\bjournal{Electron. J. Stat.}
		\bvolume{12}
		\bpages{4440 -- 4486}.
		\bdoi{10.1214/18-EJS1513}
	\end{barticle}
	\endbibitem
	
	\bibitem[\protect\citeauthoryear{Harchaoui
		et~al.}{2009}]{harchaoui2009regularized}
	\begin{binproceedings}[author]
		\bauthor{\bsnm{Harchaoui},~\bfnm{Z.}\binits{Z.}},
		\bauthor{\bsnm{Vallet},~\bfnm{F.}\binits{F.}},
		\bauthor{\bsnm{Lung-Yut-Fong},~\bfnm{A.}\binits{A.}} \AND
		\bauthor{\bsnm{Capp{\'e}},~\bfnm{O.}\binits{O.}}
		(\byear{2009}).
		\btitle{A regularized kernel-based approach to unsupervised audio
			segmentation}.
		In \bbooktitle{2009 IEEE International Conference on Acoustics, Speech and
			Signal Processing}
		\bpages{1665--1668}.
		\bpublisher{IEEE}.
	\end{binproceedings}
	\endbibitem
	
	\bibitem[\protect\citeauthoryear{Hu{\v{s}}kov{\'a} and
		Kirch}{2008}]{huvskova2008bootstrapping}
	\begin{barticle}[author]
		\bauthor{\bsnm{Hu{\v{s}}kov{\'a}},~\bfnm{Marie}\binits{M.}} \AND
		\bauthor{\bsnm{Kirch},~\bfnm{Claudia}\binits{C.}}
		(\byear{2008}).
		\btitle{Bootstrapping confidence intervals for the change-point of time
			series}.
		\bjournal{J. Time Series Anal.}
		\bvolume{29}
		\bpages{947--972}.
	\end{barticle}
	\endbibitem
	
	\bibitem[\protect\citeauthoryear{Jackson et~al.}{2005}]{jackson2005algorithm}
	\begin{barticle}[author]
		\bauthor{\bsnm{Jackson},~\bfnm{B.}\binits{B.}},
		\bauthor{\bsnm{Scargle},~\bfnm{J.~D.}\binits{J.~D.}},
		\bauthor{\bsnm{Barnes},~\bfnm{D.}\binits{D.}},
		\bauthor{\bsnm{Arabhi},~\bfnm{S.}\binits{S.}},
		\bauthor{\bsnm{Alt},~\bfnm{A.}\binits{A.}},
		\bauthor{\bsnm{Gioumousis},~\bfnm{P.}\binits{P.}},
		\bauthor{\bsnm{Gwin},~\bfnm{E.}\binits{E.}},
		\bauthor{\bsnm{Sangtrakulcharoen},~\bfnm{P.}\binits{P.}},
		\bauthor{\bsnm{Tan},~\bfnm{L.}\binits{L.}} \AND
		\bauthor{\bsnm{Tsai},~\bfnm{T.~T.}\binits{T.~T.}}
		(\byear{2005}).
		\btitle{An algorithm for optimal partitioning of data on an interval}.
		\bjournal{IEEE Signal Process. Lett.}
		\bvolume{12}
		\bpages{105--108}.
	\end{barticle}
	\endbibitem
	
	\bibitem[\protect\citeauthoryear{Killick, Fearnhead and
		Eckley}{2012}]{killick2012optimal}
	\begin{barticle}[author]
		\bauthor{\bsnm{Killick},~\bfnm{R.}\binits{R.}},
		\bauthor{\bsnm{Fearnhead},~\bfnm{P.}\binits{P.}} \AND
		\bauthor{\bsnm{Eckley},~\bfnm{I.~A.}\binits{I.~A.}}
		(\byear{2012}).
		\btitle{Optimal detection of changepoints with a linear computational cost}.
		\bjournal{J. Amer. Statist. Assoc.}
		\bvolume{107}
		\bpages{1590--1598}.
	\end{barticle}
	\endbibitem
	
	\bibitem[\protect\citeauthoryear{Kim, Morley and
		Nelson}{2005}]{kim2005structural}
	\begin{barticle}[author]
		\bauthor{\bsnm{Kim},~\bfnm{C.~J.}\binits{C.~J.}},
		\bauthor{\bsnm{Morley},~\bfnm{J.~C.}\binits{J.~C.}} \AND
		\bauthor{\bsnm{Nelson},~\bfnm{C.~R.}\binits{C.~R.}}
		(\byear{2005}).
		\btitle{The structural break in the equity premium}.
		\bjournal{J. Finance}
		\bvolume{23}
		\bpages{181--191}.
	\end{barticle}
	\endbibitem	
	
	\bibitem[\protect\citeauthoryear{Kov{\'a}cs et~al.}{2020}]{kovacs2020seeded}
	\begin{barticle}[author]
		\bauthor{\bsnm{Kov{\'a}cs},~\bfnm{S.}\binits{S.}},
		\bauthor{\bsnm{Li},~\bfnm{H.}\binits{H.}},
		\bauthor{\bsnm{B{\"u}hlmann},~\bfnm{P.}\binits{P.}} \AND
		\bauthor{\bsnm{Munk},~\bfnm{A.}\binits{A.}}
		(\byear{2023}).
		\btitle{Seeded Binary Segmentation: {A} general methodology for fast and
			optimal change point detection}.
		\bjournal{Biometrika}
		\bvolume{110}
		\bpages{249--256}.
	\end{barticle}
	\endbibitem
	
	\bibitem[\protect\citeauthoryear{Li, Guo and Munk}{2019}]{li2019multiscale}
	\begin{barticle}[author]
		\bauthor{\bsnm{Li},~\bfnm{H.}\binits{H.}},
		\bauthor{\bsnm{Guo},~\bfnm{Q.}\binits{Q.}} \AND
		\bauthor{\bsnm{Munk},~\bfnm{A.}\binits{A.}}
		(\byear{2019}).
		\btitle{Multiscale change-point segmentation: {B}eyond step functions}.
		\bjournal{Electron. J. Stat.}
		\bvolume{13}
		\bpages{3254--3296}.
	\end{barticle}
	\endbibitem
	
	\bibitem[\protect\citeauthoryear{Li, Munk and Sieling}{2016}]{li2016fdr}
	\begin{barticle}[author]
		\bauthor{\bsnm{Li},~\bfnm{H.}\binits{H.}},
		\bauthor{\bsnm{Munk},~\bfnm{A.}\binits{A.}} \AND
		\bauthor{\bsnm{Sieling},~\bfnm{H.}\binits{H.}}
		(\byear{2016}).
		\btitle{{FDR}-control in multiscale change-point segmentation}.
		\bjournal{Electron. J. Stat.}
		\bvolume{10}
		\bpages{918--959}.
	\end{barticle}
	\endbibitem
	
	\bibitem[\protect\citeauthoryear{Li et~al.}{2022}]{li2022automatic}
	\begin{barticle}[author]
		\bauthor{\bsnm{Li},~\bfnm{Jie}\binits{J.}},
		\bauthor{\bsnm{Fearnhead},~\bfnm{Paul}\binits{P.}},
		\bauthor{\bsnm{Fryzlewicz},~\bfnm{Piotr}\binits{P.}} \AND
		\bauthor{\bsnm{Wang},~\bfnm{Tengyao}\binits{T.}}
		(\byear{2022}).
		\btitle{Automatic change-point detection in time series via deep learning}.
		\bjournal{arXiv preprint arXiv:2211.03860}.
	\end{barticle}
	\endbibitem
	
	\bibitem[\protect\citeauthoryear{Liehrmann and Rigaill}{2023}]{liehrmann2023ms}
	\begin{barticle}[author]
		\bauthor{\bsnm{Liehrmann},~\bfnm{Arnaud}\binits{A.}} \AND
		\bauthor{\bsnm{Rigaill},~\bfnm{Guillem}\binits{G.}}
		(\byear{2023}).
		\btitle{Ms. FPOP: An Exact and Fast Segmentation Algorithm With a Multiscale
			Penalty}.
		\bjournal{arXiv preprint arXiv:2303.08723}.
	\end{barticle}
	\endbibitem
	
	\bibitem[\protect\citeauthoryear{Lin et~al.}{2016}]{lin2016approximate}
	\begin{barticle}[author]
		\bauthor{\bsnm{Lin},~\bfnm{K.}\binits{K.}},
		\bauthor{\bsnm{Sharpnack},~\bfnm{J.}\binits{J.}},
		\bauthor{\bsnm{Rinaldo},~\bfnm{A.}\binits{A.}} \AND
		\bauthor{\bsnm{Tibshirani},~\bfnm{R.~J.}\binits{R.~J.}}
		(\byear{2016}).
		\btitle{Approximate Recovery in Changepoint Problems, from l\_2 Estimation
			Error Rates}.
		\bjournal{arXiv preprint arXiv:1606.06746}.
	\end{barticle}
	\endbibitem
	
	\bibitem[\protect\citeauthoryear{Maidstone et~al.}{2017}]{maidstone2017optimal}
	\begin{barticle}[author]
		\bauthor{\bsnm{Maidstone},~\bfnm{R.}\binits{R.}},
		\bauthor{\bsnm{Hocking},~\bfnm{T.}\binits{T.}},
		\bauthor{\bsnm{Rigaill},~\bfnm{G.}\binits{G.}} \AND
		\bauthor{\bsnm{Fearnhead},~\bfnm{P.}\binits{P.}}
		(\byear{2017}).
		\btitle{On optimal multiple changepoint algorithms for large data}.
		\bjournal{Stat. Comput.}
		\bvolume{27}
		\bpages{519--533}.
	\end{barticle}
	\endbibitem
	
	\bibitem[\protect\citeauthoryear{Niu, Hao and Zhang}{2016}]{niu2016multiple}
	\begin{barticle}[author]
		\bauthor{\bsnm{Niu},~\bfnm{Y.~S.}\binits{Y.~S.}},
		\bauthor{\bsnm{Hao},~\bfnm{N.}\binits{N.}} \AND
		\bauthor{\bsnm{Zhang},~\bfnm{H.}\binits{H.}}
		(\byear{2016}).
		\btitle{Multiple change-point detection: {A} selective overview}.
		\bjournal{Statist. Sci.}
		\bpages{611--623}.
	\end{barticle}
	\endbibitem
	
	\bibitem[\protect\citeauthoryear{Olshen et~al.}{2004}]{olshen2004circular}
	\begin{barticle}[author]
		\bauthor{\bsnm{Olshen},~\bfnm{A.~B.}\binits{A.~B.}},
		\bauthor{\bsnm{Venkatraman},~\bfnm{E.}\binits{E.}},
		\bauthor{\bsnm{Lucito},~\bfnm{R.}\binits{R.}} \AND
		\bauthor{\bsnm{Wigler},~\bfnm{M.}\binits{M.}}
		(\byear{2004}).
		\btitle{Circular binary segmentation for the analysis of array-based {DNA} copy
			number data}.
		\bjournal{Biostatistics}
		\bvolume{5}
		\bpages{557--572}.
	\end{barticle}
	\endbibitem
	
	\bibitem[\protect\citeauthoryear{Pein, Eltzner and
		Munk}{2021}]{pein2021analysis}
	\begin{barticle}[author]
		\bauthor{\bsnm{Pein},~\bfnm{F.}\binits{F.}},
		\bauthor{\bsnm{Eltzner},~\bfnm{B.}\binits{B.}} \AND
		\bauthor{\bsnm{Munk},~\bfnm{A.}\binits{A.}}
		(\byear{2021}).
		\btitle{Analysis of patchclamp recordings: model-free multiscale methods and
			software}.
		\bjournal{Eur. Biophys. J.}
		\bvolume{50}
		\bpages{187--209}.
	\end{barticle}
	\endbibitem
	
	\bibitem[\protect\citeauthoryear{Pein, Sieling and
		Munk}{2017}]{pein2017heterogeneous}
	\begin{barticle}[author]
		\bauthor{\bsnm{Pein},~\bfnm{F.}\binits{F.}},
		\bauthor{\bsnm{Sieling},~\bfnm{H.}\binits{H.}} \AND
		\bauthor{\bsnm{Munk},~\bfnm{A.}\binits{A.}}
		(\byear{2017}).
		\btitle{Heterogeneous change point inference}.
		\bjournal{J. R. Stat. Soc. Ser. B. Stat. Methodol.}
		\bvolume{79}
		\bpages{1207--1227}.
	\end{barticle}
	\endbibitem
	
	\bibitem[\protect\citeauthoryear{Rigaill}{2015}]{rigaill2015pruned}
	\begin{barticle}[author]
		\bauthor{\bsnm{Rigaill},~\bfnm{Guillem}\binits{G.}}
		(\byear{2015}).
		\btitle{A pruned dynamic programming algorithm to recover the best
			segmentations with $1 $ to $ K\_ $\{$max$\}$ $ change-points}.
		\bjournal{Journal de la Soci{\'e}t{\'e} Fran{\c{c}}aise de Statistique}
		\bvolume{156}
		\bpages{180--205}.
	\end{barticle}
	\endbibitem
	
	\bibitem[\protect\citeauthoryear{Sharipov, Tewes and
		Wendler}{2016}]{sharipov2016sequential}
	\begin{barticle}[author]
		\bauthor{\bsnm{Sharipov},~\bfnm{Olimjon}\binits{O.}},
		\bauthor{\bsnm{Tewes},~\bfnm{Johannes}\binits{J.}} \AND
		\bauthor{\bsnm{Wendler},~\bfnm{Martin}\binits{M.}}
		(\byear{2016}).
		\btitle{Sequential block bootstrap in a Hilbert space with application to
			change point analysis}.
		\bjournal{Canad. J. Statist.}
		\bvolume{44}
		\bpages{300--322}.
	\end{barticle}
	\endbibitem
	
	\bibitem[\protect\citeauthoryear{Truong, Oudre and
		Vayatis}{2020}]{truong2020selective}
	\begin{barticle}[author]
		\bauthor{\bsnm{Truong},~\bfnm{C.}\binits{C.}},
		\bauthor{\bsnm{Oudre},~\bfnm{L.}\binits{L.}} \AND
		\bauthor{\bsnm{Vayatis},~\bfnm{N.}\binits{N.}}
		(\byear{2020}).
		\btitle{Selective review of offline change point detection methods}.
		\bjournal{Signal Process.}
		\bvolume{167}
		\bpages{107299}.
	\end{barticle}
	\endbibitem
	
	\bibitem[\protect\citeauthoryear{Verzelen et~al.}{2020}]{verzelen2020optimal}
	\begin{barticle}[author]
		\bauthor{\bsnm{Verzelen},~\bfnm{N.}\binits{N.}},
		\bauthor{\bsnm{Fromont},~\bfnm{M.}\binits{M.}},
		\bauthor{\bsnm{Lerasle},~\bfnm{M.}\binits{M.}} \AND
		\bauthor{\bsnm{Reynaud-Bouret},~\bfnm{P.}\binits{P.}}
		(\byear{2020}).
		\btitle{Optimal Change-Point Detection and Localization}.
		\bjournal{arXiv preprint arXiv:2010.11470}.
	\end{barticle}
	\endbibitem
	
	\bibitem[\protect\citeauthoryear{Vostrikova}{1981}]{vostrikova1981detecting}
	\begin{barticle}[author]
		\bauthor{\bsnm{Vostrikova},~\bfnm{L.~Y.}\binits{L.~Y.}}
		(\byear{1981}).
		\btitle{Detecting ``disorder'' in multidimensional random processes}.
		\bjournal{Dokl. Akad. Nauk}
		\bvolume{259}
		\bpages{270--274}.
	\end{barticle}
	\endbibitem
	
	\bibitem[\protect\citeauthoryear{Wong}{1983}]{wong1983consistency}
	\begin{barticle}[author]
		\bauthor{\bsnm{Wong},~\bfnm{W.~H.}\binits{W.~H.}}
		(\byear{1983}).
		\btitle{On the consistency of cross-validation in kernel nonparametric
			regression}.
		\bjournal{Ann. Statist.}
		\bvolume{11}
		\bpages{1136--1141}.
	\end{barticle}
	\endbibitem
	
	\bibitem[\protect\citeauthoryear{Yang}{2007}]{yang2007consistency}
	\begin{barticle}[author]
		\bauthor{\bsnm{Yang},~\bfnm{Y.}\binits{Y.}}
		(\byear{2007}).
		\btitle{Consistency of cross validation for comparing regression procedures}.
		\bjournal{Ann. Statist.}
		\bvolume{35}
		\bpages{2450--2473}.
	\end{barticle}
	\endbibitem
	
	\bibitem[\protect\citeauthoryear{Yao}{1988}]{yao1988estimating}
	\begin{barticle}[author]
		\bauthor{\bsnm{Yao},~\bfnm{Y.~C.}\binits{Y.~C.}}
		(\byear{1988}).
		\btitle{Estimating the number of change-points via Schwarz' criterion}.
		\bjournal{Statist. Probab. Lett.}
		\bvolume{6}
		\bpages{181--189}.
	\end{barticle}
	\endbibitem
	
	\bibitem[\protect\citeauthoryear{Yao and Au}{1989}]{yao1989least}
	\begin{barticle}[author]
		\bauthor{\bsnm{Yao},~\bfnm{Y.~C.}\binits{Y.~C.}} \AND
		\bauthor{\bsnm{Au},~\bfnm{S.~T.}\binits{S.~T.}}
		(\byear{1989}).
		\btitle{Least-squares estimation of a step function}.
		\bjournal{Sankhya A}
		\bpages{370--381}.
	\end{barticle}
	\endbibitem
	
	\bibitem[\protect\citeauthoryear{Yu and Feng}{2014}]{yu2014modified}
	\begin{barticle}[author]
		\bauthor{\bsnm{Yu},~\bfnm{Y.}\binits{Y.}} \AND
		\bauthor{\bsnm{Feng},~\bfnm{Y.}\binits{Y.}}
		(\byear{2014}).
		\btitle{Modified cross-validation for penalized high-dimensional linear
			regression models}.
		\bjournal{J. Comput. Graph. Statist.}
		\bvolume{23}
		\bpages{1009--1027}.
	\end{barticle}
	\endbibitem
	
	\bibitem[\protect\citeauthoryear{Zhang and Siegmund}{2007}]{zhang2007modified}
	\begin{barticle}[author]
		\bauthor{\bsnm{Zhang},~\bfnm{N.~R.}\binits{N.~R.}} \AND
		\bauthor{\bsnm{Siegmund},~\bfnm{D.~O.}\binits{D.~O.}}
		(\byear{2007}).
		\btitle{A modified {B}ayes information criterion with applications to the
			analysis of comparative genomic hybridization data}.
		\bjournal{Biometrics}
		\bvolume{63}
		\bpages{22--32}.
	\end{barticle}
	\endbibitem
	
	\bibitem[\protect\citeauthoryear{Zou, Wang and Li}{2020}]{zou2020consistent}
	\begin{barticle}[author]
		\bauthor{\bsnm{Zou},~\bfnm{C.}\binits{C.}},
		\bauthor{\bsnm{Wang},~\bfnm{G.}\binits{G.}} \AND
		\bauthor{\bsnm{Li},~\bfnm{R.}\binits{R.}}
		(\byear{2020}).
		\btitle{Consistent selection of the number of change-points via
			sample-splitting}.
		\bjournal{Ann. Statist.}
		\bvolume{48}
		\bpages{413--439}.
	\end{barticle}
	\endbibitem
	
\end{thebibliography}

\begin{thebibliography}{8}
	% BibTex style file: imsart-nameyear.bst, 2017-11-03
	% Default style options (sort=1,type=nameyear).
	% Used options (sort=1,type=nameyear).
	
	\bibitem[\protect\citeauthoryear{Balsubramani}{2014}]{balsubramani2014sharp}
	\begin{barticle}[author]
		\bauthor{\bsnm{Balsubramani},~\bfnm{A.}\binits{A.}}
		(\byear{2014}).
		\btitle{Sharp finite-time iterated-logarithm martingale concentration}.
		\bjournal{arXiv preprint arXiv:1405.2639}.
	\end{barticle}
	\endbibitem
	
	\bibitem[\protect\citeauthoryear{Boucheron, Lugosi and
		Massart}{2013}]{boucheron2013concentration}
	\begin{bbook}[author]
		\bauthor{\bsnm{Boucheron},~\bfnm{S.}\binits{S.}},
		\bauthor{\bsnm{Lugosi},~\bfnm{G.}\binits{G.}} \AND
		\bauthor{\bsnm{Massart},~\bfnm{P.}\binits{P.}}
		(\byear{2013}).
		\btitle{Concentration inequalities: {A} nonasymptotic theory of independence}.
		\bpublisher{Oxford university press}.
	\end{bbook}
	\endbibitem
	
	\bibitem[\protect\citeauthoryear{Horv{\'a}th}{1993}]{horvath1993maximum}
	\begin{barticle}[author]
		\bauthor{\bsnm{Horv{\'a}th},~\bfnm{L.}\binits{L.}}
		(\byear{1993}).
		\btitle{The maximum likelihood method for testing changes in the parameters of
			normal observations}.
		\bjournal{Ann. Statist.}
		\bpages{671--680}.
	\end{barticle}
	\endbibitem
	
	\bibitem[\protect\citeauthoryear{Verzelen et~al.}{2020}]{verzelen2020optimal}
	\begin{barticle}[author]
		\bauthor{\bsnm{Verzelen},~\bfnm{N.}\binits{N.}},
		\bauthor{\bsnm{Fromont},~\bfnm{M.}\binits{M.}},
		\bauthor{\bsnm{Lerasle},~\bfnm{M.}\binits{M.}} \AND
		\bauthor{\bsnm{Reynaud-Bouret},~\bfnm{P.}\binits{P.}}
		(\byear{2023}).
		\btitle{Optimal Change-Point Detection and Localization}.
		\bjournal{Ann. Statist.}.
		\bvolume{51}		
		\bpages{1586--1610}.
	\end{barticle}
	\endbibitem
	
	\bibitem[\protect\citeauthoryear{Wang, Zou and Qiu}{2021}]{wang2021data}
	\begin{barticle}[author]
		\bauthor{\bsnm{Wang},~\bfnm{G.}\binits{G.}},
		\bauthor{\bsnm{Zou},~\bfnm{C.}\binits{C.}} \AND
		\bauthor{\bsnm{Qiu},~\bfnm{P.}\binits{P.}}
		(\byear{2021}).
		\btitle{Data-Driven Determination of the Number of Jumps in Regression Curves}.
		\bjournal{Technometrics}
		\bpages{1--11}.
	\end{barticle}
	\endbibitem
	
	\bibitem[\protect\citeauthoryear{Zhang and Chen}{2020}]{zhang2020concentration}
	\begin{barticle}[author]
		\bauthor{\bsnm{Zhang},~\bfnm{H.}\binits{H.}} \AND
		\bauthor{\bsnm{Chen},~\bfnm{S.~X.}\binits{S.~X.}}
		(\byear{2020}).
		\btitle{Concentration inequalities for statistical inference}.
		\bjournal{arXiv preprint arXiv:2011.02258}.
	\end{barticle}
	\endbibitem
	
	\bibitem[\protect\citeauthoryear{Zhang and Wei}{2021}]{zhang2021sharper}
	\begin{barticle}[author]
		\bauthor{\bsnm{Zhang},~\bfnm{H.}\binits{H.}} \AND
		\bauthor{\bsnm{Wei},~\bfnm{H.}\binits{H.}}
		(\byear{2021}).
		\btitle{Sharper Sub-Weibull Concentrations: {N}on-asymptotic {B}ai-{Y}in
			Theorem}.
		\bjournal{arXiv preprint arXiv:2102.02450}.
	\end{barticle}
	\endbibitem
	
	\bibitem[\protect\citeauthoryear{Zou, Wang and Li}{2020}]{zou2020consistent}
	\begin{barticle}[author]
		\bauthor{\bsnm{Zou},~\bfnm{C.}\binits{C.}},
		\bauthor{\bsnm{Wang},~\bfnm{G.}\binits{G.}} \AND
		\bauthor{\bsnm{Li},~\bfnm{R.}\binits{R.}}
		(\byear{2020}).
		\btitle{Consistent selection of the number of change-points via
			sample-splitting}.
		\bjournal{Ann. Statist.}
		\bvolume{48}
		\bpages{413--439}.
	\end{barticle}
	\endbibitem
	
\end{thebibliography}

\end{cbunit}

\end{document}